\documentclass[final, 3p,times]{elsarticle}

\usepackage{amssymb}
 \usepackage{amsthm}
\usepackage{listings}
\usepackage{url}
\usepackage{amsmath}
\usepackage[numbers]{natbib} 
\usepackage[autostyle]{csquotes}
\journal{Journal of Automata, Languages and Combinatorics}
 
\begin{document}

\begin{frontmatter}



\title{Union of Finitely Generated Congruences on Ground Term Algebra}


\author{S\'andor V\'agv\"olgyi}
\ead{vagvolgy@inf.u-szeged.hu}
\affiliation{{Department  of Foundations of Computer Science University of Szeged},
             addressline={\'Arp\'ad t\'er 2}, 
            city={Szeged},
            postcode={H-6720}, 
          country={Hungary}}

\begin{abstract}
We show that for any   ground term equation systems  $E$ and $F$, (1) the union of the generated congruences by $E$ and $F $
is a congruence on the ground term algebra if and only if  there exists  a ground term equation system $H$ such that the congruence generated by $H$ is equal to 
the union of the congruences generated  by $E$ and $F$ if and only if the congruence generated by the union of $E $ and $F$ is equal to the union of the congruences generated by $E $ and $F$, and (2) it is decidable  in square time  whether the congruence generated by the union of $E$ and $F$ is equal to the union of the congruences generated by $E $ and $F$, 
where the size of the input is  the number of occurrences of symbols in $E$ plus the number of occurrences of symbols in $F$.
\end{abstract}



\end{frontmatter}



%
%
\let\s=\sigma    \let\c=\circ     \let\b=\Box
\let\S=\Sigma   \let\D=\Delta
\let\th=\Theta
\def\sgt{SGT_\S (X)}
\def\ts{T_\Sigma}
\def\tsm{T_{\S,m}}
\def\les#1#2#3{#1\leq #2\leq #3}
\def\yk{Y_1\c\ldots\c Y_m}
\def\yp{Y_1\cdot\ldots\cdot Y_m}
\def\zk{W_1\c\ldots\c W_n}
\def\zp{W_1\cdot\ldots\cdot W_n}
\def\seq#1#2#3{#1_{#2},\ldots,#1_{#3}}
\def\seqq#1#2#3#4{#1_{#3}(#2_{#3}),\ldots,#1_{#4}(#2_{#4})}
\def\red#1{\mathop\rightarrow_{#1}}
\def\tred#1{\mathop\rightarrow_{#1}^{*}}
\def\thue#1{\mathop\leftrightarrow_{#1}}
\def\tthue#1{\mathop\leftrightarrow^*_{#1}}
\def\plusthue#1{\mathop\leftrightarrow_{#1}^{+}}
\newcommand{\Nat}{{\mathbb N}} 
%
%
\newtheorem{thm}{Theorem}[section]
\newtheorem{cons}[thm]{Consequence}
\newtheorem{sta}[thm]{Statement}
\newtheorem{lm}[thm]{Lemma}
\newtheorem{cla}[thm]{Claim}
\newtheorem{prop}[thm]{Proposition}
\newtheorem{lem}[thm]{Lemma}
\newtheorem{df}[thm]{Definition}
\newtheorem{alg}[thm]{Algorithm}
\newtheorem{exa}[thm]{Example}
\newtheorem{prob}[thm]{Problem}

\section{Introduction}

Kozen \cite{Kozen77}, Nelson and  Oppen \cite{jacm/NelsonO80}, and  Tarjan \cite{journals/jacm/Tarjan75} studied the word problem for ground term equation systems (GTESs)
and applied the graph-based method of congruence closure. In order to decide for a given GTES $E$ and ground terms $s$ and $t$, whether $s \tthue  E t$ holds, 
they constructed the subterm graph of $E$ and $s$ and $t$ and computed the congruential consequences of $E$ on the vertices of the graph.
Kozen  \cite{Kozen77}
showed that the  word problem for every  GTES is decidable in polynomial
time.  Nelson and  Oppen \cite{jacm/NelsonO80}
 and Tarjan \cite{journals/jacm/Tarjan75} showed that the 
 word problem for every  GTES is decidable   in $O(n \, \mathrm{log} \, n )$ time applying the congruence closure procedure on the subterm graph of the GTES $E$ and ground terms $s$ and $t$.
Continuing this research line, 
Snyder \cite{jsc/Snyder93} gave  an efficient algorithm
which, given a GTES  $E$ over a signature $\S$, produces an equivalent, reduced ground term rewriting system $R$ over $\S$ in  $O(n \, \mathrm{log} \, n )$ time. Then 
F\"ul\"op and V\'agv\"olgyi \cite{eatcs/FuloopV91} gave a simple and efficient algorithm
which, given a GTES  $E$  over a signature $\S$, produces a reduced ground term rewriting system (GTRS for short) $R$ over a signature containing $\S$ such that two terms are congruent modulo $E$ if and only if their $R$-normal forms are equal.
The algorithm calls a congruence closure procedure 
\cite{journals/jacm/Tarjan75} on the subterm graph of the GTES $E$ and runs in  $O(n \, \mathrm{log} \, n )$ time, where $n$ is the size of $E$. 
We say that a  tree language is congruential if it is the union of finitely many congruence classes of a finitely generated congruence. Brainerd \cite{Brainerd69}, Kozen  \cite{Kozen77} 
showed that congruential tree languages are the same as recognizable tree languages. Applying the above algorithm of F\"ul\"op and V\'agv\"olgyi \cite{eatcs/FuloopV91},
for a given GTES  $E$  and ground terms $p_1,\ldots ,p_k$ of size $n$, a
 deterministic bottom-up tree automaton
can be  constructed  in $O(n \,\mathrm{log} \,n)$ time which recognizes the
congruential tree language $p_1/_{\tthue
E}\cup\cdots\cup p_k/_{\tthue E}$  \cite{tcs/Vagvolgyi93}.

V\'agv\"olgyi \cite{tcs/Vagvolgyi03b}
showed that it is decidable for any given 
GTESs  $E$ and $F$   whether
there is a GTES $H$ such that 
$\tthue H=\tthue E\cap \tthue F$. If the answer is \enquote{yes}, then we can construct such a  GTES  $H$. 
We show the dual to the above result, first we
show that 
for any GTESs $E$ and $F$, $\tthue E\cup \tthue F$
 is a congruence on the ground term algebra
if and only if 
 there exists  a GTES $H$ such that $\tthue H=\tthue E\cup \tthue F$
  if and only if $\tthue {E\cup F}=\tthue E\cup \tthue F$.
  Second, we show that it is decidable  in  $O(n^2)$ 
 time for any given 
  GTESs  $E$ and $F$ whether $\tthue {E\cup F}=\tthue E\cup \tthue F$, where 
  $n$ is   the number of occurrences of symbols in $E$ plus the number of occurrences of symbols in $ F$.
  For each  signature $\Sigma$ without  constants,  the set of ground terms over $\Sigma$ is the empty set, each GTES over $\Sigma $ is the empty set, and our results above are obvious. Thus, throughout the paper we consider signatures with  constants. 

We now recall the fast ground completion algorithm of F\"ul\"op and V\'agv\"olgyi \cite{eatcs/FuloopV91}.
Let 
$E$ and $F$ be GTESs over $\Sigma$, and let $W$ denote the set of subterms appearing in the GTESs $E$ and $F$.
F\"ul\"op and V\'agv\"olgyi \cite{eatcs/FuloopV91}
presented a
fast ground completion algorithm which, given GTESs  $E$ and $F$  over a  signature   $\S$,
produces in $O(n \,\mathrm{log}\, n)$ time a finite set $C\langle E, F\rangle$ of  new  constants and 
a reduced GTRS $ R \langle E, F \rangle  $
over $\S$ extended with the new constants
such  that $\tthue  E = \tthue {R\langle E, F \rangle }  \cap \ts \times \ts $.
 Here $n$ stands for the number of occurrences of symbols in $E$ plus the number of occurrences of symbols in  $F$.
Furthermore, the  set $C\langle E, F\rangle$ of new constants is the set of 
the  equivalence classes of
the restriction   $\Theta\langle E, F \rangle  $
of $\tthue E  $ to the set $W$. For this purpose we proceed as follows:
In the first step  in $O(n)$
time, we create a subterm directed, labeled,  and acyclic multigraph (dag) $G$ for $E$ and $F$  \cite{jsc/Snyder93}. We consider the vertices in the  subterm dag  $G$   as terms in $W$.
 The relation $\tau \langle E, F \rangle $ on the vertices of $G$  corresponds to $E$. We construct the  relation $\tau \langle E, F \rangle $.
 Then 
we run on $G$  a  congruence closure algorithm for the relation $\tau \langle E, F \rangle $ in $O(n\,\mathrm{log}\, n)$ time  \cite{jacm/DowneyST80}. In this way, we obtain   on
$G$ the congruence closure $\rho \langle E, F \rangle $ of the relation $\tau \langle E, F \rangle $. Considering the vertices in the  subterm dag  $G$   as terms, 
$\rho \langle E, F \rangle $ yields 
       the restriction $\Theta\langle E, F \rangle$ of $\tthue E $ to $W$. Then using $\rho \langle E, F \rangle $, 
     $\Theta\langle E, F \rangle$,  and $C\langle E, F\rangle$, we 
produce the GTRS $ R \langle E, F \rangle $ in $O(n)$  time. 
GTRS $ R \langle E, F \rangle $ is the set of all 
rewrite rules 
 $\s({b}_1, \ldots, {b}_m) \dot{\rightarrow}{b}$, where 
 $\s\in \S_m$, $m\in \Nat$, and   $ {b}_1, \ldots, {b}_m, {b} \in C\langle E, F\rangle$, and $\{\, \s(t_1, \ldots, t_m) \mid t_1\in {b}_1, \ldots, t_m\in {b}_m\, \}\subseteq{b}$.
 F\"ul\"op and V\'agv\"olgyi \cite{eatcs/FuloopV91} showed that GTRS $ R \langle E, F \rangle $ is reduced, and hence is convergent.
We construct  GTRSs $ R \langle  F, E  \rangle $ and $R\langle E\cup F, \emptyset \rangle$ in a similar way.

The intuitive  idea of the paper is as follows. Let 
$E$ and $F$ be GTESs over $\Sigma$. For every  subterm $t$ appearing in $E \cup F$,  
GTRS $R\langle E, F \rangle$  reduces all elements  of the congruence  class $t/_{\tthue {\langle E,   F \rangle}}$   
    to the constant  (equivalence class)  $t/_{\Theta\langle E, F \rangle}$,    
and 
GTRS $R\langle E\cup F, \emptyset \rangle$ reduces every  element  of the congruence  class $t/_{\tthue {\langle E\cup F, \emptyset \rangle}}$ 
 to the  constant (equivalence class)   $t/_{\Theta\langle E\cup  F , \emptyset \rangle}$.    
For every subterm $t$  appearing in $E \cup F$, the   congruence class   $t/_{{\tthue {\langle E\cup F, \emptyset \rangle}}}$ is equal to the congruence class  $t/_{{\tthue {\langle E,  F \rangle}}}$  
if and only if $t/_{\Theta\langle E\cup  F , \emptyset \rangle} =t/_{\Theta\langle E, F \rangle}$ and 
for every  $p \in \ts$ and reduction sequence ${\cal RS}1$  from  $p$ to   $t/_{\Theta\langle E\cup F, \emptyset  \rangle}$   by
$R\langle {E \cup F},\emptyset\rangle$,  
there exists a reduction sequence ${\cal RS}2$  from $p$ to   $t/_{\Theta\langle E, F \rangle}$   by
$R\langle E, F \rangle $, 
where
each  rewrite rule 
$\s({b}_1, \ldots, {b}_m) \dot{\rightarrow} {b}\in  R\langle {E \cup F},\emptyset\rangle$  applied  at any node of $p$  along the  reduction sequence ${\cal RS}1$ is matched with
 some rewrite rule 
 $\s({c}_1, \ldots, {c}_m ) \dot{\rightarrow}  {c}$ in $ R\langle E, F\rangle$
applied by GTRS $R\langle E, F \rangle$ at the same node along the  reduction sequence ${\cal RS}2$. By matching we mean that 
   ${c}_i \subseteq  {b}_i$
 for all $i\in \{\, 1, \ldots, m\, \}$,  
and 
  ${c}\subseteq{b}$. 
 For every ${b}\in C\langle E\cup F, \emptyset \rangle$,
we say that  GTRS $R\langle E, F \rangle$  keeps up with the  GTRS 
$R\langle E\cup F, \emptyset \rangle$ writing the constant     ${b}$ 
if the following holds. 
If  GTRS 
$R\langle E\cup F, \emptyset \rangle$  applies a  rewrite rule  of the form 
 $\s({b}_1, \ldots, {b}_m) \dot{\rightarrow} {b}$ along a reductions sequence from $p$ to   $t/_{\Theta\langle E, F \rangle}$, then  
$ R\langle E, F\rangle$ also applies a rewrite rule $\s({c}_1, \ldots, {c}_m ) \dot{\rightarrow}  {c}$ for some 
   ${c}_1, \ldots, {c}_m, c \in C\langle E, F\rangle$ at the same node,  that is,  $ R\langle E, F\rangle$ also move forward along  its reduction sequence.
For every ${b}\in C\langle E\cup F, \emptyset \rangle$,
we define  that  GTRS $R\langle F, E \rangle$  keeps up with the  GTRS 
$R\langle E\cup F, \emptyset \rangle$ writing the constant     ${b}$ 
in a similar way. 

  For each   constant 
${d}\in C\langle E\cup F, \emptyset \rangle$, we keep track of all constants ${c}\in C\langle E\cup F, \emptyset\rangle$ which appear as a right-hand side of a rewrite  rule applied along 
 the reduction sequence from  some ground term $t\in \ts$  to the constant 
${d}$    by $R\langle E\cup F, \emptyset \rangle$.
To this end, we construct a directed pseudograph without parallel arcs (dpwpa), $AUX[E; F]$ with vertices $C\langle E\cup F, \emptyset \rangle$. At each vertex ${b}\in C\langle E\cup F, \emptyset \rangle$, we indicate whether it is an element of $C\langle E, F\rangle$ whether it is an 
element of 
$C\langle F, E\rangle$, whether GTRS $R\langle E, F \rangle$ keeps up with GTRS $R\langle E\cup F, \emptyset \rangle$ writing the constant     ${b}$, and whether  GTRS $R\langle F, E \rangle$ keeps up with GTRS $R\langle E\cup F, \emptyset \rangle$ writing the constant     ${b}$. Every directed arc $({b}, {c})$ is associated with some rewrite rule
 $ \s({b}_1, \ldots, {b}_m) \dot{\rightarrow}{b}$ in $R\langle E\cup F, \emptyset \rangle$, where    ${c}={b}_i$ for some $i \in \{\, 1, \ldots, m\, \}$. That is, ${b}$ is the right-hand side of the
 rewrite  rule, and ${c}$ appears on the left-hand side of the rewrite rule.
We say that GTRS 
 $R\langle E, F \rangle$
simulates the  GTRS $R\langle E\cup F, \emptyset \rangle$ on a constant    ${d}\in C\langle  E \cup F, \emptyset \rangle $ if  ${d}\in C\langle  E, F \rangle $ and  
 for all constants ${c}$ appearing in a reduction sequence from a  ground term  $p \in \ts$ to ${d}$ by GTRS  $R\langle E\cup F, \emptyset \rangle$, GTRS
$R\langle E, F \rangle$ keeps up with GTRS $R\langle E\cup F, \emptyset \rangle$ writing $c$.
 That is, for all constants ${c}$ with a path of positive length or a cycle     from  
   ${c}$ to    ${d}$ in  $AUX[E; F]$,  
$R\langle E, F \rangle$ keeps up with GTRS $R\langle E\cup F, \emptyset \rangle$ writing $c$.
We define that GTRS 
 $R\langle  F, E \rangle$
simulates the  GTRS $R\langle E\cup F, \emptyset \rangle$ on a constant    ${d}\in C\langle  E \cup F, \emptyset \rangle $ in a similar way.
During the run of  the decision algorithm, we carry out  a partial depth first search on the auxiliary dpwpa $AUX[E; F]$ to decide whether  $R\langle E, F \rangle$
simulates the  GTRS $R\langle E\cup F, \emptyset \rangle$ on a constant    ${d}\in C\langle  E \cup F, \emptyset \rangle $, and 
to decide whether  $R\langle F, E \rangle$
simulates the  GTRS $R\langle E\cup F, \emptyset \rangle$ on a constant    ${d}\in C\langle  E \cup F, \emptyset \rangle $.
 We say that GTRS $R\langle E, F \rangle$ and GTRS $R\langle F, E \rangle$
together  simulate the  GTRS $R\langle E \cup F, \emptyset \rangle$
 if 
   for each constant    ${d}\in C\langle  E \cup F, \emptyset \rangle $, GTRS  $R\langle E, F \rangle$  or GTRS  $R\langle F, E \rangle$
    simulate  the GTRS  $R\langle E \cup F, \emptyset \rangle$ on     ${d} $.   We show that   
for every subterm $t$  appearing in $E \cup F$,
  $\left( t/_{\tthue{ E\cup F} }=  t/_{{\tthue E}} \mbox{ or } t/_{\tthue{ E\cup F} }=  t/_{{\tthue F}} 
  \right) $
  if and only if 
  GTRS $R\langle E, F \rangle$ and GTRS $R\langle F, E \rangle$
together  simulate the  GTRS $R\langle E \cup F, \emptyset \rangle$.
We consider the following four main cases.

Main Case 1: $\Sigma $ is a unary signature. Then $\tthue {E\cup F}= \tthue E \cup \tthue F$ if and only if GTRS $R\langle E, F \rangle$ and GTRS $R\langle F, E \rangle$
together  simulate the  GTRS $R\langle E \cup F, \emptyset \rangle$. 

Main Case 2:  Both GTRS $R\langle E, F\rangle$ and GTRS $R\langle F, E\rangle$ are  total. 
Then  $\tthue {E\cup F}=\tthue E\cup \tthue F$
if and only if each $\Theta\langle E\cup F, \emptyset \rangle$ equivalence class is equal to a 
 $\Theta\langle E, F \rangle$ equivalence class or a $\Theta\langle F, E \rangle$  equivalence class.

Main Case 3:  $\S_k\neq \emptyset$ for some $k\geq 2$ and  $R\langle E, F\rangle$ 
or  $R\langle F, E\rangle$ is total. Then $\tthue {E\cup F}= \tthue E \cup \tthue F$ if and only if GTRS $R\langle E, F \rangle$ and GTRS $R\langle F, E \rangle$
together  simulate the  GTRS $R\langle E \cup F, \emptyset \rangle$. 

Main Case 4:  $\S$ has a symbol of arity at least $2$,
  and GTRS  $ R \langle E, F \rangle$ and GTRS $ R \langle F, E \rangle$  are not total.   
  Then $\tthue {E\cup F}=\tthue E\cup \tthue F$
 if and only if   $E \subseteq \tthue F$ or
 $ F \subseteq \tthue E $. 



 In Main Case 1, the  overall complexity of our  computation is   $O(n^2)$.
 In each of Main Cases 2, 3, and 4,  the  overall complexity of our  computation is   $O(n\, \mathrm{log}\, n)$.
  Main  Case 2 overlaps with Main Case 1 and Main Case 3, and is a subcase of the union of 
 Main Case 1 and Main Case 3.  
 
 We now describe an  alternative approach for our decision problem.
 We can think of the  GTRSs $ R \langle E, F \rangle$, $ R \langle F, E \rangle$, and $ R \langle E \cup F, \emptyset \rangle$ 
  as the sets of rewrite rules of bottom-up tree automata. Thence we can also apply the decision algorithm of   Champav{\`{e}}re  et al \cite{franciak}. Their 
algorithm decides  inclusion between two tree languages recognized by given bottom-up 
  tree automata. Applying their algorithm  to our decision problem in  each of Main Cases 1, 2, and 3  yields an alternative decision  algorithm, which
    is more sophisticated than ours, and in general takes more time than ours,  it runs in  $O(n^3)$ time.  
   
   For any given 
  GTESs  $E$ and $F$, we can construct the  GTES $E \cup F$ in  $O(n \, \mathrm{log} \, n)$ time, where 
  $n$ is  as above. 
  
The rest of this paper is organized as follows: 
In Section \ref{prel}, we recall some basic definitions and fix the notations. In Section 
\ref{runningpelda},  we start our running examples. 
We consider in each main case an  example of GTESs $E$ and $F$ such that
$\tthue {E\cup F}=\tthue E\cup \tthue F$ and an  example
for GTESs $E$ and $F$ such that
$\tthue E\cup \tthue F\subset \tthue {E\cup F}$.
In Section \ref{gyors}, we recall the  fast ground completion algorithm of 
F\"ul\"op and V\'agv\"olgyi \cite{eatcs/FuloopV91}.
In Section \ref{peldak}, we continue  our running examples.
  In Section \ref{three-way}, we
 define  the  auxiliary dpwpa $AUX[ E;  F] =
(C\langle E\cup F, \emptyset \rangle, A[E; F])$ 
for the GTESs $E$ and $F$.
Then  we compute  $AUX[ E;  F] $ for our running  examples. 
In Section \ref{folytat}, we further continue  our running examples. In Section \ref{epit}, 
we construct the auxiliary  dpwpa  for the GTESs $E$ and $F$.
In Section \ref{preparatory}, we show some preparatory results. In Section \ref{jointly}, we define the notion
that GTRS $R\langle E, F \rangle$ and GTRS $R\langle F, E \rangle$
together  simulate the  GTRS $R\langle E \cup F, \emptyset \rangle$.  In Section \ref{esetek}, we 
define the four main cases of our decision algorithm.  We show that we can  decide which Main Cases hold in   $O(n \, \mathrm{log} \, n )$   time.
 In Sections \ref{tavasz}, \ref{nyar}, \ref{osz}, and \ref{tel}, we study Main Cases 1, 2, 3, and 4, respectively
one by one.
For each main case,  we present our decidability and complexity results, and we continue our relevant running examples.
 In Section \ref{algoritmus},  we  sum up our decidability and complexity results.
In Section \ref{tanulsag} we present our concluding remarks.

 \section{Preliminaries}\label{prel}

In this section, we present a brief review of the notions, notation and
preliminary results used in the paper.  

\paragraph{Algorithm concepts} In algorithmic analysis we use  the RAM (random access machine) model \cite[Section~2.2]{Sanders},
\cite[Section~2.1]{skiena20}. 
 For this model of computation,
each simple operation $(+,\, *,\, –, \, =,\,  if, \, call)$ takes exactly one time step.
An array is a collection of items of same data type stored at contiguous memory locations. 
For the concept of a red-black tree, see {\rm \cite[Section~13]{cormen}}.
\begin{prop}\label{keres} {\rm \cite[Section~13]{cormen}}
The search operation takes $O(\mathrm{log}\, n)$ time on a red-black tree, where $n$ is the  size of the tree.
\end{prop}
\begin{prop}\label{redblack} {\rm \cite[Section~13]{cormen}}
The time-complexity of inserting $n\in \Nat$  elements into an empty  red-black tree, not allowing  duplicate values in the resulting red-black tree, is $O(n \,\mathrm{log}\, n)$. 
\end{prop}
\begin{proof}
The time-complexity of inserting a single element into a red-black tree is $O(\mathrm{log}\, n)$, 
where $n$ is the  size of the tree \cite{cormen}. Consequently it  takes  $O(n \,\mathrm{log}\, n)$ time to insert $n\in \Nat$  elements into an empty  red-black tree.
We modify the insertion routine of the red-black tree to not allow duplicate values. The insertion function takes the same amount of time i.e. $O(\mathrm{log}\, n)$,
 as before,  and  does  not insert a  value if it is already in the tree. 
\end{proof}

\paragraph{Basic notions} 
The set of nonnegative integers is denoted by $\Nat$, and  $\Nat^*$ stands
for the free monoid generated by  $\Nat$ with empty word
 $\varepsilon$ as identity element. The  length of a word $\alpha\in \Nat^*$ is defined as  the number of integers it contains and is 
 denoted by
$\mathrm{length}(\alpha)$.
A word $\alpha\in \Nat^*$ is a prefix of a  word $\beta \in \Nat^*$ if there is  a word $\gamma\in \Nat^*$ 
such that $\alpha\gamma=\beta$. Moreover, if $\gamma \neq \varepsilon$, then  $\alpha$ is a proper prefix of $\beta$. 

  A binary sequence is a string of ones and zeros,  $BIN$ stands for the set of binary sequences. 
For an integer  $\ell \in \Nat$,  we write $BIN_\ell$
to denote the set of binary sequences
of length $\ell$. We define the binary sequence $0^k$, $k\in \Nat$ by recursion on $k$. Let $0^0= \varepsilon$. For each $k\in \Nat$, let
 $0^{k+1}= 0 0^k$. 
Let $0\leq k \leq \ell$ and $\alpha \in BIN_k$. We pad  the binary sequence $\alpha$ with leading zeros to length  $\ell$ by
producing  the binary sequence $0^{\ell-k} \alpha$.

We now recall the definition of the lexicographic order on  $BIN$ \cite [Definition 6.4.1.]{johnsonbaugh}. 
Given two distinct binary strings, to determine whether one precedes the other in the  lexicographic order
we compare the numbers in the words. There are two possibilities:
\begin{enumerate}\rm 
\item The words have different lengths, and each number in the shorter word
is identical to the corresponding number in the longer word.
\item The words have the same or different lengths, and at some
position, the numbers in the words differ. 
\end{enumerate}
If Condition 1 holds, then the shorter word precedes the longer. If Condition 2 holds, we locate the leftmost position $\ell$ at which the
numbers differ. The order of the words is determined by the order of the numbers at position $\ell$. 
\begin{df}[{\cite[Definition 6.4.1.]{johnsonbaugh}}] \rm
Let  $\alpha= i_1i_2 \ldots i_m$ and $\beta=j_1 j_2 \ldots j_n$ be binary strings. We say that $\alpha$ is lexicographically less than $\beta$ and write
$\alpha <_{lex} \beta$ if 

\begin{enumerate}
\item 
$m< n$ and $i_k = j_k $ for every $k=1, 2, \ldots, m$
or 
 \item 
$i_k\neq j_k$ for some $k$  and  for the smallest such $k$, we have $i_k < j_k$. 
\end{enumerate}
\end{df}

\paragraph{Sets and Relations {\rm \cite{baader1999term}}}
For any sets $A$ and $B$, 
$A \subseteq B$ means that $A$ is a subset of $B$. Proper inclusion is denoted by $A \subset B$. The union of two sets $A $ and $B$ is denoted by $A\cup B$. 
 The cardinality of a set $ A$ is  denoted $|A|$.
A relation on a set $A$ is a subset
   ${\rightarrow}$ of $A\times A$ \cite{baader1999term}.
We write    ${b}{\rightarrow} {c}$ for $({b},{c})\in \,{\rightarrow}$.
We denote by    ${\rightarrow}^*$ the reflexive
transitive closure of    ${\rightarrow}$,
by $\leftrightarrow^+$ the  transitive symmetric closure of    ${\rightarrow}$,
 and by $\leftrightarrow^*$ the reflexive
transitive  symmetric closure of    ${\rightarrow}$. 
 Note that $\leftrightarrow^*$ is an
equivalence relation on $A$.

Let    ${\rightarrow}$ be a relation on a set $A$.
An element    ${b}\in A$ is in normal form  (irreducible for     ${\rightarrow}$) if there
 exists no    ${c}\in A$ such that
   ${b}{\rightarrow} {c}$.
 We say that 
    ${c}\in A $ is a normal form of    ${b}\in A$ if    ${b}{\rightarrow}^* {c}$ and    ${c}$ is in normal form.

A relation    ${\rightarrow}$ on a set $A$ is
\begin{itemize}
\item normalizing if every element has a normal form, 
\item  terminating if there exists no infinite sequence of
elements    ${b}_1, {b}_2, {b}_3, \ldots$ in $A$
such that    ${b}_1{\rightarrow} {b}_2{\rightarrow} {b}_3{\rightarrow}\ldots$,
\item confluent if for any elements    ${b}_1, {b}_2, {b}_3$ in $A$,
whenever    ${b}_1 {\rightarrow}^*  {b}_2$ and     ${b}_1 {\rightarrow}^*  {b}_3$, there exists an
element    ${b}_4$ in $A$ such that    ${b}_2{\rightarrow}^* {b}_4 $ and    ${b}_3{\rightarrow}^* {b}_4 $,
\item convergent if it is  terminating and confluent.

\end{itemize}
Any terminating relation is normalizing. Hence we have the following proposition.
 \begin{prop}[{\cite[Lemma 2.1.8]{baader1999term}}]\label{nulladik}
 For any convergent relation     ${\rightarrow}$ on a set $A$, every
    ${b}\in A$ has a unique normal form, called the normal form of    ${b}$. 
\end{prop}
 If    ${b}\in A$ has a unique normalform, then we denote the latter
 by    ${b}{\downarrow}$.

Let $\rho$ be an
equivalence relation on $A$.  Then for
every    ${b}\in A$, we denote by    ${b}/_\rho$
the $\rho$ equivalence class containing    ${b}$,
i.e. $ {b}/_\rho=\{\, {c}\mid {b}\rho {c}\,\}$.
For each $B\subseteq A$, let $B/_\rho=\{\, b/_\rho\mid b\in B\, \}$. 

The empty product $\Pi_{i=1} ^ 0 n_i=1$ over $\Nat$ 
is the result of multiplying no factors.
The ceiling function maps a real number $x$ to the least integer greater than or equal to $x$, denoted $\lceil x \rceil $. 

\paragraph{Ground terms  {\rm \cite{baader1999term}}} A  signature   $\S$ is a finite set of symbols in
which every element $\s$ is associated with a non-negative integer, the arity of $\s$ denoted by  $\mathrm{arity}(\s)$. 
For each integer $k\in \Nat$, $\S_k$
denotes the elements of $\S$ which have arity $k$. We lose no generality assuming  that $\S_0\not=\emptyset$, see the Introduction.
 If $\Sigma=\Sigma_0 \cup \Sigma_1$, then $\Sigma$ is a unary  signature. 
$mar_{\Sigma} $  denotes the maximal arity of the symbols in $\Sigma$. 

We denote by $T_\S$ the set of ground terms over $\S$. It is the smallest set
$U$ for which
\begin{itemize}
\item[(i)] $\S_0\subseteq U$ and
\item[(ii)] $\s(t_1,\ldots ,t_m)\in U$ whenever $\s\in \S_m$ with $m\geq 1$ and $t_1,\ldots ,t_m\in U$.
\end{itemize}
Ground terms are also called trees. A set of ground terms is also called a tree language. 
For a ground  term $t\in \ts$, the set $sub(t)$ of subterms
of $t$ is defined by recursion as follows:
\begin{itemize}
\item[(i)] if $t\in \S_0$, then
$sub(t)=\{\,t\,\}$,
\item[(ii)]if $t=\s(t_1,\ldots,t_m)$
for some $m\geq 1$, $\s\in \S_m$, and $t_1,\ldots ,t_m\in\ts$,
then we have $sub(t)=
\bigcup\{\, sub(t_i)\mid 1\leq i\leq m \, \}\cup \{\,t\,\}$.
\end{itemize}
For a finite set $Z  \subseteq \ts$ of ground terms,
\begin{quote} $sub(Z )=
\bigcup\{\, sub(t)\mid t \in Z   \, \}$.
\end{quote} 
For a ground term $t\in \ts$, the height of $t$  is denoted by
$height(t)$ and is defined by recursion:
\begin{itemize}
\item[(i)] if $t\in \S_0$, then
$height(t)=0$, and
\item[(ii)] if $t=\s(\seq t1m)$ with $m\geq 1$ and $\s \in \S_m$, then
$height(t)=max \{\,height(t_i)\mid 1\leq i\leq m\, \}+1$.
\end{itemize}
For a ground term $t\in\ts$, the set of positions $\mathrm{POS}(t)\subseteq \Nat^*$ 
of $t$ is defined by tree induction.
\begin{itemize}
\item If $t\in\S_0,$ then  $\mathrm{POS}(t)=\{\, \varepsilon\, \}$. 
\item If $t=\s (t_1,\ldots, t_m)$ with $\s\in\S_m$, $ m \geq 1$, then
$\mathrm{POS}(t)=\{\, \varepsilon\, \}\cup \{\, i \alpha \mid 1\leq i \leq m, \alpha\in \mathrm{POS}(t_i) \, \}$.
\end{itemize}
The size of  a ground term $t\in\ts$ is  $\mathrm{size}(t)=|POS(t)|$.
For a finite set $Z  \subseteq \ts$ of ground terms, 
$\mathrm{size}(Z )$ is the number of occurrences of symbols of $\S$ in $Z$, i.e., 
$\mathrm{size}(Z )=\sum _{t \in Z} \mathrm{size}(t)$.
For any  ground term $t\in T_\S$ and position $\alpha\in \mathrm{POS}(t)$, we introduce the subterm 
$t|_\alpha\in \ts$ of $t$ at $\alpha$ 
as follows:
\begin{itemize}
\item for $t\in \S_0$, $t|_\varepsilon=t$;

\item for $t=\s(t_1, \ldots, t_m)$ with $m\geq 1$ and $\s\in \S_m$ if 
$\alpha =\varepsilon$ then $t|_\alpha=t$, 
 otherwise if
 $\alpha =i\beta$ with $1\leq i\leq m$, then $t|_\alpha=t_i|_\beta$. 
\end{itemize}

A context $\delta$ over $\S$ can be viewed as an incomplete term, which may contain several empty places  or holes. 
Formally, a context over $\S$  is a ground term containing zero, one or more occurrences of a special constant symbol 
$ \diamondsuit\not \in \Sigma$
denoting holes, i.e., a ground term over the extended signature $\S\cup 
\{\, \diamondsuit\, \}$. If $\delta$ contains exactly $n\in \Nat$ holes, then we say that $\delta$ is an $n$-context.
The set of all  $n$-contexts over $\S$ 
is denoted by $CON_{\S, n}$. 
For example, let $\S=\S_0\cup \S_2$, $\S_0=\{\, \#, \$\, \}$, and $\S_2=\{\, f, g\,  \}$.
Then 
$g(f(\diamondsuit, \#), g(\diamondsuit, \diamondsuit)) \in CON_{\S, 3}$.
Let $\delta$ be an $n$-context over $\S$, and let 
$t_1, \ldots, t_n\in \ts$ be ground  terms. Then  $\delta[t_1, \ldots, t_n]$ denotes the result of replacing the holes of $\delta$ from left to right by $t_1, \ldots, t_n$.
For example, 
\begin{quote}
$g(f(\diamondsuit, \#), f(\diamondsuit, \diamondsuit))[g(\$, \$), \#, \$]=
g(f(g(\$, \$), \#), f(\#, \$))$.
\end{quote}
Let $\delta$ be a $1$-context  over $\S$, and let $\alpha \in POS(\delta)$
be such that 
$t|_\alpha= \diamondsuit$. Then we call $\alpha$ the address of $\diamondsuit$ in $\delta$, and    $addr(\delta)$ stands for $\alpha$.
Let $\delta$ be a  $1$-context over $\S$ such that $\delta\neq \diamondsuit$, and  let    ${d}\in \S_0$ be a constant. 
Then we say that the ground term  $\delta[{d}]$ is   a proper extension  of the constant    ${d}\in \S_0$.

\paragraph{Ground term rewrite systems and equation systems {\rm\cite{jsc/Snyder93}}} Let $\S$ be a  signature, and $R$ be 
a finite subset of  $\ts\times \ts$. 
The size  of  $R$  is denoted by $\mathrm{size}(R)$,  and is the number of
occurrences of symbols of $\S$ in $R$, ie.,   $\mathrm{size}(R)=\sum_{(l, r) \in R} \mathrm{size}(l)+\mathrm{size}(r)$.
 Let
 $\Sigma\langle R\rangle $ denote the signature of  those symbols which occur in  $R$.
 For any ground  terms $s, t\in \ts$, we say that
$s$ rewrites to $t$, and denote this  by $s\red R t$ if
there exist an ordered pair $(l,r)$ in $R$, a context $\delta \in CON_{\Sigma, 1}$ such that
$s=\delta [l]$ and $t=\delta [r]$.
We fully specify such a rewrite step  by
$s \rightarrow _{[\alpha, (l, r)]}t$, where  $\alpha= addr(\delta)$. 
Furthermore,  
$s \leftrightarrow_{[\alpha, (l, r)]}t$  stands for 
($s \rightarrow _{[\alpha, (l, r)]}t$ or $t \rightarrow _{[\alpha, (l, r)]}s$).
When we focus on the relations $\red R$ and $\tred R$, we call $R$  a ground term rewrite system (GTRS for short), and
we call an ordered pair $(l, r)$ a rewrite rule, and denote it as $l\dot{\rightarrow} r$.
Moreover, we say that $l$ is the
left-hand side and $r$ is the right-hand side of the rewrite rule
$l\dot{\rightarrow} r$.
When we focus on the symmetric relations $\thue R$ and $\tthue R$,
we call $R$  a ground term equation system (GTES for short), and
we denote $(l, r)$ as $l \doteq r$, and call it an equation. Moreover, we denote a GTES  as $E$ rather than $R$  \cite{jsc/Snyder93}.
Let $E$  and $F$ be  GTESs over  a   signature   $\Sigma$. Then we define the  GTES $E \cup F$ over $\S$ as the union of the sets $E$ and $F$. 
\begin{sta}\label{erre}
Let $E$  and $F$ be  GTESs over  a   signature   $\Sigma$.
 We can construct  the  GTES $E \cup F$ over $\S$ in $O(n \, \mathrm{log} \, n)$  time.
\end{sta}
\begin{proof}
In   
 $O(n \, \mathrm{log} \, n)$ time we mergesort the  GTES $E $ and the GTES $F$ \cite[Section~2.2]{books/daglib/0029345}. Then in $O(n)$
time,  we concurrently scan through both of the  sorted relations to produce the result. When the concurrent scan of both GTESs reveals the same equation in both relations, only one of the equations is retained. 
\end{proof}

\begin{df} \rm
Let $E$  and $F$ be  GTESs over  a   signature   $\Sigma$.
 Then   let $\Sigma\langle E, F\rangle = \Sigma\langle E \rangle \cup \Sigma\langle F \rangle $. That is,  $\Sigma\langle E, F\rangle $ denotes the signature of  those symbols which occur in  $E$ or $F$.
 The  Boolean value $FLAG\langle E, F\rangle $  indicates  whether 
 $\S=\Sigma\langle E, F \rangle $. \newline
\end{df}
\begin{sta} \label{koszonom}
Let $E$  and $F$ be  GTESs over  a   signature   $\Sigma$, and  $n=\mathrm{size}(E)+\mathrm{size}(F)$. 
We can count the elements of $\S\langle E, F \rangle$  and compute the  Boolean value $FLAG\langle E, F\rangle $ in $O(n)$ time.
\end{sta}

Let $ E$  be a GTES  over $\S$. Then 
\begin{quote}
$ST\langle E \rangle=\{\, t\in \ts\mid t \in sub(u) \mbox{ or } t \in sub(v)  \mbox{ for some  }
u\doteq v \mbox{ in } E\, \}$
\end{quote} is the set of subterms occurring $E$.
The
below claim is immediate. 
\begin{sta}\label{meret}
1. Let $Z \subseteq \ts$ be a finite set  of ground terms over a signature $\S$. Then $|sub(Z)| \leq \mathrm{size}(Z)$. \newline
2. Let $ E$  be a GTES  over a signature $\S$. Then 
 $|ST\langle E \rangle |\leq \mathrm{size}(E)$.
 \end{sta}

Let $\S$ be a signature. 
A $\S$-algebra    ${\cal A} $ consists of a carrier set $A$, and a mapping that associates with each symbol $\s\in \Sigma_k$ a mapping 
$\s^{\cal A} : A^k\rightarrow A$ for every $k\in \Nat$ \cite[Definition 3.2.1]{baader1999term}.
The $\S$-ground-term algebra    ${\cal T}(\Sigma)$ has $\ts$ as carrier set, and interprets the symbols $\s\in \S_k$ for every  $k\in \Nat$ as follows:
\begin{quote}
$\s^{{\cal T}(\Sigma) }: \ts^k  \rightarrow \ts : (t_1, \ldots, t_k) \mapsto \s(t_1, \ldots, t_k)$.
\end{quote} 
Let $ E$  be a GTES  over $\S$.
The congruence relation generated by $E$ on the $\Sigma$-ground term algebra      ${\cal T}(\Sigma)$ is $\tthue E$. It is well known that $\tthue E$ is the least congruence on      ${\cal T}(\Sigma)$
containing $E$.
Let $R$ be a GTRS  over $\S$.
We say that GTRS $R$ is equivalent to the GTES $E$ if $\tthue R=\tthue E$. 
We say that  GTRS $R$  reaches a ground term $s\in \ts$ from a ground term $t \in \ts$ if $t \tred R s$ holds.
We say that  GTRS $R$ is  terminating, (confluent, etc.) if $\red R$ is
 terminating (confluent, etc.). A term $t\in \ts$ is irreducible for $R$ if it is irreducible  for $\red  R$.
If a term $t\in \ts$ has a unique $\red R$ normal form, then the latter is
denoted by $t{\downarrow}_R$.
A GTRS $R$ is reduced if for every rewrite rule $l\dot{\rightarrow} r$ in $R$,
$l$ is  irreducible for $R-\{\, l\dot{\rightarrow} r\,\}$
and $r$ is irreducible for  $R$.
We recall the following important result.
\begin{prop} [{\cite[Theorem 2.14]{jsc/Snyder93}}]\label{mozart} Every
  reduced GTRS $R$ is convergent.
\end{prop}

\begin{df}\label{alfa}\rm Let $E$ and $F$ be GTESs over a
 signature   $\S$, and   $W= ST\langle E \cup F\rangle$.
 Then let 
$\Theta\langle E,F \rangle=\tthue E \cap W\times W$.
\end{df}
It should be clear that
$\Theta\langle E, F \rangle$ is an equivalence relation on $W$. 
For each $t\in W$, we have $ t/_{\Theta\langle E, F \rangle}= t/_{\tthue E}\cap  W\times W$.

\paragraph{Bottom-up tree automata {\rm \cite{Engelfriet75}}, {\rm \cite{gs}}} 
Let $\S$ be a  signature, a bottom-up tree automaton (bta for short)  over 
$\S$  is a quadruple    ${\cal A}=(\S,A,R, A_f)$,
 where  $A$ is a finite set of states of arity $0$, 
$\S \cap A=\emptyset$, 
$A_f(\subseteq A)$
is the set of final states, 
$R$ is a finite set of rewrite rules of the following
form:
\begin{quote}
$ \s({b}_1,\ldots, {b}_k) \dot{\rightarrow}{b}$ with $k\in \Nat$, $\s\in \S_k$,
$\seq {b}1k, {b}\in A$.
\end{quote}
We consider $R$ as a GTRS  over $\S\cup A$.
 The size  of $\cal A$  is denoted by $\mathrm{size}({\cal A})$, and is $|A|+\mathrm{size}(R)$. 
The tree 
language recognized by $\cal A$ is $L({\cal A})=\{\,
t\in \ts\mid (\exists {b}\in A_f) \;t\tred {R} {b}\}$.
 We say that     ${\cal A}=(\S, A,  R, A_f)$ is a deterministic bta if 
for any $\s\in \S_k$, $k \in \Nat$,    ${b}_1, \ldots, {b}_k\in A$, there is
at most  one
rule with left-hand side $\s({b}_1, \ldots, {b}_k)$ in $R$.

\begin{prop}\label{mezo}{\em \cite{tcs/Vagvolgyi93}}
Let $E$ and $F$  be  GTESs over  a   signature   $\Sigma$,  and let 
 $W=ST \langle E\cup F \rangle$.
   For every  $p \in W$,  
   we can construct in  $O(n \, \mathrm{log} \, n )$   time the deterministic bottom-up tree automaton
       ${\cal A}=(\S, C\langle E, F \rangle,  R\langle E, F\rangle,\{\, p/_{\Theta \langle E, F\rangle }\, \} )$, and
        ${\cal A}$ recognizes  the tree language $ p/_{\tthue E}$.
         \end{prop}

\begin{prop}\label{erdo}{\em 
    \cite{franciak}}
  Given two btas      ${\cal A}=(  \S, A, R_{\cal A}, A_f)$ and
     ${\cal B}=( \S, B, R_{\cal B}, B_f)$,
  we can decide whether $L({\cal A})\subseteq L({\cal B})$ in $O (\mathrm{size}({\cal A})  \cdot \mathrm{size}({\cal B}) )$ time.
\end{prop}

\paragraph{Directed pseudographs without parallel arcs \rm{\cite{bang}}}
A directed pseudograph is a type of directed graph that allows for both parallel arcs  (i.e. pairs of directed arcs  with the same tail and the same head) and loops 
(i.e, directed arcs whose head and tail coincide) \cite[top of page 4]{bang}. 
A directed pseudograph without parallel arcs (dpwpa  for short) $G$ is an ordered pair $(V, A)$, 
where $V$ is a ﬁnite set, and $A\subseteq V \times V $
is a  relation on $V$. The set $V$ is called the vertex set of $G$, and its elements
are called vertices.  The set $A$ is called the directed arc  set of $G$, and its
elements are called directed arcs. 
A walk in $G$  is an alternating sequence $w=x_1\,  (x_1, x_2) \,  x_2 \,  (x_2, x_3)\,   x_3 \,  \ldots\,   x_{k-1}\,   (x_{k-1}, x_k) \,    x_k$, $k\geq 1$, 
of vertices $ x_i$ and arcs $(x_j, x_{j+1}) $. 
We say that  $w$ is from vertex $x_1$ to vertex $x_k$. The length of $w$ is the number $k-1$ of its arcs   \cite[page 12]{bang}.
A single vertex $x_1$ is considered as a walk from $x_1$ to $x_1$ of length $0$.
 A trail is a walk $w$ in which all arcs are distinct.
If the vertices of a  trail $w$  are distinct, then $w$ is a path. 
If the vertices
$x_1 , x_2 , \ldots ,  x_{k-1} $ of a  trail $w$ are distinct, $k \geq 2$, and $x_1 = x_k$, then $w$ is a cycle.
We say that a vertex $x$ is reachable  from a vertex $y$ if there exists a walk  $w$  from vertex $y$ to vertex $x$.
By definition, every vertex is reachable from itself. We say that a  vertex $x$ is positive-step reachable  from a vertex $y$ 
 if there exists a walk  $w$ of positive length  from vertex $y$ to vertex $x$.
 Observe that for any  vertices $x$ and $y$, the vertex $x$ is positive-step reachable  from the vertex $y$ if and only if there exists a 
 path of positive length or a cycle from the vertex $x$ to the vertex $y$.

The adjacency-list representation of a dpwpa  $G=(V, A)$ is an array  $Adj$ of $|V |$ lists, one list for each vertex in $V $ \cite{cormen}. For each $x\in V$, 
the adjacency list $Adj(x)$ contains all the vertices $y$ such that there is a directed arc $(x, y) \in A$.  That is, $Adj(x)$ 
consists of all the vertices adjacent to $x$ in $G$.  In pseudocode we treat the list $Adj$ as an attribute of the graph.  Let us assume that we are given
the set
$V$ and the set of arcs.  
It is well known that there exists an algorithm that builds the adjacency-list representation of the dpwpa $G$  with a time complexity of  $O(V)+O(A)$.

 \paragraph{Congruence closure on subterm dags \rm{\cite{jsc/Snyder93}}} Let $\Sigma$  be a signature. A subterm dag over $\S$   is a directed, labeled, and acyclic multigraph $G = ( V, A, \lambda)$, where $V$ is the set of vertices,
  $A \subseteq   V\times V $ is the multiset of directed arcs and
  $\lambda: V \rightarrow \S$
  is the labeling function such that if $outdegree(v)=m$
  for some vertex $v \in V$, 
then $\lambda(v)\in \S_m$. The multiset of  directed arcs leaving a vertex is ordered. We denote by $v/i$ the $i$th
successor of a vertex $v$.
The vertices of $G$ can be interpreted as ground terms over $\S$ by the function $\widehat{\; }: V \rightarrow W$ which is defined by recursion:
for a vertex $v \in V$, we have $\widehat{v}= \lambda(v)$ if $outdegree
(v)=0$ and
we have $\widehat{v}= \s(\widehat{v/1}, \ldots, \widehat{v/m})$
 if $outdegree (v)=m$
 for some $m > 0$ and $\lambda(v)=\s$. Let $E$  be  a GTES over  $\Sigma$.
       We say that  $G$
 is a subterm dag for a GTES  $E$  if \hspace{1 mm}  $\widehat{\; }$\hspace{1 mm}   is a bijection
 between $ V$ and $ST \langle E \rangle$. 
    Let $E$  and $F$ be  GTESs over  $\Sigma$.
  We say that  $G$
 is a subterm dag for the  GTES  $E$ and GTES $F$  if \hspace{1 mm}  $\widehat{\; }$\hspace{1 mm}   is a bijection
 between $ V$ and $ST \langle E \rangle \cup ST \langle F \rangle $.

 \begin{df}\rm
Let $E$ and $F$  be  GTESs over  a   signature   $\Sigma$,  $G = ( V, A, \lambda)$   be  a subterm dag for  $E$ and $F$. Then let
 $\tau\langle E, F\rangle=\{\, (u, v)\in V\times V \mid (\hat{u}, \hat{v}) \in E\, \}$
and  let $\rho\langle E, F\rangle$ be the congruence  closure of  $\tau\langle E, F\rangle$.
\end{df} 

We now cite 
Proposition 3.4 in \cite{jsc/Snyder93}.
\begin{prop}[{\cite[Proposition 3.4]{jsc/Snyder93}}]\label{haho}
Let $E$  be  a GTES over  a   signature   $\Sigma$, and   $n=\mathrm{size}(E)$.
 We can construct 
 in $O(n)$ time a subterm dag for the GTES $E$ and 
the relation
$\tau \langle E\rangle $.
\end{prop}
By Proposition \ref{haho}, we have the following result. 
\begin{sta}\label{szolo}
Let $E$  and $F$ be  GTESs over  a   signature   $\Sigma$, and  $n=\mathrm{size}(E)+\mathrm{size}(F)$. We can construct 
 in $O(n)$ time a subterm dag for the GTESs $E$ and $ F$. 
\end{sta}

\begin{df}\rm 
Let $G=(V, A, \lambda)$ be a subterm dag over $\Sigma$. An equivalence relation $\rho$ on $V$ is a congruence on $V$ if for each ordered pair of vertices 
$u, v \in V$, the conditions $\lambda(u)= \lambda(v)$  and $u/i\rho v/i$ for each $1\leq i \leq outdegree(v)$  imply that $u\rho v$. 
 Given a relation $\tau$  on $V$, the smallest congruence on $V$ containing
 $\tau$  is called the congruence closure of $\tau$.
 \end{df}
 Let $\Sigma$ be a signature,  and  $ G=( V, A, \lambda)$ be a subterm dag.
      Then   for any relation  $\tau$ on $V$,   there exists a unique congruence closure of $\tau$ on $V$,
 and algorithms for calculating congruence closures
are given in \cite{jacm/NelsonO80} and \cite{jacm/DowneyST80}.
 The latter can be implemented with a worst-case time of
$O(|A|\, \mathrm{log}\, |A|)$;
 the former runs in $O(|A|^2)$ time but seems to be faster in practice. We now apply  the congruence closure  algorithm of  \cite{jacm/DowneyST80}.
 Assume that it performs  the weighted union operation \cite{journals/jacm/Tarjan75} on equivalence classes.  Then the congruence closure algorithm 
  keeps track the cardinality of each equivalence class during its run, and can output  the cardinality of all  $\rho$-classes. We rename the congruence classes 
  produced by the congruence closure algorithm from $1$ to the number of the congruence classes in $O(|V|)$ time. 
  
\begin{prop}\label{lezaras}  {\rm \cite{jacm/DowneyST80}} 
Let $E$ be a GTES over $\Sigma$. 
 There exists an algorithm which, given a subterm dag  $ G=( V, A, \lambda)$
 for $E$, and a  relation  $\tau$ on $V$, produces the congruence closure
 $\rho$  of $\tau$ on $V$ in
 $O(n \, \mathrm{log} \, n)$ time,
 where $n=\mathrm{size}(E)$. The algorithm gives each  $\rho$-class a unique name in the set $\{\, 1, \ldots,|V/\rho|\,  \}$, and also outputs
 an array of the cardinalities  of all $\rho$-classes.
  \end{prop} 
    Let $E $ be a  GTES over $\Sigma$, and $G=(V, A, \lambda)$
 be a subterm dag for  $E$, and    $\rho$  be a  congruence on $V$.
 We shall use a function $\mathrm{FIND}_\rho (x)$ that returns a unique name of
the congruence class of $\rho$ containing $ x$ for each $x \in V$.
A possible implementation of
such a function was described in \cite{journals/jacm/Tarjan75}.

\section{Running Examples for the Main Cases}\label{runningpelda}

Along the proof of our main result, 
we shall distinguish four main cases. 
We define these  four main cases, and then  we present eight  running examples to illustrate our concepts and results. In each main case, a running  example presents  GTESs $E$ and $F$
such that   $\tthue {E\cup F} \subseteq \tthue E\cup \tthue F$ and a running  example presents  $E$ and $F$ such that    $\tthue {E\cup F} \not  \subseteq \tthue E\cup \tthue F$.
  From now on throughout the paper, $E$ and $F$ are GTESs over a  signature   $\S$, and  $W$ stands for
$ST\langle E \cup F\rangle$.

\begin{itemize} 
\item [Main Case 1:] $\S$ is a unary  signature.  Example \ref{1oneegy} and its sequels,   Examples \ref{1oneegyfolytatas},  \ref{1oneegymasodikfolytatas},   \ref{1oneegyharmadikfolytatas}. and  \ref{1oneegynegyedikfolytatas}, and \ref{1oneegyotodikfolytatas}, also  Example \ref{2szepen}  and its sequels,  \ref{2szepenfolytatas},   \ref{2szepenmasodikfolytatas},  \ref{2szepenharmadikfolytatas},  \ref{2szepennegyedikfolytatas}, and   \ref{2szepenotodikfolytatas} 
    illustrate this case.
\item  [Main Case 2:] Both GTRS $R\langle E, F\rangle$ and GTRS $R\langle F, E\rangle$ are  total. 
Example \ref{3szilva} and its sequels 
 Examples \ref{3szilvafolytatas},  \ref{3szilvamasodikfolytatas}, 
 further
 Example  \ref{4u2staring} and its sequels, Examples \ref{4u2staringfolytatas},  \ref{4u2staringmasodikfolytatas}, 
 exemplify this case.
\item  [Main Case 3:] 
   $\S_k\neq \emptyset$ for some $k\geq 2$, and  at least  one of $R\langle E, F\rangle$ 
and   $R\langle F, E\rangle$ is  total. Example \ref{5uborka} and its sequels, Examples \ref{5uborkafolytatas},  \ref{5uborkamasodikfolytatas}, \ref{5uborkaharmadikfolytatas},  \ref{5uborkanegyedikfolytatas},  and \ref{5uborkaotodikfolytatas}, moreover 
 Example \ref{6korte}, and its sequels,  Examples \ref{6kortefolytatas}, \ref{6kortemasodikfolytatas}, \ref{6korteharmadikfolytatas}, 
 \ref{6kortenegyedikfolytatas}, and 
 \ref{6korteotodikfolytatas} present this case.  

\item   [Main Case 4:]  $\S$ has a symbol of arity at least $2$, 
  and GTRS  $ R \langle E, F \rangle$ and GTRS $ R \langle F, E \rangle$  are not total.   
  Example
\ref{7twoketto} and its  sequel Example   \ref{7twokettofolytatas},  furthermore, Example 
 \ref{8threeharom} and its sequels Example \ref{8threeharomfolytatas} and Example \ref{8threeharommasodikfolytatas}
demonstrate  this case.
 \end{itemize}
 Main  Case 2 overlaps with Main Case 1 and Main Case 3, and is a subcase of the union of 
 Main Case 1 and Main Case 3.
 
For each running example,
we  compute the sets $ST\langle E \rangle$,  $ST\langle F \rangle $, $W$ and the  relations $\tthue E$, $\tthue F$, $\tthue {E \cup F}$,   $\Theta\langle E, F \rangle$, $\Theta\langle F, E \rangle$,
$\Theta\langle E \cup  F , \emptyset \rangle$, and state whether 
  $\tthue {E\cup F} \subseteq \tthue E\cup \tthue F$ or  $\tthue {E\cup F} \not  \subseteq \tthue E\cup \tthue F$.


Recall that 
$E$ and $F$ are GTESs over a
 signature   $\S$, and  $W$ stands for
$ST\langle E \cup F\rangle$.
\begin{exa}\label{1oneegy}\rm 
Let 
$\S=\S_0\cup \S_1$, $\S_0=\{\, \#,\; \$\, \}$,
$\S_1=\{\, f, \;  g \, \}$,
$E=\{\,  f(\#)\doteq g(\#) \, \}$, and 
$F=\{\, f(\$)\doteq g(\$) \, \}$.
Observe the following equations:
\begin{quote}
 $\tthue E = \{\,(\delta(s), \delta (t))\mid \delta \in CON_{\Sigma, 1}, s, t \in \{\, f(\#), \;  g(\#) \, \}\, \}$,
\newline
$\tthue F = \{\,(\delta(s), \delta (t))\mid \delta \in CON_{\Sigma, 1}, s, t \in \{\, f(\$), \;  g(\$) \, \}\, \}$, 
\newline
 $\tthue {E\cup F} = \tthue E \cup \tthue F$, 
 \newline
 $ST\langle E \rangle=\{\, \#, \;   f(\#), \; g(\#) \, \}   $, \, 
 $ST\langle F \rangle=\{\, \$, \;   f(\$), \; g(\$) \, \}   $, and
$W=\{\, \#, \;\$, \;  f(\#), \;   f(\$),
  \; g(\#), \;   g(\$) \, \}   $.
\newline
The equivalence classes of  $\Theta\langle E, F \rangle$ are
     $\{\, \#\,\}$, $\{\,   \$\, \}$,  $\{\, f(\#), \;g(\#) \, \}$,
 $\{\, f(\$)\,\}$, $\{\,  g(\$) \, \}$.
 \newline   The equivalence classes of  $\Theta\langle F, E \rangle$ are
   $\{\, \#\,\}$, $\{\,   \$\, \}$, 
    $\{\, f(\#)\,\}$, $\{\,  g(\#) \, \}$, 
    $\{\, f(\$), \;g(\$) \, \}$. \newline
  The equivalence classes of  $\Theta\langle E \cup  F, \emptyset \rangle$ are
    $\{\, \#\,\}$, $\{\,   \$\, \}$, 
    $\{\, f(\#), \;  g(\#) \, \}$, 
    $\{\, f(\$), \;  g(\$) \, \}$.
    \end{quote}
 We note that 
for each $t\in W $, 
\begin{quote} 
$ t/_{\Theta\langle {E\cup F}, \emptyset\rangle}=
  t/_{\Theta\langle E, F \rangle}  $ or 
  $ t/_{\Theta\langle {E\cup F}, \emptyset \rangle}=  t/_{\Theta\langle {F}, E \rangle}$.
  \end{quote}
 That is,  each $\Theta\langle E\cup F, \emptyset \rangle$ equivalence class is equal to a 
 $\Theta\langle E, F \rangle$ equivalence class or a $\Theta\langle F, E \rangle$  equivalence class.

\end{exa}

\begin{exa}\label{2szepen}\rm 
Let 
$\S=\S_0\cup \S_1$, $\S_0=\{\, \#,\; \$\, \}$,
$\S_1=\{\, f\, \}$,
$E=\{\,  f(f(\#))\doteq \# \, \}$, and 
$F=\{\, f(f(f(\#)))\doteq \#    \, \}$.

For every $n\in \Nat$, we define $f^n(\#)$ as follows:
\begin{itemize}

\item[(i)] $f^0(\#)=\#$ and 

\item[(ii)] $f^n(\#)=f(f^{n-1}(\#))$ for $n\geq 1$.
\end{itemize}
Observe the following equations:
\begin{quote}
 $\tthue E = \{\,(f^m(\#), f^n(\#))\mid m, n\in \Nat,\; m-n \mbox{ is divisible by } 2\,  \}$,
\newline
$\tthue F =  \{\,(f^m(\#), f^n(\#))\mid m, n\in \Nat, \;   m-n  \mbox{ is   divisible by }  3\,  \}$,
\newline
 $\tthue {E\cup F} = \ts \times \ts  $, 
 \newline
 $\tthue E\cup \tthue F\subset \tthue {E\cup F} $,\newline
 $ST\langle E \rangle=\{\, \#, \;   f(\#), \; f^2(\#)  \, \}   $,\;
 $ST\langle F \rangle= \{\, \#, \;   f(\#), \; f^2(\#), \;f^3(\#) \, \}   $,  \ and
 $W=ST\langle F \rangle$.
 \newline
The equivalence classes of  $\Theta\langle E, F \rangle$ are
     $\{\, \#, \;  f^2(\#)\}$, $\{\, f(\#), \; f^3(\#)\, \}$.
  \newline
  The equivalence classes of  $\Theta\langle F, E \rangle$ are
    $\{\, \#, \;  f^3(\#)\}$, $\{\, f(\#)\, \}$, \; $\{\, f^2(\#)\, \}$.
\newline
  The only equivalence class of  $\Theta\langle E \cup  F, \emptyset  \rangle$ is 
    $\{\, \#, \; f(\#), \;  f^2(\#), \;f^3(\#)\, \}$.
     \end{quote}
  We note that 
  \begin{quote} 
for each $t\in W $, 
$\left( t/_{\Theta\langle E, F \rangle }  \subset t/_{\Theta\langle {E\cup F},  \emptyset  \rangle}
 \mbox{ and } t/_{\Theta\langle {F}, E \rangle} \subset 
   t/_{\Theta\langle {E\cup F}, \emptyset  \rangle}\right)$.
  \end{quote}
\end{exa}

\begin{exa}\label{3szilva} \rm  

Let 
$\S=\S_0\cup \S_2$, $\S_0=\{\, \#,\,  \$, \;\mathsterling \, \}$, 
$\S_2=\{\, f \, \}$. Before 
 defining  the  GTESs $E$ and $F$, we 
present statements 1 -- 4  as intuitive explanations.

    1. $\tthue E$ has four congruence classes:
 \begin{itemize}
 \item The singleton $\{\, \#\, \}$. Because there is no equation in $E$ with left-hand side or right-hand side $\#$. We consider  $\#$ as the representative
of this class. 

 \item The set of all ground terms which are different from the constant $\#$, and contain only  the symbols $f$ and $\#$. 
Applying  the first three equations,  $E$ rewrites each ground term $t$  of height $2$ in the above set to $f(\#, \#)$. 
Hence  for  each ground term $t$ in this set, we have  $t \tthue E f(\#, \#)$. We consider  $f(\#, \#)$  as the representative
of this class. 

\item  The set of all ground terms containing only the symbols $f$ and $\$$.  Applying the fourth equation of $E$ 
 for  each such ground term $t$ in this set, we have  $t \tthue E \$$.  We consider  $ \$$ as  the representative
of this class. 

\item  The set of all ground terms  which contain the constant $\mathsterling$ or at least two different constants. 
Let   $s, t \in  \{\, \#, \,f(\#, \#),  \$, \mathsterling\, \}$. If   the term 
 $f(s,t)$  contains the constant $\mathsterling$ or at least two different constants, then $f(s, t)$ is the left-hand side of one  of the  5th -- 15th equations of $E$.
 Hence
 for  each ground term $p$ which contains the constant $\mathsterling$ or at least two different constants, we have  $p \tthue E \mathsterling$. 
  \end{itemize}

2. We obtain GTES $F$ from GTES $E$ by swapping the constants $\#$ and $\$$. Intuitively we can think of  $F$ as a \enquote{mirror image} of $E$. 

3.  $\tthue F$ has four congruence classes:
 \begin{itemize}
 \item the singleton $\{\, \$\, \}$,
 \item the set of all ground terms which are different from the constant  $\$$, and contain only  the symbols $f$ and $\$$, 
 \item the set of all ground terms containing only the symbols $f$ and $\#$, and 
\item  the set of all ground terms which contain  the constant $\mathsterling$ or  at least two different constants.
  \end{itemize}

4. $\tthue {E\cup F}$ has three congruence classes:
 \begin{itemize}
 \item  the set of all ground terms which  contain only the symbols $f$ and $\$$, 
 \item the set of all ground terms containing only the symbols $f$ and $\#$, and 
 \item the set of all ground terms which contain  the constant $\mathsterling$ or  at least two different constants.
  \end{itemize}

Let 
\begin{quote}
$E=\{\, f(\#, f(\#, \#)) \doteq f(\# , \#), \;  f(f(\#, \#), \#) \doteq f(\# , \#), \; f(f(\#,\#), f(\#, \#)) \doteq f(\# , \#), \;\newline
  f(\$, \$)\doteq \$, \;\newline
 f(\#, \$)\doteq \mathsterling, \; f(\#, \mathsterling)\doteq \mathsterling, \; \newline
 f(\$, \#)\doteq \mathsterling, \;  f(\$, \mathsterling)\doteq \mathsterling, \; \newline
  f(\mathsterling, \#)\doteq \mathsterling, \;   f(\mathsterling, \$)\doteq \mathsterling, \; f(\mathsterling, \mathsterling)\doteq \mathsterling,
    \newline
 f(\$, f(\#, \#))\doteq \mathsterling, \;   f(\mathsterling, f(\#, \#))\doteq \mathsterling, \; \newline
   f(f(\#, \#), \$)\doteq \mathsterling, \; f(f(\#, \#), \mathsterling)\doteq \mathsterling
    \, \}$,
    \end{quote}
   \begin{quote}
$F=
\{\, f(\$, f(\$, \$)) \doteq f(\$ , \$), \;  f(f(\$, \$), \$) \doteq f(\$ , \$), \; f(f(\$,\$), f(\$, \$)) \doteq f(\$ , \$), \;\newline
f(\#, \#)\doteq \#, \; \newline
 f(\$, \#)\doteq \mathsterling, \; f(\$, \mathsterling)\doteq \mathsterling, \; \newline
   f(\#, \$)\doteq \mathsterling, \; f(\#, \mathsterling)\doteq \mathsterling, \; \newline
  f(\mathsterling, \#)\doteq \mathsterling, \;   f(\mathsterling, \$)\doteq \mathsterling, \; f(\mathsterling, \mathsterling)\doteq \mathsterling
    \newline
 f(\#, f(\$, \$))\doteq \mathsterling, \;   f(\mathsterling, f(\$, \$))\doteq \mathsterling, \; \newline
   f(f(\$, \$), \#)\doteq \mathsterling, \; f(f(\$, \$), \mathsterling)\doteq \mathsterling
    \, \}$.
\end{quote} 
    Note the following equations:
\begin{itemize}

  \item 
 $ST\langle E \rangle= \{\, \#, \;\$, \;\mathsterling, \;\newline
 f(\#, \#), \; f(\#, \$), \; f(\#, \mathsterling), \; \newline
 f(\$, \#), \;   f(\$, \$), \; f(\$, \mathsterling), \; \newline
  f(\mathsterling, \#), \;   f(\mathsterling, \$), \; f(\mathsterling, \mathsterling), \;\newline
   f(\#, f(\#, \#)) , \;  f(f(\#, \#), \#) , \; \newline 
   f(f(\#,\#), f(\#, \#)), \;
  f(\$, f(\#, \#)), \;   f(\mathsterling, f(\#, \#)), \; \newline
   f(f(\#, \#), \$), \; f(f(\#, \#), \mathsterling)
   \, \}$, 
\item
 $ST\langle F \rangle= \{\, \$, \;\#, \;\mathsterling, \;\newline
 f(\$, \$), \; f(\$, \#), \; f(\$, \mathsterling), \; \newline
 f(\#, \$), \;   f(\#, \#), \; f(\#, \mathsterling), \; \newline
  f(\mathsterling, \$), \;   f(\mathsterling, \#), \; f(\mathsterling, \mathsterling), \;\newline
     f(\$, f(\$, \$)) , \;  f(f(\$, \$), \$) , \; 
   f(f(\$,\$), f(\$, \$)), \;\newline 
    f(\#, f(\$, \$)), \;   f(\mathsterling, f(\$, \$)), \; \newline 
   f(f(\$, \$), \#), \; f(f(\$, \$), \mathsterling)
   \, \}$,

\item
 $W=
 \{\, \#, \;\$, \;\mathsterling, \;\newline
 f(\#, \#), \; f(\#, \$), \; f(\#, \mathsterling), \; \newline
 f(\$, \#), \;   f(\$, \$), \; f(\$, \mathsterling), \; \newline
  f(\mathsterling, \#), \;   f(\mathsterling, \$), \; f(\mathsterling, \mathsterling), \;\newline
    f(\#, f(\#, \#)) , \;  f(f(\#, \#), \#) , \; \newline 
   f(f(\#,\#), f(\#, \#)), \;
  f(\$, f(\$, \$)) , \;  \newline
    f(f(\$, \$), \$) , \; 
   f(f(\$,\$), f(\$, \$)), \;\newline 
     f(\$, f(\#, \#)), \;   f(\mathsterling, f(\#, \#)), \;
   f(f(\#, \#), \$), \; f(f(\#, \#), \mathsterling), \;\newline
  f(\#, f(\$, \$)), \;   f(\mathsterling, f(\$, \$)), \; 
   f(f(\$, \$), \#), \; f(f(\$, \$), \mathsterling)
   \, \}$. 
\end{itemize}

\begin{quote}
$\tthue E=  \{\, (\#, \#)\, \} \cup (T_{\{\, \#, f\, \} }  \setminus \{\, \#\, \})  \times (T_{\{\, \#, f\, \} }  \setminus \{\, \#\, \}) \cup
(T_{\{\, \$, f\, \} } \times T_{\{\, \$, f\, \} })\cup  \newline
\left(T_\S\setminus  (T_{\{\, \#, f\, \}} \cup    T_{\{\, \$, f\, \} })\right)\times \left(T_\S\setminus  (T_{\{\, \#, f\, \}} )\cup    T_{\{\, \$, f\, \} }) \right)$.

  $\tthue E$ has the congruence  classes  $ \{\, \#\, \}$, \, $T_{\{\, \#, f\, \} }  \setminus \{\, \#\, \}$, \,
   $T_{\{\, \$, f\, \} }$, \,  and 
  $T_\S\setminus  (T_{\{\, \#, f\, \}} \cup T_{\{\, \$, f\, \} })$.
  
Then
\begin{itemize}
\item  $T_{\{\, \#, f\, \} } $ is the set of all ground terms containing only  the symbols $f$ and $\#$. 
\item  $T_{\{\, \$, f\, \} } $ is the set of all ground terms containing only  the symbols $f$ and $\$$. 
\item
 $T_{\{\, \#, f\, \}} \cup    T_{\{\, \$, f\, \} }$  is the set of all ground terms which contain only  the symbols $f$ and $\#$
or only the symbols $f$ and $\$$.
\item  $T_\S\setminus  (T_{\{\, \#, f\, \}} \cup    T_{\{\, \$, f\, \} })$ is the set of all ground terms which contain 
   the constant $\mathsterling$ or  at least two different constants.
 \end{itemize}

    $\tthue F=(T_{\{\, \#, f\, \} })\times T_{\{\, \#, f\, \} })\cup 
\left( T_{\{\, \$, f\, \} }  \setminus \{\, \$\, \}\right)\times \left( T_{\{\, \$, f\, \} }  \setminus \{\, \$\, \}\right) \cup \newline
  \left(T_\S\setminus  (T_{\{\, \#, f\, \}} \cup    T_{\{\, \$, f\, \} }\right)\times 
  \left(T_\S\setminus  (T_{\{\, \#, f\, \}} \cup    T_{\{\, \$, f\, \} }))\right)
  \cup
     \{\, (\$, \$)\, \} $.

  $\tthue F$ has the congruence classes  $T_{\{\, \#, f\, \} } $, \, $\{\, \$\, \}$,  \,
   $T_{\{\, \$, f\, \} } \setminus \{\, \$\, \}$,\, and 
  $T_\S\setminus  (T_{\{\, \#, f\, \}} \cup    T_{\{\, \$, f\, \} })$.

  $\tthue {E \cup  F}= 
  (T_{\{\, \#, f\, \} })\times T_{\{\, \#, f\, \} }) \cup 
  (T_{\{\, \$, f\, \} }\times T_{\{\, \$, f\, \} }) \cup \newline 
   \left(T_\S\setminus  (T_{\{\, \#, f\, \}} \cup    T_{\{\, \$, f\, \} })\right)\times 
 \left( (T_\S\setminus  (T_{\{\, \#, f\, \}} \cup    T_{\{\, \$, f\, \} })\right)
   $.

  $\tthue {E \cup  F}$ has the congruence classes 
 $T_{\{\, \#, f\, \} }$, \,   $T_{\{\, \$, f\, \} }$, \, and $T_\S\setminus  (T_{\{\, \#, f\, \}} \cup    T_{\{\, \$, f\, \} })$.
  
 \end{quote}

Observe that $E \not \subseteq \tthue F$ and $F \not\subseteq \tthue E$ and $\tthue E \cup \tthue F = \tthue {E \cup F}$.


  The equivalence classes of  $\Theta\langle E, F \rangle$ are
  \begin{quote}
  $   \{\, \#\, \},  \;  \newline
 \{\, \$, \;  f(\$, \$) , \;
 f(f(\$,\$),\$), \; f(\$,f(\$,\$)),\; f(f(\$,\$),f(\$,\$)) \, \}, 
\newline 
  \{\,f(\#, \#), \;   f(\#, f(\#, \#)), \;  f(f(\#, \#), \#), \; f(f(\#,\#), f(\#, \#))   \, \},  
   \newline\;\{\,\mathsterling, \;  f(\#, \$), \; f(\#, \mathsterling), \;   f(\$, \#), \; 
  f(\$, \mathsterling), \; 
  f(\mathsterling, \#), \;   f(\mathsterling, \$), \; f(\mathsterling, \mathsterling), \; \newline
   f(\$, f(\#, \#)), \;   f(\mathsterling, f(\#, \#)), \;
   f(f(\#, \#), \$), \; f(f(\#, \#), \mathsterling), \;\newline
  f(\#, f(\$, \$)), \;   f(\mathsterling, f(\$, \$)), \; 
   f(f(\$, \$), \#), \; f(f(\$, \$), \mathsterling)
  \, \}$.
   \end{quote}
The equivalence classes of  $\Theta\langle  F, E \rangle$  are
 \begin{quote}
 $ \{\, \#, \;   f(\#, \#), \; f(f(\#,\#),\#), \; f(\#,f(\#,\#)), \; f(f(\#,\#),f(\#,\#)) \, \}, \;\newline  
 \{\, \$ \,\}, \newline \{\, f(\$, \$), \; f(\$, f(\$, \$)), \;  f(f(\$, \$), \$), \; f(f(\$,\$), f(\$, \$)) \}, \;\newline
  \{\, \mathsterling, \; 
 f(\#, \$), \; f(\#, \mathsterling), \;     f(\$, \#), \; f(\$, \mathsterling), \; 
   f(\mathsterling, \#), \; f(\mathsterling, \$), \;  f(\mathsterling, \mathsterling), \; \newline
   f(\$, f(\#, \#)), \;   f(\mathsterling, f(\#, \#)), \;
   f(f(\#, \#), \$), \; f(f(\#, \#), \mathsterling), \;\newline
  f(\#, f(\$, \$)), \;   f(\mathsterling, f(\$, \$)), \; 
   f(f(\$, \$), \#), \; f(f(\$, \$), \mathsterling)
    \, \}$.
    \end{quote}
   The only equivalence class  of  $\Theta\langle E\cup F, \emptyset \rangle$ is $W$.

\end{exa}

\begin{exa}\label{4u2staring}\rm We now define the  GTESs $E$ and $F$.
  Intuitively,  applying the first two rules, GTES $E$ rewrites each occurrence  of $\$$ and  $\mathsterling$ to $\#$ in every ground term $t\in \ts$. Then applying the last three rules, GTES $E$ rewrites
 the ground term $t\in \ts$ with $height(t)\geq 1$ to $\#$.    GTES $F$ rewrites each occurrence of
  $\flat$  to $\mathsterling$ in every ground term. In this way, it  rewrites each ground term $t\in \ts$ with $height(t)= 1$ to $\mathsterling$. Hence GTES $F$ rewrites each ground term $t\in \ts$ with $height(t)\geq 1$ to $\mathsterling$.

 Let 
$\S=\S_0\cup \S_2$, $\S_0=\{\, \#,\;  \$, \;\mathsterling, \; \flat \, \}$,
$\S_2=\{\, f \, \}$, 
\begin{quote}
  $E=\{\, \#\doteq \$, \;\#\doteq \mathsterling,\;  \#\doteq f(\#, \#), \;     \#\doteq f(\#, \flat), \;   
 \#\doteq f(\flat, \#)
  \, \}$, 
\end{quote}
and
\begin{quote}
  $F=\{\, \mathsterling\doteq \flat, 
  \newline
  \mathsterling\doteq f(\#, \#), \;
  \mathsterling\doteq f(\#, \$), \;
 \mathsterling\doteq f(\#, \mathsterling), \;
     \newline
 \mathsterling\doteq f(\$, \#), \;
   \mathsterling\doteq f(\$, \$), \;
  \mathsterling\doteq f(\$, \mathsterling), \;
     \newline
  \mathsterling\doteq f(\mathsterling, \#), \;
  \mathsterling\doteq f(\mathsterling, \$), \;
   \mathsterling\doteq f(\mathsterling, \mathsterling)
      \, \}$.
\end{quote}
We note that  for each ground term $t\in \ts$ with $height(t)\geq 1$,
 \begin{quote}
$t\tthue E  \#$ and $t\tthue F \mathsterling$.
 \end{quote}
Therefore, 
\begin{quote}
$\tthue E=
\{\,  (\flat, \flat)   \, \}
\cup
((\ts\setminus \{\, \flat \, \})
\times (\ts\setminus \{\, \flat \, \}))$, 

$\tthue F =\{\, 
(\#, \#), \; (\$, \$) \, \} \cup
((\ts\setminus \{\, \#, \;   \$ \, \})
\times( \ts\setminus \{\, \#,  \; \$     \, \}))$, and \newline 
    $\tthue {E \cup F}=\ts \times \ts $.
  \end{quote}
  Note that 
 $(\#, \flat)\not\in \tthue E\cup \tthue F$. Consequently, 
 \begin{quote}
$\tthue E\cup \tthue F\subset \tthue {E\cup F}$.
\end{quote}
We notice  the following equations:
\begin{quote}
    $W=\{\, \#, \; \$, \;\mathsterling, \; \flat, \;
    \newline
     f(\#, \#), \;
    f(\#, \$), \;
    f(\#, \mathsterling), \;
  f(\#, \flat), \;
 \newline
  f(\$, \#), \;
  f(\$, \$), \;
   f(\$, \mathsterling), \;
  f(\$, \flat), \;
    \newline
   f(\mathsterling, \#), \;
    f(\mathsterling, \$), \;
  f(\mathsterling, \mathsterling), \;
    f(\mathsterling, \flat), \;
     \newline
 f(\flat, \#), \;
 f(\flat, \$), \;
 f(\flat, \mathsterling), \;
 f(\flat, \flat)
 \, \}$.
   \end{quote}
  The equivalence classes of  $\Theta\langle E, F \rangle$ are
  \begin{quote}
    $\{\, \#, \;  \$, \;  \mathsterling, \;\newline
    \, f(\#, \#), \;   
  f(\#, \$), \;
 f(\#, \mathsterling), \;
f(\#, \flat), \; \newline
 f(\$, \#), \;   
  f(\$, \$), \;
 f(\$, \mathsterling), \;
 f(\$, \flat), \;
  \newline
 f(\mathsterling, \#), \;   
  f( \mathsterling, \$), \;
  f(\mathsterling, \flat),
   f(\mathsterling, \mathsterling), \;
\newline
 f(\flat, \#), \;   
  f(\flat, \$), \;
 f(\flat, \mathsterling), \;
 f(\flat, \flat)\, \},\;\newline
 \{\, \flat\, \}$.
   \end{quote}
   The equivalence classes of  $\Theta\langle F, E \rangle$ are
    \begin{quote}
    $\{\,  \mathsterling, \; \flat, \; \newline 
    \, f(\#, \#), \;   
  f(\#, \$), \;
 f(\#, \mathsterling), \;
f(\#, \flat), \;
\newline
 f(\$, \#), \;   
  f(\$, \$), \;
 f(\$, \mathsterling), \;
 f(\$, \flat), \;
  \newline
 f(\mathsterling, \#), \;   
  f( \mathsterling, \$), \;
  f(\mathsterling, \flat),
   f(\mathsterling, \mathsterling), \;
\newline
 f(\flat, \#), \;   
  f(\flat, \$), \;
 f(\flat, \mathsterling), \;
 f(\flat, \flat)\, \},\newline
 \{\, \#\, \}, \{\, \$\, \}$.
\end{quote}
  The only equivalence class of  $\Theta\langle E\cup F, \emptyset \rangle$ is $W$. 
    We note that 
    \begin{quote} 
$\#/_{\Theta\langle E, F \rangle} \subset \#/_{\Theta\langle {E\cup F},  \emptyset \rangle}$  
 and $ \#/_{\Theta\langle {F}, E \rangle}
  \subset  \#/_{\Theta\langle {E\cup F},  \emptyset \rangle}$.
\end{quote}
\end{exa}

\begin{exa}\label{5uborka}\rm 
Let 
$\S=\S_0\cup \S_2$, $\S_0=\{\, \#, \;\$  \, \}$,
$\S_2=\{\, f \, \}$,
\begin{quote}
$E=\{\, f(\#, \#) \doteq \$, \; f(\#, \$) \doteq \$, \;
 f(\$, \#) \doteq \$, \; f(\$, \$) \doteq \$  \, \}$ 
and 
 $F=\{\, f(\#, \#) \doteq \#   \, \}$.
\end{quote}
Note the following equations:
  \begin{quote}
 $ST\langle E \rangle=
    \{\,  \#, \; \$, \;  f(\#, \#), \; 
    f(\#, \$), \; f(\$, \#), \;f(\$, \$)\, \}$,\;
 $ST\langle F \rangle=\{\, \#, \; f(\#, \#)  \, \}$,\;
 $W=ST\langle E \rangle$.
\end{quote}
  \begin{quote}
  The equivalence classes of  $\Theta\langle E, F \rangle$ are
$    \{\,  \# \,\}$,  
   $\{\, \$, \; f(\#, \#), \; f(\#, \$), \; f(\$, \#), \;   f(\$, \$) \, \}$.
 \newline
The equivalence classes of  $\Theta\langle F, E \rangle$ are
$\{\, \#, \;  f(\#, \#)\,\}, \; \{\, \$\, \}, \;
     \{\, f(\#, \$)\, \}, \; \{\, f(\$, \#)\, \}, \; \{\,  f(\$, \$)\, \}$.
        \newline The only  equivalence class of  $\Theta\langle E\cup  F, \emptyset \rangle$  is
 $    \{\,  \#, \;  \$, \; f(\#, \#), \;f(\#, \$), \; f(\$, \#), \;   f(\$, \$) \, \}$.
  \end{quote}
Observe that $E \not \subseteq \tthue F$ and $F \not\subseteq \tthue E$. Since $(\#, \$)\not \in \tthue E \cup \tthue F$, and 
 $(\#, \$)\in\tthue {E \cup F}$, we have $\tthue E \cup \tthue F \subset \tthue {E \cup F}$.
\end{exa}

\begin{exa}\label{6korte} \rm  
 Let 
$\S=\S_0\cup \S_1\cup \S_2$, $\S_0=\{\, \#, \;\$  \, \}$, $\S_1=\{\, g  \, \}$,
$\S_2=\{\, f \, \}$.  We now define the  GTESs $E$ and $F$. Intuitively,  applying the first five rules, GTES $E$ rewrites each ground term $t$ -- where  $height(t) \geq 1$ and $t$  contains only the symbols $f$, $g$, and $\#$ --  to $g(\#)$. 
Aplying the last four rules,  GTES $E$ rewrites each ground term $t$ containing an occurrence of $\$$ to $\$$. 
Let 
\begin{quote}
$E=\{\, f(\#, \#) \doteq g(\#), \;  f(\#, g(\#)) \doteq g(\#), \;
 f(g(\#), \#) \doteq g(\#), \; f(g(\#), g(\#)) \doteq g(\#), \; g(g(\#)) \doteq g(\#), \; \newline 
 f(\#, \$)\doteq \$, \; 
 f(\$, \#)\doteq \$, \; 
  f(\$, \$)\doteq \$, \;
 g(\$)\doteq \$ \, \}$ and 
 
$F=\{\, f(\#, \#) \doteq \#, \;g(\#)\doteq \# \,   \}$.
\end{quote}
Intuitively, $F$ can simulate the first five equations of $E$, that is, the congruence generated by the first five equations  of $E$ is a proper subset of $\tthue F$.
We need the constant $\$$ and the  last four equations of $E$  to ensure that 
$\tthue E \not \subseteq \tthue F$. 
Note the following equations:
\begin{itemize}
  \item 
 $ST\langle E \rangle=
   \{\,  \#, \;  \$, \;  f(\#, \#), \;
  f(\#, \$), \; f(\$, \#), \;   f(\$, \$), \;g(\#), \;g(\$), \;
  f(\#, g(\#)), \; f(g(\#), \#), \; f(g(\#), g(\#)), \;  
    g(g(\#)) \, \}$,

\item
 $ST\langle F \rangle= \{\,  \#, \;f(\#, \#), \;g(\#)\,\}$,

\item
 $W=ST\langle E \rangle$.
\end{itemize}
  The equivalence classes of  $\Theta\langle E, F \rangle$ are
  \begin{quote}
$  \{\, \#\, \},   \{\,    f(\#, \#), \;   f(\#, g(\#)), \;
 f(g(\#), \#), \; f(g(\#), g(\#)), \; g(\#), \;g(g(\#))  \, \}, \; \newline
  \{\, \$, \;  f(\#, \$), \; f(\$, \#), \;f(\$, \$), \;g(\$), \; g(g(\$)) \, \}$,
 \end{quote}
 The equivalence classes of  $\Theta\langle  F, E \rangle$  are
 \begin{quote}
 $    \{\,  \#, \; f(\#, \#), \;   f(\#, g(\#)), \;
 f(g(\#), \#), \; f(g(\#), g(\#)), \; g(\#), \;g(g(\#))  \, \}$, \newline
  $  \{\, \$ \, \}, \; \{\,  f(\#, \$)\, \}, \; \{\, f(\$, \#)\, \}, \;\{\,  f(\$, \$)\, \}, \;\{\, g(\$)\, \}, \; \{\, g(g(\$)) \, \}$.

  \end{quote}
   The equivalence classes of  $\Theta\langle E\cup F, \emptyset \rangle$ are
  \begin{quote}
$  \{\, \#, \;  f(\#, \#), \;   f(\#, g(\#)), \;
 f(g(\#), \#), \; f(g(\#), g(\#)), \; g(\#), \;g(g(\#))  \, \}, \; \newline
  \{\, \$, \;  f(\#, \$), \; f(\$, \#), \;f(\$, \$), \;g(\$), \; g(g(\$)) \, \}$.
 \end{quote}
  Observe that $E \not \subseteq \tthue F$ and $F \not\subseteq \tthue E$.
  \begin{quote}
  $\$/_{\tthue E}= \$/_{\tthue {E\cup F}}=\{\, t\in \ts \mid \mbox{ there exists } \alpha \in POS(t) \mbox{ such that } t|_\alpha=\$ \, \}$.\newline 
  $\#/_{\tthue F}= \#/_{\tthue {E\cup F}}=\{\, t\in \ts \mid \mbox{ there does not exist } \alpha \in POS(t) \mbox{ such that } t|_\alpha=\$ \, \}$.\newline
    Accordingly    $\tthue {E \cup F} \subseteq \tthue E \cup \tthue F$.
\end{quote}
\end{exa}

\begin{exa}\label{7twoketto}\rm 
Let 
$\S=\S_0\cup \S_2$, $\S_0=\{\, \#,\; \$\, \}$,
$\S_2=\{\, f\, \}$,
\begin{quote}
  $E=\{\,f(\#, \#)\doteq f(\#, \$), \; f(\#, \#)\doteq f(\$, \#), \;
  f(\#, \#)\doteq f(\$, \$) \, \}$ and
$F=\{\,\#\doteq \$  \, \}$.
\end{quote}
Observe the following equations and inclusion:
\begin{quote}
 $ST\langle E \rangle=\{\, \#, \; \$, \;
  f(\#, \#), \; f(\#, \$), \;f(\$, \#), \;  f(\$, \$)  \, \}$,\;
$ST\langle F \rangle=  \{\, \#,\$ \, \}$, \;
 $W=ST\langle E \rangle$,\;
 $E \subseteq \tthue F$, \;
 $\tthue {E\cup F}=\tthue F$,\;
 $\tthue {E \cup F} \subseteq \tthue E \cup \tthue F$.

\end{quote}
Since $E \subseteq \tthue F$, each  equivalence class of  $\Theta\langle E, F \rangle$ is a subset of some equivalence class of  $\Theta\langle F, E \rangle$.

\begin{quote}
The equivalence classes of  $\Theta\langle E, F \rangle$ are
$\{\, \#\, \}, \;\{\, \$ \, \}, \; 
  \{\, f(\#, \#), \; f(\#, \$), \;  f(\$, \#), \;
    f(\$, \$)  \, \}$.
\newline
    The equivalence classes of  $\Theta\langle F, E\rangle$ are
$\{\, \#, \; \$ \, \}, \; 
  \{\, f(\#, \#), \; f(\#, \$), \;  f(\$, \#), \;
    f(\$, \$)  \, \}$.
\newline 
The equivalence classes of  $\Theta\langle E\cup F, \emptyset \rangle$ are
$\{\, \#, \; \$ \, \}, \; 
  \{\, f(\#, \#), \; f(\#, \$), \;  f(\$, \#), \;
    f(\$, \$)  \, \}$.
\end{quote}
\end{exa}

\begin{exa}\label{8threeharom}\rm 
Let 
$\S=\S_0\cup \S_2$, $\S_0=\{\, \#, \;\flat, \;\$,\; \mathsterling \, \}$,
$\S_2=\{\, f \, \}$,
\begin{quote}
$E=\{\,  \#\doteq \$  \, \}$, 
and 
  $F=\{\,   \mathsterling \doteq \flat  \, \}$.
\end{quote}
We remark  that
$f(\#, \mathsterling) \tthue {E\cup F} f(\$, \flat)$,\;
$(f(\#, \mathsterling), f(\$, \flat))\not\in  \tthue E  $, \; and 
$(f(\#, \mathsterling), f(\$, \flat))\not\in  \tthue F  $.
Thence
\begin{quote}
$\tthue E\cup \tthue F\subset \tthue {E\cup F}$.
\end{quote}
Observe the following equations:
\begin{quote}
  $ST\langle E \rangle=\{\, \#, \; \$ \, \}$, \;
 $ST\langle F \rangle=\{\, \mathsterling, \; \flat \, \}$, \;
 $W =\{\, \#, \; \$, \;\mathsterling, \; \flat \, \}$.
\end{quote}
  \begin{quote}
  The equivalence classes of  $\Theta\langle E, F \rangle$ are
 $    \{\, \#, \; \$ \, \}$,
    $\{\, \mathsterling\,\}$, $\{\,   \flat\, \}$.
       \newline 
The equivalence classes of  $\Theta\langle F, E \rangle$ are
  $\{\, \#\,\}$, $\{\,  \$ \, \}$,
    $\{\, \mathsterling, \; \flat\, \}$.
      \newline
The equivalence classes of  $\Theta\langle E\cup  F, \emptyset  \rangle$ are
  $\{\, \#, \; \$ \, \}$,
    $\{\, \mathsterling, \; \flat\, \}$.
\end{quote} 
Notice  that $E \not \subseteq \tthue F$ and $F \not\subseteq \tthue E$.
\end{exa}

\section{Fast Ground Completion Introducing Additional Constants}\label{gyors}

From now on throughout the paper, let $\Sigma$ be a signature,  $E$ and $F$ be  GTESs over $\S$, 
$W=ST \langle E\cup F \rangle$, 
and   $n= \mathrm{size}(E)+\mathrm{size}(F)$.  
We recall  the fast ground completion  algorithm of F\"ul\"op and V\'agv\"olgyi \cite[Algorithm 3.6.]{eatcs/FuloopV91}. Its input is $\S$, $E$, and
$F $.  Their fast ground completion  algorithm \cite[Algorithm 3.6.]{eatcs/FuloopV91} produces in 
$O(n \,\mathrm{log}\, n)$ time a finite set of  new  constants and 
a reduced GTRS  $ R \langle E, F \rangle  $
over $\S$ extended with the new constants
such  that $\tthue E=\tthue {R\langle E, F  \rangle } \cap \ts \times \ts $.
 Their algorithm calls a
 congruence closure algorithm  on  the  subterm dag for $E\cup F$ \cite{jacm/DowneyST80}, and runs in 
$O(n \, \mathrm{log}\, n)$ time by Proposition \ref{lezaras}. 

We now cite the  Claim on \cite[p.~167]{Kozen77}.
  \begin{prop}[{\cite [Claim on page 167]{Kozen77}}] \label{kozen}
     $\Theta\langle E, F \rangle= \{\, (\widehat{x}, \widehat{y}) \mid (x, y)\in \rho
\langle E, F  \rangle\, \}$.
  \end{prop}

\begin{df}\label{szeder}  \rm We put
$C\langle E, F  \rangle = \{\,  t/_{\Theta\langle E, F  \rangle}  \mid t \in W \, \}$.
\end{df} 
By Proposition \ref{kozen}, 
for each equivalence class $\vartheta \in W /_{  \Theta\langle E, F \rangle }$, there exists exactly one   equivalence  class $\eta \in V/_{\rho\langle E, F  \rangle  } $ such that 
$\vartheta= \{\, \widehat{x} \mid x\in \eta\, \}$.
To simplify our notations, we give  $\vartheta$
 the  unique name of $\eta$.  The elements of $ C\langle E, F  \rangle$
are viewed as constants, and we  define the  signature   $\Sigma \cup  C\langle E, F  \rangle$.
By Definition \ref{szeder}, we have the following statement. 
\begin{sta}\label{darabsz}
$|C\langle E, F   \rangle|\leq |W|\leq \mathrm{size}(W )\leq n$. 
\end{sta}

\begin{df}[{\cite[Section 3]{eatcs/FuloopV91}}]\label{szabaly} \rm
Let
 \begin{quote} $R\langle E, F  \rangle= \newline
  \{\, \s( t_1/_{\Theta\langle E, F \rangle}, \ldots,  t_m/_{\Theta\langle E, F  \rangle}) \dot{\rightarrow}  \s(t_1, \ldots,t_m)/_{\Theta\langle E, F  \rangle}
  \mid \s\in \S_m, m\in \Nat, t_1, \ldots, t_m \in \ts,  \s(t_1, \ldots, t_m)\in W \, \}$ 
  \end{quote}
  be a GTRS over the signature  $\S \cup C\langle E, F  \rangle$.
  \end{df}
Note that for any $\s(t_1, \ldots, t_m) \in W$ with  $\s\in \S_m$, $m\in \Nat$,
and $t_1, \ldots, t_m \in \ts  $ there exists at most one
rewrite rule 
 with left-hand side $ \s( t_1/_{\Theta\langle E, F  \rangle}, \ldots,  t_m/_{\Theta\langle E, F  \rangle})$ in $R \langle E, F \rangle $.
  \begin{df}\label{csurig} \rm
  We say that GTRS $R \langle E, F \rangle $
is  total if for any  $\s\in \S_m$ with $m\in \Nat$, and      ${b}_1, \ldots, {b}_m \in C\langle E, F \rangle$,  there exists a rewrite rule in
 GTRS $R \langle E, F \rangle $ with left-hand side 
$\s({b}_1, \ldots, {b}_m)$. 
\end{df}
F\"ul\"op and V\'agv\"olgyi \cite{eatcs/FuloopV91}
 showed the following properties of $R\langle E, F  \rangle$.

\begin{prop}[{\cite[Proposition 3.1]{eatcs/FuloopV91}}]\label{elso}
   $R\langle E, F  \rangle$ is reduced.
\end{prop}
By Proposition \ref{elso} and  Proposition \ref{mozart}, we have the following result.
 \begin{prop}\label{rhcp}{\em \cite{eatcs/FuloopV91}}
   $R\langle E, F  \rangle$ is convergent.
\end{prop}

 \begin{prop}[{\cite[Proposition 3.2]{eatcs/FuloopV91}}] \label{masodik}
   For every $t\in W  $, we have $t \tred {R\langle E, F  \rangle}  t/_{\Theta\langle E, F  \rangle}$.
\end{prop}

\begin{prop} \label{beszur}
{\em \cite{eatcs/FuloopV91}} Every $t\in W $ has a unique $ R\langle E, F  \rangle$ normal form.
  For every $t\in W $, we have $t{\downarrow}_{R\langle E, F  \rangle}=
   t/_{\Theta\langle E, F  \rangle}$.
\end{prop}
\begin{proof} Let $t\in W  $. 
 By Proposition \ref{masodik}, 
 we have $t \tred {R\langle E, F  \rangle}  t/_{\Theta\langle E, F  \rangle}$. 
 By direct inspection of the rewrite rules of $R\langle E, F  \rangle$, we get that for each
$t\in W $,
$ t/_{\Theta\langle E, F  \rangle}$ is irreducible for $R\langle E, F  \rangle$. Accordingly $t/_{\Theta\langle E, F  \rangle}$ is a $\red {R\langle E, F  \rangle} $ normal form 
 of $t$.
By Proposition \ref{nulladik} and  Proposition \ref{rhcp},
 $t\in W $ has a unique normal form, 
 $t{\downarrow}_{R\langle E, F  \rangle}$.  Then we have $t{\downarrow}_{R\langle E, F  \rangle}=
   t/_{\Theta\langle E, F  \rangle}$.
\end{proof}

\begin{prop}[{\cite[Lemma 3.3]{eatcs/FuloopV91}}]\label{harmadik}
  For all $s\in \ts$
  and $t \in W $, 
  \begin{quote}
 if $s \tred {R\langle E, F  \rangle}  t/_{\Theta\langle E, F  \rangle}$,
  then $s \tthue E t$.
  \end{quote} 
\end{prop}

\begin{prop}[{\cite[Theorem 3.4]{eatcs/FuloopV91}}]\label{hatodik} 
{\em \cite{eatcs/FuloopV91}}
  $\tthue E =\tthue {R\langle E, F  \rangle} \cap \ts\times \ts$.
\end{prop}

\begin{prop}[{\cite[Theorem 3.5]{eatcs/FuloopV91}}]    \label{negyedik}
  For all $s, t \in \ts$,  
  \begin{quote} 
   $s \tthue E t$ if and only if $s{\downarrow}_{R\langle E, F  \rangle}= t{\downarrow}_{R\langle E, F  \rangle}$.
  \end{quote}
\end{prop}


Proposition \ref{beszur} and Proposition \ref{negyedik} imply the following proposition.
\begin{prop}\label{otodik}
  For all $s\in \ts$ and
  $t \in W $,
  \begin{quote}   
   if   $s \tthue E t$, then  $s \tred {R\langle E, F  \rangle}   t/_{\Theta\langle E, F  \rangle}$.
   \end{quote} 
\end{prop}


  By Proposition \ref{harmadik} and Proposition
  \ref{otodik}, we have the following result.
\begin{prop}\label{szazadik}
  For all $s\in \ts$   and $t \in W $,
  \begin{quote} 
  $s \tred {R\langle E, F  \rangle}  t/_{\Theta\langle E, F  \rangle}$
 if and only if  $s \tthue E t$.
 \end{quote} 
\end{prop}
\begin{sta}\label{zala}
If $R \langle E, F \rangle =\emptyset$, then $R \langle E, F \rangle $ is not total. 
\end{sta}
\begin{proof}
Let $\s\in \S_0$. There does not exists a rule in  $R \langle E, F \rangle $ with left-hand side $\s$. Consequently, 
 $R \langle E, F \rangle $ is not total. 
\end{proof}
 \begin{lem} \label{abecsikapu} 
If   \,$\bigcup W /_{\tthue E}=\ts$,
 then  
  $R\langle E, F  \rangle$ is
 total. 
\end{lem}
\begin{proof}
 Let  $\s\in \S_m$ with $m\in \Nat$, and let  $t_1/_{\Theta\langle E, F  \rangle}, \ldots,  t_m/_{\Theta\langle E, F  \rangle}$
 be  equivalence classes
of $\Theta\langle E, \; F \rangle$ for some $t_1, \ldots t_m \in W $.  By our assumption, there exists $s\in W $ such that  
$\s(t_1, \ldots, t_m)\tthue E s$. By Proposition \ref{negyedik}, 
$\s(t_1, \ldots, t_m){\downarrow}_{R\langle E, F  \rangle}= s{\downarrow}_{R\langle E, F  \rangle}$. By Proposition
\ref{beszur}, 
$s{\downarrow}_{R\langle E, F  \rangle}=   s/_{\Theta\langle E, F  \rangle}$. Thus 
$\s(t_1, \ldots, t_m){\downarrow}_{R\langle E, F  \rangle}= s/_{\Theta\langle E, F  \rangle}$.
Therefore, 
\begin{quote}
 $\s(t_1, \ldots, t_m)\tred {R\langle E, F  \rangle}
\s(t_1{\downarrow}_{R\langle E, F  \rangle}, \ldots, t_m{\downarrow}_{R\langle E, F  \rangle})= $ (by Proposition \ref{masodik})\newline
$\s( t_1/_{\Theta\langle E, F  \rangle}, \ldots,  t_m/_{\Theta\langle E, F  \rangle}) \red {R\langle E, F  \rangle} 
 s/_{\Theta\langle E, F  \rangle}$ \hspace{3 mm} (by direct inspection of the form of the rewrite rules of $R\langle E, F  \rangle$).
\end{quote}
By direct inspection of the rewrite rules of $R$, we get that
in the last step of the sequence of reductions we applied the rewrite rule \begin{quote}
  $\s( t_1/_{\Theta\langle E, F  \rangle}, \ldots,  t_m/_{\Theta\langle E, F  \rangle}) \dot{\rightarrow}  s/_{\Theta\langle E, F  \rangle}$
     in $R\langle E, F  \rangle$. \qedhere \end{quote}
\end{proof}

\begin{lem}\label{ezredik} If 
$R\langle E, F  \rangle$  is a total GTRS,  then  for each $s\in \ts$, 
there exists $t \in W $ such that 
$s{\downarrow}_ {R\langle E, F  \rangle}=t/_{\Theta\langle E, F  \rangle}$.
 \end{lem}
 \begin{proof}
 We proceed by induction on
$height(s)$.

{\em Base Case:} Let 
$height(s)=0$.
Then $s=\s$ for some $\s\in \S_0$. By our assumption that  $R\langle E, F  \rangle$ is  a total GTES and by Definitions \ref{szabaly}  and \ref{csurig},
there is a rewrite  rule
$\s\dot{\rightarrow}  \s/_{\Theta\langle E, F  \rangle}$ in $R\langle E, F  \rangle$. 
Then $\s\in W $ and $s{\downarrow}_ {R\langle E, F  \rangle} = \s/_{\Theta\langle E, F  \rangle}$.

{\em Induction Step:} 
Let   $k\in \Nat$, and let us assume that we have shown the
lemma for all integers less than or equal to $k$. Let $height(s)=k+1$.
Then $s=\s(s_1, \ldots,s_m)$ for some $\s \in \S_m$, $m\geq 1$,  and $s_1, \ldots,s_m\in \ts$.
By the induction hypothesis, for each $i \in \{\, 1, \ldots, m\,\}$, 
 there exist $t_i \in W   $ such that
 $s_i{\downarrow}_ {R\langle E, F  \rangle}= t_i/_{\Theta\langle E, F  \rangle}$.
 By our assumption that  $R\langle E, F  \rangle$  
  is  a total GTRS, there is a rewrite rule 
  with left-hand side     $\s( t_1/_{\Theta\langle E, F  \rangle}, \ldots,  t_m/_{\Theta\langle E, F  \rangle})$ in 
  $R\langle E, F  \rangle$.  
  By Definition \ref{szabaly}, this rule is of the form
  $$\s( t_1/_{\Theta\langle E, F  \rangle}, \ldots,  t_m/_{\Theta\langle E, F  \rangle}) \dot{\rightarrow}  \s(p_1, \ldots, p_m)/_{\Theta\langle E, F  \rangle}$$ 
 in $R\langle E, F  \rangle$ for some $p_i\in t_i/_{\Theta\langle E, F  \rangle}$  for every $i\in \{\, 1, \ldots, m\, \}$. 
 Moreover  $\s(p_1, \ldots, p_m)\in W$.
 Then 
 \begin{quote}
  $s=\s(s_1, \ldots,s_m)\tred {R\langle E, F  \rangle} s( t_1/_{\Theta\langle E, F  \rangle}, \ldots,  t_m/_{\Theta\langle E, F  \rangle})
  \red {R\langle E, F  \rangle}  \s(p_1, \ldots, p_m)/_{\Theta\langle E, F  \rangle}$.
  \end{quote} 
By direct inspection of  Definition \ref{szabaly}, we get that $\s(p_1, \ldots, p_m)/_{\Theta\langle E, F  \rangle}$ is irreducible for   $R\langle E, F \rangle $.
   Thus 
$s{\downarrow}_ {R\langle E, F  \rangle}=   \s(p_1, \ldots, p_m)/_{\Theta\langle E, F  \rangle}$. \qedhere
\end{proof}
The following proposition is a 
 consequence of Lemma \ref{abecsikapu}  and Lemma \ref{ezredik} and Proposition \ref{negyedik}.
 \begin{prop} \label{halle} 
  $\bigcup W /_{\tthue E}=\ts$ if and only if 
  $R\langle E, F  \rangle$ is
 total. 
\end{prop}
\begin{sta}\label{olvaseler} GTRS
  $ R \langle E, F  \rangle$ reaches a constant     ${b}\in C\langle E , F \rangle$ from   a ground term $t$ over the signature $\Sigma \cup C\langle E , F \rangle$
    if  and only if $t{\downarrow}_ {R\langle E, F \rangle} ={b}$. 
   \end{sta}
   \begin{proof}
   $(\Rightarrow)$
 The constant      ${b}\in C\langle E , F \rangle$  is irreducibe for  $ R \langle E, F  \rangle$ by direct inspection of  Definition \ref{szabaly}.   By Proposition \ref{rhcp}, 
      ${R\langle E, F  \rangle} $ is convergent.
 Accordingly
       ${b}\in C\langle E , F \rangle$ is the unique $ R \langle E, F  \rangle$-normal form of $t$.    
    
     $(\Leftarrow)$ This direction is trivial.
   \end{proof}
  \begin{df}\label{zug} \rm Let
\begin{quote} 
 $BSTEP\langle E, F  \rangle =\{\, ({b}, {c})\in C\langle E, F \rangle \times C\langle E, F \rangle \mid \mbox{  there exists  a rewrite rule }
 \s({b}_1, \ldots, {b}_m) \dot{\rightarrow} {b} \in R\langle E, F  \rangle \mbox{ and }  1\leq i \leq m \mbox{ such that } {c}={b}_i \, \}$.
 \end{quote}
 We call an element of  $BSTEP\langle E, F  \rangle$  a transition  backstep.
  \end{df}
       \begin{sta}\label{drrszam} Let  $n=\mathrm{size}(E)+\mathrm{size}(F)$. Then 
       \begin{quote}
 1. $|R\langle E, F   \rangle|\leq |W|\leq \mathrm{size}(W) \leq size (E \cup F) \leq \mathrm{size}(E) + \mathrm{size}(F)= n $. \newline
 2. $\mathrm{size}(R\langle E,  F \rangle)\leq \mathrm{size}(W)+ |W|\leq 2 n$\, and \, \newline 
3. $|BSTEP\langle E, F   \rangle|\leq \mathrm{size}(R\langle E,  F \rangle)\leq 2 n$.
\end{quote}
\end{sta} 
\begin{proof}
By direct inspection of  Definition \ref{szabaly} and by Statement \ref{meret}  we have Statement 1. 

We now show Statement 2. By Definition \ref{szabaly}, 
for each rule  $\s({b}_1, \ldots, {b}_m) \dot{\rightarrow} {b}\in R\langle E, F \rangle$, there exists a ground term $\s(t_1, \ldots, t_m)\in W$
such that    ${b}_i=t_i/_ {\Theta\langle E, F \rangle}$ for every $i\in \{\, 1, \ldots, m\, \}$.
Then we assign to the occurrence of    ${b}_i$ the root symbol of the ground term $t_i$  for every $i\in \{\, 1, \ldots, m\, \}$. Consequently $\mathrm{size}(R\langle E,  F \rangle)\leq \mathrm{size}(W)+ |W|$. By Statement 1, $\mathrm{size}(W)+ |W|\leq 2 n$.

By direct inspection of   Definition \ref{zug} and by Statement 2  we have Statement 3. 
\end{proof}
We now adopt the fast ground completion  algorithm of 
 F\"ul\"op and V\'agv\"olgyi \cite{eatcs/FuloopV91}
  that produces the GTRS $R\langle E, F  \rangle$.
 \begin{df}[{\cite[Algorithm 3.6]{eatcs/FuloopV91}}]\label{alma}\rm 
Algorithm Fast Ground Completion Introducing Additional Constants (FGC for short) 
\newline
Input: A  signature   $\S$,  GTES $E$ and GTES $F$ over $\Sigma$.
 \newline
Output:\begin{itemize}
\item 
 the relation $\rho\langle E, F \rangle$, 
 \item the  set of nullary symbols $C\langle E, F \rangle$, 
 \item 
 the GTRS $R\langle E, F \rangle$ stored without duplicates  in a  red-black tree $RBT1$,  and 
 \item the relation 
 $BSTEP\langle E, F  \rangle$ stored  without duplicates  in a  red-black tree $RBT2$.
  \end{itemize}
Data structures: 
The congruence closure algorithm appearing in Proposition 
\ref{lezaras},  gives each  $\rho$-class a unique name, which is an element of the set $\{\, 1, \ldots, |W|\, \}$.
We assign to each symbol of $\Sigma$ a unique  number in the set
\begin{quote}
 $\{\, |W|+1, \ldots, 2\cdot |W|\, \}$. 
  \end{quote} 
Let
\begin{quote}
$\ell= \lceil\mathrm{log} _2|W|\rceil+1$. 
  \end{quote} 
 For  each symbol $\s\in \Sigma\langle E, F \rangle $,  we  
pad the binary form of its number
with leading zeros to length  $\ell$. The binary sequence  obtained in this way is called 
the binary representation 
of  the  symbol $\s\in \Sigma\langle E, F \rangle $, and is  denoted as $\phi(\s)$.
Similarly,  for  each constant    ${b}\in  C\langle E, F \rangle  $, we  
pad the  binary form of its number
with leading zeros to length  $\ell$. The binary sequence  obtained in this way is called 
the binary representation 
of  the constant    ${b}\in  C\langle E, F \rangle  $, and is  denoted as $\phi({b})$.

The binary representation of a 
rule $\s({b}_1, \ldots, {b}_m) \dot{\rightarrow} {b}\in R\langle E, F \rangle$ is the string  $\phi(\s)\phi({b}_1), \ldots, \phi({b}_m)\phi({b})$. 
We store the binary representations of the  rewrite rules in $R\langle E, F \rangle$  without duplicates in the   red-black tree  $RBT1$ in lexicographic order.
 The binary representation of an ordered pair  $({b}, {c})\in BSTEP\langle E ,  F \rangle$ is the ordered pair $(\phi({b}), \phi({c}))$. 
We store the binary representations of 
the elements of  $BSTEP\langle E, F \rangle$ without duplicates in the   red-black tree  $RBT2$ in lexicographic order.
Here 
for every ordered pair 
$({b}, {c})\in BSTEP\langle E , F\rangle$, the search key is the first component $\phi({b})$, the second component, $\phi({c})$ is  satellite data.\newline
 {\bf var}  $i$, $m$: $\Nat$; $\s$:  $\Sigma$;    ${b},  {b}_1, \ldots, {b}_{mar_{\S}}: C \langle E,  F  \rangle$; $RBT1$, $RBT2$: red-black tree;
 \begin{lstlisting}
(*1. Create a subterm dag $ G=( V, A, \lambda)$ for  $E$ and $ F $.*)
(*By  Proposition \ref{lezaras}, compute the congruence  closure $\rho\langle E, F \rangle$  of the relation $\tau\langle E, F \rangle=\{\, (u, v)\mid (\hat{u}, \hat{v}) \in E\, \}.$*)
(* Compute $C\langle E, F \rangle= \{\, x/ _{\rho\langle E, F \rangle} \mid x \in V\,\}$.*)
(*2. \{ Computing   the GTRS $R\langle E, F \rangle$ \}*)
(*$R\langle E, F \rangle:=\emptyset$;*)
(*$BSTEP\langle E, F  \rangle:=\emptyset$;*)
(*For each $x\in V$ {\bf do}*)
  (*{\bf begin}*)
    (*$\s:= \lambda(x)$;*)
    (*$m := outdegree(x);$*)
    (*   ${b} := \mathrm{FIND}_{ \rho\langle E, F  \rangle  }(x)$;*)
    (*for each $1 \leq i \leq m $ {\bf do}*)
      (*   ${b}_i := \mathrm{FIND}_{  \rho\langle E, F  \rangle }(x/i)$;*)
    (*$R\langle E, F \rangle:=R\langle E, F \rangle\cup \{\, \s({b}_1, \ldots, {b}_m) \dot{\rightarrow} {b}\, \}$;*)
    (*$BSTEP\langle E, F  \rangle:=BSTEP\langle E, F  \rangle\cup \{\, ({b}_1, {b}), \ldots, ({b}_m, {b})\, \}$;*)
  (*{\bf end}*)
 \end{lstlisting}
 \end{df}
 By Statement \ref{drrszam} and Proposition \ref{redblack}, the time complexity of 
storing $R\langle E, F   \rangle$ without duplicates  in the  red-black tree  $RBT1$
is $O(n\, \mathrm{log}\, n)$, and 
the time complexity of storing 
$BSTEP\langle E, F   \rangle$ without duplicates  in the   red-black tree  $RBT2$ is $O(n\, \mathrm{log}\, n)$, where $n=\mathrm{size}(E)+\mathrm{size}(F)$.

\begin{prop}[{\cite[Algorithm 3.6]{eatcs/FuloopV91}}] \label{kecskemet}
  Algorithm FGC takes as input a signature $\Sigma$, and  GTESs  $E$ and $F$  over  $\S$. Let $n=\mathrm{size}(E)+\mathrm{size}(F)$. 

1. Algorithm FGC  produces  in $O(n\, \mathrm{log}\, n)$ time the relation $\rho\langle E, F \rangle$ and 
the  set of nullary symbols $C\langle E, F \rangle$.
 
 2.   In $O(n\, \mathrm{log}\, n)$ time,   Algorithm FGC computes
  the  GTRS $R\langle E, F  \rangle$ over $\Sigma \cup C\langle E, F \rangle$,  
  and stores $R\langle E, F \rangle$ without duplicates  in the  lexicographic order $<_{lex}$ on binary sequences  
 in the   red-black tree  $RBT1$. 
 
 3. In $O(n\, \mathrm{log}\, n)$ time,  Algorithm FGC computes the relation
 $BSTEP\langle E, F  \rangle$, and  
 stores $BSTEP\langle E, F   \rangle$  without duplicates  
 in the   red-black tree  $RBT2$. The search key is the first component, the second component  is  satellite data.
  \end{prop}
\begin{proof} 
F\"ul\"op and V\'agv\"olgyi \cite{eatcs/FuloopV91} showed that we can construct the GTRS $R\langle E, F \rangle$ in  $O(n \,\mathrm{log}\, n)$ time.
Hence we can construct the  relation $BSTEP\langle E, F  \rangle $ in  $O(n \,\mathrm{log}\, n)$ time as well. 
In addition to the result of  F\"ul\"op and V\'agv\"olgyi \cite{eatcs/FuloopV91}, we now compute the cost of storing  the GTRS $R\langle E, F \rangle$,
and the cost of storing the  relation $BSTEP\langle E, F  \rangle $.
 First we show Statement 1. 
 Algorithm FGC produces the  set of nullary symbols $C\langle E, F \rangle$. By Definition \ref{szeder}, 
for every  equivalence class $\vartheta\in \Theta\langle E, F \rangle$, 
Algorithm FGC gives a unique name  for $\vartheta$ which is the  unique name of the $\rho\langle E, F  \rangle $   equivalence class $\eta$ such that 
$\vartheta= \{\, \widehat{x} \mid x\in \eta\, \}$.

We now show Statement 2. 
Algorithm FGC produces the GTRS $R\langle E, F \rangle$ apart
from that, 
 for every  equivalence class $\vartheta\in \Theta\langle E, F \rangle$,
Algorithm FGC writes the unique name  for $\vartheta$ \cite{eatcs/FuloopV91}.  
 By Statement \ref{drrszam}, 
 $\mathrm{size}(R\langle E,  F \rangle)\leq 2\cdot  |W |\leq 2 n$. 
Then by Proposition \ref{redblack}, 
 the time-complexity of inserting the  elements
of $R\langle E\cup F ,  \emptyset \rangle$ into an empty  red-black tree, not allowing  duplicate values in the resulting red-black tree
$RBT1$, is $O(n \,\mathrm{log}\, n)$,

We now show Statement 3.   
In the relation  $BSTEP\langle E, F  \rangle \subseteq C\langle E, F \rangle \times C\langle E, F \rangle$, we collect all pairs $({b}_i , {b})$ such that there is a rewrite rule 
 $\s({b}_1, \ldots, {b}_m) \dot{\rightarrow}{b}$ in  $R\langle E, F \rangle$.
   By Statement \ref{drrszam}, $|BSTEP\langle E, F\rangle|\leq 2 n$.
Then by Proposition \ref{redblack}, 
the time-complexity of inserting the  elements
of  $BSTEP\langle E, F \rangle$ into an empty  red-black tree, not allowing  duplicate values in the resulting red-black tree, is $O(n \,\mathrm{log}\, n)$.
 \end{proof}

We give the GTRS  $R\langle E, F \rangle  $ as  input to an algorithm by giving the  red-black tree $RBT1$ 
output by Algorithm FGC storing the GTRS $R\langle E,  F \rangle$. 
  We now define an algorithm deciding whether $R\langle E,  F \rangle$ is  total.

 \begin{df}\label{dectot}\rm Let 
 $E$ and $F$ be GTESs over  a signature $\S$.
   \newline
Algorithm Deciding whether  GTRS $R\langle E,  F \rangle$   is Total (DTOT for short).  
 \newline
Input: The signature $\Sigma$, the Boolean value $FLAG\langle E, F \rangle$,   the set of constants $ C\langle E, F \rangle$, 
and  the  red-black tree $RBT$ storing the GTRS $R\langle E,  F \rangle$. 
  \newline
Output: Boolean.
\{\ $\mathrm{true}$ if $R\langle E,  F \rangle$ total. Otherwise $\mathrm{false}$. \ \} \newline
{\bf var} $i$, $m, count: \Nat$;
\newline
If $FLAG\langle E, F\rangle= \mathrm{false}$, then    return $\mathrm{false}$. If $R\langle E, F \rangle=\emptyset$, then return $\mathrm{false}$.
Otherwise, traverse the red-black tree $RBT$ in in-order  traversal, and read all binary sequences  stored in the red-black tree $RBT$
in the  lexicographic order $<_{lex}$ on binary sequences.
 For every $m\in \Nat$, $\s\in \Sigma\langle R\langle E, F \rangle _m$,  read the rules with $\s$ appering on the left-hand side 
one after another. 
During this process,
\begin{lstlisting}
(*{\bf for every} $\s\in \S\langle E, F \rangle$ {\bf do}*)
  (*{\bf begin}*)
    (*count  all rules with $\s$ appearing on the left-hand side, and store the result in the variable $count$;*)
    (*$m:=\mathrm{arity}(\s)$;*)
    (*   ${\bf if }\; count< |C\langle E, F \rangle|^m$*)
      (*{\bf then} {\bf return} $\mathrm{false}$*) 
  (*{\bf end}*)
(*   ${\bf output}(\mathrm{true})$*)    
\end{lstlisting}
\end{df}

\begin{lem}\label{decidetotal} 
Algorithm DTOT decides  in $O(n)$ time  whether  $R\langle E,  F\rangle$ is total, where $n=\mathrm{size}(E)+\mathrm{size}(F)$.
\end{lem}
\begin{proof}
If  $R\langle E, F \rangle=\emptyset$, then by Statement \ref{zala}, we are done. 

Assume that $R\langle E, F \rangle\neq\emptyset$. If $FLAG\langle E, F\rangle= \mathrm{false}$, that is,  $|\S\langle E, F \rangle|< |\S|$,
 then 
 $R\langle E,  F \rangle$ is not total.   
 Assume that $FLAG\langle E, F\rangle= \mathrm{true}$, that is,  $|\S\langle E, F \rangle|=|\S|$. Then for each  symbol $\s\in \S$, 
we  count the rewrite rules in $R\langle E,  F \rangle$ with $\s$ appearing in the left-hand side and store the result in the variable $count$.
If  there exists a symbol $\s\in \S_m$, $m\in \Nat$, 
such that  for $\s$, $count <   |\langle  C\langle E, F  \rangle\rangle|^m $, 
 then  there   exist constants      ${b}_1, \ldots, {b}_m \in C\langle E, F  \rangle$ such that 
 there does not exist a  rewrite rule with left-hand side
 $\s({b}_1, \ldots, {b}_m) $
 in  $R\langle E, F  \rangle$. 
  Therefore, $R\langle E,  F \rangle$ is not total. 
    Conversely if $R\langle E,  F \rangle$ is not total, then there exists a symbol $\s\in \S_m$ and constants    ${b}_1, \ldots, {b}_m \in C\langle E, F  \rangle$ such that 
 there does not exist   a rewrite rule with left-hand side  $\s({b}_1, \ldots, {b}_m) $
 in  $R\langle E, F  \rangle$. Then for $\s$, $ count <   |\langle  C\langle E, F  \rangle\rangle|^m $.  
  Consequently if $R\langle E,  F\rangle$ is  total, then  Algorithm DTOT   outputs true; otherwise it returns false.

 We now consider the  time complexity of the Algorithm DTOT. 
 If $|\S\langle E, F \rangle|< |\S|$  or  $R\langle E, F \rangle=\emptyset$,   then Algorithm DTOT takes less than $O(n)$ time.  
 Assume that $|\S\langle E, F \rangle|= |\S|$ and  $R\langle E, F \rangle\neq \emptyset$.
  Then 
Algorithm DTOT traverse the red-black tree $RBT$ in in-order  traversal, and it reads all binary sequences  stored in the red-black tree $RBT$
in the  lexicographic order $<_{lex}$ on binary sequences. For every $\s\in \Sigma\langle R\langle E, F \rangle  \rangle$, it reads the rules with $\s$ appering on the left-hand side 
one after another, and 
 counts the rewrite rules in $R\langle E,  F\rangle$ with $\s$ appearing in the left-hand side  for all symbols $\s\in \Sigma$. This takes $O(n)$ time.

  If there exists a symbol $\s \in \Sigma$ such that  $\s$ does   not appear in $R\langle E,  F\rangle$,  then Algorithm DTOT outpouts false, and 
  does not evaluate the expression  $count< |C\langle E, F\rangle|^m$.
   Otherwise, for every  symbol $\s \in \Sigma_m $, $m\in \Nat$,   we carry out  $m$ multiplications when computing $|C\langle E, F\rangle|^m$. 
Because each symbol 
$\s\in \Sigma\langle R\langle E, F \rangle  \rangle_m $, $m\in \Nat$,  has  $m $ constants in its arguments in the left-hand side of some rewrite rule in 
$R\langle E, F\rangle  $.  
The number of multiplications   is less than or equal to 
$\sum _{m\in \Nat} \sum_ {\s\in \Sigma\langle R\langle E, F  \rangle \rangle_m} m\leq $
$\mathrm{size}(R\langle E,  F\rangle)\leq 2 n$
by Statement \ref{drrszam}. 
Consequently, the number of multiplications is less than or equal to $2  n$.
  The number of comparisons of nonnegative integers is $|\Sigma|\leq n$.
Therefore the  time complexity of algorithm DTOT is
 $O(n )$. 
           \end{proof}

\section{Running Examples Continued}\label{peldak}
We now illustrate  the concepts of Section \ref{gyors} by continuing our running examples. 
We compute the sets of constants $ C\langle E, F \rangle$,
$C\langle F, E\rangle$,  $C\langle E \cup F, \emptyset\rangle$, and the GTRSs  $ R \langle E, F \rangle$,  $ R \langle F, E \rangle$, 
$R\langle E\cup F,  \emptyset \rangle$. 

Recall that $E$ and $F$ are GTESs over a
 signature   $\S$, and  $W$ stands for
$ST\langle E \cup F\rangle$.
\begin{exa}\label{1oneegyfolytatas} \rm We continue Example \ref{1oneegy}.
\begin{quote}
    $C\langle E, F\rangle=\{\, \{\, \#\,\}$, $\{\,   \$\, \}$,  $\{\, f(\#), \;g(\#) \, \}$,
 $\{\, f(\$)\,\}$,  $\{\, g(\$)\,\}\, \}$.
\end{quote}

 \begin{quote}     
$R\langle E, F\rangle=\{\,  
\# \dot{\rightarrow}   \#/_{\Theta\langle E ,F \rangle}, \;
\$ \dot{\rightarrow}   \$/_{\Theta\langle E ,F \rangle}, \;\newline
f(   \#/_{\Theta\langle E ,F \rangle}) \dot{\rightarrow}
  f(\#)/_{\Theta\langle E ,F \rangle}, \;
f(   \$/_{\Theta\langle E ,F \rangle}) \dot{\rightarrow}
  f(\$)/_{\Theta\langle E ,F \rangle}, \;\newline
g(\#/_{\Theta\langle E ,F \rangle}) \dot{\rightarrow}
  f(\#)/_{\Theta\langle E ,F \rangle}, \;
 g(\$/_{\Theta\langle E ,F \rangle}) \dot{\rightarrow}, \;
  g(\$)/_{\Theta\langle E ,F \rangle}\;\}$,
\end{quote}

\begin{quote}
  $C\langle F, E\rangle=\{\, \{\, \#\,\}$, $\{\,   \$\, \}$,
  $\{\, f(\#)  \, \} $,  $\{\, g(\#) \, \}$,
 $\{\, f(\$), \;g(\$)\, \}\}$, 
\end{quote}

 \begin{quote}     
   $R\langle F, E\rangle=\{\,
\# \dot{\rightarrow}   \#/_{\Theta\langle F ,E \rangle}, \;
\$ \dot{\rightarrow}   \$/_{\Theta\langle F ,E \rangle}, \;\newline
f(   \#/_{\Theta\langle F ,E \rangle}) \dot{\rightarrow}
  f(\#)/_{\Theta\langle F ,E \rangle}, \;
f(   \$/_{\Theta\langle F ,E \rangle}) \dot{\rightarrow}
  f(\$)/_{\Theta\langle F ,E \rangle}, \;\newline
  g(   \#/_{\Theta\langle F ,E \rangle}) \dot{\rightarrow}
  g(\#)/_{\Theta\langle F ,E \rangle}, \;
   g(\$/_{\Theta\langle F ,E \rangle}) \dot{\rightarrow}
  f(\$)/_{\Theta\langle F ,E \rangle}\;\}$,
\end{quote}

\begin{quote}
  $C\langle E \cup F, \emptyset\rangle=\{\, \{\, \#\,\}$, $\{\,   \$\, \}$,
  $\{\, f(\#), \;  g(\#) \, \}$,
 $\{\, f(\$), \; \;g(\$)\, \}\, \}$,
\end{quote}

 \begin{quote}     
   $R\langle E \cup F, \emptyset\rangle=\{\,
\# \dot{\rightarrow}   \#/_{\Theta\langle E \cup F ,\emptyset \rangle}, \;
\$ \dot{\rightarrow}   \$/_{\Theta\langle E \cup F ,\emptyset \rangle}, \;
\newline
f(   \#/_{\Theta\langle E \cup F ,\emptyset  \rangle}) \dot{\rightarrow}
  f(\#)/_{\Theta\langle E\cup F, \emptyset \rangle}, \;
f(   \$/_{\Theta\langle E \cup F ,\emptyset  \rangle}) \dot{\rightarrow}
  f(\$)/_{\Theta\langle E \cup F ,\emptyset \rangle}, \;
\newline
g(   \#/_{\Theta\langle E \cup F ,\emptyset  \rangle}) \dot{\rightarrow}
  f(\#)/_{\Theta\langle E \cup F ,\emptyset  \rangle},
g(   \$/_{\Theta\langle E \cup F ,\emptyset  \rangle}) \dot{\rightarrow}
  f(\$)/_{\Theta\langle E \cup F ,\emptyset  \rangle} \,\}$.
\end{quote}

\end{exa}

\begin{exa}\label{2szepenfolytatas} \rm We continue Example \ref{2szepen}.
\begin{quote}
    $C\langle E, F\rangle= \{\,  \{\, \#, \;  f^2(\#)\}, \;\{\, f(\#), \;f^3(\#)\, \}\, \}$,
\end{quote}

 \begin{quote}     
$R\langle E, F\rangle=\{\,  
\# \dot{\rightarrow}   \#/_{\Theta\langle E ,F \rangle}, \;
f(   \#/_{\Theta\langle E ,F \rangle}) \dot{\rightarrow}
  f(\#)/_{\Theta\langle E ,F\rangle}, \;
f(  f(\#)/_{\Theta\langle E ,F \rangle})
\dot{\rightarrow}
\#/_{\Theta\langle E ,F \rangle} \, 
 \}$,
\end{quote}

\begin{quote}
  $C\langle F, E\rangle=\{\, \{\, \#, \;f^3(\#) \,\}$, 
  $\{\, f(\#)\, \}$,
 $\{\, f^2(\#)\, \}\}$, 
\end{quote}

 \begin{quote}     
$R\langle F, E\rangle=\{\,  
\# \dot{\rightarrow}   \#/_{\Theta\langle F ,E \rangle}, \;
f(   \#/_{\Theta\langle F ,E \rangle}) \dot{\rightarrow}
  f(\#)/_{\Theta\langle F ,E\rangle}, \;
f(f(\#)/_{\Theta\langle F ,E\rangle})
\dot{\rightarrow} f^2(\#)/_{\Theta\langle F ,E\rangle}, \;\;\newline
f(f^2(\#)/_{\Theta\langle F ,E\rangle})
\dot{\rightarrow} \#/_{\Theta\langle F ,E\rangle}\, \}$,
\end{quote}

 \begin{quote}     
 $C\langle E\cup F, \emptyset\rangle=\{\, \{\, \#, \; f(\#), \;  f^2(\#), \;f^3(\#)\, \}\, \}=\{\, W\, \}$,  
     \end{quote}

 \begin{quote}     
$R\langle E\cup F, \emptyset \rangle=\{\,  
\# \dot{\rightarrow}   \#/_{\Theta\langle E \cup F,  \emptyset  \rangle}, \;f( \#/_{\Theta\langle E \cup F,  \emptyset \rangle}) \dot{\rightarrow}
  \#/_{\Theta\langle E \cup F,  \emptyset \rangle}
 \, \}$.
\end{quote}
We note that $R\langle E, F\rangle$, $R\langle F, E\rangle$, and 
 $R\langle E\cup F, \emptyset\rangle$ 
 are total.
\end{exa}

\begin{exa}\label{3szilvafolytatas} \rm We continue Example \ref{3szilva}.

\begin{quote}
    $C\langle E, F\rangle= \{\,     
     \{\, \#\, \},  \;  \newline \{\, \$, \;  f(\$, \$), \;f(\$, f(\$, \$)) , \; 
    f(f(\$, \$), \$) , \; 
   f(f(\$,\$), f(\$, \$))       \,  \}, \;\newline 
  \{\,f(\#, \#), \;   f(\#, f(\#, \#)), \;  f(f(\#, \#), \#), \; f(f(\#,\#), f(\#, \#))   \, \},  
   \newline\;\{\,\mathsterling, \;  f(\#, \$), \; f(\#, \mathsterling), \;   f(\$, \#), \; 
  f(\$, \mathsterling), \; 
  f(\mathsterling, \#), \;   f(\mathsterling, \$), \; f(\mathsterling, \mathsterling), \; \newline
   f(\$, f(\#, \#)), \;   f(\mathsterling, f(\#, \#)), \;
   f(f(\#, \#), \$), \; f(f(\#, \#), \mathsterling), \;\newline
  f(\#, f(\$, \$)), \;   f(\mathsterling, f(\$, \$)), \; 
   f(f(\$, \$), \#), \; f(f(\$, \$), \mathsterling)
  \, \}
    \, 
  \}$,
  \end{quote}

 \begin{quote}     
$R\langle E, F\rangle=\{\,
 \# \dot{\rightarrow} \#/_{\Theta\langle E ,F\rangle}, \;\;
\$ \dot{\rightarrow} \$/_{\Theta\langle E ,F\rangle}, \;\; 
\mathsterling \dot{\rightarrow} \mathsterling/_{\Theta\langle E ,F\rangle}, \;\; \newline  
 f( \#/_{\Theta\langle E ,F\rangle}, \#/_{\Theta\langle E ,F\rangle}) \dot{\rightarrow} 
  f(\#, \#)/_{\Theta\langle E ,F\rangle}, \;\;
   f( \#/_{\Theta\langle E ,F\rangle}, \$/_{\Theta\langle E ,F\rangle}) \dot{\rightarrow} 
 \mathsterling/_{\Theta\langle E ,F\rangle}, \;\;  \newline   
   f( \#/_{\Theta\langle E ,F\rangle}, \mathsterling/_{\Theta\langle E ,F\rangle}) \dot{\rightarrow} 
 \mathsterling/_{\Theta\langle E ,F\rangle}, \;\;  
 f( \$/_{\Theta\langle E ,F\rangle}, \#/_{\Theta\langle E ,F\rangle}) \dot{\rightarrow} 
 \mathsterling/_{\Theta\langle E ,F\rangle}, \;\;\newline  
   f( \$/_{\Theta\langle E ,F\rangle}, \$/_{\Theta\langle E ,F\rangle}) \dot{\rightarrow} 
 \$/_{\Theta\langle E ,F\rangle}, \;\;
   f( \$/_{\Theta\langle E ,F\rangle}, \mathsterling/_{\Theta\langle E ,F\rangle}) \dot{\rightarrow} 
 \mathsterling/_{\Theta\langle E ,F\rangle}, \;\;\newline 
  f( \mathsterling/_{\Theta\langle E ,F\rangle}, \#/_{\Theta\langle E ,F\rangle}) \dot{\rightarrow} 
 \mathsterling/_{\Theta\langle E ,F\rangle}, \;\;
   f( \mathsterling/_{\Theta\langle E ,F\rangle}, \$/_{\Theta\langle E ,F\rangle}) \dot{\rightarrow} 
 \mathsterling/_{\Theta\langle E ,F\rangle}, \;\;\newline  
  f( \mathsterling/_{\Theta\langle E ,F\rangle}, \mathsterling/_{\Theta\langle E ,F\rangle}) \dot{\rightarrow} 
 \mathsterling/_{\Theta\langle E ,F\rangle}, \;\; 
    f(\#/_{\Theta\langle E ,F\rangle}, f(\#, \#)/_{\Theta\langle E ,F\rangle})\dot{\rightarrow} 
 f(\#, \#)/_{\Theta\langle E ,F\rangle}, \;\;\newline  
   f(f(\#, \#)/_{\Theta\langle E ,F\rangle}, \#/_{\Theta\langle E ,F\rangle})\dot{\rightarrow} 
 f(\#, \#)/_{\Theta\langle E ,F\rangle}, \;\; \,
   f(f(\#,\#)/_{\Theta\langle E ,F\rangle}, f(\#, \#))/_{\Theta\langle E ,F\rangle}\dot{\rightarrow} 
 f(\#, \#)/_{\Theta\langle E ,F\rangle}  \;\; \newline
       f(\$/_{\Theta\langle E ,F\rangle}, f(\#, \#)/_{\Theta\langle E ,F\rangle})\dot{\rightarrow} 
\mathsterling/_{\Theta\langle E ,F\rangle}, \;\;  
 f(\mathsterling /_{\Theta\langle E ,F\rangle}, f(\#, \#)/_{\Theta\langle E ,F\rangle})\dot{\rightarrow} 
 \mathsterling/_{\Theta\langle E ,F\rangle}, \;\; \newline
  f(f(\#, \#)/_{\Theta\langle E ,F\rangle}, \$/_{\Theta\langle E ,F\rangle})\dot{\rightarrow} 
 \mathsterling/_{\Theta\langle E ,F\rangle}, \;\; 
  f(f(\#, \#)/_{\Theta\langle E ,F\rangle}, \mathsterling /_{\Theta\langle E ,F\rangle})\dot{\rightarrow}  \mathsterling /_{\Theta\langle E ,F\rangle}
       \, \}$.

\end{quote}

\begin{quote}  
 
  $C\langle F, E\rangle= \{\,     
     \{\, \$\, \},  \;  \newline \{\, \#, \;  f(\#, \#), \,  f(\#, f(\#, \#)) , \;  f(f(\#, \#), \#) , \;    f(f(\#,\#), f(\#, \#))
   \,  \}, \;\newline 
  \{\,f(\$, \$), \;   f(\$, f(\$, \$)), \;  f(f(\$, \$), \$), \; f(f(\$,\$), f(\$, \$))   \, \},  
   \newline\;\{\,\mathsterling, \;  f(\$, \#), \; f(\$, \mathsterling), \;   f(\#, \$), \; 
  f(\#, \mathsterling), \; 
  f(\mathsterling, \$), \;   f(\mathsterling, \#), \; f(\mathsterling, \mathsterling), \; \newline
   f(\#, f(\$, \$)), \;   f(\mathsterling, f(\$, \$)), \;
   f(f(\$, \$), \#), \; f(f(\$, \$), \mathsterling), \;\newline
  f(\$, f(\#, \#)), \;   f(\mathsterling, f(\#, \#)), \; 
   f(f(\#, \#), \$), \; f(f(\#, \#), \mathsterling)
  \, \}
    \, 
  \}$,
  
  \end{quote}

 \begin{quote}     
$R\langle F, E\rangle=\{\,
 \$ \dot{\rightarrow} \$/_{\Theta\langle F ,E\rangle}, \;\;
\# \dot{\rightarrow} \#/_{\Theta\langle F ,E\rangle}, \;\; 
\mathsterling \dot{\rightarrow} \mathsterling/_{\Theta\langle F ,E\rangle}, \;\; \newline  
 f( \$/_{\Theta\langle F ,E\rangle}, \$/_{\Theta\langle F ,E\rangle}) \dot{\rightarrow} 
  f(\$, \$)/_{\Theta\langle F ,E\rangle}, \;\;
   f( \$/_{\Theta\langle F ,E\rangle}, \#/_{\Theta\langle F ,E\rangle}) \dot{\rightarrow} 
 \mathsterling/_{\Theta\langle F ,E\rangle}, \;\;  \newline   
   f( \$/_{\Theta\langle F ,E\rangle}, \mathsterling/_{\Theta\langle F ,E\rangle}) \dot{\rightarrow} 
 \mathsterling/_{\Theta\langle F ,E\rangle}, \;\;  
 f( \#/_{\Theta\langle F ,E\rangle}, \$/_{\Theta\langle F ,E\rangle}) \dot{\rightarrow} 
 \mathsterling/_{\Theta\langle F ,E\rangle}, \;\;\newline  
   f( \#/_{\Theta\langle F ,E\rangle}, \#/_{\Theta\langle F ,E\rangle}) \dot{\rightarrow} 
 \#/_{\Theta\langle F ,E\rangle}, \;\;
   f( \#/_{\Theta\langle F ,E\rangle}, \mathsterling/_{\Theta\langle F ,E\rangle}) \dot{\rightarrow} 
 \mathsterling/_{\Theta\langle F ,E\rangle}, \;\;\newline 
  f( \mathsterling/_{\Theta\langle F ,E\rangle}, \$/_{\Theta\langle F ,E\rangle}) \dot{\rightarrow} 
 \mathsterling/_{\Theta\langle F ,E\rangle}, \;\;
   f( \mathsterling/_{\Theta\langle F ,E\rangle}, \#/_{\Theta\langle F ,E\rangle}) \dot{\rightarrow} 
 \mathsterling/_{\Theta\langle F ,E\rangle}, \;\;\newline  
  f( \mathsterling/_{\Theta\langle F ,E\rangle}, \mathsterling/_{\Theta\langle F ,E\rangle}) \dot{\rightarrow} 
 \mathsterling/_{\Theta\langle F ,E\rangle}, \;\; 
    f(\$/_{\Theta\langle F ,E\rangle}, f(\$, \$)/_{\Theta\langle F ,E\rangle})\dot{\rightarrow} 
 f(\$, \$)/_{\Theta\langle F ,E\rangle}, \;\;\newline  
   f(f(\$, \$)/_{\Theta\langle F ,E\rangle}, \$/_{\Theta\langle F ,E\rangle})\dot{\rightarrow} 
 f(\$, \$)/_{\Theta\langle F ,E\rangle}, \;\; \,
   f(f(\$,\$)/_{\Theta\langle F ,E\rangle}, f(\$, \$))/_{\Theta\langle F ,E\rangle}\dot{\rightarrow} 
 f(\$, \$)/_{\Theta\langle F ,E\rangle}  \;\; \newline
       f(\#/_{\Theta\langle F ,E\rangle}, f(\$, \$)/_{\Theta\langle F ,E\rangle})\dot{\rightarrow} 
 \mathsterling/_{\Theta\langle F ,E\rangle}, \;\;  
 f(\mathsterling /_{\Theta\langle F ,E\rangle}, f(\$, \$)/_{\Theta\langle F ,E\rangle})\dot{\rightarrow} 
 \mathsterling/_{\Theta\langle F ,E\rangle}, \;\; \newline
  f(f(\$, \$)/_{\Theta\langle F ,E\rangle}, \#/_{\Theta\langle F ,E\rangle})\dot{\rightarrow} 
 \mathsterling/_{\Theta\langle F ,E\rangle}, \;\; 
  f(f(\$, \$)/_{\Theta\langle F ,E\rangle}, \mathsterling /_{\Theta\langle F ,E\rangle})\dot{\rightarrow}  \mathsterling /_{\Theta\langle F ,E\rangle}
       \, \}$.

\end{quote}

\begin{quote}  
  
       $C\langle E \cup F, \emptyset\rangle=  \{\,
       \{\, \#, \;\$, \;\mathsterling, \;\newline
 f(\#, \#), \; f(\#, \$), \; f(\#, \mathsterling), \; \newline
 f(\$, \#), \;   f(\$, \$), \; f(\$, \mathsterling), \; \newline
  f(\mathsterling, \#), \;   f(\mathsterling, \$), \; f(\mathsterling, \mathsterling), \;\newline
    f(\#, f(\#, \#)) , \;  f(f(\#, \#), \#) , \; 
   f(f(\#,\#), f(\#, \#)), \;
  \;\newline  
  f(\$, f(\$, \$)) , \;      f(f(\$, \$), \$) , \; 
   f(f(\$,\$), f(\$, \$)), \;\newline 
   f(\$, f(\#, \#)), \;   f(\mathsterling, f(\#, \#)), \;
   f(f(\#, \#), \$), \; f(f(\#, \#), \mathsterling), \;\newline
  f(\#, f(\$, \$)), \;   f(\mathsterling, f(\$, \$)), \; 
   f(f(\$, \$), \#), \; f(f(\$, \$), \mathsterling)
   \, \}
     \, \}$,
   \end{quote}

\begin{quote}     
$R\langle E \cup F, \emptyset\rangle=\{\,  
\# \dot{\rightarrow}   \#/_{\Theta\langle   E\cup F, \emptyset\rangle}, \;
\$ \dot{\rightarrow}   \#/_{\Theta\langle   E\cup F, \emptyset\rangle}, \;
\mathsterling \dot{\rightarrow}   \#/_{\Theta\langle   E\cup F, \emptyset\rangle}, \;
f(  \#/_{\Theta\langle   E\cup F, \emptyset\rangle}, \#/_{\Theta\langle   E\cup F, \emptyset\rangle})\dot{\rightarrow}  \#/_{\Theta\langle   E\cup F, \emptyset\rangle}\, \}$.
\end{quote}

  We observe  that $R\langle E, F\rangle$, $R\langle F, E\rangle$, and 
 $R\langle E\cup F, \emptyset \rangle$ 
 are total.
  
\end{exa}

  \begin{exa}\label{4u2staringfolytatas}\rm 
  We now continue Example \ref{4u2staring}.  
  \begin{quote}  
  $C\langle E, F\rangle=\{\,   \#/_{\Theta\langle E ,F\rangle},  \;
     \flat/_{\Theta\langle E ,F\rangle }\, \}$,
  \end{quote}
   \begin{quote}     
$R\langle E, F\rangle=\{\, 
\# \dot{\rightarrow}   \#/_{\Theta\langle E ,F\rangle}, \;
\$ \dot{\rightarrow}   \#/_{\Theta\langle E ,F\rangle}, \;
\mathsterling \dot{\rightarrow}   \#/_{\Theta\langle E ,F\rangle}, \;
\flat \dot{\rightarrow}   \flat/_{\Theta\langle E ,F\rangle}, \;\newline
   f(   \#/_{\Theta\langle E ,F\rangle},    \#/_{\Theta\langle E ,F\rangle}) \dot{\rightarrow}
 \#/_{\Theta\langle E ,F\rangle}, \; f(   \#/_{\Theta\langle E ,F\rangle},    \flat/_{\Theta\langle E ,F\rangle}) \dot{\rightarrow}
 \#/_{\Theta\langle E ,F\rangle}, \; 
   f(   \flat/_{\Theta\langle E ,F\rangle},    \#/_{\Theta\langle E ,F\rangle}) \dot{\rightarrow}
 \#
   \, \}$.
  \end{quote}

\begin{quote}  
$C\langle F, E\rangle=\{\,   \#/_{\Theta\langle F ,E\rangle},
 \$/_{\Theta\langle F ,E\rangle},  \mathsterling/_{\Theta\langle F ,E\rangle}\, \}$,
    \end{quote}

\begin{quote}     
$R\langle F, E\rangle=\{\,  
\# \dot{\rightarrow}   \#/_{\Theta\langle F,E\rangle}, \;
\$ \dot{\rightarrow}   \$/_{\Theta\langle F,E\rangle}, \;
\mathsterling \dot{\rightarrow}   \mathsterling/_{\Theta\langle F,E\rangle}, \;
\flat \dot{\rightarrow}   \mathsterling/_{\Theta\langle F,E\rangle}, \;\newline
   f(   \#/_{\Theta\langle F,E\rangle},    \#/_{\Theta\langle F,E\rangle}) \dot{\rightarrow}
 \mathsterling/_{\Theta\langle F,E\rangle}, \;
 f(   \#/_{\Theta\langle F,E\rangle},    \$/_{\Theta\langle F,E\rangle}) \dot{\rightarrow}
 \mathsterling/_{\Theta\langle F,E\rangle}, \;
 f(   \#/_{\Theta\langle F,E\rangle},    \mathsterling /_{\Theta\langle F,E\rangle}) \dot{\rightarrow}
 \mathsterling/_{\Theta\langle F,E\rangle}, \;
 \newline 
  f(   \$/_{\Theta\langle F,E\rangle},    \#/_{\Theta\langle F,E\rangle}) \dot{\rightarrow}
 \mathsterling/_{\Theta\langle F,E\rangle}, \;
 f(   \$/_{\Theta\langle F,E\rangle},    \$/_{\Theta\langle F,E\rangle}) \dot{\rightarrow}
 \mathsterling/_{\Theta\langle F,E\rangle}, \;
 f(   \$/_{\Theta\langle F,E\rangle},    \mathsterling /_{\Theta\langle F,E\rangle}) \dot{\rightarrow}
 \mathsterling/_{\Theta\langle F,E\rangle}, \;
 \newline 
  f(   \mathsterling/_{\Theta\langle F,E\rangle},    \#/_{\Theta\langle F,E\rangle}) \dot{\rightarrow}
 \mathsterling/_{\Theta\langle F,E\rangle}, \;
 f(   \mathsterling/_{\Theta\langle F,E\rangle},    \$/_{\Theta\langle F,E\rangle}) \dot{\rightarrow}
 \mathsterling/_{\Theta\langle F,E\rangle}, \;
 f(   \mathsterling/_{\Theta\langle F,E\rangle},    \mathsterling /_{\Theta\langle F,E\rangle}) \dot{\rightarrow}
 \mathsterling/_{\Theta\langle F,E\rangle} 
  \, \}$.
  \end{quote}

\begin{quote}  
  $C\langle E \cup F, \emptyset\rangle=\{\,   \#/_{\Theta\langle   E\cup F, \emptyset\rangle}\, \}$.
    \end{quote}

\begin{quote}     
$R\langle E \cup F, \emptyset\rangle=\{\,  
\# \dot{\rightarrow}   \#/_{\Theta\langle   E\cup F, \emptyset\rangle}, \;
\$ \dot{\rightarrow}   \#/_{\Theta\langle   E\cup F, \emptyset\rangle}, \;
\mathsterling \dot{\rightarrow}   \#/_{\Theta\langle   E\cup F, \emptyset\rangle}, \;
\flat \dot{\rightarrow}   \#/_{\Theta\langle   E\cup F, \emptyset\rangle}, \;
\newline
f(  \#/_{\Theta\langle   E\cup F, \emptyset\rangle}, \#/_{\Theta\langle   E\cup F, \emptyset\rangle})\dot{\rightarrow}  \#/_{\Theta\langle   E\cup F, \emptyset\rangle}\, \}$.
\end{quote}
 We note that $R\langle E, F\rangle$, $R\langle F, E\rangle$, and 
 $R\langle E\cup F, \emptyset \rangle$ 
 are total.
\end{exa}

\begin{exa}\label{5uborkafolytatas} \rm 
  We now continue Example \ref{5uborka}.  
  \begin{quote}  
  $C\langle E, F\rangle=\{\,   \#/_{\Theta\langle E ,F\rangle}, \;  
   \$/_{\Theta\langle E ,F\rangle}\, \}$.
  \end{quote}
  
      \begin{quote}     
$R\langle E, F\rangle=\{\,  
\# \dot{\rightarrow}   \#/_{\Theta\langle E ,F\rangle}, \;\$ \dot{\rightarrow}   \$/_{\Theta\langle E ,F\rangle}, \;\newline 
f(\#/_{\Theta\langle E ,F\rangle}, \#/_{\Theta\langle E ,F\rangle})
\dot{\rightarrow}   \$/_{\Theta\langle E ,F\rangle}, \;
f(\#/_{\Theta\langle E ,F\rangle}, \$/_{\Theta\langle E ,F\rangle})
\dot{\rightarrow}  \$/_{\Theta\langle E ,F\rangle}, \; \newline 
f(\$/_{\Theta\langle E ,F\rangle}, \#/_{\Theta\langle E ,F\rangle})
\dot{\rightarrow}  \$/_{\Theta\langle E ,F\rangle}, \; 
f(\$/_{\Theta\langle E ,F\rangle}, \$/_{\Theta\langle E ,F\rangle})
\dot{\rightarrow}  \$/_{\Theta\langle E ,F\rangle}\; \}$. 
 \end{quote}
Observe that $R\langle E, F\rangle$ is total.

  \begin{quote}  
  $C\langle F, E\rangle=
 \{\,  \{\, \#, \; f(\#, \#)\,\}, \; \{\, \$\, \}, \;
     \{\, f(\#, \$)\, \}, \; \{\, f(\$, \#)\, \}, \; \{\,  f(\$, \$)\, \} \, \}$.
        \end{quote}
   \begin{quote}     
 $R\langle F, E\rangle=\{\, 
  \# \dot{\rightarrow}   \#/_{\Theta\langle F ,E\rangle},\; 
   \$ \dot{\rightarrow}   \$/_{\Theta\langle F ,E\rangle}, \; \newline
       f(\#/_{\Theta\langle F ,E\rangle}, \#/_{\Theta\langle F ,E\rangle})
\dot{\rightarrow}   \#/_{\Theta\langle F ,E\rangle}, 
    f(\#/_{\Theta\langle F ,E\rangle}, \$/_{\Theta\langle F ,E\rangle}) \dot{\rightarrow}  f(\#, \$)/_{\Theta\langle F ,E\rangle},\;\newline
   f(\$/_{\Theta\langle F ,E\rangle}, \#/_{\Theta\langle F ,E\rangle}) \dot{\rightarrow}   f(\$, \#)/_{\Theta\langle F ,E\rangle},\;
   f(\$/_{\Theta\langle F ,E\rangle}, \$/_{\Theta\langle F ,E\rangle})\dot{\rightarrow} f(\$, \$)/_{\Theta\langle F ,E\rangle}
   \, \}$.
  \end{quote}
 Note that  $R\langle F, E\rangle$ is not total.
   \begin{quote} $C\langle E \cup F, \emptyset\rangle=
         \{\, \{\,  \#, \; \$, \;  f(\#, \#), \; 
    f(\#, \$), \; f(\$, \#), \;f(\$, \$)\, \}\, \}$.   
         \end{quote}     
      \begin{quote}     
$R\langle E\cup F, \emptyset\rangle=\{\,  
\# \dot{\rightarrow}   \#/_{\Theta\langle E \cup F, \emptyset \rangle}, \; 
\$ \dot{\rightarrow}   \#/_{\Theta\langle E \cup F, \emptyset\rangle}, \; 
f(\#/_{\Theta\langle E \cup F , \emptyset \rangle}, \#/_{\Theta\langle E \cup F, \emptyset\rangle})
\dot{\rightarrow}  
\#/_{\Theta\langle E \cup F, \emptyset\rangle}\, \}$.
\end{quote}
We see that 
$R\langle E \cup F, \emptyset\rangle$ is total.

\end{exa}

\begin{exa}\label{6kortefolytatas} \rm 
  We now continue Example \ref{6korte}.  
Recall that  the set of equivalence classes of  $\Theta\langle E, F \rangle$ is:
\begin{quote}
$C\langle E, F \rangle =\{\,  \{\, \#\, \},  \;  \{\,    f(\#, \#), \;   f(\#, g(\#)), \;
 f(g(\#), \#), \; f(g(\#), g(\#)), \; g(\#), \;g(g(\#))  \, \}, \; \newline
  \{\, \$, \;  f(\#, \$), \; f(\$, \#), \;f(\$, \$), \;g(\$), \; g(g(\$)) \, \}\, \}$.
 \end{quote}
     That is,  
   $C\langle E, F\rangle=\{\,   \#/_{\Theta\langle E ,F\rangle}, \;
  \$/_{\Theta\langle E ,F\rangle}, \,  f(\#, \#)/_{\Theta\langle E ,F\rangle}\, \} $.
         \begin{quote}     
$R\langle E, F\rangle=\{\,  
\# \dot{\rightarrow}   \#/_{\Theta\langle E ,F\rangle}, \;
 \$ \dot{\rightarrow}   \$/_{\Theta\langle E ,F\rangle}, \; \newline 
 f(\#/_{\Theta\langle E ,F\rangle}, \#/_{\Theta\langle E ,F\rangle})
\dot{\rightarrow}  
g(\#)/_{\Theta\langle E ,F\rangle}, \; 
f(\#/_{\Theta\langle E ,F\rangle}, g(\#)/_{\Theta\langle E ,F\rangle})
\dot{\rightarrow}  
g(\#)/_{\Theta\langle E ,F\rangle}, \; \newline
f(g(\#)/_{\Theta\langle E ,F\rangle}, \#/_{\Theta\langle E ,F\rangle})
\dot{\rightarrow}  
g(\#)/_{\Theta\langle E ,F\rangle}, \; 
f(g(\#)/_{\Theta\langle E ,F\rangle},g( \#)/_{\Theta\langle E ,F\rangle})
\dot{\rightarrow}  
g(\#)/_{\Theta\langle E ,F\rangle}, \;  \newline
g(\#/_{\Theta\langle E ,F\rangle})
\dot{\rightarrow}  
g(\#)/_{\Theta\langle E ,F\rangle}, \; 
g(g(\#)/_{\Theta\langle E ,F\rangle})
\dot{\rightarrow}  
g(\#)/_{\Theta\langle E ,F\rangle}, \;  \newline
f(\#/_{\Theta\langle E ,F\rangle}, \$/_{\Theta\langle E ,F\rangle})
\dot{\rightarrow}  
\$/_{\Theta\langle E ,F\rangle}, \;
f(\$/_{\Theta\langle E ,F\rangle}, \#/_{\Theta\langle E ,F\rangle})
\dot{\rightarrow}  \$/_{\Theta\langle E ,F\rangle},\;\newline
f(\$/_{\Theta\langle E ,F\rangle}, \$/_{\Theta\langle E ,F\rangle})
\dot{\rightarrow}  \$/_{\Theta\langle E ,F\rangle},  \;  
g(\$/_{\Theta\langle E ,F\rangle})\dot{\rightarrow}  g( \$)/_{\Theta\langle E ,F\rangle}, \;
g( g( \$)/_{\Theta\langle E ,F\rangle})
\dot{\rightarrow}   \$/_{\Theta\langle E ,F\rangle}$.
 \end{quote}
Note that  $R\langle E, F\rangle$ is total.

    Recall that 
the set of equivalence classes of  $\Theta\langle  F, E \rangle$ is:
 \begin{quote}
  $C\langle F, E\rangle= \{\,  \{\,  \#, \; f(\#, \#),  \;   f(\#, g(\#)), \; 
 f(g(\#), \#), \;  f(g(\#), g(\#)), \;  g(\#), \;g(g(\#))  \, \}$, \; \newline
  $  \{\, \$\, \},\;\{\, g(\$)\, \},\; \{\,  f(\#, \$)\, \}, \;  \{\, f(\$, \#)\, \}, \;\{\,  f(\$, \$)\, \}, \; \{\, g(g(\$)) \, \}
\, \}$. 
 \end{quote}
That is, 
  \begin{quote}  
  $C\langle F, E\rangle=\{\, \{\,   \#/_{\Theta\langle F,E\rangle}, \;
    \$/_{\Theta\langle F,E\rangle}, \; g(\$)/_{\Theta\langle F ,E\rangle}, \; f(\#, \$)/_{\Theta\langle F, E \rangle}, \;
     f(\$, \#)/_{\Theta\langle F, E \rangle}, \;\newline
     f(\$, \$)/_{\Theta\langle F, E \rangle}, \;
     g(g(\$))/_{\Theta\langle F, E \rangle} \, \}
    \, \} $.
  \end{quote}
     
      \begin{quote}     
$R\langle F, E\rangle=\{\,  
 \# \dot{\rightarrow}   \#/_{\Theta\langle F ,E\rangle}, \;
 \$ \dot{\rightarrow}   \$/_{\Theta\langle F ,E\rangle}, \; \newline 
f(\#/_{\Theta\langle F ,E\rangle}, \#/_{\Theta\langle F ,E\rangle})
\dot{\rightarrow}  
\#/_{\Theta\langle F ,E\rangle}, \;
f(\#/_{\Theta\langle F ,E\rangle}, g(\#)/_{\Theta\langle F ,E\rangle})
\dot{\rightarrow}  \#/_{\Theta\langle F ,E\rangle}, \;\newline
f(g(\#)/_{\Theta\langle F,E\rangle}, \#/_{\Theta\langle F ,E\rangle})
\dot{\rightarrow} \#/_{\Theta\langle F ,E\rangle}, \;
f(g(\#)/_{\Theta\langle F ,E\rangle}, g(\#)/_{\Theta\langle F ,E\rangle})
\dot{\rightarrow} \#/_{\Theta\langle F ,E\rangle}, \; \newline
g(\#/_{\Theta\langle F ,E\rangle})  \dot{\rightarrow} \#/_{\Theta\langle F ,E\rangle}, \;
g(\$/_{\Theta\langle F ,E\rangle})  \dot{\rightarrow} g(\$)/_{\Theta\langle F ,E\rangle}, \; \newline
f(\#/_{\Theta\langle F ,E\rangle}, \$/_{\Theta\langle F ,E\rangle})
\dot{\rightarrow}  f(\#, \$)/_{\Theta\langle F ,E\rangle}, \,
f(\$/_{\Theta\langle F ,E\rangle}, \#_{\Theta\langle F ,E\rangle})
\dot{\rightarrow}  f(\$, \#)/_{\Theta\langle F ,E\rangle}\newline
f(\$/_{\Theta\langle F ,E\rangle}, \$_{\Theta\langle F ,E\rangle})
\dot{\rightarrow}  f(\$, \$)/_{\Theta\langle F ,E\rangle}, \;
g(g(\$)/_{\Theta\langle F ,E\rangle})  \dot{\rightarrow} g(g(\$))/_{\Theta\langle F ,E\rangle}\, \}$.
 \end{quote}
 Note that 
$R\langle F, E\rangle$ is not total.

  Recall that 
the set of equivalence classes of  $\Theta\langle E \cup F, \emptyset \rangle$ is:
  \begin{quote}
$ C\langle E \cup F, \emptyset\rangle=\{\, \  \{\, \#, \;  f(\#, \#), \; f(\#, g(\#)), \;
 f(g(\#), \#), \; f(g(\#), g(\#)), \; g(\#), \;g(g(\#))  \, \}, \; \newline
  \{\, \$, \;  f(\#, \$), \; f(\$, \#), \;f(\$, \$), \;g(\$), \; g(g(\$)) \, \}\, \}$.
 \end{quote}
 That is, 
  $C\langle E\cup F, \emptyset\rangle=\{\, \#/_{\Theta\langle E\cup F, \emptyset \rangle}, \$/_{\Theta\langle E\cup F, \emptyset \rangle}  \, \}$.
    \begin{quote}
$R\langle E\cup F, \emptyset\rangle=\{\, \# \dot{\rightarrow}   \#/_{\Theta\langle E\cup F, \emptyset \rangle}, \;
\$\dot{\rightarrow} \$/_{\Theta\langle E\cup F, \emptyset \rangle}, \; \newline
f(\#/_{\Theta\langle E \cup F ,\emptyset  \rangle}, \#/_{\Theta\langle E \cup F ,\emptyset  \rangle})
\dot{\rightarrow} \#/_{\Theta\langle   E\cup F, \emptyset\rangle}, \; 
f(\#/_{\Theta\langle E \cup F ,\emptyset  \rangle}, \$/_{\Theta\langle E \cup F ,\emptyset  \rangle})
\dot{\rightarrow} \$/_{\Theta\langle   E\cup F, \emptyset\rangle}, \; \newline 
f(\$_{\Theta\langle E \cup F ,\emptyset  \rangle}, \#/_{\Theta\langle E \cup F ,\emptyset  \rangle})
\dot{\rightarrow} \$/_{\Theta\langle   E\cup F, \emptyset\rangle}, \; 
f(\$/_{\Theta\langle E \cup F ,\emptyset  \rangle}, \$/_{\Theta\langle E \cup F ,\emptyset  \rangle})
\dot{\rightarrow}  
\$/_{\Theta\langle   E\cup F, \emptyset\rangle}, \;\newline 
g(\#/_{\Theta\langle E \cup F, \emptyset \rangle}) \dot{\rightarrow} \#/_{\Theta\langle E \cup F, \emptyset \rangle}, \;
g(\$/_{\Theta\langle   E\cup F, \emptyset\rangle}) \dot{\rightarrow}  \$/_{\Theta\langle   E\cup F, \emptyset\rangle}
\, \}$.
\end{quote}
Observe  that  $R\langle E \cup F, \emptyset \rangle$ is total.
\end{exa}

\begin{exa}\label{7twokettofolytatas} \rm We continue Example \ref{7twoketto}.
\begin{quote}
  $C\langle E, F\rangle=\{\, \{\, \# \, \}, \;\{\,  \$\, \}, \; 
 \{\,  f( \#, \#), \;   f( \#, \$), \;f( \$, \#), \;   f( \$, \$) \,
 \}\, \}$.
\end{quote}
 \begin{quote}
$R\langle E, F\rangle=\{\,  
\# \dot{\rightarrow}   \#/_{\Theta\langle E ,F\rangle}, \;
\$ \dot{\rightarrow}   \$/_{\Theta\langle E ,F\rangle}, \;\newline 
f(  \#/_{\Theta\langle E ,F\rangle},  \#/_{\Theta\langle E ,F\rangle})
\dot{\rightarrow}  f(\#, \#)/_{\Theta\langle E ,F\rangle}, \;
f(  \#/_{\Theta\langle E ,F\rangle},  \$/_{\Theta\langle E ,F\rangle})
\dot{\rightarrow}  f(\#, \#)/_{\Theta\langle E ,F\rangle}, \; \newline 
f(  \$/_{\Theta\langle E ,F\rangle},  \#/_{\Theta\langle E ,F\rangle})
\dot{\rightarrow}  f(\#, \#)/_{\Theta\langle E ,F\rangle}, \; 
f(  \$/_{\Theta\langle E ,F\rangle},  \$/_{\Theta\langle E ,F\rangle})
\dot{\rightarrow}  f(\#, \#)/_{\Theta\langle E ,F\rangle}\, \},
$
\end{quote}
\begin{quote}
  $C\langle F, E\rangle=\{\, \{\, \# , \;\$\, \}, \;
 \{\,  f( \#, \#), \;   f( \#, \$), \;f( \$, \#), \;   f( \$, \$) \,
 \}\, \}$,
\end{quote}
\begin{quote}
$R\langle F, E\rangle=\{\,  
\# \dot{\rightarrow}   \#/_{\Theta\langle F ,E\rangle}, \;
\$ \dot{\rightarrow}   \#/_{\Theta\langle F ,E\rangle}, \;
f(  \#/_{\Theta\langle F ,E\rangle},  \#/_{\Theta\langle F,E\rangle})
\dot{\rightarrow}  f(\#, \#)/_{\Theta\langle F,E\rangle}\,\}$,
\end{quote}
\begin{quote}
$C\langle E \cup F, \emptyset\rangle=C\langle F, E\rangle$, \,
  $R\langle E \cup F, \emptyset\rangle=R\langle F, E\rangle$.
 \end{quote} 
We notice that $R\langle E, F\rangle$, $R\langle F, E\rangle$, and 
$R\langle E\cup F, \emptyset \rangle$ are not total.

\end{exa}

\begin{exa}\label{8threeharomfolytatas} \rm 
  We now continue Example \ref{8threeharom}.  
   \begin{quote}  
     $C\langle E, F \rangle=\{\,  \{\, \#, \;  \$\, \}, \;
     \{\, \mathsterling\, \}, \;
     \{\,  \flat \, \}\, \}$,
  \end{quote}
  
     \begin{quote}     
$R\langle E, F\rangle=\{\,  
\# \dot{\rightarrow}   \#/_{\Theta\langle E ,F\rangle}, \;
\$ \dot{\rightarrow}   \#/_{\Theta\langle E ,F\rangle}, \;
\mathsterling \dot{\rightarrow}   \mathsterling/_{\Theta\langle E ,F\rangle}, \;
\flat \dot{\rightarrow}   \flat/_{\Theta\langle E ,F\rangle}  \,\}$,

 \end{quote}

\begin{quote}  
$C\langle F, E \rangle=\{\, \{\, \# \, \}, \;
  \{\, \$ \,\}, \; \{\, \mathsterling, \;\flat \, \}  \, \}$,
    \end{quote}

\begin{quote}     
$R\langle F, E\rangle=\{\,  
\# \dot{\rightarrow}   \#/_{\Theta\langle F ,E\rangle}, \;
\$ \dot{\rightarrow}   \$/_{\Theta\langle F ,E\rangle}, \;
\mathsterling \dot{\rightarrow}   \mathsterling/_{\Theta\langle F ,E\rangle}, \;
\flat \dot{\rightarrow}   \mathsterling/_{\Theta\langle F ,E\rangle} \,\}$,
\end{quote}
  
\begin{quote}  
  $C\langle E \cup F, \emptyset \rangle=
\{\,  \{\, \#, \;  \$\, \}, \;
\, \{\, \mathsterling, \;\flat \, \}  \, \}$, and
\end{quote}
\begin{quote}     
$R\langle E \cup F, \emptyset\rangle=\{\,  
\# \dot{\rightarrow}   \#/_{\Theta\langle   E\cup F, \emptyset\rangle}, \;
\$ \dot{\rightarrow}   \#/_{\Theta\langle   E\cup F, \emptyset\rangle}, \;
\mathsterling \dot{\rightarrow}
               \mathsterling/_{\Theta\langle   E\cup F, \emptyset\rangle}, \;
\flat \dot{\rightarrow}   \mathsterling/_{\Theta\langle E \cup F , \emptyset \rangle}
\,\}$.
\end{quote}
Observe that GTRS  $ R \langle E, F \rangle$,  GTRS $ R \langle F, E \rangle$, and GTRS 
$R\langle E\cup F,  \emptyset \rangle$ are not total.
\end{exa}

\section{The Auxiliary  dpwpa  for the GTESs $E$ and $F$}
\label{three-way} 
Recall that  $E$ and $F$ are GTESs over a  signature   $\S$, $n=\mathrm{size}(E)+\mathrm{size}(F)$, and  $W$ stands for
$ST\langle E \cup F\rangle$. 
 We  define the auxiliary dpwpa $ AUX[ E;  F ]=
(C\langle E\cup F, \emptyset \rangle, A[ E;  F])$ 
for the GTESs $E$ and $F$.
To simplify our notations, each vertex  of the dpwpa $ AUX [E; F] $ is a constant     ${b}\in C\langle E\cup F, \emptyset \rangle$,
i.e, an  equivalence class of  $\Theta\langle E\cup F, \emptyset\rangle$. 
 Furthermore, we study the $R\langle {E \cup F}, \emptyset \rangle$ reachability of a constant     ${b}\in C\langle E\cup F,  \emptyset \rangle$
   from    a proper extension    of a constant     ${c}\in C\langle E\cup F,  \emptyset \rangle$.

 \begin{df}\label{mimika}\rm  Let
 $E$ and $F$ be  GTESs over  a signature $\Sigma$. 
 Then  let 
\begin{quote} 
$INC( E\cup F, E) = \{\, ({b}, {c}) \mid {b}\in C \langle E \cup F, \emptyset \rangle,  {c}\in C\langle E, F\rangle,
  \mbox{  and } {c}\subseteq {b} \, \}$,  \newline
   $NUM(  E\cup F, E) = \{\, ({b}, k) \mid {b}\in C \langle E \cup F, \emptyset \rangle, k=|\{ {c}\in C\langle E, F\rangle \mid 
   {c}\subseteq {b}\, \}| \, \}$,  \newline
   $INC( E\cup F, F) = \{\, ({b}, {c}) \mid {b}\in C\langle E \cup F, \emptyset\rangle, {c}\in C \langle F, E \rangle \mbox{  and } {c}\subseteq {b} \, \}$, and \newline
  $NUM( E\cup F, F) = \{\, ({b}, k) \mid {b}\in C \langle E \cup F, \emptyset \rangle, k=|\{ {c}\in C\langle F, E\rangle\mid 
  {c}\subseteq {b}\, \}| \, \}$.
 \end{quote}
 \end{df}
 Intuitively, $INC( E\cup F, E)$ shows the inclusions between the elements of $C \langle E \cup F, \emptyset \rangle$ and $C\langle E, F\rangle$. 
 $NUM( E\cup F, E)$ consists of all pairs $({b}, k)$, where the second component $k$ is the number of all pairs 
 in $INC( E\cup F, E)$, with first component    ${b}$. Similar statements hold true for $INC(E \cup F, F)$ and $NUM(E \cup F, F)$ as well.
 
 \begin{sta} \label{csorvas} Let
 $E$ and $F$ be  GTESs over  a signature $\Sigma$. Then 
 $|INC( E\cup F, E)|\leq n $ and $|INC( E\cup F, F)|\leq n $. 
 \end{sta}
 \begin{proof}
 By Statement \ref{darabsz}, $|C\langle E, F  \rangle|\leq n$. For each    ${c}\in C \langle E, F \rangle $, there exists at most one
   ${b}\in C \langle E \cup F, \emptyset \rangle $ such that $ ({b}, {c}) \in INC( E\cup F, E)  $.
Therefore,   $|INC( E\cup F, E)|\leq |C\langle E, F  \rangle|\leq n $. 

We can show similarly that  $|INC( E\cup F, F)|\leq |C\langle F, E  \rangle|\leq n $. 
\end{proof}
  We get the following statement by direct inspection of Definition \ref{mimika}.
 \begin{sta}\label{dan}
 1.  For every    ${b}\in C \langle E \cup F, \emptyset \rangle $,  there exists a unique ordered pair $ ({b}, {c}) \in INC( E\cup F, E) $
 if and only if    ${b}\in C \langle E, F \rangle $. Similarly, for every   ${b}\in C \langle E \cup F, \emptyset \rangle $,  there exists a unique ordered pair $ ({b}, {c}) \in 
 INC( E\cup F, F) $
 if and only if    ${b}\in C \langle F, E \rangle $. 
 
 2. For every       ${b}\in C \langle E \cup F, \emptyset \rangle $,  there exists a unique ordered pair $ ({b}, k) \in NUM (E\cup F, E) $. Here  $k=1$
 if and only if    ${b}\in C \langle E, F \rangle $. Similarly, for every   ${b}\in C \langle E \cup F, \emptyset \rangle $,  there exists a unique ordered pair $ ({b}, \ell) \in 
 NUM (E\cup F, F) $. Here  $\ell=1$ 
 if and only if    ${b}\in C \langle F, E \rangle $. 
 
 \end{sta}

 \begin{df}\label{lepesttart} \rm Let
 $E$ and $F$ be  GTESs over  a signature $\Sigma$, and     ${b}\in C\langle E\cup F, \emptyset\rangle$. 
     We say that GTRS $R\langle E, F \rangle$  keeps up with the  GTRS 
$R\langle E\cup F, \emptyset \rangle$ writing the constant     ${b}\in C\langle E\cup F, \emptyset\rangle$ if 
\begin{quote}
for every rewrite rule  of the form 
 $\s({b}_1, \ldots, {b}_m) \dot{\rightarrow} {b}\in  R\langle {E \cup F}, \emptyset \rangle$ with \newline
$m\in \Nat$, $\sigma \in \Sigma^{(m)}$, and    ${b}_1,\ldots,{b}_m \in C \langle E \cup F, \emptyset \rangle$, \newline 
 for every     ${c}_1, \ldots, {c}_m \in C\langle E, F\rangle$ with    ${c}_i \subseteq  {b}_i$
 for every $i\in \{\, 1, \ldots, m\, \}$,  \newline
  there exists a rewrite rule in $ R\langle E, F\rangle$ of the form 
 $\s({c}_1, \ldots, {c}_m )\dot{\rightarrow} {c}$ for some    ${c}\in C\langle E, F\rangle$.\end{quote} 
 \end{df}
 Note that  for the constant    ${c}\in C\langle E, F\rangle$ appearing in the definition above, we have  
      ${c} \subseteq{b}$. 
Definition \ref{lepesttart} also defines the notion that GTRS $R\langle F, E \rangle$  keeps up with the  GTRS 
$R\langle E\cup F, \emptyset \rangle$ writing the constant     ${b}\in C\langle E\cup F, \emptyset\rangle$.

\begin{df}\label{inaskodik}\rm  Let 
 $E$ and $F$ be  GTESs over  a signature  $\Sigma$. 
The auxiliary dpwpa
for the GTESs $E$ and $F$  is  $ AUX  [ E;  F] = (C\langle E\cup F, \emptyset\rangle,  A[E; F])$, where 
$ A[E; F]=BSTEP\langle E\cup F, \emptyset \rangle$.
      We assign the attributes 
      \begin{quote} $equal_E,\; equal_F,\;  keeps_E, \;  keeps_F, \;       
       visited\_outer\_loop, \;
      visited\_inner\_loop_E, \;
      visited\_inner\_loop_F, \; and \;
      visited: \mathrm{Boolean}$ 
      \end{quote} to the  vertices of $C\langle E\cup F, \emptyset \rangle$.  
     Furthermore, we  assign the satellite data
      $ counter_E, counter_F:  \Nat$
      to the rewrite rules of $ R \langle E\cup F, \emptyset \rangle$.      
 \end{df}

  For each vertex     ${b} \in C \langle E \cup F, \emptyset \rangle$,  our  decision algorithm computes the values of the above attributes such that the following conditions hold:  
  \begin{itemize}
\item [1.]   ${b}.equal_E=\mathrm{true}$ if and only if 
$  {b} \in C\langle E, F\rangle$. 
\item [2.]   ${b}.equal_F=\mathrm{true}$ if and only if 
$  {b} \in C\langle F, E\rangle$. 
\item [3.]   ${b} .keeps_E= \mathrm{true}$ if and only if  $ R \langle E, F \rangle$   keeps up with  $ R \langle E\cup F, \emptyset \rangle$ writing    ${b} $.
\item [4.]   ${b} .keeps_F= \mathrm{true}$ if and only if  $ R \langle F, E \rangle$   keeps up with  $ R \langle E\cup F, \emptyset \rangle$ writing    ${b} $.
\end{itemize}
Recall  that for each vertex    ${b}\in C\langle E \cup F, \emptyset\rangle $, 
list $Adj({b})$  consists of all  vertices  $  {c}\in C\langle E \cup F, \emptyset\rangle $ such that  there
exists a rewrite rule $\s({b}_1, \ldots, {b}_m)\dot{\rightarrow}{b}$ in $ R \langle E \cup  F, \emptyset\rangle$ with $\s\in \S_m$, $m\geq 1$, and    ${c}={b}_i$ for some $i\in \{\, 1, \ldots, m \, \}$.

\begin{sta}\label{iszap} Let 
 $E$ and $F$ be  GTESs over a signature   $\Sigma$. Then 
  $|A[E; F]|\leq 2  n$.
 \end{sta}
 \begin{proof} 
 By Definition \ref{inaskodik}, Statement \ref{drrszam},  we have 
 $|A[E; F]|=
 |BSTEP\langle E\cup F, \emptyset \rangle| \leq 2 n$.
 \end{proof}

  Intuitively, the basic idea of the proof of the following lemma  is similar to that of the pumping lemma for regular languages \cite{sipser2006},  
 context-free languages  \cite{sipser2006}, and regular tree languages  \cite{gs}, \cite{eng2}.  The main difference from the proof of the pumping lemma is that we iterate contracting the ground term rather than pumping it.
  We stop when we get such a  small ground term that  we cannot contract it any more.  
If  $ R \langle E\cup F, \emptyset \rangle$ 
 reaches a constant     ${b}\in C\langle E\cup F,  \emptyset \rangle$ from   a proper extension   of a constant       ${c}\in C\langle E\cup F,  \emptyset \rangle$,
then there exists a  $1$-context 
 $\zeta$  over $\S$ such that 
such that \begin{itemize}
\item $\zeta\neq \diamondsuit$,
\item $ R \langle E\cup F, \emptyset \rangle$ reaches     ${b}$ from  $\zeta[{c}]$, and 
\item  $ R \langle E\cup F, \emptyset \rangle$  reduces the subterms of $\zeta[{c}]$ at different  prefixes of $\mathrm{addr}(\zeta)$ -- except possibly for the pair  
   of $\varepsilon$ and   $\mathrm{addr}(\zeta)$ --
  to different constants in  $ C\langle E\cup F,  \emptyset \rangle$.
   \end{itemize}
  Thus the length of  $\mathrm{addr}(\zeta)$ is less than equal to  $|C\langle E\cup F, \emptyset \rangle| $.
 
 \begin{lem} \label{kokusz}  Let 
  $E$ and $F$  be  GTESs over a  signature   $\S$.
 Let 
  $ R \langle E\cup F, \emptyset \rangle$ reach a constant     ${b}\in C\langle E\cup F,  \emptyset \rangle$ from a proper extension  of a constant      ${c}\in C\langle E\cup F,  \emptyset \rangle$. 
          Then  there exists a $1$-context $\zeta$ over $\Sigma$, where
   \begin{itemize}
   \item [{\rm ({b})}] $\zeta\neq \diamondsuit$,
      \item  [{\rm ({c})}] $k\geq 1$ denotes   $\mathrm{length}(\mathrm{addr}(\zeta) )$, 
  \item [{\rm (c)}]
  for every $i\in\{\, 0, \ldots, k\,\}$,   $\alpha_i$ denotes  the prefix of $\mathrm{addr}(\zeta) $ of length $i$ 
  \end{itemize}   such that  
    \begin{itemize}
       \item  [{\rm (d)}] $\zeta[{c}]{\downarrow}_ {R\langle E\cup F, \emptyset \rangle} ={b}$,
       \item  [{\rm (e)}] $\zeta[{c}]|_{\alpha_i }  {\downarrow}_ {R\langle E\cup F, \emptyset \rangle} \neq \zeta[{c}]|_{\alpha_j } {\downarrow}_ {R\langle E\cup F, \emptyset \rangle}$  for any  $0\leq i < j \leq k$ such that 
   $1\leq i $ or $  j \leq k-1$,
      and 
  \item  [{\rm (f)}] $k \leq |C\langle E\cup F, \emptyset \rangle |$.
    \end{itemize}        
 \end{lem} 
 \begin{proof}   
  It is sufficient to prove that there exists a $1$-context $\zeta$ over $\Sigma$ such that  
Conditions (a)--(e) hold. Because for any  $1$-context $\zeta$ over $\Sigma$ satisfying Conditions (a)--(e),  Condition (f) holds as well.
 By Statement \ref{olvaseler}, 
     there exists a $1$-context $\delta$ over $\Sigma$ such that $\delta\neq\diamondsuit$ and $\delta[{c}]{\downarrow}_ {R\langle E\cup F, \emptyset \rangle} ={b}$, i.e. 
     Conditions (a)--(d) hold for $\zeta=\delta $.
    We now distinguish two cases:
    
 {\em Case 1:}    
 Conditions  (e)  holds for $\zeta=\delta $. Then we are done.

    {\em Case 2:} Condition (e) does not hold for $\zeta=\delta $. 
    For every $i\in\{\, 0, \ldots, n\,\}$,   let $\beta_i$ denote the prefix of $\mathrm{addr}(\delta) $ of length $i$.
 There exist $i, j \in \{\, 0, \ldots, k \, \}$ 
  such that  
 \begin{quote}
 $\left(0\leq i < j \leq k-1  \mbox{ or } 1\leq i < j \leq k\right)$ and 
   $\zeta[{c}]|_{\beta_i }  {\downarrow}_ {R\langle E\cup F, \emptyset \rangle} = \zeta[{c}]|_{\beta_j } {\downarrow}_ {R\langle E\cup F, \emptyset \rangle}$.
 \end{quote}  
        We define  $\delta'\in CON_{\Sigma, 1}$  from  $\delta$  by replacing  the subterm $ \delta|_{\beta_i}$ at $\beta_i$ by the subterm   $\delta|_{\beta_j}$
   at $\beta_j$.  Then 
     \begin{quote}
     $\delta'\neq \diamondsuit$, 
    $\mathrm{length}(\mathrm{addr}(\delta')) < \mathrm{length}(\mathrm{addr}(\delta))$,  and 
     $\delta'[{c}]{\downarrow}_ {R\langle E\cup F, \emptyset \rangle} ={b}$ i.e.   Conditions (a)--(d) hold for $\zeta=\delta' $.
      \end{quote}
  We now distinguish two cases:

    {\em Case 2.1:} Condition (e) holds for $\zeta=\delta'$. Then Conditions (a)--(e) hold for $\zeta=\delta'$. In this case we are done.

    {\em Case 2.2:}  Condition (e) does  not hold for $\zeta=\delta'$. In this case we let $\delta =\delta'$ and repeat the procedure of contracting $\delta$  described above
    finitely many times until Condition (e) holds for $\zeta=\delta'$.  Then Conditions (a)--(e) hold for $\zeta=\delta'$. 
 \end{proof}
 Intuitively, 
 we now consider two constants,    ${b}$ and    ${c}$ in $C\langle E\cup F, \emptyset\rangle$   and a sequence of rewrite rules in  $R\langle {E \cup F}, \emptyset \rangle$
 of which consecutive applications leads  us from     ${c}$  to 
    ${b}$.
   Let    ${b}, {c} \in C\langle E\cup F, \emptyset \rangle$, $k\geq 1$,  and
  for each $j=1, \ldots, k$, let 
  \begin{itemize}
 \item  $\s_j({b}_{j1}, \ldots, {b}_{jm_j}) \dot{\rightarrow} {b}_j\in  R\langle {E \cup F}, \emptyset \rangle$, $m_j\geq 1$, 
  be a rewrite rule and 
\item  $i_j \in \{\, 1, \ldots, m_j\, \}$ be an integer 
 \end{itemize}
 such that \begin{itemize}
 \item    ${c}= {b}_{1i_1}$, 
 \item 
for each $j=1, \ldots, k-1$,    ${b}_j= {b}_{(j+1)i_{j+1}}$, and 
 \item 
    ${b}={b}_k$.
 \end{itemize}
The above  sequence of rewrite rules  is connected in the sense that   we can apply them one after another, the right-hand side of each rule appears on the left-hand side of the next rule.
  Furthermore,  the consecutive applications of these rewrite rules leads us from the constant     ${c}$  to the constant
    ${b}$ since     ${c}$ appears on the left-hand side of the first rule, and    ${b}$ is the right-hand side of the last rule.
 We can apply the above rewrite rules one after another in the following way: 
  \newline    ${c}$ appears in the left-hand side of the first rewrite rule, 
   \newline
   we apply the first rewrite rule to its left-hand side, $\s_1({b}_{11}, \ldots, {b}_{1m_1}) $, we get the right-hand side    ${b}_1$, 
    which appears on  the left-hand side of the second rewrite rule, 
\newline we apply the second  rewrite rule to its left-hand side, $\s_2({b}_{21}, \ldots, {b}_{2m_2}) $, we get the right-hand side    ${b}_2$,
    which appears on  the left-hand side of the third rewrite rule, and son on.\newline
   Finally, the right-hand side    ${b}_{k-1}$ of the last but one rewrite rule appears in the left-hand side of the last rewrite rule,  \newline
   we apply the last rewrite rule to 
    its left-hand side, $\s_k({b}_{k1}, \ldots, {b}_{km_k}) $, we get its right-hand side    ${b}_k={b}$.
   
  We now state formally and show the following intuitive statement: if   $ R \langle E\cup F, \emptyset \rangle$ reaches a constant     ${b}\in C\langle E\cup F,  \emptyset \rangle$ from   a proper extension   of a  constant       ${c}\in C\langle E\cup F,  \emptyset \rangle$, then 
  there exists a connected sequence of   rewrite rules in  $ R \langle E\cup F, \emptyset \rangle$ that  leads us from  the constant    ${c}$  to the constant    ${b}$.
   Moreover,    ${c}$ and the 
  intermediate 
  constants between    ${c}$ and    ${b}$, appearing as the right-hand sides  of the applied rewrite rules except for  the last one, are pairwise different. 
  \begin{lem} \label{banan} 
   Let 
  $E$ and $F$  be  GTESs over a  signature   $\S$.
  Let 
  $ R \langle E\cup F, \emptyset \rangle$ reach a constant     ${b}\in C\langle E\cup F, \emptyset \rangle$ from   a proper extension   of a constant       ${c}\in C\langle E\cup F,  \emptyset \rangle$. 
          Then there exists an integer     $1\leq  k \leq |C\langle E\cup F, \emptyset \rangle |$ and 
    for each $j=1, \ldots, k$, there exists  a rewrite rule 
 $\s_j({b}_{j1}, \ldots, {b}_{jm_j}) \dot{\rightarrow} {b}_j\in  R\langle {E \cup F}, \emptyset  \rangle$, $m_j\geq 1$, and an integer $i_j \in \{\, 1, \ldots, m_j\, \}$ 
 such that 
  \begin{itemize}
 \item    ${c}= {b}_{1i_1}$, 
  \item 
for each $j=1, \ldots, k-1$,     ${b}_j= {b}_{(j+1)i_{j+1}}$, 
  \item    ${c} \not\in \{\, {b}_1, \ldots, {b}_{k-1}\, \}$ and 
  for any  $ 1 \leq  i < j \leq  k-1$,     ${b}_i \neq {b}_j$, 
 and 
 \item 
    ${b}={b}_k$.
 \end{itemize}
  \end{lem} 
 \begin{proof}
 By Lemma \ref{kokusz},
 there exists a $1$-context $\zeta$ over $\Sigma$, where
   \begin{itemize}
   \item  [{\rm ({b})}] $\zeta\neq \diamondsuit$, 
      \item  [{\rm ({c})}] $k\geq 1$ denotes   $\mathrm{length}(\mathrm{addr}(\zeta) )$, 
  \item  [{\rm (c)}]
  for every $i\in\{\, 0, \ldots, k\,\}$,   $\alpha_i$ denotes  the prefix of $\mathrm{addr}(\zeta) $ of length $i$ 
  \end{itemize}   such that 
   \begin{itemize}
       \item  [{\rm (d)}] $\zeta[{c}]{\downarrow}_ {R\langle E\cup F, \emptyset \rangle} ={b}$,
       \item  [{\rm (e)}] $\zeta[{c}]|_{\alpha_i }  {\downarrow}_ {R\langle E\cup F, \emptyset \rangle} \neq \zeta[{c}]|_{\alpha_j } {\downarrow}_ {R\langle E\cup F, \emptyset \rangle}$  for any  $0\leq i < j \leq k$ such that 
   $1\leq i $ or $  j \leq k-1$,
      and 
  \item  [{\rm (f)}] $k \leq |C\langle E\cup F, \emptyset \rangle |$.
    \end{itemize} 
    As $\mathrm{addr}(\xi)=\alpha_k$, we have $\xi[{c}]|_{\alpha_k} ={c}$.
   \begin{quote}  
For each $i\in \{1, \ldots, k-1\, \}$, let 
     ${b}_i=\zeta[{c}]|_{\alpha_i }  {\downarrow}_ {R\langle E\cup F, \emptyset \rangle} $. 
  \end{quote} 
  Since 
     ${c}=\zeta[{c}]|_{\alpha_k }  {\downarrow}_ {R\langle E\cup F, \emptyset \rangle} $,   
   by Condition (e) 
 we have  \begin{quote}  
    ${c} \not\in \{\, {b}_1, \ldots, {b}_{k-1}\, \}$ and 
  for any  $ 1 \leq  i < j \leq  k-1$,     ${b}_i \neq {b}_j$.
    \end{quote} 
     Then consider a sequence of  $ R\langle {E \cup F}, \emptyset \rangle$   reductions, where we indicate the rewrite rules applied at the positions 
     $\alpha_{k-1}$, $\alpha_{k-2}$, $\ldots$, $\alpha_0$:
  \begin{quote}
$\zeta[{c}]=t_0 \tred {  R\langle {E \cup F}, \emptyset \rangle} s_1     \red {[ \alpha_{k-1}, \s_1({b}_{11}, \ldots, {b}_{1m_1}) \dot{\rightarrow} {b}_1]}t_1 
  \tred {  R\langle {E \cup F}, \emptyset \rangle} s_2 \red {[ \alpha_{k-2}, \s_2({b}_{21}, \ldots, {b}_{2m_2}) \dot{\rightarrow} {b}_2]}
  t_2 \tred {  R\langle {E \cup F}, \emptyset \rangle} \newline
  \cdots
      \red {[ \alpha_1, \s_{k-1}({b}_{(k-1)1}, \ldots, {b}_{(k-1)m_1}) \dot{\rightarrow} {b}_{k-1}]}t_{k-1} 
   \tred {  R\langle {E \cup F}, \emptyset \rangle} s_k \red {[ \alpha_0, \s_k({b}_{k1}, \ldots, {b}_{km_k}) \dot{\rightarrow} {b}_k ]} {b}_k= {b}$.
\end{quote}
Here
 \begin{itemize}
 \item $s_i\in T_{\Sigma \cup C\langle E\cup F, \emptyset \rangle }  $ for every $i\in \{1, \ldots, k\, \}$, and $t_i \in 
  T_{\Sigma \cup C\langle E\cup F, \emptyset \rangle} $ for every $i\in \{0, \ldots, k-1\, \}$,
 \item for each $i\in \{1, \ldots, k\, \}$,
     ${b}_i=\zeta[{c}]|_{\alpha_{k-i }}  {\downarrow}_ {R\langle E\cup F, \emptyset \rangle} $, 
 
 \item    ${c}= {b}_{1i_1}$ for some  $i_1\in \{\, 1, \ldots, m_1 \, \}$, 
      \item for each $j\in \{1, \ldots, k\, \}$,
   $\s_j({b}_{j1}, \ldots, {b}_{jm_j}) \dot{\rightarrow} {b}_j\in  R\langle {E \cup F}, \emptyset \rangle$, $m_j\geq 1$,  

 \item 
 for each $j=1, \ldots, k-1$, there exists  $i_{j+1}\in \{\, 1, \ldots, m_{j+1} \, \}$ such that 
    ${b}_j= {b}_{(j+1)i_{j+1}}$, and
 
 \item 
    ${b}={b}_k$. \qedhere
   \end{itemize} 
 \end{proof} 
    We have the following result by Lemma \ref{banan},  Statement  \ref{olvaseler},     
     and by the definition of $A[E; F]$ in Definition \ref{inaskodik}. 
 \begin{lem}\label{napsutes} Let 
  $E$ and $F$  be  GTESs over a  signature   $\S$.
 If  $ R \langle E\cup F, \emptyset \rangle$ 
 reaches a constant     ${b}\in C\langle E\cup F,  \emptyset \rangle$ from   a proper extension   of a constant       ${c}\in C\langle E\cup F,  \emptyset \rangle$, then there exists 
  a path of positive length or a cycle   from    the vertex    ${b}$ to the vertex    ${c}$  in the auxiliary dpwpa   
$ AUX[ E;  F] =
(C\langle E \cup F, \emptyset\rangle,A[E; F])$.
    \end{lem}

 We now show the converse of Lemma \ref{napsutes}. 
Intuitively,  
  if there exists a connected sequence of  rewrite rules in  $ R \langle E\cup F, \emptyset \rangle$ 
  that leads us from a constant    ${c}\in C\langle E\cup F,  \emptyset \rangle$  to a constant    ${b}\in C\langle E\cup F,  \emptyset \rangle$,
 then   $ R \langle E\cup F, \emptyset \rangle$ reaches the  constant     ${b}$  from a proper extension of the constant     ${c}$.  
 By Definition \ref{szeder}, Definition \ref{szabaly},  and Proposition \ref{masodik} we have the following lemma. 
  \begin{lem}\label{millio} Let 
  $E$ and $F$  be  GTESs over a  signature   $\S$.
 Let 
 $\s({b}_1, \ldots, {b}_m) \dot{\rightarrow} {b}\in  R\langle {E \cup F}, \emptyset \rangle$
 be a rewrite rule, where $m\geq 1$,    ${b}_1, \ldots, {b}_m, {b}\in C\langle E \cup F, \emptyset\rangle$. Let $i\in \{\, 1, \ldots, m\, \}$.
 Then there exist ground terms $s_1, \ldots, s_{i-1}, s_{i+1}, \ldots, s_m\in \ts$ such that 
 \begin{quote}
   $\s(s_1, \ldots, s_{i-1},{b}_i,  s_{i+1}, \ldots, s_m )\tred {R\langle E \cup F, \emptyset\rangle}\s({b}_1, \ldots, {b}_m) \red {R\langle E \cup F, \emptyset\rangle}{b}$.
   \end{quote}
   \end{lem}

 \begin{lem}\label{trillio}  Let 
  $E$ and $F$  be  GTESs over a  signature   $\S$.
 Let   $k\geq 1$,   and 
  $({b}_k, {b}_{k-1}), \, ({b}_{k-1}, {b}_{k-2}),\,  \ldots , \, ({b}_1, {b}_0) \in A[E; F]$. 
  Then there exists a $1$-context $\delta$ over $\Sigma$ such that  $\delta\neq \diamondsuit$  and 
  $\delta[{b}_0]\tred {R\langle E\cup F, \emptyset\rangle} {b}_k$.
     \end{lem}
 \begin{proof} By Definition \ref{inaskodik} and Definition \ref{zug}, 
  for each $j=1, \ldots, k$, there exists a rewrite rule 
 $\s_j({b}_{j1}, \ldots, {b}_{jm_j}) \dot{\rightarrow} {b}_j\in  R\langle {E \cup F}, \emptyset \rangle$, $m_j\geq 1$, 
 such that 
 \begin{quote}
   ${b}_0\in \{\, {b}_{11}, \ldots, {b}_{1m_1} \, \}$ and 
for each $j=1, \ldots, k-1$, there exists  $i_{j+1}\in \{\, 1, \ldots, m_{j+1} \, \}$ such that 
    ${b}_j= {b}_{(j+1)i_{j+1}}$. \end{quote} 
 We proceed by induction on $k$.
 
   {\em Base Case:} $k=1$. By Lemma \ref{millio}, we are done.  
 
 {\em Induction Step:}  Let $k\geq 1$, and  assume that the statement of the lemma is true for $k$. We now show that it is true for $k+1$.
 By the induction hypothesis, 
  there exists a $1$-context $\delta$ over $\Sigma$ such that  
   \begin{quote}
  $\delta\neq \diamondsuit$  and $\delta[{b}_0]\tred {R\langle E \cup F, \emptyset\rangle} {b}_k$. 
\end{quote}
By Lemma \ref{millio},  there exist ground terms $s_1, \ldots, s_{i-1}, s_{i+1}, \ldots,  s_{m_{k+1}}\in \ts$ such that 
\begin{quote}
$\s_{k+1}(s_1, \ldots, s_{i_{k+1} -1},  {b}_k, s_{i_{k+1} +1}, \ldots, s_{m_{k+1}})= 
\s_{k+1}(s_1, \ldots, s_{i_{k+1} -1},   {b}_{(k+1)i_{k+1} }, s_{i_{k+1} +1}, \ldots, s_{m_{k+1}})
\tred {R\langle E\cup F, \emptyset\rangle}
\newline 
 \s_{k+1}({b}_{(k+1)1},  \ldots, {b}_{(k+1)i_{k+1} -1},   {b}_{(k+1)i_{k+1} }, {b}_{(k+1)i_{k+1} +1}, \ldots , {b}_{(k+1)m_{k+1}})
\red {R\langle E\cup F, \emptyset\rangle}{b}_{k+1}$.
\end{quote}
Thence we have 
 \begin{quote}
$\s_{k+1}(s_1, \ldots, s_{i_{k+1} -1},  \delta[{b}_0], s_{i_{k+1} +1}, \ldots, s_{m_{k+1}})\tred {R\langle E\cup F, \emptyset\rangle} 
\s_{k+1}(s_1, \ldots, s_{i_{k+1} -1},  {b}_k, s_{i_{k+1} +1}, \ldots, s_{m_{k+1}})= \newline
\s_{k+1}(s_1, \ldots, s_{i_{k+1} -1},   {b}_{(k+1)i_{k+1} }, s_{i_{k+1} +1}, \ldots, s_{m_{k+1}})
\tred {R\langle E\cup F, \emptyset\rangle}\newline
\s_{k+1}({b}_{(k+1)1}, \ldots, {b}_{ (k+1) i_{k+1} -1},   {b}_{(k+1)i_{k+1} }, {b}_{(k+1) i_{k+1} +1}, \ldots,  {b}_{(k+1)m_{k+1}})
\red {R\langle E\cup F, \emptyset\rangle}
 {b}_{k+1}$. \qedhere
\end{quote}
 \end{proof}

    We have the following statement 
by    Lemma \ref{trillio}.        
 \begin{lem}\label{napkelte}  Let 
  $E$ and $F$  be  GTESs over a  signature   $\S$.
  If there exists 
  a path of positive length or a cycle   from   a vertex    ${b}$ to a vertex    ${c}$  in the auxiliary dpwpa   
$ AUX[ E;  F] =
(C\langle E \cup F, \emptyset\rangle,A[E; F])$, then 
 $ R \langle E\cup F, \emptyset \rangle$ 
 reaches the constant     ${b}\in C\langle E\cup F,  \emptyset \rangle$ from a proper extension of the constant     ${c}\in C\langle E\cup F,  \emptyset \rangle$.
     \end{lem}

 By Lemma \ref{napsutes} and Lemma \ref{napkelte}, we  have the following result. 
\begin{lem}\label{rozmaring} Let
 $E$ and $F$ be  GTESs over  a signature $\Sigma$.
  $ R \langle E\cup F, \emptyset \rangle$ reaches a constant    ${b}\in C\langle E\cup F, \emptyset \rangle$ from   a proper extension   of a constant       ${c}\in C\langle E\cup F,   \emptyset \rangle$
     if and only if there is a path of positive length or a cycle   from  
  the vertex    ${b}$ to the vertex    ${c}$
 in the auxiliary dpwpa   
$ AUX[ E;  F] =
(C\langle E \cup F, \emptyset\rangle, A[E; F])$.
   
 \end{lem}

\section {Running Examples Further Continued}\label{folytat}
 We  continue our running examples, for each one of them we  present the  dpwpa $AUX[ E;  F] $. 
\begin{exa}\label{1oneegymasodikfolytatas}\rm
We now continue Examples \ref{1oneegy} and  \ref{1oneegyfolytatas}. We construct
 the auxiliary  dpwpa $ AUX[ E;  F] =
(C\langle E\cup F,  \emptyset \rangle, A[E; F])$ for the GTESs $E$ and $F$.
Recall that
\begin{quote}
  $C\langle E\cup F, \emptyset  \rangle=\{\, \{\, \#\,\}$, $\{\,   \$\, \}$,
  $\{\, f(\#), \;  g(\#) \, \}$,
 $\{\, f(\$), \;g(\$)\, \}\, \}$.
\end{quote}
Observe that
 $A[E; F]=\{\, (f(\#)/_{\Theta\langle E\cup F, \emptyset \rangle}, \#/_{\Theta\langle E \cup F ,\emptyset\rangle}),\;
 (f(\$)/_{\Theta\langle E\cup F, \emptyset \rangle},   \;\$/_{\Theta\langle E \cup F ,\emptyset \rangle})\,\}$.
  \begin{itemize}
\item [1.] $equal_E,\; equal_F,\; keeps_E,\; keeps_F$ are 
 attributes of the   vertices in  $C\langle E\cup F, \emptyset  \rangle$. 
\item [2.] $  \#/_{\Theta\langle E \cup F ,\emptyset \rangle}.equal_E= \mathrm{true}$,\;
$  \$/_{\Theta\langle E \cup F ,\emptyset \rangle}.equal_E= \mathrm{true}$, \;
 $f(\#)/_{\Theta\langle E\cup F, \emptyset \rangle}.equal_E= \mathrm{true}$, and \newline 
$f(\$)/_{\Theta\langle E\cup F, \emptyset \rangle}.equal_E= \mathrm{false}$.

 $  \#/_{\Theta\langle E \cup F ,\emptyset \rangle}.equal_F= \mathrm{true}$,\;
$  \$/_{\Theta\langle E \cup F ,\emptyset \rangle}.equal_F= \mathrm{true}$, \;
 $f(\#)/_{\Theta\langle E\cup F, \emptyset \rangle}.equal_F= \mathrm{false}$, and  \newline 
$f(\$)/_{\Theta\langle E\cup F, \emptyset \rangle}.equal_F= \mathrm{true}$.

\item [3.] $  \#/_{\Theta\langle E \cup F ,\emptyset \rangle}.keeps_E= \mathrm{true}$,\;
$  \$/_{\Theta\langle E \cup F ,\emptyset \rangle}.keeps_E= \mathrm{true}$, \;
 $f(\#)/_{\Theta\langle E\cup F, \emptyset \rangle}).keeps_E= \mathrm{true}$, and  \newline 
$f(\$)/_{\Theta\langle E\cup F, \emptyset \rangle}).keeps_E= \mathrm{true}$.

 $  \#/_{\Theta\langle E \cup F ,\emptyset \rangle}.keeps_F= \mathrm{true}$,\; 
$  \$/_{\Theta\langle E \cup F ,\emptyset \rangle}.keeps_F= \mathrm{true}$, \;
 $f(\#)/_{\Theta\langle E\cup F, \emptyset \rangle}.keeps_F= \mathrm{true}$, and  \newline 
$f(\$)/_{\Theta\langle E\cup F, \emptyset \rangle}.keeps_F= \mathrm{true}$.

\end{itemize}
    \end{exa}

\begin{exa}\label{2szepenmasodikfolytatas}\rm 
We now continue Examples \ref{2szepen} and  \ref{2szepenfolytatas}. We construct
 the auxiliary  dpwpa $ AUX[ E;  F] =
(C\langle E\cup F,  \emptyset \rangle, A[E; F])$ for the GTESs $E$ and $F$.
Recall that 
 \begin{quote}     
 $C\langle E\cup F, \emptyset  \rangle=\{\, \{\, \#, \; f(\#), \;  f^2(\#), \;f^3(\#)\, \}\, \}= \{\, W\, \}$. 
     \end{quote}
Observe that 
 $A[E; F]=\{\, (\#/_{\Theta\langle E \cup F ,\emptyset \rangle}, \#/_{\Theta\langle E\cup F, \emptyset \rangle})  \,\}$.
  \begin{itemize}
\item [1.] 
$equal_E,\; equal_F,\; keeps_E,\; keeps_F$ are 
 attributes of the   vertices in  $(C\langle E\cup F, \emptyset\rangle, A[E; F])$. 

\item [2.] $  \#/_{\Theta\langle E \cup F ,\emptyset \rangle}.equal_E= \mathrm{false}$ \;  and \; 
  $\#/_{\Theta\langle E \cup F ,\emptyset \rangle}.equal_F= \mathrm{false}$.

\item [3.] $  \#/_{\Theta\langle E \cup F ,\emptyset \rangle}.keeps_E= \mathrm{false}$ \; and \; 
 $  \#/_{\Theta\langle E \cup F ,\emptyset \rangle}.keeps_F= \mathrm{false}$.

\end{itemize}
    \end{exa}

\begin{exa}\label{5uborkamasodikfolytatas}\rm 
We now continue Example \ref{5uborka} and  its sequel \ref{5uborkafolytatas}. We construct
 the auxiliary dpwpa $ AUX[ E;  F] =
(C\langle E\cup F, \emptyset  \rangle, A[E; F])$ for the GTESs $E$ and $F$.
 Recall that 
\begin{quote} $C\langle E \cup F, \emptyset\rangle=
         \{\, \{\,  \#, \; \$, \;  f(\#, \#), \; 
    f(\#, \$), \; f(\$, \#), \;f(\$, \$)\, \}\, \}$.   
         \end{quote}   
Observe  that  
      \begin{quote}
$ A[E; F])=
\{\,( \#/_{\Theta\langle E \cup F, \emptyset \rangle}, \#/_{\Theta\langle E \cup F ,\emptyset \rangle}) \,\}$. 
 \end{quote}
         By direct inspection of $C\langle E, F\rangle$,  $ C\langle F, E\rangle$, and   $C\langle E \cup F, \emptyset\rangle$, we get that
         \begin{quote} 
$ \#/_{\Theta\langle E \cup F, \emptyset \rangle}\not\in C\langle E, F\rangle$ hence $ \#/_{\Theta\langle E \cup F, \emptyset \rangle}.equal_E=\mathrm{false}$,\newline
      $ \#/_{\Theta\langle E \cup F, \emptyset \rangle}\not\in C\langle F, E\rangle$ hence $ \#/_{\Theta\langle E \cup F, \emptyset \rangle}.equal_F=\mathrm{false}$.
      \end{quote}     
By direct inspection of $C\langle E, F \rangle$,   $R\langle E, F \rangle  $,
$C\langle E \cup F, \emptyset \rangle $, and    $R\langle E \cup F, \emptyset \rangle  $, we get that  
$R\langle E, F\rangle$  keeps up with $R\langle E\cup F, \emptyset\rangle$ writing 
$\#/_{\Theta\langle E\cup F  ,\emptyset \rangle}$, hence        
$ \#/_{\Theta\langle   E \cup F, \emptyset \rangle}.keeps_E= \mathrm{true}$.

   The rewrite rule 
$f(\#/_{\Theta\langle E\cup F  ,\emptyset \rangle}, \#/_{\Theta\langle E\cup F  ,\emptyset \rangle})
\dot{\rightarrow}  
\#/_{\Theta\langle E\cup F  ,\emptyset \rangle}$
is in $R\langle E\cup F, \emptyset\rangle$ and
 $ f(\#, \$)/_{\Theta\langle F  ,E  \rangle}\subseteq \#/_{\Theta\langle E\cup F  ,\emptyset \rangle}$. 
However,   in $R\langle F, E\rangle$
there does not exist a rewrite  rule with left-hand side 
$f( f(\#, \$)/_{\Theta\langle F  ,E  \rangle}, f(\#, \$)/_{\Theta\langle F  ,E \rangle})$. Hence $R\langle F, E\rangle$ does not keep up with $R\langle E\cup F, \emptyset\rangle$ writing 
$\#/_{\Theta\langle E\cup F  ,\emptyset \rangle}$, and          
$ \#/_{\Theta\langle   E \cup F, \emptyset \rangle}.keeps_F = \mathrm{false}$.

   \end{exa}
    \begin{exa}\label{6kortemasodikfolytatas}\rm 
We now continue Examples \ref{6korte}  and  \ref{6kortefolytatas}. We construct
 the auxiliary dpwpa $ AUX[ E;  F] =
(C\langle E\cup F, \emptyset\rangle, A[E; F])$ for the GTESs $E$ and $F$.
Observe  that
    \begin{quote}
$A[E; F]=
\{\,( \#/_{\Theta\langle E \cup F, \emptyset \rangle}, \#/_{\Theta\langle E \cup F ,\emptyset \rangle}), \;
( \$/_{\Theta\langle E \cup F, \emptyset \rangle}, \#/_{\Theta\langle E \cup F ,\emptyset \rangle}), \;
( \$/_{\Theta\langle E \cup F, \emptyset \rangle}, \$/_{\Theta\langle E \cup F ,\emptyset \rangle})
\,\}$. 
 \end{quote}

    By direct inspection of $C\langle E, F\rangle$,  $ C\langle F, E\rangle$, and   $C\langle E \cup F, \emptyset\rangle$, we get that
 \begin{quote} 
   $\#/_{\Theta\langle E \cup F, \emptyset \rangle}.equal_E=\mathrm{false}$,
    $\#/_{\Theta\langle E \cup F, \emptyset \rangle}.equal_F =\mathrm{true}$,\newline
   $\$/_{\Theta\langle E \cup F, \emptyset \rangle}.equal_E=\mathrm{true}$, 
    $\$/_{\Theta\langle E \cup F, \emptyset \rangle}.equal_F =\mathrm{false}$,    
       
 \end{quote}

 We now show that GTRS $R\langle F, E \rangle$  keeps up with GTRS 
$R\langle E\cup F, \emptyset \rangle$ writing $\#/_{\Theta\langle E \cup F, \emptyset \rangle}$.
Recall that
  \begin{itemize} 
 \item    $\#/_{\Theta\langle F ,E \rangle}= \#/_{\Theta\langle E \cup F ,\emptyset \rangle}$. Furthermore, $\#/_{\Theta\langle F, E \rangle}$ is the only element of $C \langle F, E  \rangle $ which is a subset of $\#/_{\Theta\langle E \cup F, \emptyset \rangle}$.
   
\item  For the rewrite rule   $\# \dot{\rightarrow}   \#/_{\Theta\langle E\cup F ,\emptyset\rangle}\in R\langle E\cup F, \emptyset\rangle$,  there is a rewrite rule with left-hand side $\#$  
 in $R\langle F, E\rangle$.  
   
 \item   For the rewrite rule   
$f(\#/_{\Theta\langle E \cup F ,\emptyset\rangle}, \#/_{\Theta\langle E \cup F ,\emptyset\rangle})
\dot{\rightarrow}  
\#/_{\Theta\langle   E\cup F, \emptyset\rangle}   \in R\langle E\cup F, \emptyset\rangle$,  there is a rewrite rule with left-hand side 
   $f(\#/_{\Theta\langle F ,E\rangle}, \#/_{\Theta\langle F ,E \rangle})$  in $R\langle  F, E\rangle$. 
     \end{itemize}
   Consequently, 
     GTRS $R\langle F, E \rangle$  keeps up with GTRS 
$R\langle E\cup F, \emptyset \rangle$ writing $\#$, and \begin{quote} $\#/_{\Theta\langle E \cup F, \emptyset \rangle}.keeps_F =\mathrm{true}$.  \end{quote}

    We now show that  GTRS $R\langle F, E \rangle$ does not  keep up with  GTRS 
$R\langle E\cup F, \emptyset \rangle$ writing $\$/_{\Theta\langle E \cup F, \emptyset \rangle}$.
 Recall that  
   \begin{itemize}
  \item  $g(\$/_{\Theta\langle E\cup F, \emptyset\rangle}) \dot{\rightarrow}  \$/_{\Theta\langle E\cup F, \emptyset\rangle}\in  R\langle E \cup F, \emptyset \rangle$, 
   
 \item   $g(g(\$))/_{\Theta\langle F, E \rangle}) \subseteq \$/_{\Theta\langle E \cup F, \emptyset \rangle})$, and 
   
  \item  in  $R\langle E \cup F, \emptyset \rangle$, there does not exist a rewrite rule with left-hand side 
   $g(g(g(\$))/_{\Theta\langle F, E \rangle})$. 
   \end{itemize}
      Thus 
    GTRS $R\langle F, E \rangle$ does not   keep up with  GTRS 
$R\langle E\cup F, \emptyset \rangle$ writing $\$/_{\Theta\langle E \cup F, \emptyset \rangle}$, and  
\begin{quote}
 $\$/_{\Theta\langle E \cup F, \emptyset \rangle}.keeps_F =\mathrm{false}$.  
  \end{quote}

  \end{exa}
  
 \section{Constructing the Auxiliary  dpwpa  for the GTESs $E$ and $F$}\label{epit}

We adopt and slightly modify the fast ground completion  algorithm of 
 F\"ul\"op and V\'agv\"olgyi in Definition \ref{szabaly}. This way, we produce  the GTRSs  $R\langle E, F  \rangle$, GTRS $R\langle  F, E\rangle$, and $R\langle E\cup F, \emptyset\rangle$, and 
construct the  sets and relations and dpwpa, which we introduced in Section \ref{three-way}  and our decision procedure needs. 

We define a mapping $ counter_E: R \langle E \cup F, \emptyset\rangle \rightarrow \Nat$, 
to each rewrite rule $\s({b}_1, \ldots, {b}_m) \dot{\rightarrow}{b}$, $m\in \Nat$, 
in $ R \langle E \cup  F,  \emptyset \rangle$, we assign the  number of rules $\s({c}_1, \ldots, {c}_m) \dot{\rightarrow} {c}$ in $ R \langle  E, F\rangle $,
where
   ${c}_i \subseteq  {b}_i$ for every  $i\in \{\, 1, \ldots, m\, \}$.
 In pseudocode we treat the nonnegative integer $counter_E(\s({b}_1, \ldots, {b}_m) \dot{\rightarrow} {b})   $ as  satellite data 
 of   the rewrite rule
 $\s({b}_1, \ldots, {b}_m) \dot{\rightarrow}{b}$, and denote it as  
 $\s({b}_1, \ldots, {b}_m) \dot{\rightarrow} {b}.counter_E$. In this way, we store the mapping  $counter_E$.
    We define the  mapping $ counter_F: R \langle E \cup  F, \emptyset \rangle \rightarrow \Nat$ and treat $counter_F(\s({b}_1, \ldots, {b}_m) \dot{\rightarrow} {b})   $  in a similar way.

\begin{df}\label{augbarack}\rm 
Algorithm Constructing Auxiliary   dpwpa $ AUX[ E;  F] $
(CAD for short)\newline 
Input: 
 A signature $ \Sigma$  and  GTESs $E$ and $F$ over $\S$.
\newline 
Output:
\begin{enumerate}
\item the relations  $\rho\langle E, F\rangle$, $\rho\langle F, E\rangle$,  and $\rho\langle E\cup F, \emptyset\rangle$;
\item the  set of nullary symbols  $C\langle E, F\rangle$, $C\langle F, E\rangle$, and
$C\langle E \cup F, \emptyset\rangle$;
\item  $INC( E\cup F, E)$ stored  without duplicates in a red-black tree saved in a  variable  $RBT1$,
 for every ordered pair$({b}, {c})\in INC\langle E\cup F, E \rangle$, the search key is the first component    ${b}$, the second component,    ${c}$ is  satellite data;
\item  $INC(E\cup F, F)$ stored in a red-black tree  saved   without duplicates in a  variable    $RBT2$,
for every ordered pair$({b}, {c})\in INC( E\cup F, F)$, the search key is the first component    ${b}$, the second component,    ${c}$ is  satellite data;
\item $NUM(  E\cup F, E)$ stored in a red-black tree  saved   without duplicates in a  variable    $RBT3$,
for every ordered pair$({b}, k)\in NUM(E\cup F, E )$, the search key is the first component    ${b}$, the second component, $k$ is  satellite data;
\item $NUM( E\cup F, F)$ stored in a red-black tree  saved   without duplicates in a  variable    $RBT4$
for every ordered pair$({b}, k)\in NUM(E\cup F, F )$, the search key is the first component    ${b}$, the second component, $k$ is  satellite data;
\item the GTRS $R\langle E\cup F, \emptyset\rangle$  stored in a red-black tree saved   without duplicates in a  variable  $RBT5$;
 each node  in  $RBT5$ contains a serach key which is  a rule in $R\langle E\cup F, \emptyset\rangle$, and satellite data $counter_E$ and $counter_F$;
\item $BSTEP\langle E\cup F, \emptyset \rangle$ stored in a red-black tree saved   without duplicates in a  variable   $RBT6$; 
for every ordered pair$({b}, {c})\in BSTEP\langle E\cup F, \emptyset \rangle$, the search key is the first component    ${b}$, the second component,    ${c}$ is  satellite data;
 \item the GTRS $R\langle E, F\rangle$ stored in a red-black tree saved   without duplicates in a  variable  $RBT7$;
  \item the GTRS $R\langle  F, E\rangle$ stored in a red-black tree  saved   without duplicates in a  variable   $RBT8$; 
 \item 
the auxiliary dpwpa $ AUX[ E;  F] =
(C\langle E\cup F, \emptyset \rangle, A[E; F])$ for the GTESs $E$ and $F$, and 
   for every vertex    ${b}\in C\langle E \cup  F,  \emptyset \rangle$,
the values of  the  attributes    ${b}.equal_E$,    ${b}.equal_F$,    ${b}.keeps_E$, and    ${b}.keeps_F$;
\end{enumerate} 
\begin{quote} {\bf var} $i$, $m$: integer;  $\s$:  symbol in $\Sigma$; \newline
   ${b}, {b}_1, \ldots, {b}_{mar_{\Sigma}}: C \langle E \cup F, \emptyset \rangle$;\;
   ${c}, {c}_1,  \ldots, {c}_ {mar_{\Sigma}}:C \langle E, F \rangle$; \;
   ${d},{d}_1,  \ldots, {d}_{mar_{\Sigma}}:C \langle F, E \rangle$; \newline
$RBT1$, $RBT2$, $RBT3$, $RBT4$, $RBT6$, $RBT5$, $RBT7$, $RBT8$: red-black trees.
\end{quote}
(Here comes the main part of the algorithm CAD. We run  the fast ground completion  algorithm 
  in Definition \ref{szabaly} for each of GTESs $E$, $F$, and $E\cup F$, in this way we compute data items 1, 2, 7,  9, and 10 above. During the run of the algorithm CAD, it  computes  also  the other data items listed above. Our algorithm consists of six steps:
 \begin{enumerate}
\item We create a subterm dag $ G=( V, A, \lambda)$ for  $E$  and $F$. 
\item  We compute the relations  $INC(  E\cup F, E )$,  $INC( E\cup F , F )$, $NUM( E\cup F, E)$, $NUM(E\cup F, F)$, $BSTEP\langle E\cup F, \emptyset\rangle$, the GTRSs $R\langle E \cup F, \emptyset\rangle$, $R\langle E, F\rangle$, $R\langle F, E\rangle$, the values of the attributes  $equal_E$, $equal_F$ for every vertex in $C\langle E, F\rangle$,  and the   satellite data $counter_E$, and $counter_F$ for every rule in $R\langle E \cup F, \emptyset\rangle$.
\item We compute the relations  $NUM( E\cup F, E)$ and   $NUM(E\cup F, F)$. 
\item We create the  adjacency-list representation of the auxiliary dpwpa $ AUX[ E;  F] =
(C\langle E\cup F, \emptyset \rangle, A[E; F])$ for the GTESs $E$ and $F$. 
\item  We compute the values of the attribute $keeps_E$. \item   We compute the values of the attribute  $keeps_F$.)
\end{enumerate}

\begin{lstlisting}
(*1.\{ Create a subterm dag $ G=( V, A, \lambda)$ for  $E$  and $F$. \}*)
(*By  Proposition \ref{lezaras}, compute the congruence  closure $\rho\langle E, F\rangle$  of the relation $\tau\langle E, F\rangle=\{\, (u, v)\mid (\hat{u}, \hat{v}) \in E\, \}.$*)
(*$C\langle E, F\rangle:= \{\, x/ _{\rho\langle E, F\rangle} \mid x\in V \,\}$.*)
(*By  Proposition \ref{lezaras}, compute the congruence  closure $\rho\langle F, E\rangle$  of the relation $\tau\langle F, E\rangle=\{\, (u, v)\mid (\hat{u}, \hat{v}) \in F\, \}.$*)
(*$C\langle F, E\rangle:= \{\, x/ _{\rho\langle F, E\rangle} \mid x \in V\,\}$.*)
(*Compute  $\tau\langle E\cup F, \emptyset\rangle$ using the equality  $\tau\langle E\cup F, \emptyset\rangle=\tau\langle E, F\rangle \cup \tau\langle F, E\rangle$.*) 
(*By  Proposition \ref{lezaras}, compute the congruence  closure $\rho\langle E\cup F, \emptyset\rangle$  of the relation $\tau\langle E\cup F, \emptyset\rangle=\{\, (u, v)\mid (\hat{u}, \hat{v}) \in E\cup F\, \}.$*)
(*$C\langle E\cup F, \emptyset\rangle:= \{\, x/ _{\rho\langle E\cup F, \emptyset\rangle} \mid x \in V \,\}$.*)

(* 2.\{ Computing the relations  $INC(  E\cup F, E )$,  $INC( E\cup F , F )$, $NUM( E\cup F, E)$, $NUM(E\cup F, F)$, $BSTEP\langle E\cup F, \emptyset\rangle$, the GTRSs $R\langle E \cup F, \emptyset\rangle$, $R\langle E, F\rangle$, $R\langle F, E\rangle$, the values of the attributes  $equal_E$, $equal_F$ for every vertex in $C\langle E, F\rangle$,  and the   satellite data $counter_E$, and $counter_F$ for every rule in $R\langle E \cup F, \emptyset\rangle$. \}*)
 (*{\bf begin}*)
    (*$R\langle E, F\rangle:=\emptyset$;*)
    (*$R\langle F, E\rangle:=\emptyset$;*)  
    (*$R\langle E \cup F, \emptyset\rangle:=\emptyset$;*) 
    (*$INC( E\cup F, E):=\emptyset$;*)
    (*$INC(  E\cup F, F):=\emptyset$;*)
    (*$NUM( E\cup F, E):=\emptyset$;*)
    (*$NUM(  E\cup F, F):=\emptyset$;*)
    (*$BSTEP\langle E\cup F, \emptyset\rangle :=\emptyset$*)
 (*{\bf end}*)
 (*{\bf for} each    ${d}\in C\langle E\cup F, \emptyset \rangle $ {\bf do}*) 
   (*{\bf begin}*)
     (*   ${d}.equal_E=\mathrm{\mathrm{false}};$*)
     (*   ${d}.equal_F=\mathrm{\mathrm{false}};$*) 
   (*{\bf end}*)  
 (*{\bf for} each $x\in V$ {\bf do}*)
   (*{\bf begin}*)
     (*$\s:= \lambda(x)$;*)
     (*$m := outdegree(x);$*)
     (*   ${b} := \mathrm{FIND}_{ \rho\langle E, F \rangle  }(x)$;*)
     (*   ${c} := \mathrm{FIND}_{ \rho\langle F, E \rangle  }(x)$;*)
     (*   ${d} := \mathrm{FIND}_{ \rho\langle E \cup F, \emptyset \rangle  }(x)$;*)
     (*$INC(   E\cup F, E):=INC( E\cup F, E )\cup \{\, ( {d},{b}) \, \}$;*)
     (*(store $INC(  E\cup F, E)$ without duplicates   in the  red-black tree $RBT1$.)*)
     (*$INC( E\cup F, F):=INC(E\cup F, F )\cup \{\, ({d},{c}) \, \}$;*)
     (*(store $INC( E\cup F, F)$ without duplicates   in the  red-black tree $RBT2$.)*)
     (*{\bf if} $|{b}|=|{d}|$*)
       (*   ${\bf then}$    ${d}.equal_E:=\mathrm{\mathrm{true}};$*)  
     (*{\bf if} $|{c}|=|{d}|$*)
       (*   ${\bf then}$    ${d}.equal_F:=\mathrm{\mathrm{true}};$*) 
     (*(By Proposition  \ref{lezaras} we know the cardinalities of all  $W/_{\Theta\langle E, F \rangle}$-classes,  $W/_{\Theta\langle F, E\rangle}$-classes, and  $W/_{\Theta\langle E \cup F, \emptyset\rangle}$-classes.)*)    
     (*{\bf for} {\bf each} $1 \leq i \leq m $ {\bf do}*)
       (*{\bf begin}*)
         (*   ${b}_i := \mathrm{FIND}_{  \rho\langle E, F \rangle }(x/i)$;*)
         (*   ${c}_i := \mathrm{FIND}_{  \rho\langle F, E \rangle }(x/i)$;*)
         (*   ${d}_i := \mathrm{FIND}_{  \rho\langle E \cup F, \emptyset \rangle }(x/i)$;*)
       (*{\bf end}*)  
       (*{\bf if} $\s({d}_1, \ldots, {d}_m) \dot{\rightarrow} {d}\not \in R\langle E \cup F, \emptyset\rangle$*)
         (*{\bf then begin}*)
           (*$R\langle E \cup F, \emptyset\rangle:=R\langle E \cup F, \emptyset\rangle\cup \{\, \s({d}_1, \ldots, {d}_m) \dot{\rightarrow} {d} \, \}$;*)
           (*(Store $R\langle E \cup F, \emptyset\rangle$ without duplicates   in the  red-black tree $RBT5$.)*)
           (*$\s({d}_1, \ldots, {d}_m) \dot{\rightarrow} {d}.counter_E :=0$;*)
           (*$ \s({d}_1, \ldots, {d}_m) \dot{\rightarrow} {d}.counter_F:=0$;*)
           (*$BSTEP\langle  E\cup F, \emptyset\rangle :=BSTEP\langle  E\cup F, \emptyset\rangle \cup \{\, ({d}_1, {d}), \ldots, ({d}_m, {d})\, \}$;*)
           (*(Store $BSTEP\langle  E\cup F, \emptyset\rangle $  without duplicates in the  red-black tree $RBT6$)*)
         (*{\bf end}*)  
       (*{\bf if} $\s({b}_1, \ldots, {b}_m) \dot{\rightarrow} {b}\not \in R\langle E, F\rangle$*)
         (*{\bf then begin}*)
           (*$R\langle E, F\rangle:=R\langle E, F\rangle\cup \{\, \s({b}_1, \ldots, {b}_m) \dot{\rightarrow} {b}\, \}$;*)
           (*(Store $R\langle E, F\rangle$ without duplicates   in the  red-black tree $RBT7$.)*)
           (*$ \s({d}_1, \ldots, {d}_m) \dot{\rightarrow} {d}.counter_E  :=   \s({d}_1, \ldots, {d}_m) \dot{\rightarrow} {d}.counter_E+1$;*)
         (*{\bf end}*)
       (*{\bf if} $\s({c}_1, \ldots, {c}_m) \dot{\rightarrow} {c}\not \in R\langle F, E\rangle$*)
         (*{\bf then begin}*)
           (*$R\langle F, E\rangle:=R\langle F, E\rangle\cup \{\, \s({c}_1, \ldots, {c}_m) \dot{\rightarrow} {c}\, \}$;*)
           (*(Store $R\langle F, E\rangle$ without duplicates   in the  red-black tree $RBT8$.)*)
           (*$ \s({d}_1, \ldots, {d}_m) \dot{\rightarrow} {d}.counter_F  :=   \s({d}_1, \ldots, {d}_m) \dot{\rightarrow} {d}.counter_F+1$;*)
         (*{\bf end}*)
   (*{\bf end}*)

(*3. \{ Computing the relations  $NUM( E\cup F, E)$ and   $NUM(E\cup F, F)$. \}*)
(*Read the  nodes  of the red-black tree $RBT1$ in search key  order, during this process create the red-black tree $RBT3$ containing without duplicates the relation $NUM(  E\cup F, E)$;*)
(*Read the   nodes  of the red-black tree $RBT2$ in search key  order, during this process create the red-black tree $RBT4$ containing  without duplicates the relation $NUM( E\cup F, F)$;*)

(*4. \{ Creating the  adjacency-list representation of the auxiliary dpwpa $ AUX[ E;  F] =
(C\langle E\cup F, \emptyset \rangle, A[E; F])$ for the GTESs $E$ and $F$ \}*)
(* We build the adjacency-list representation of the dpwpa    $ AUX[ E;  F] =
(C\langle E\cup F, \emptyset\rangle,A[E; F])$ by reading in search key  order the nodes  of the black-red tree   $RBT6$ storing the relation  $BSTEP\langle E\cup F, \emptyset\rangle$. Recall that for every ordered pair $({b}, {c})\in BSTEP\langle E\cup F, \emptyset \rangle$, the search key is the first component    ${b}$, the second component,    ${c}$ is  satellite  data.*)

(*5. \{ Computing the values of the attribute $keeps_E$. \}*) 
(*{\bf for} {\bf each}     ${b}\in C\langle E\cup F, \emptyset\rangle$  {\bf do} *)
  (*   ${b}.keeps_E:=\mathrm{true}$;*)
(*{\bf for} {\bf each} rewrite rule $\s({b}_1, \ldots, {b}_m) \dot{\rightarrow} {b}\in R\langle E \cup F,   \emptyset \rangle  $  with $m\in\Nat$ {\bf do}*)  
  (*{\bf if} $\s({b}_1, \ldots, {b}_m) \dot{\rightarrow} {b}.counter_E  \neq \Pi_{i=1} ^ m |{b}_i/_{\Theta \langle E, F\rangle}|$*)
      (*{\bf then}    ${b}.keeps_E:=\mathrm{false}$*)
   
(*6. \{ Computing the values of the attribute  $keeps_F$. \}*) 
(*{\bf for} {\bf each}     ${b}\in C\langle E \cup  F, \emptyset \rangle$  {\bf do} *)
  (*   ${b}.keeps_F:=\mathrm{true}$;*)
(*{\bf for} {\bf each} rewrite rule $\s({b}_1, \ldots, {b}_m) \dot{\rightarrow} {b}\in R\langle E \cup F,  \emptyset  \rangle  $ with $m\in\Nat$ {\bf do}*)  
  (*{\bf if} $\s({b}_1, \ldots, {b}_m) \dot{\rightarrow} {b}.counter_F  \neq \Pi_{i=1} ^ m |{b}_i/_{\Theta \langle F, E\rangle}|$*)
      (*{\bf then}    ${b}.keeps_F:=\mathrm{false}$*)
\end{lstlisting}
\end{df}
\begin{lem}\label{eso} Let 
 $E$ and $F$ be  GTESs over  a signature  $\Sigma$. 
 \newline
  1. Algorithm CAD on $\Sigma$, $E$, and $F$ produces    
  \begin{itemize}
\item the relations  $\rho\langle E, F\rangle$, $\rho\langle F, E\rangle$,  and $\rho\langle E\cup F, \emptyset\rangle$;
\item the  set of nullary symbols  $C\langle E, F\rangle$, $C\langle F, E\rangle$, and
$C\langle E \cup F, \emptyset\rangle$;
\item  $INC( E\cup F, E) $ stored  without  duplicate values  in a red-black tree   saved in a  variable  $RBT1$;
\item  $INC( E\cup F, F)$ stored  without  duplicate values  in a red-black tree  saved in a  variable    $RBT2$;
\item $NUM(E\cup F, E)$ stored  without  duplicate values  in a red-black tree  saved in a  variable    $RBT3$;
\item $NUM(E\cup F, F)$ stored  without  duplicate values  in a red-black tree  saved in a  variable    $RBT4$;
\item the GTRS $R\langle E\cup F, \emptyset\rangle$  stored without  duplicate values in a red-black tree saved in a  variable  $RBT5$,
 each node  in  $RBT5$ contains a search key  which is  a rule in $R\langle E\cup F, \emptyset\rangle$, and satellite data $counter_E$ and $counter_F$;
 \item $BSTEP\langle E\cup F, \emptyset \rangle$ stored without  duplicate values in a red-black tree saved in a  variable   $RBT6$; 
for every ordered pair$({b}, {c})\in BSTEP\langle E\cup F, \emptyset \rangle$, the search key is the first component    ${b}$, the second component,    ${c}$ is  satellite data;
  \item the GTRS $R\langle E, F\rangle$ stored without  duplicate values in a red-black tree saved in a  variable  $RBT7$;
  \item the GTRS $R\langle  F, E\rangle$ stored without  duplicate values in a red-black tree  saved in a  variable   $RBT8$; 
 \item 
the auxiliary dpwpa $ AUX[ E;  F] =
(C\langle E\cup F, \emptyset \rangle, A[E; F])$ for the GTESs $E$ and $F$.
\end{itemize}  
   
 2.    Algorithm CAD on $\Sigma$, $E$, and $F$  decides
 for every constant    ${b}\in C\langle E\cup F, \emptyset \rangle$, 
 whether     ${b}\in C\langle E, F \rangle$ and whether     ${b}\in C\langle F, E \rangle $, and  
    sets the value of the attribute    ${b}.equal_E$ to true if and only if    ${b}\in C\langle E, F \rangle$, and sets the value of the attribute    ${b}.equal_F$ to true if and only if    ${b}\in C\langle F, E \rangle$.

   3.  Algorithm CAD on $\Sigma$, $E$, and $F$ decides for any    ${d}\in  R\langle {E \cup F}, \emptyset\rangle$ and any rewrite rule 
 $\s({d}_1, \ldots, {d}_m) \dot{\rightarrow} {d}\in  R\langle {E \cup F}, \emptyset \rangle$ with right-hand side    ${d}$,
 whether  for all    ${b}_1, \ldots, {b}_m  \in C\langle E, F\rangle$  such that 
    ${b}_i \subseteq  {d}_i$ for every  $i\in \{\, 1, \ldots, m\, \}$, 
 there exists a  rewrite rule  in $R\langle E, F \rangle$ with left-hand side 
 $\s({b}_1, \ldots, {b}_m) $. If so, then it sets  the value of the attribute    ${d}.keeps_E$ to $\mathrm{true}$, otherwise it sets 
 the value of the attribute    ${d}.keeps_E$ to $\mathrm{false}$.

   4. Algorithm CAD on $\Sigma$, $E$, and $F$ decides for any      ${d}\in  R\langle {E \cup F}, \emptyset \rangle$ and any rewrite rule 
 $\s({d}_1, \ldots, {d}_m) \dot{\rightarrow} {d}\in  R\langle {E \cup F}, \emptyset\rangle$ with right-hand side    ${d}$,
     whether  for all    ${c}_1, \ldots, {c}_m  \in C\langle F, E\rangle$  such that 
    ${c}_i \subseteq  {d}_i$ for every  $i\in \{\, 1, \ldots, m\, \}$, 
 there exists a  rewrite rule  in $R\langle F, E \rangle$ with left-hand side 
 $\s({c}_1, \ldots, {c}_m) $. If so, then it sets  the value of the attribute    ${d}.keeps_F$ to $\mathrm{true}$, otherwise it sets 
 the value of the attribute    ${d}.keeps_F$ to $\mathrm{false}$. 
 
  \end{lem}
\begin{proof}
By Proposition \ref{kecskemet} and by direct inspection of Definition \ref{augbarack}, we obtain Statement 1. 

By Proposition \ref{lezaras},  we know the cardinality of all  $W/_{\Theta\langle E, F \rangle}$-classes,  $W/_{\Theta\langle F, E\rangle}$-classes, and 
 $W/_{\Theta\langle E \cup F, \emptyset\rangle}$-classes. For all     ${b} \in R\langle E,F  \rangle$
 and 
     ${d} \in R\langle E\cup F, \emptyset  \rangle$ such that    ${b}\cap {d} \neq \emptyset$, we have    ${b}={d}$ if and only if $|{b}|=|{d}|$.
  Hence we have Statement 2.
 
We now show Statement 3.
 For every  rewrite rule $\s({d}_1, \ldots, {d}_m)\dot{\rightarrow} {d} \in R\langle E\cup F, \emptyset  \rangle$, we decide 
whether for all    ${b}_1, \ldots, {b}_m\in C\langle E, F\rangle$ with     ${b}_1\subseteq {d}_1, \ldots, {b}_m\subseteq {d}_m$,
there exists a rewrite rule with left-hand side 
 $\s({b}_1, \ldots, {b}_m)$ in $R\langle E, F \rangle$.  When constructing the 
 GTRS $R\langle E, F \rangle$, whenever we add to  $R\langle E, F \rangle$ a new rule with left-hand side $\s({b}_1, \ldots, {b}_m)$,
 where    ${b}_1\subseteq {b}_1, \ldots, {b}_m\subseteq {d}_m$,
 we add $1$ to the
 attribute $ \s({d}_1, \ldots, {d}_m) \dot{\rightarrow} {d}.counter_E $. 
  In this way, we count the rules  in $R\langle  E, F\rangle$  with left-hand side 
 $\s({b}_1, \ldots, {b}_m)$,
  where    ${b}_1\subseteq {d}_1, \ldots, {b}_m\subseteq {d}_m$,
   and store the result in the attribute $ \s({d}_1, \ldots, {d}_m) \dot{\rightarrow} {d}.counter_E $. 
  Then $\s({d}_1, \ldots, {d}_m) \dot{\rightarrow} {d}.counter_E  =\Pi_{i=1} ^ m |{d}_i/_{\Theta \langle E, F\rangle}|$ if and only if 
  for all    ${b}_1, \ldots, {b}_m\in C\langle E, F\rangle$
 with     ${b}_1\subseteq {d}_1, \ldots, {b}_m\subseteq {d}_m$,
 there exists a  rewrite rule with left-hand side $\s({b}_1, \ldots, {b}_m)$
 in $ R\langle E, F \rangle$. 
   When Algorithm   CAD terminates,  for every vertex    ${d}\in C\langle E\cup F, \emptyset \rangle$,
\begin{quote} 
   ${d}.keeps_E=\mathrm{true}$ if and only if \newline 
 for every rewrite rule  of the form 
 $\s({d}_1, \ldots, {d}_m) \dot{\rightarrow} {d}\in  R\langle {E \cup F}, \emptyset \rangle$, \newline
 for every     ${b}_1, \ldots, {b}_m \in C\langle E, F\rangle$ with    ${b}_i \subseteq  {d}_i$
 for all $i\in \{\, 1, \ldots, m\, \}$,  \newline
  there exists a rewrite rule in $ R\langle E, F\rangle$ of the form 
 $\s({b}_1, \ldots, {b}_m )\dot{\rightarrow}{b}$ for some    ${b}\in C\langle E, F\rangle$.
\end{quote}

The proof of Satement 4 is similar to that of Statement 3.
  \end{proof}

 \begin{sta}\label{porosz}
 1. We can decide in    $O(n)$  time  whether 
  for all   ${b}\in C \langle E \cup F, \emptyset \rangle $,   ${b}\in C \langle E, F \rangle $. \newline  
  2. We can decide in    $O(n)$  time whether 
  for all   ${b}\in C \langle E \cup F, \emptyset \rangle $,    ${b}\in C \langle F, E \rangle $. 
 \end{sta}
 \begin{proof} First we show Statement 1. We read each  nodes of the   red-black tree  $RBT3$ in lexicographic order. 
 For every   ${b}\in C \langle E \cup F, \emptyset \rangle $, there exists a unique ordered pair $({b}, k) \in NUM( E\cup F, E)$. If $k=1$, 
  then 
  by Statement \ref{dan},   ${b}\in C \langle E, F \rangle $. Otherwise,   ${b} \not\in C \langle E, F \rangle $.

We can show Statement 2 similarly to Statement 1.
 \end{proof}

  \begin{lem}\label{bonn} Let 
 $E$ and $F$ be  GTESs over  a signature  $\Sigma$. 
Algorithm CAD runs in $O(n\, \mathrm{log}\, n)$   time  on $\Sigma$, $E$, and $F$, 
  where $n= \mathrm{size}(E)+\mathrm{size}(F)$.  \end{lem}

 \begin{proof} We show  Statements \ref{egy1} -- \ref{kilenc9}.

\begin{sta}\label{egy1}
   The time complexity of constructing the red-black trees $RBT1$ and $RBT2$ is $O(n)$. 
  \end{sta}
\begin{proof} 
  By Statement \ref{csorvas},   $INC( E\cup F, E)\leq n $. Consequently, 
when we construct the  red-black tree $RBT1$, we carry out at most $n$ insertion. By Proposition \ref{redblack}, 
  we construct the  red-black tree $RBT1$ in $O(n\, \mathrm{log}\, n)$ time. 
 Similarly, we construct the  red-black tree $RBT2$ in $O(n\, \mathrm{log}\, n)$ time. 
 \end{proof}

By Statement \ref{erre}, we have the following result.    
\begin{sta}\label{ketto2}
   The time complexity of constructing  the relation $\tau\langle E\cup F, \emptyset\rangle$ 
and  the  GTES $E \cup F$  stored without  duplicate values in a red-black tree saved in the  variable  $RBT5$,
  is  $O(n^2)$ time.
  \end{sta}

We carry out $2 n$ comparisons of nonnegative numbers in $O(n)$ time. Hence we have the following fact.
\begin{sta}\label{harom3}
The time complexity of computing the value attributes     ${b}.equal_E$ and     ${b}.equal_F$ for all constants    ${b}\in C\langle E\cup F, \emptyset \rangle$, is $O(n)$. 
     \end{sta}
  
By Statement 2 and Statement 3 of \ref{kecskemet}, we have the following fact. 
\begin{sta}\label{negy4}
 The time complexity of constructing the following data is  $O(n\, \mathrm{log}\, n)$.
 \begin{itemize} 
\item the GTRS $R\langle E\cup F, \emptyset\rangle$  stored without  duplicate values in a red-black tree saved in the  variable  $RBT5$,
 each node  in  $RBT5$ contains a search key  which is  a rule in $R\langle E\cup F, \emptyset\rangle$;
 \item $BSTEP\langle E\cup F, \emptyset \rangle$ stored in a red-black tree saved in the  variable   $RBT6$; 
   \item the GTRS $R\langle E, F\rangle$ stored without  duplicate values in a red-black tree saved in the  variable  $RBT7$;
  \item the GTRS $R\langle  F, E\rangle$ stored without  duplicate values in a red-black tree  saved in the  variable   $RBT8$; 
 \end{itemize}
  \end{sta}

\begin{sta}\label{ot5}
 The time complexity of constructing
\begin{itemize}
\item $NUM( E\cup F, E)$ stored  without  duplicate values  in a red-black tree  saved in a  variable    $RBT3$, and 
\item $NUM( E\cup F, F)$ stored  without  duplicate values  in a red-black tree  saved in a  variable    $RBT4$
\end{itemize} 
is $O(n)$. 
   \end{sta}
\begin{proof}
When we construct the  red-black tree $RBT3$, we read the whole red-black tree $RBT1$ in search key order in $O(n)$ time. Each  modification takes 
constant  time.
By Statement \ref{csorvas},   $|INC( E\cup F, E)|\leq n $.   Consequently, we construct the  red-black tree $RBT3$ in $O(n)$ time. 
 Similarly, we construct the  red-black tree $RBT4$ in $O(n)$ time. 
\end{proof}

 \begin{sta}\label{hat6}
  The time complexitity  of computing the values of 
 the   satellite data $counter_E$, and $counter_F$ for all rules in $R\langle E \cup F, \emptyset\rangle$
  is $O(n)$. 
  \end{sta}
\begin{proof}
 When Algorithm CAD computes the values of the  satellite data  
    $counter_E$ and $counter_F$ for all rules of  the GTRS $R\langle E\cup F, \emptyset \rangle$,  
    it carries out  at most $2  n$  additions, since
   $|  C\langle E, F\rangle|\leq    |W|\leq n$ and       $|  C\langle F, E\rangle|\leq    |W|\leq n$.   
    \end{proof}

 \begin{sta}\label{het7}
  The time complexity of constructing the 
 auxiliary dpwpa $ AUX[ E;  F] =
(C\langle E\cup F, \emptyset\rangle, A[E; F])$
  is $O(n)$.
   \end{sta}
\begin{proof}
 We build the adjacency-list representation of the dpwpa    $ AUX[ E;  F] =
(C\langle E\cup F, \emptyset\rangle,A[E; F])$ by reading in search key  order the nodes  of the black-red tree   $RBT6$
storing the relation  $BSTEP\langle E\cup F, \emptyset\rangle$. Recall that for every ordered pair
$({b}, {c})\in BSTEP\langle E\cup F, \emptyset \rangle$, the search key is the first component    ${b}$, the second component,    ${c}$ is  satellite 
data. 
Then the time complexity of building  the adjacency-list representation of the dpwpa    $ AUX[ E;  F]  $ is 
$O(|C\langle E\cup F, \emptyset\rangle|+ |A[E; F]| )$.
By Statement \ref{darabsz},
$|C\langle E\cup F, \emptyset\rangle|\leq |W|\leq n$.
By Statement \ref{iszap}, 
 $|A[E; F]|=|BSTEP\langle E\cup F, \emptyset\rangle|\leq \mathrm{size}(E) + size (F)=2 n$. 
Therefore the time complexity of building  the adjacency-list representation of the dpwpa    $ AUX[ E;  F] $ 
is    $O(n)$. 
\end{proof}
 Statement \ref{porosz} implies the following fact.
 \begin{sta}\label{nyolc8}
  The time complexitity  of computing the values of  the  attributes $equal_E$ 
 and $equal_F$   of  all constants    ${b}\in C\langle E\cup F, \emptyset \rangle$ is   $O(n)$. 
   \end{sta}

 \begin{sta}\label{kilenc9}
 The time complexitity  of computing the values of the  attributes  $keeps_E$ and 
  $keeps_F$ of  all constants    ${b}\in C\langle E\cup F, \emptyset \rangle$  is  $O(n\, \mathrm{log}\, n)$ time. 
 \end{sta}
\begin{proof}
  We now estimate the number of operations when computing 
 the values of the attributes    ${b}.keeps_E$  for all    ${b}\in C\langle E\cup F, \emptyset\rangle$.
  For
 each rewrite rule  
$\s({b}_1, \ldots, {b}_m) \dot{\rightarrow} {b}\in R\langle E \cup F, \emptyset\rangle$, 
Algorithm CAD carries out $m$ searches in the red-black tree $RBT3$
and  less than or equal to  $m$ multiplications. 
By Statement \ref{drrszam}, 
 $\mathrm{size}(R\langle E\cup F, \emptyset\rangle)\leq 2 n$.
In this way, when we compute the values of the attribute $keeps_E$ for all  rewrite rules  
$\s({b}_1, \ldots, {b}_m) \dot{\rightarrow} {b}\in R\langle E \cup F, \emptyset\rangle$, the number of searches in the red-black tree $RBT3$
is less than equal to $n$. Consequently, by Proposition \ref{keres}, 
 the  searches in the red-black tree $RBT3$ take  $O(n \, \mathrm{log} \, n )$   time.

Furthermore, 
the number of multiplications is less than or equal to 
 $\sum _{\s({b}_1, \ldots, {b}_m) \dot{\rightarrow} {b}
\in  R\langle E \cup F, \emptyset\rangle} m\leq 
\mathrm{size}(R\langle E \cup F, \emptyset \rangle)\leq 
2 n$ by Statement \ref{drrszam}.
\noindent 
Thus, the number of multiplications is less than or equal to $2 n$.

 For each    ${b}\in C\langle E\cup F, \emptyset\rangle$, the number of comparisons of nonnegative integers during computing the value of the attribute    ${b}.keeps_E$ is equal to 
the number of rules in $R\langle  E \cup F, \emptyset    \rangle$ with right-hand side    ${b}$. 
 Consequently, the number of comparisons of nonnegative integers during computing the values of the attributes    ${b}.keeps_E$  for all    ${b}\in C\langle E\cup F, \emptyset\rangle$ is equal to $|R\langle  E\cup F, \emptyset    \rangle|$. By Statement \ref{drrszam}
  $|R\langle  E\cup F, \emptyset    \rangle| \leq n $. 
  Consequently the time complexity of computing for all    ${b}\in C\langle E\cup F, \emptyset\rangle$, 
the values of the attribute    ${b}.keeps_E$ is   $O(n \, \mathrm{log} \, n )$.   
Similarly, the time complexity of computing for all    ${b}\in C\langle E\cup F, \emptyset\rangle$, the values of the attribute    ${b}.keeps_F$  is 
  $O(n \, \mathrm{log} \, n )$.  
\end{proof}

The Lemma follows from  Statements \ref{egy1} -- \ref{kilenc9}. 
\end{proof}


\section{Preparatory  Results}\label{preparatory}
 The  decidability and complexity results of the following sections require some preparatory results which we group together here. 
\begin{prop} \label{valodi}
    Let $A$ be a set,    ${b}\in A$, and let $\rho,\tau, \omega$ be 
equivalence relations on $A$ such that     ${b}/_\rho \subset  {b}/_\omega   $ and
   ${b}/_\tau \subset  {b}/_\omega$. Then    ${b}/_\rho \cup  {b}/_\tau\subset  {b}/_\omega$.
    \end{prop}
  \begin{proof}
 There exist    ${c}, {d}\in A$, such that    ${c}\in  {b}/_\omega \setminus  {b}/_\rho $ and
     ${d}\in  {b}/_\omega \setminus  {b}/_\tau $. 
  We now distinguish three cases.

  {\em Case 1: }
  $({c}, {d} )\in \rho $. Then  by the definition of    ${c}$, $({b}, {d})\not \in \rho$. By the definition of    ${d}$, $({b}, {d})\not \in \tau$. Thence
  $({b}, {d})\not \in \rho\cup \tau$.

  {\em Case 2: }
  $({c}, {d} )\in \tau $. Then by the definition of    ${d}$,
  $({b}, {c})\not \in \tau$.  By the definition of    ${c}$, $({b}, {c})\not \in \rho$. Consequently
  $({b}, {c})\not \in \rho\cup \tau$.

   {\em Case 3: }
   $({c}, {d} )\not \in \rho \cup \tau$. Then by the definition of    ${c}$ and    ${d}$,
   $({c}, {d} )\in  \omega$, and hence 
   $({c}, {d} )\in  \omega \setminus (\rho \cup \tau)$.

 In all  three cases we have     ${b}/_\rho \cup  {b}/_\tau\subset  {b}/_\omega$.
  \end{proof}

\begin{prop} \label{kovetkezmeny}
    Let $A$ be a set,    ${b}\in A$, and let $\rho,\tau, \omega$ be 
    equivalence relations on $A$ such that
   ${b}/_\rho \cup  {b}/_\tau=   {b}/_\omega$. Then 
       ${b}/_\rho =  {b}/_\omega$ or    ${b}/_\tau =  {b}/_\omega$. 
    \end{prop}
\begin{proof}
 By contradiction, assume that    ${b}/_\rho \neq   {b}/_\omega   $ and     ${b}/_\tau \neq  {b}/_\omega$.
Then      ${b}/_\rho \subset  {b}/_\omega   $ and
   ${b}/_\tau \subset  {b}/_\omega$. Consequently, by Proposition \ref{valodi}, 
   ${b}/_\rho \cup  {b}/_\tau\subset  {b}/_\omega$.
\end{proof}
\begin{lem} \label{teljes} 
 Let 
  $E$ and $F$  be  GTESs over a  signature   $\S$.
If 
GTRS $R\langle E, F\rangle$ is total, then $R\langle E\cup F, \emptyset\rangle$ 
is total as well.
\end{lem}
\begin{proof}
Let $\s\in \Sigma_m$, $m\in \Nat$, and    ${b}_1, \ldots, {b}_m\in  C\langle E\cup F, \emptyset\rangle$. Then there exist $ t_1, \ldots, t_m \in W$ such that
for every $i \in \{\, 1, \ldots, m\, \}$, $t_i/_ {\Theta\langle E\cup F ,  \emptyset \rangle}={b}_i$. 
As GTRS $R\langle E, F\rangle$ is total, there exists a rewrite rule 
$\s(t_1/_  {\Theta\langle E, F \rangle},  \ldots, t_m/_  {\Theta\langle E, F \rangle}) \dot{\rightarrow} t/_  {\Theta\langle E, F \rangle}$ in $R\langle E, F\rangle$. Accordingly, by the definition of $R\langle E, F\rangle$,
 $\s(t_1, \ldots, t_m)\in W$ and  
$\s(t_1, \ldots, t_m)\tthue E t$. Then $\s(t_1, \ldots, t_m)\tthue {E\cup F}  t$. 
For every $i \in \{\, 1, \ldots, m\, \}$, 
$t_i/_  {\Theta\langle E\cup F, \emptyset \rangle}\in C \langle E\cup F, \emptyset\rangle  $.
Thus, by the definition of $R\langle E\cup F,  \emptyset \rangle$, the rewrite rule 
$\s(t_1/_  {\Theta\langle E \cup F , \emptyset \rangle}, \ldots, t_m/_  {\Theta\langle E \cup F, \emptyset \rangle}) \dot{\rightarrow} t/_ {\Theta\langle E \cup F, \emptyset \rangle}$ is in $R\langle E\cup F, \emptyset\rangle$. 
\end{proof}

\begin{prop}\label{egyenlo}
  For any 
  GTESs  $E$ and $F$ over a  signature   $\S$,
  the following four statements are equivalent:
    \begin{itemize}
    \item[1.] $\tthue {E\cup F}\subseteq \tthue E\cup \tthue F$.
    \item[2.] $\tthue {E\cup F}= \tthue E\cup \tthue F$.
    \item[3.]
    There exists a GTES $H$ over $\S$ such that $\tthue H=\tthue E\cup \tthue F$.
    \item[4.] $\tthue E\cup \tthue F$ is a congruence on the ground term algebra      ${\cal T}(\Sigma)$.
      \end{itemize}
\end{prop}
\begin{proof}  
$(1\Rightarrow 2)$ For the GTES  $ E\cup F$, we have
$\tthue {E\cup F}\subseteq \tthue E\cup \tthue F\subseteq \tthue {E\cup F}$.
Thence $\tthue {E\cup F}= \tthue E\cup \tthue F$.

$(2 \Rightarrow 3)$ For the GTES $E\cup F$, we have  $\tthue {E\cup F}= \tthue E\cup \tthue F$.

$(3 \Rightarrow 4)$
$\tthue H$ is a congruence on the ground term algebra      ${\cal T}(\Sigma)$, hence  $\tthue E\cup \tthue F$ is a congruence on the ground term algebra      ${\cal T}(\Sigma)$.

$(4 \Rightarrow 1)$ 
$\tthue E\cup \tthue F$ is a congruence on the ground term algebra      ${\cal T}(\Sigma)$.
 Since $\tthue {E\cup F}$ is the least congruence on $\ts$ containing $E\cup F$,
  we have 
   $\tthue {E\cup F}\subseteq \tthue E\cup \tthue F$.
   \end{proof}
In  light of Proposition \ref{egyenlo}, from now on we decide for 
GTESs  $E$ and $F$ over a  signature   $\S$,  whether $\tthue {E\cup F}\subseteq \tthue E\cup \tthue F$.

 \begin{lem} \label{alapvetes} Let $E$ and $F$ be GTESs over a unary  signature   $\Sigma$. Let $W= ST\langle E \cup F\rangle$. 
  \begin{quote}
  $\tthue {E\cup F}\subseteq \tthue E \cup \tthue F$
  if and only if 
 for each $p\in W  $, $\left( p/_{\tthue {E\cup F}}=  p/_{\tthue E} \mbox{ or }
p/_{\tthue {E\cup F}}=  p/_{\tthue F}\right)$.
\end{quote}
\end{lem}
\begin{proof}
($\Rightarrow$) If $W=\emptyset$, then we are done. Assume that $W\neq \emptyset$.  
Let $p\in W$, $s\in \ts$,  and $p\tthue {E\cup F}  s$. Then
 $p\tthue E  s$ or  $p\tthue F  s$. Consequently 
 $ p/_{\tthue {E\cup F}}\subseteq  p/_{\tthue E} \cup   p/_{\tthue F}$.
 Since $p/_{\tthue E} \cup   p/_{\tthue F}\subseteq  p/_{\tthue {E\cup F}}$, we have $ p/_{\tthue {E\cup F}}= p/_{\tthue E} \cup   p/_{\tthue F}$.
By Proposition \ref{kovetkezmeny},  $ p/_{\tthue {E\cup F}}= p/_{\tthue E}$ or 
$ p/_{\tthue {E\cup F}}= p/_{\tthue F}$.

($\Leftarrow$)
Let $s, t \in \ts$ be arbitrary such that $s\tthue {E\cup F} t$. Then 
\begin{quote} 
$s=q_0
  \thue  { [\alpha_1,(l_1, r_1) ]} q_1 \thue  { [\alpha_2,(l_2, r_2) ]}
  q_2 \thue  { [ \alpha_3, (l_3, r_3) ]} \cdots  \thue  { [ \alpha_n, (l_n, r_n) ]}
  q_n =t$
  for some $n\in \Nat$,  
\end{quote} 
where for every $i\in \{\, 0, \ldots, n\, \}$, $q_i \in \ts $ and for every $i\in \{\, 1, \ldots, n\, \}$, 
 $l_i \doteq r_i \in E\cup F$.
  We now show that 
 \begin{quote}
  $s\tthue E t$ or
    $s\tthue F t$.
\end{quote} 
If $n=0$, then $s=t$ and hence $s\tthue E t$. Let us assume that $n\geq 1$.
Let $\alpha_i$ be a  shortest word in the set
$\{\, \alpha_1, \alpha_2, \ldots, \alpha_n\, \}$, where  $i\in \{\, 1, \ldots, n\, \}$. Then
there exist
$\delta\in CON_{\Sigma, 1}$,  $v, z \in \ts$  such that
\begin{itemize}
\item     $addr(\delta)= \alpha_i$, 
  \item 
$s=\delta [v]$,
   $t=\delta [z]$, and
   \item  $q_{i-1}=\delta [l_i]$,   $q_i=\delta [r_i]$.
  \end{itemize}
So,  \begin{quote} $v\tthue {E\cup F} l_i \thue {E\cup F} r_i \tthue {E\cup F} z$\,  or \, $v\tthue {E\cup F} r_i \thue {E\cup F} l_i \tthue {E\cup F} z$. \end{quote}
Consequently, 
$v, z \in  l_i/_{\tthue {E\cup F}}$.
As 
 $ l_i/_{\tthue {E\cup F}}=  l_i/_{\tthue E}$ or 
$ l_i/_{\tthue {E\cup F}}=  l_i/_{\tthue F}$, we have 
$v, z \in  l_i/_{\tthue E}$ or $v, z \in  l_i/_{\tthue F}$. 
Therefore 
\begin{quote}
  $v \tthue E z$ or   
  $v\tthue  F z$.
\end{quote}
Then 
\begin{quote}
  $s=\delta [v]\tthue E \delta [z] =t$ or
    $s=\delta [v]\tthue F \delta [z] =t$. \qedhere
\end{quote}
\end{proof}


By Proposition \ref{mezo} and Proposition \ref{erdo} we have  the following result. 
\begin{prop}\label{franciaeldontsimple} Let 
  $E$ and $F$  be  GTESs over a  signature   $\S$, and  $n=\mathrm{size}(E)+\mathrm{size}(F)$.
 Let $W= ST\langle E \cup F\rangle$. 
   We can decide in $O(n^3)$ time whether for each $p\in W  $, 
  \begin{quote}
  $ p/_{\tthue {E\cup F}}=  p/_{\tthue E}$ or 
$ p/_{\tthue {E\cup F}}=  p/_{\tthue F}$. \qedhere
\end{quote}
\end{prop}

  

 \begin{lem} \label{het} Let 
  $E$ and $F$  be  GTESs over a  signature   $\S$ such that GTRS $R\langle E \cup F , \emptyset\rangle$ is total. Let $W= ST\langle E \cup F\rangle$. 
  Then $\tthue {E\cup F}\subseteq \tthue E \cup \tthue F$
  if and only if for each $p\in W  $, 
  \begin{quote}
  $ p/_{\tthue {E\cup F}}=  p/_{\tthue E}$ or 
$ p/_{\tthue {E\cup F}}=  p/_{\tthue F}$.
\end{quote}
\end{lem}
\begin{proof}
($\Rightarrow$)
Let $p\in W$, $s\in \ts$,  and $p\tthue {E\cup F}  s$. Then
 $p\tthue E  s$ or  $p\tthue F  s$. Consequently 
 $ p/_{\tthue {E\cup F}}\subseteq  p/_{\tthue E} \cup   p/_{\tthue F}$.  
 Observe that 
 $p/_{\tthue E} \cup   p/_{\tthue F}\subseteq  p/_{\tthue {E\cup F}}  $.  
 Thus  $ p/_{\tthue {E\cup F}}=  p/_{\tthue E} \cup   p/_{\tthue F}$. 
By Proposition \ref{kovetkezmeny},  $ p/_{\tthue {E\cup F}}= p/_{\tthue E}$ or 
$ p/_{\tthue {E\cup F}}= p/_{\tthue F}$.

($\Leftarrow$) 
Let $s, t \in \ts$ be arbitrary such that $s\tthue {E\cup F} t$. Then by Lemma \ref{ezredik},
$s \tred {R\langle E\cup F, \emptyset \rangle}  p/_{\Theta\langle E\cup F, \emptyset \rangle}$ for some $p \in W$. Consequently, 
 $s\tthue {E\cup F}p$.
 Since $s\tthue {E\cup F} t$,
 we have  $t\tthue {E\cup F}p$.
 Then by our assumption, 
 \begin{quote}
 $s\tthue E p$ and $t\tthue E p$ or 
 $s\tthue F p$ and $t\tthue F p$.
 \end{quote}
  For this reason $s\tthue E t$ or  $s\tthue F t$. Thence  $\tthue {E\cup F}\subseteq \tthue E \cup \tthue F$.
 \end{proof}

 \begin{lem} \label{duna} Let 
  $E$ and $F$  be  GTESs over a  signature   $\S$ such that GTRS
  $R\langle E, F \rangle$
  is total. Let $W= ST\langle E \cup F\rangle$. Then for each $p\in W  $,
  \begin{quote}
   $p/_{\tthue {E\cup F}}=  p/_{\tthue E} \mbox{  if and only if }
   p/_{\Theta\langle E\cup F, \emptyset \rangle}= p/_{\Theta\langle E, F \rangle}$.
   \end{quote}
\end{lem}
\begin{proof}
($\Rightarrow$) Let $p \in W$. By the definition of    ${\Theta\langle E, F \rangle}$ and    ${\Theta\langle E\cup F, \emptyset  \rangle}$,
\begin{quote}
$p/_{\Theta\langle E\cup F, \emptyset \rangle}= p_{\tthue {E\cup F}} \cap \ts \times \ts$ and
$p/_{\Theta\langle E, F \rangle}= p_{\tthue E}  \cap \ts \times \ts$.
\end{quote}
 Thus, 
\begin{quote}
$ p/_{\Theta\langle E\cup F, \emptyset \rangle}= p/_{\Theta\langle E, F \rangle}$.
\end{quote}

 ($\Leftarrow$)
  Let $p \in W$.
 Since $\tthue E \subseteq \tthue {E\cup F}$, we have  $ p/_{\tthue E}\subseteq p/_{\tthue {E\cup F}}$. Our aim is to show that  
 $ p/_{\tthue {E\cup F}}\subseteq   p/_{\tthue E}$. To this end, let 
 $s\in \ts $ be such that $s\tthue {E \cup F} p$. We now show that   $s\tthue E  p$ as well.
As $R\langle E, F \rangle$ is total,  by Lemma  \ref{ezredik}, 
\begin{quote}
$s{\downarrow}_ {R\langle E, F \rangle} = t/_{\Theta\langle E, F \rangle}$ for some $t \in W$.
\end{quote}
By Propositon \ref{harmadik},   $s\tthue E  t$. 
Then \begin{quote}
$t \tthue  E s \tthue {E \cup F} p$. 
\end{quote}
Thus
\begin{quote}
$t \tthue {E \cup F} p$. 
\end{quote} Since $t, p \in W$, by the definition of $\Theta\langle E\cup F, \emptyset \rangle$, we have 
\begin{quote}
$(t, p)\in {\Theta\langle E\cup F, \emptyset \rangle}$.
\end{quote}
 As $\tthue E \subseteq \tthue {E\cup F}$, we have 
  \begin{quote}
 $t_{\Theta\langle E, F \rangle}  \subseteq t_{\Theta\langle E\cup F, \emptyset \rangle}=  p_{\Theta\langle E\cup F, \emptyset \rangle}= p/_{\Theta\langle E, F \rangle}$.
 \end{quote}
Consequently 
\begin{quote}
$t/_{\Theta\langle E, F \rangle}\subseteq p_{\Theta\langle E, F \rangle}$.
\end{quote}
Thence
\begin{quote}
$t/_{\Theta\langle E, F \rangle}= p_{\Theta\langle E, F \rangle}$.
\end{quote}
 Consequently,  $s\tthue E t \tthue E p$.  
 For this reason, by the definition of $s$, \begin{quote}
  $ p/_{\tthue {E\cup F}}\subseteq   p/_{\tthue E}$. 
  \end{quote}
  Thus, 
 $ p/_{\tthue {E\cup F}}=  p/_{\tthue E}$. 
\end{proof}

    \begin{lem} \label{asas}
Let $E$ and $F$ be GTESs over  a  signature   $\Sigma$ such that  
 GTRS $R\langle E, F\rangle$ and GTRS $R\langle F, E\rangle$ are  total.  
   Then 
   \begin{quote}
  $\tthue {E\cup F}\subseteq \tthue E \cup \tthue F$
  if and only if for each $p\in W  $, 
  $\left( p/_{\tthue {E\cup F}}=  p/_{\tthue E} \mbox{ or }
 p/_{\tthue {E\cup F}}=  p/_{\tthue F}\right)$.
 \end{quote}
   \end{lem}
   \begin{proof}
   By Lemma \ref{teljes},
  $R\langle E\cup F, \emptyset\rangle$ 
is total as well. Consequently, by Lemma \ref{het} and Lemma \ref{duna}, we have the lemma.
 \end{proof}

 \section{GTRS $R\langle E, F \rangle$ and GTRS $R\langle F, E \rangle$ 
together  simulate the  GTRS $R\langle E \cup F, \emptyset \rangle$}\label{jointly}
 We now introduce the notion of GTRS $R\langle E, F \rangle$
simulating  the  GTRS $R\langle E\cup F, \emptyset \rangle$ on a constant    ${b}\in C\langle  E \cup F, \emptyset \rangle$.
We study the case when GTRS $R\langle E, F\rangle$ is   total. 
Furthermore,  we  introduce the  
 concept of GTRS $R\langle E, F \rangle$ and GTRS $R\langle F, E \rangle$
together  simulate the  GTRS $R\langle E \cup F, \emptyset \rangle$. 
Then we describe the connection between the latter concept and 
the auxiliary   dpwpa   $ AUX[ E;  F] =
(C\langle E \cup F, \emptyset\rangle, A[E; F] )$.

  \begin{df}\label{singlesimulate}  
\rm Let $E$ and $F$ be GTESs over  a signature   $\Sigma$, and let    ${b}\in C\langle  E \cup F, \emptyset \rangle $.
  We say that GTRS $R\langle E, F \rangle$
simulates the  GTRS $R\langle E\cup F, \emptyset \rangle$ on the constant    ${b}\in C\langle  E \cup F, \emptyset \rangle $ if Conditions 1 and 2 hold:
 \begin{itemize}
\item[1.]    ${b}\in C\langle  E, F \rangle $.

\item[2.] For every    ${c}\in C\langle E\cup F, \emptyset \rangle$,
  if
 $R\langle {E \cup F}, \emptyset \rangle$ reaches      ${b}$ from    a proper extension    of the constant    ${c}$, then \newline
 for every rewrite rule  of the form
 $\s({c}_1, \ldots, {c}_m) \dot{\rightarrow} {c}\in  R\langle {E \cup F}, \emptyset \rangle$ with $m\in \Nat$, \newline
 for all    ${d}_1, \ldots, {d}_m  \in  C \langle E, F \rangle$   with
    ${d}_i \subseteq  {c}_i$ for all $i\in \{\, 1, \ldots, m\, \}$, \newline
 there exists    ${d} \in  C \langle E, F \rangle$ 
 such that the rewrite rule   
 $\s({d}_1, \ldots, {d}_m ) \dot{\rightarrow} {d}$ is in $R\langle E, F \rangle$.
 \end{itemize}
    \end{df}\begin{sta}\label{sobriety}
     ${d}\subseteq {c}$ for    ${d}, {c}$ appearing in Condition 2 in  Definition \ref{singlesimulate}.
  \end{sta}
Intuitively, in Condition 2  in  Definition \ref{singlesimulate},  the rewrite rule   
 $\s({d}_1, \ldots, {d}_m ) \dot{\rightarrow} {d}$  in $R\langle E, F \rangle$ matches the rewrite rule
 $\s({c}_1, \ldots, {c}_m) \dot{\rightarrow} {c}\in  R\langle {E \cup F}, \emptyset \rangle$.

  By direct inspection of Definition \ref{csurig} and
 Definition  \ref{singlesimulate}, we get the following statement. 
 \begin{sta} \label{hasznos}
   Let $E$ and $F$ be GTESs over  a signature   $\Sigma$ such that  
 GTRS $R\langle E, F\rangle$ is   total,   and let    ${b}\in C\langle  E \cup F, \emptyset \rangle $.
 GTRS $R\langle E, F \rangle$
simulates the  GTRS $R\langle E\cup F, \emptyset \rangle$ on the constant     ${b}\in C\langle  E \cup F, \emptyset \rangle $ if and only if       ${b}\in C\langle  E, F \rangle $.
  \end{sta}

  \begin{df}\label{defsimulate}\rm  Let $E$ and $F$ be GTESs over  a signature   $\Sigma$.
    We say that GTRS $R\langle E, F \rangle$ and GTRS $R\langle F, E \rangle$
together  simulate the  GTRS $R\langle E \cup F, \emptyset \rangle$
 if 
   for each constant    ${b}\in C\langle  E \cup F, \emptyset \rangle $, GTRS  $R\langle E, F \rangle$  or GTRS  $R\langle F, E \rangle$
    simulate  the GTRS  $R\langle E \cup F, \emptyset \rangle$ on     ${b} $. 
   \end{df}
  
 We have the following statement by direct inspection of  Definition \ref{defsimulate} and Statement \ref{hasznos}.
 \begin{sta}\label{szabill} Let  
 $E$ and $F$ be  GTESs over a signature   $\Sigma$ such that GTRS  $R\langle E, F\rangle$ is total.  
 GTRS $R\langle E, F \rangle$ and  GTRS $R\langle F, E \rangle$ together  simulate the  GTRS $R\langle E \cup F, \emptyset \rangle$  
  if and only if for each constant    ${b}\in C\langle  E \cup F, \emptyset \rangle $, 
    Conditions 1 or  2 hold: 
    \newline
    1.        $ {b}\in C\langle  E, F \rangle $. \newline  
 2.   $ {b}\in C\langle  F, E\rangle $ and   GTRS $R\langle F, E \rangle$
simulates the  GTRS $R\langle E\cup F, \emptyset \rangle$ on    ${b}$.  
 \end{sta}

  By Definitions \ref{singlesimulate} and \ref{defsimulate}, and  Lemma \ref{rozmaring}, and Definition \ref{lepesttart}, we have the following lemma.
   \begin{lem}\label{fax} Let $E$ and $F$ be GTESs over a  signature   $\Sigma$.
   GTRS $R\langle E, F \rangle$ and GTRS $R\langle F, E \rangle$
together  simulate the  GTRS $R\langle E \cup F, \emptyset \rangle$
   if and only if  for each constant    ${b}\in C\langle  E \cup F, \emptyset \rangle $, 
    Conditions 1 or  2 hold:

 1.       $ {b}\in C\langle  E, F \rangle $  and 
        for each    ${c} \in C\langle  E \cup F, \emptyset \rangle $ if there is a  path of positive length or a cycle     from  
     ${b}$ to    ${c}$
 in the auxiliary   dpwpa   $ AUX[ E;  F] =
(C\langle E \cup F, \emptyset\rangle, A[E; F] )$, then  
$ R \langle E, F \rangle$  keeps up with  $ R \langle E\cup F, \emptyset \rangle$  writing    ${c} $.

 2.  $ {b}\in C\langle  F, E \rangle $  and 
        for each    ${c} \in C\langle  E \cup F, \emptyset \rangle $ if there is a  path of positive length or a cycle     from  
     ${b}$ to    ${c}$ 
 in the auxiliary   dpwpa   $ AUX[ E;  F] =
(C\langle E \cup F, \emptyset\rangle,A[E; F])$, then  
$ R \langle F, E \rangle$  keeps up with  $ R \langle E\cup F, \emptyset \rangle$  writing    ${c} $.
  \end{lem}

     By  Statement  \ref{szabill},  Definition \ref{singlesimulate},  Lemma \ref{rozmaring}, and Definition \ref{lepesttart}
we have the following lemma.
  \begin{lem}\label{telex}     Let 
 $E$ and $F$ be  GTESs over  a signature $\Sigma$ such that $R\langle E, F\rangle$ is total.
   GTRS $R\langle E, F \rangle$ and GTRS $R\langle F, E \rangle$
together  simulate the  GTRS $R\langle E \cup F, \emptyset \rangle$
   if and only if  for each  constant    ${b}\in C\langle  E \cup F, \emptyset \rangle $, 
    Conditions 1 or  2 hold:

 1.       $ {b}\in C\langle  E, F \rangle $.

 2.  $ {b}\in C\langle  F, E \rangle $  and 
        for each  constant    ${c} \in C\langle  E \cup F, \emptyset \rangle $ if there is a  path of positive length or a cycle     from  
     ${b}$ to    ${c}$ 
 in the auxiliary   dpwpa   $ AUX[ E;  F] =
(C\langle E \cup F, \emptyset\rangle, A[E; F])$, then  
$ R \langle F, E \rangle$  keeps up with  $ R \langle E\cup F, \emptyset \rangle$  writing    ${c} $.
  \end{lem}

\begin{exa}\label{1oneegyharmadikfolytatas}\rm 
We now continue Example \ref{1oneegy}, Example \ref{1oneegyfolytatas}, and Example \ref{1oneegymasodikfolytatas}.
There exist  two paths  of positive length  in $ AUX[ E;  F] $.
   There do not exist cycles in $ AUX[ E;  F] $. 
   
 The length of the first path    is $1$, and it  starts from the  vertex   
   $f(\#)/_{\Theta\langle E\cup F  ,\emptyset \rangle}$,  and leads to   $\#/_{\Theta\langle E\cup F  ,\emptyset \rangle}$.    
   By Lemma \ref{rozmaring},
     $ R \langle E\cup F, \emptyset \rangle$ reaches   $f(\#)/_{\Theta\langle E\cup F  ,\emptyset \rangle}$
            from a proper extension of    $\#/_{\Theta\langle E\cup F  ,\emptyset \rangle}$.
           Recall that 
    \begin{quote}    
$f(\#)/_{\Theta\langle E\cup F, \emptyset \rangle}.equal_E= \mathrm{true}$ \; and \;
  $\#/_{\Theta\langle E \cup F, \emptyset \rangle}.keeps_E=\mathrm{true}$.
    \end{quote}

 The length of the second path is $1$, and it  starts from the  vertex     $f(\$)/_{\Theta\langle E\cup F  , \emptyset \rangle}$,
 and leads to   $\$/_{\Theta\langle E\cup F  ,\emptyset \rangle}$.
   By Lemma \ref{rozmaring}, 
    $ R \langle E\cup F, \emptyset \rangle$ reaches  $f(\$)/_{\Theta\langle E\cup F  ,\emptyset \rangle}$ from a proper extension of 
      $\$/_{\Theta\langle E\cup F  ,\emptyset \rangle}$.   Recall that 
 \begin{quote}  
  $f(\$)/_{\Theta\langle E \cup F, \emptyset \rangle}.equal_F=\mathrm{true}$ \; and \; 
 $  \$/_{\Theta\langle E \cup F ,\emptyset \rangle}.keeps_F= \mathrm{true}$.
    \end{quote}

Observe that Conditions 1 or 2 in Lemma \ref{fax}
hold true for each  vertex in 
$C\langle E \cup F, \emptyset\rangle=\{\, \{\, \#\,\}$, \;$\{\,   \$\, \}$,\;
  $\{\, f(\#), \;  g(\#) \, \}$,\;
 $\{\, f(\$), \;g(\$)\, \}\, \}$. Thence by Lemma \ref{fax}, GTRS $R\langle E, F \rangle$ and GTRS $R\langle F, E \rangle$
together  simulate the  GTRS $R\langle E \cup F, \emptyset \rangle$.

    \end{exa} 

\begin{exa}\label{2szepenharmadikfolytatas}\rm 
We now continue Example \ref{2szepen},  Example \ref{2szepenfolytatas}, and  Example \ref{2szepenmasodikfolytatas}. 
 There does not exists a  path  of positive length   in  $ AUX[ E;  F] $.
There exists one cycle   in  $ AUX[ E;  F] $. 
   It is of length  $1$, and it  starts from the  vertex    $\#/_{\Theta\langle E\cup F  ,\emptyset \rangle}$,  and leads to
   $\#/_{\Theta\langle E\cup F  ,\emptyset \rangle}$.
    By Lemma \ref{rozmaring},
     $ R \langle E\cup F, \emptyset \rangle$ reaches    
         $\#/_{\Theta\langle E\cup F  ,\emptyset \rangle}$     from a proper extension of   $\#/_{\Theta\langle E\cup F  ,\emptyset \rangle}$. 
 Recall that 
 \begin{quote}  
  $\#/_{\Theta\langle E \cup F, \emptyset \rangle}.equal_E=\mathrm{false}$,
    $\#/_{\Theta\langle E \cup F, \emptyset \rangle}.equal_F =\mathrm{false}$, and \newline
      $\#/_{\Theta\langle E \cup F, \emptyset \rangle}.keeps_E =\mathrm{false}$,  
      $\#/_{\Theta\langle E \cup F, \emptyset \rangle}.keeps_F =\mathrm{false}$.
 \end{quote}

Observe that neither Condition 1 in  nor Condition  2 in Lemma \ref{fax} holds
 for the  vertex $\#/_{\Theta\langle E\cup F  ,\emptyset \rangle}$.
  Thus by Lemma \ref{fax}, GTRS $R\langle E, F \rangle$ and GTRS $R\langle F, E \rangle$ 
together do not   simulate the  GTRS $R\langle E \cup F, \emptyset \rangle$.
    \end{exa}

Intuitively, assume that GTRS $R\langle E, F \rangle$ simulates the  GTRS $R\langle E \cup F, \emptyset \rangle$
on a vertex     ${d}\in C\langle  E \cup F, \emptyset \rangle $ of  $ AUX[ E;  F] $. 
Consider a term $\delta[{b}_1,  \ldots, {b}_k]$
for some $\delta\in CON_{\Sigma,k} $,      ${b}_1, \ldots,  {b}_k \in  C\langle {E\cup F}, \emptyset \rangle$ and a sequence of reductions
\begin{quote}
(i) \, 
$\delta[{b}_1,  \ldots, {b}_k]\tred {R \langle  E \cup F, \emptyset  \rangle }{d}$. 
\end{quote} 
Furthermore, consider  any  constants    ${c}_1, \ldots, {c}_k \in C \langle E, F \rangle$ such that 
    ${c}_j \subseteq {b}_j $ for every $j\in \{\, 1, \ldots, k\, \}$. We substitute    ${c}_j$ for    ${b}_j$ in $\delta[{b}_1,  \ldots, {b}_k]$, in this way, we get the term 
 $\delta[{c}_1,  \ldots, {c}_k]$. Since
   GTRS $R\langle E, F \rangle$ simulates the  GTRS $R\langle E \cup F, \emptyset \rangle$, 
 we construct  a fitting  sequence  of reductions 
 \begin{quote} (ii) \,  
$\delta[{c}_1,  \ldots, {c}_k]\tred {R \langle  E \cup F, \emptyset  \rangle }{e}$ 
\end{quote} 
for some    ${e}\in C \langle E, F \rangle$.  We obtain (ii) following step  by step the reductions in (i). For each reduction of (i), 
  we  apply a fitting rule at the same position in (ii). Then for each reduction step,  
for the two terms $s$ and  $t$ obtained in (i) and (ii), respectively,
each constant     ${c}\in C\langle  E, F \rangle $ appearing in $t$  corresponds to the constant    ${b}\in C\langle  E \cup F, \emptyset \rangle $ 
appearing  at the same position in $s$.
Moreover    ${c}$, as congruence class,  is  
 subset of the corresponding constant    ${b}$.  Finally, constant    ${e}$ corresponds to constant    ${d}$, and    ${e} \subseteq {d}$. 
Formally, we state and prove the following result.

  \begin{lem}\label{holnap} Let 
 $E$ and $F$ be  GTESs over  a signature $\Sigma$ such that   GTRS $R\langle E, F \rangle$ simulates the  GTRS $R\langle E \cup F, \emptyset \rangle$  
  on a constant      ${d}\in C\langle  E \cup F, \emptyset \rangle $.
 Let Conditions (a) and (b) hold. 
\begin{itemize}
 \item [\em{(a)}]
 $\delta_n[{b}_{n1},  \ldots, {b}_{nk_n}]
 \red {R \langle  E \cup F, \emptyset  \rangle } \delta_{n-1}[{b}_{(n-1)1},  \ldots, {b}_{(n-1)k_{n-1}}]\red {R \langle  E \cup F,  \emptyset   \rangle }
 \delta_{n-2}[{b}_{(n-2)1},  \ldots, {b}_{(n-2)k_{n-2}}]
  \red {R \langle E \cup  F, \emptyset  \rangle }  \cdots \newline 
 \red {R \langle E \cup  F, \emptyset  \rangle } \delta_1[{b}_{11},  \ldots, {b}_{1k_1}]\red {R \langle  E \cup F, \emptyset  \rangle }{d}$, \newline 
where 
 $n\geq 1$, 
 for every $i\in \{\, 1, \ldots, n\, \}$, 
 $\delta_i \in CON_{\Sigma, {k_i}} $, $k_i \in \Nat$, and      ${b}_{i1}, \ldots,  {b}_{i{k_i}}\in 
 C\langle {E\cup F}, \emptyset\rangle$.
  \end{itemize} 
 \begin{itemize}
 \item [{\rm (b)}] For every $j\in \{\, 1, \ldots, k_n\, \}$,     ${c}_{nj}\in C \langle E, F \rangle$ and     ${c}_{nj} \subseteq {b}_{nj}$.
 \end{itemize} 
  Then 
    for any  $i\in \{\, 1, \ldots, n-1\, \}$, $j\in \{\, 1, \ldots, k_i\, \}$, there exist    ${c}_{ij}\in C \langle E, F \rangle$
  with     ${c}_{ij} \subseteq {b}_{ij}$,   and there exists    ${e}\in C \langle E, F \rangle$ with    ${e} \subseteq {d}$ 
   such that the following sequence of reductions holds:
\begin{quote}
  $\delta_n[{c}_{n_1},  \ldots, {c}_{nk_n}]
 \red {R \langle  E, F \rangle }
 \delta_{n-1}[{c}_{(n-1)1},  \ldots, {c}_{(n-1)k_{n-1}}]\red {R \langle  E, F \rangle }
 \delta_{n-2}[{c}_{(n-2)1},  \ldots, {c}_{(n-2)k_{n-2}}]
  \red {R \langle E, F \rangle } 
 \cdots \red {R \langle E, F \rangle } \newline \delta_1[{c}_{11},  \ldots, {c}_{1k_1}]\red {R \langle E, F \rangle }{e}$. \qedhere
 \end{quote} 
 \end{lem}
 
 \begin{proof}
 We proceed by induction on $n$.
 
 {\em Base Case:} $n=1$.  Assume that  Conditions (a) and (b) hold. Then $\delta_1=\s(\diamondsuit, \ldots, \diamondsuit)$ for some $\s\in \Sigma_{k_1}$,
 $k_1\in \Nat$, where
  \begin{quote}
 $\delta_1[{b}_{11},  \ldots, {b}_{1k_1}]= \s({b}_{11},  \ldots, {b}_{1k_1})$, \, 
  $\s({b}_{11},  \ldots, {b}_{1k_1} ) \dot{\rightarrow}  {d}\in  R\langle {E \cup F}, \emptyset \rangle$, \, and \, 
 $\s({b}_{11},  \ldots, {b}_{1k_1} )\red {R \langle E \cup  F, \emptyset  \rangle } {d}$.
 \end{quote}
   By Condition (b), for every $j\in \{\, 1, \ldots, k_1\, \}$,      ${c}_{1j}\in C \langle E, F \rangle$ and     ${c}_{1j} \subseteq {b}_{1j}$.
  We assumed that 
  GTRS $R\langle E, F \rangle$ simulates the  GTRS $R\langle E \cup F, \emptyset \rangle$  
   on the constant    ${d}\in C\langle  E \cup F, \emptyset \rangle $. 
 Consequently  by the rewrite rule $\s({b}_{11},  \ldots, {b}_{1k_1} ) \dot{\rightarrow}  {d}\in  R\langle {E \cup F}, \emptyset \rangle$ and Definition \ref{singlesimulate},  and Statement \ref{sobriety}, there exists a rewrite rule
  \begin{quote}
  $\s({c}_{11},  \ldots, {c}_{1k_1}) \dot{\rightarrow} {c} \in R \langle E, F \rangle$, where    ${c} \subseteq {d}$.
\end{quote} Since 
 $\delta_1[{c}_{11},  \ldots, {c}_{1k_1}]= \s({c}_{11},  \ldots, {c}_{1k_1})$, we have 
  \begin{quote}
 $\delta_1[{c}_{11},  \ldots, {c}_{1k_1}]\red {R \langle E, F \rangle } {c} $,
 where    ${c}\in C \langle E, F \rangle$ and    ${c} \subseteq {d}$.
 \end{quote} 
 
 {\em Induction Step:} Let $n\geq 1$, and  assume that the statement of the claim  is true for $n$. We now show 
 that it is also true for $n+1$ as well.
 Assume that  Conditions ($\mathrm{a}^\prime$) and ($\mathrm{b}^\prime$) hold:
 \begin{itemize}
 \item [($\mathrm{a}^\prime$)] $\delta_{n+1}[{b}_{(n+1)1}, \ldots, {b}_{(n+1) k_{n+1}}]\red {R \langle  E\cup F, \emptyset  \rangle } 
  \delta_n[{b}_{n1},  \ldots, {b}_{nk_n}]
 \red {R \langle  E\cup F, \emptyset  \rangle } \delta_{n-1}[{b}_{(n-1)1},  \ldots, {b}_{(n-1)k_{n-1}}]\red {R \langle  E\cup F, \emptyset  \rangle } \newline
 \cdots 
\red {R \langle E\cup F, \emptyset  \rangle }
  \delta_2[{b}_{21},  \ldots, {b}_{2k_2}]
 \red {R \langle E\cup F, \emptyset  \rangle } 
 \delta_1[{b}_{11},  \ldots, {b}_{1k_1}]\red {R \langle E\cup F, \emptyset  \rangle }{d}$,
 \newline 
where 
  for every $i\in \{\, 0, \ldots, n+1\, \}$, 
 $\delta_i \in CON_{\Sigma, {k_i}} $, $k_i  \in \Nat $, and      ${b}_{i1}, \ldots,  {b}_{i{k_i}}\in 
 C\langle {E\cup F}, \emptyset\rangle$.
 \end{itemize}
 \begin{itemize}
 \item [($\mathrm{b}^\prime$)] For every $j\in \{\, 1, \ldots, k_{n+1}\, \}$,     ${c}_{(n+1)j}\in C \langle E, F \rangle$ and     ${c}_{(n+1)j} \subseteq 
 {b}_{(n+1)j}$.
 \end{itemize} 
  Then Conditions (c)--(f) hold:
   \begin{itemize}
 \item [{\rm (c)}] We obtain $\delta_{n+1}$ from $\delta_n$ by replacing the $j$th occurrence  from left to right of the hole 
  $\diamondsuit$  by the context 
  $\s(\diamondsuit, \ldots, \diamondsuit )$, where 
 $j \in \{\, 1, \ldots, k_n\, \}$,   $\s\in \Sigma_m$, $m \in \Nat$,  and $k_n= k_{n+1}-m+1$.
  \item [{\rm (d)}] $\delta_{n+1}
  [{b}_{(n+1)1}, \ldots,   {b}_{(n+1)k_{n+1}}]= 
  \delta_n [ {b}_{(n+1)1}, \ldots, {b}_{(n+1)(j-1)}, \s({b}_{(n+1)j}, \ldots, {b}_{(n+1)(j+m-1)}), {b}_{(n+1)(j+m)}, 
  \ldots,   {b}_{(n+1) k_{n+1}}]$. 
 \item [{\rm (e)}]    
For every $\ell\in \{1, \ldots, j-1\, \}$,     ${b}_{n\ell}= {b}_{(n+1)\ell}$ and
    for every $\ell\in \{j+1, \ldots, k_n \, \}$,       ${b}_{n\ell}= {b}_{(n+1)(\ell+m-1)}$.
  \item [{\rm (f)}] In the last reduction step  of (a') we apply the rule $\s({b}_{(n+1)j}, \ldots, {b}_{(n+1)(j+m-1)} ) \dot{\rightarrow}   {b}_{nj}\in  R\langle {E \cup F}, \emptyset  \rangle$.
 \end{itemize}  
  By Condition ($\mathrm{a}^\prime$), $R \langle E\cup F, \emptyset  \rangle $ reaches     ${d}$
 from a proper extension of    ${b}_{nj}$.   By Condition    ($\mathrm{b}^\prime$),
 for every $\ell \in \{\, j, \ldots, j+m-1\, \}$,     ${c}_{(n+1)\ell}\in C \langle E, F \rangle$ and      ${c}_{(n+1)\ell} \subseteq {b}_{n\ell}$.
  We assumed that 
  GTRS $R\langle E, F \rangle$ simulates the  GTRS $R\langle E \cup F, \emptyset \rangle$  
  on the constant     ${d}\in C\langle  E \cup F, \emptyset \rangle $.
Therefore  by the rewrite rule appearing in Condition (f) there exists a rewrite rule   
  \begin{quote}
  $\s({c}_{(n+1)j}, \ldots, {c}_{(n+1) (j+m-1)})\dot{\rightarrow} {c}$ in $R \langle E, F \rangle$ with    ${c}\subseteq {b}_{nj} $.
\end{quote}
 \begin{quote} 
For every $\ell\in \{1, \ldots, j-1\, \}$, let     ${c}_{n\ell}= {c}_{(n+1)\ell}$, \newline  let    ${c}_{nj}={c}$, \newline
  for every $\ell\in \{j+1, \ldots, k_n \, \}$, let     ${c}_{n\ell}= {c}_{(n+1)(\ell+j-1)}$.
  \end{quote}
  Thence
  \begin{quote}
 for every $\ell\in \{\, 1, \ldots, k_n\,\}$, 
    ${c}_{n\ell}\subseteq {b}_{n\ell}$ and    
 $\delta_{n+1}[{c}_{(n+1)1},  \ldots, {c}_{(n+1)k_{n+1}}]
 \red {R \langle  E, F \rangle } \delta_n[{c}_{n1},  \ldots, {c}_{nk_n}]$.
 \end{quote}
 By the sequence of reductions ($\mathrm{a}^\prime$) we have
 \begin{quote}
$\delta_n[{b}_{n1},  \ldots, {b}_{nk_n}]
 \red {R \langle  E\cup F, \emptyset  \rangle } \delta_{n-1}[{b}_{(n-1)1},  \ldots, {b}_{(n-1)k_{n-1}}]\red {R \langle  E\cup F, \emptyset  \rangle } 
 \cdots 
\red {R \langle E\cup F, \emptyset  \rangle }
  \delta_2[{b}_{21},  \ldots, {b}_{2k_2}]
 \red {R \langle E\cup F, \emptyset  \rangle } \newline
 \delta_1[{b}_{11},  \ldots, {b}_{1k_1}]\red {R \langle E\cup F, \emptyset  \rangle }{d}$.
 \end{quote}
 Consequently by the induction hypothesis, 
 \begin{quote}
  for any  $i\in \{\, 1, \ldots, n-1\, \}$ and  $\ell\in \{\, 1, \ldots, k_i\, \}$, there exist    ${c}_{i\ell}\in C \langle E, F \rangle$,
  with    ${c}_{i\ell} \subseteq {b}_{i\ell}$,    and  \newline
there exists
    ${e}\in C \langle E, F \rangle$ with    ${e} \subseteq {d}$ 
 \end{quote} such that the following sequence of reductions holds:
   \begin{quote} 
$\delta_n[{c}_{n1},  \ldots, {c}_{nk_n}]\red {R \langle  E, F \rangle }
 \delta_{n-1}[{c}_{(n-1)1},  \ldots, {c}_{(n-1)k_{n-1}}]
 \red {R \langle E, F \rangle } 
 \cdots \red {R \langle E, F \rangle } \delta_1[{c}_{11},  \ldots, {c}_{1k_1}]\red {R \langle E, F \rangle }{e}$.
   \end{quote} 
  Then we have 
     \begin{quote} 
 $\delta_{n+1}[{c}_{(n+1)1},  \ldots, {c}_{(n+1)k_{n+1}}]
 \red {R \langle  E, F \rangle } \delta_n[{c}_{n1},  \ldots, {c}_{nk_n}]\red {R \langle  E, F \rangle }
 \delta_{n-1}[{c}_{(n-1)1},  \ldots, {c}_{(n-1)k_{n-1}}]
 \red {R \langle E, F \rangle } \newline
 \cdots \red {R \langle E, F \rangle } \delta_1[{c}_{11},  \ldots, {c}_{1k_1}]\red {R \langle E, F \rangle }{e}$.\qedhere
    \end{quote}
  \end{proof}
   Lemma \ref{holnap} implies the following lemma.
  \begin{lem}\label{hurnegyszog} Let 
 $E$ and $F$ be  GTESs over  a signature $\Sigma$ such that   GTRS $R\langle E, F \rangle$ simulates the  GTRS $R\langle E \cup F, \emptyset \rangle$  
   on a constant     ${d}\in C\langle  E \cup F, \emptyset \rangle $. Let $s\in \ts$ and $s\tred {R\langle E \cup F, \emptyset \rangle}{d}$. 
  Then 
 there exists    ${e}\in C \langle E, F \rangle$ with    ${e} \subseteq {d}$ 
  such that $s\tred {R\langle E, F\rangle}{e}$. 
 \end{lem}

  \begin{lem}\label{sportos}  Let 
 $E$ and $F$ be  GTESs over a signature   $\Sigma$.
 \begin{quote}
    For each   $t\in W $, 
  $\left( t/_{\tthue{ E\cup F} }=  t/_{{\tthue E}} \mbox{ or } t/_{\tthue{ E\cup F} }=  t/_{{\tthue F}} 
  \right) $
   \end{quote}
  if and only if 
  GTRS $R\langle E, F \rangle$ and GTRS $R\langle F, E \rangle$
together  simulate the  GTRS $R\langle E \cup F, \emptyset \rangle$.
 \end{lem}
 
 \begin{proof} 
 $(\Rightarrow)$ Let
  $t\in W $.  Then by our assumption, 
  \begin{quote}
    $ t/_{\Theta\langle E\cup F, \emptyset \rangle }=  t/_{\Theta\langle E, F \rangle } \mbox{ or }
   t/_{\Theta\langle E\cup F, \emptyset \rangle }=  t/_{\Theta\langle F, E \rangle }  $.
   \end{quote}
   Assume that $ t/_{\tthue{ E\cup F} }=  t/_{{\tthue E}} $. The case  $t/_{\tthue{ E\cup F} }=  t/_{{\tthue F}} $
      is similar. We now show that  GTRS $R\langle E, F \rangle$ simulates the  GTRS $R\langle E \cup F, \emptyset \rangle$ on $ t/_{\Theta\langle E\cup F, \emptyset \rangle }$.
 To this end, let     ${d} \in C\langle E\cup F, \emptyset \rangle$ be such that 
 $R\langle {E \cup F}, \emptyset \rangle$ reaches   $t/_{\Theta\langle {E\cup F}, \emptyset \rangle}$ from   a proper extension    of the constant     ${d}$. 
 Then there exists $\delta\in CON_{\Sigma, 1}$
 such that $\delta[{d}] \tred {R\langle E \cup F, \emptyset \rangle} t/_{\Theta\langle {E\cup F}, \emptyset \rangle}$. 
 Let 
 $\s({b}_1, \ldots, {b}_m) \dot{\rightarrow} {d}\in  R\langle {E \cup F}, \emptyset\rangle$ be a rewrite rule with    ${b}_1, \ldots, {b}_m \in 
 C\langle {E \cup F}, \emptyset\rangle$, and let     ${c}_1, \ldots, {c}_m  \in C \langle E, F \rangle$   with
    ${c}_i \subseteq {b}_i$ for all $i\in \{\, 1, \ldots, m\, \}$. We now show that 
  there exists a rewrite rule  in  $R\langle E, F \rangle$ with left-hand side 
 $\s({c}_1, \ldots, {c}_m) $. 
  Let $t_1\in {c}_1, \ldots, t_m \in {c}_m$. By Definition \ref{szeder},  we have  
  \begin{quote}
  $t_1, \ldots, t_m \in W$, \newline
 $t_1/_{\Theta\langle E, F \rangle}={c}_1, \ldots,t_m/_{\Theta\langle E, F \rangle}={c}_m$, \newline
 $t_1/_{\Theta\langle {E\cup F}, \emptyset \rangle}={b}_1, \ldots,t_m/_{\Theta\langle {E\cup F} , \emptyset \rangle}={b}_m$. \newline
  Hence 
  $\s(t_1/_{\Theta\langle {E\cup F}, \emptyset\rangle}, \ldots,t_m/_{\Theta\langle {E\cup F} , \emptyset \rangle}) \dot{\rightarrow} {d}
 \in  R\langle {E \cup F}, \emptyset \rangle$.
 \end{quote}  
 Then by Proposition    \ref{masodik},  \begin{quote}
 $\delta [\s(t_1,  \ldots, t_m)] \tred {R\langle E \cup F, \emptyset \rangle} \delta [\s({b}_1,  \ldots, {b}_m)] \red {R\langle E \cup F, \emptyset \rangle} 
 \delta[{d}]\tred {R\langle E \cup F, \emptyset \rangle} t/_{\Theta\langle {E\cup F}, \emptyset \rangle}$. 
 \end{quote}  
 By Proposition \ref{harmadik}, 
 $\delta [\s(t_1,  \ldots, t_m)] \tthue { E\cup F} t$.
    By  the equation $t/_{\tthue {E\cup F}}=  t/_{\tthue E}$, we have 
   $\delta [\s(t_1,  \ldots, t_m)] \tthue E  t$. 
     By Proposition \ref{otodik}, \begin{quote}
     $\delta [\s(t_1,  \ldots, t_m)]\tred {R\langle E, F \rangle}  t/_{\Theta\langle  E, F \rangle}$. \end{quote} 
        By Proposition \ref{rhcp},  
   $R\langle E, F \rangle $ is convergent. Observe that $ t/_{\Theta\langle  E, F \rangle}$ is irreducible for   $R\langle E, F \rangle $.
   Thence  we have 
   \begin{quote}
     $\delta [\s(t_1,  \ldots, t_m)] {\downarrow}_{R\langle E, F \rangle} =t/_{\Theta\langle E, F \rangle}$.
  \end{quote}     
 By the set memberships $t_1, \ldots, t_m \in W$ and by Proposition \ref{masodik}, we have    
 $t_j \tred {R\langle E, F \rangle} t_j/_{\Theta\langle { E}, F\rangle}$ for every $j \in \{\, 1, \ldots, m\, \}$.
  Then 
\begin{quote}
$\delta [\s(t_1,  \ldots, t_m)]\tred {R\langle E, F \rangle} 
\delta[\s(t_1/_{\Theta\langle { E}, F\rangle}, \ldots,t_m/_{\Theta\langle {E}, F\rangle} ) ]$.
      \end{quote}  
       By Proposition \ref{rhcp},  
   $R\langle E, F \rangle $ is convergent.    Thence  we have    
   \begin{quote}  
    $\delta[\s(t_1/_{\Theta\langle { E}, F\rangle}, \ldots,t_m/_{\Theta\langle {E}, F\rangle} )]
    {\downarrow}_{R\langle E, F \rangle} =t/_{\Theta\langle E, F \rangle}$.
   \end{quote}    
 Therefore,  there exists
  a rewrite rule  with left-hand side $\s(t_1/_{\Theta\langle { E}, F\rangle}, \ldots,t_m/_{\Theta\langle {E}, F\rangle} ) $
   in $R\langle {E}, F\rangle$.

  $(\Leftarrow)$    Let $t \in W$.  We assumed that  GTRS $R\langle E, F \rangle$ and GTRS $R\langle F, E \rangle$
together  simulate the  GTRS $R\langle E \cup F, \emptyset \rangle$. By Definition \ref{defsimulate} GTRS $R\langle E, F \rangle$ or 
GTRS $R\langle F, E \rangle$ simulate the  GTRS $R\langle E \cup F, \emptyset \rangle$   
  on  the constant  $  t/_{\tthue {E\cup F}}\in C\langle  E \cup F, \emptyset \rangle $.  
    Assume that  GTRS $R\langle E, F \rangle$ simulates the  GTRS $R\langle E \cup F, \emptyset \rangle$   
  on $  t/_{\tthue {E\cup F}}$.   Then $t/_{\Theta\langle E\cup F, \emptyset \rangle}= t/_{\Theta\langle E, F \rangle}$.   
     Let   $s\in  t/_{\tthue {E\cup F}}$. 
 Consequently by Proposition \ref{otodik},
  $s\tred {R \langle E\cup F, E \rangle } t/_{\Theta\langle {E\cup F},\emptyset\rangle}$. 
  By Lemma \ref{hurnegyszog}, there exists     ${e}\in C \langle E, F \rangle$ such that     ${e} \subseteq  t/_{\Theta\langle {E \cup F}, \emptyset \rangle}$ and 
  $s\tred {R \langle E, F \rangle } {e}$.
   Then    ${e}=p/_{\Theta\langle E, F \rangle}$ for some $p\in W$.
   By Proposition \ref{harmadik},
   $s \tthue { E\cup F}t$ and  $s \tthue { E} p$. 
   Then $s \tthue { E\cup F}p$. Thus    
     ${e} = p/_{\Theta\langle E, F \rangle}\subseteq t/_{\Theta\langle E\cup F, \emptyset \rangle}= t/_{\Theta\langle E, F \rangle}$.  Consequently 
      ${e} = t/_{\Theta\langle E, F \rangle}$. 
    Hence 
     $s\in  t/_{\tthue E}$.  
  Consequently, 
  for any  $s\in  t/_{\tthue {E\cup F}}$, we have
    $s\in  t/_{\tthue E}$. 
Thus
  $ t/_{\tthue {E\cup F}}\subseteq  t/_{\tthue E}$, and hence  $ t/_{\tthue {E\cup F}}=  t/_{\tthue E} $. 
 We showed that if  GTRS $R\langle E, F \rangle$ simulates the  GTRS $R\langle E \cup F, \emptyset \rangle$   
  on $  t/_{\tthue {E\cup F}}$, then  $ t/_{\tthue {E\cup F}}=  t/_{\tthue E} $.

 We can show similarly that if  GTRS $R\langle F, E \rangle$ simulates the  GTRS $R\langle E \cup F, \emptyset \rangle$   
  on $  t/_{\tthue {E\cup F}}$, then  $ t/_{\tthue {E\cup F}}=  t/_{\tthue F} $. 
    Consequently,  \begin{quote}
    for each   
  $t\in W $, 
  $\left( t/_{\tthue {E\cup F}}=  t/_{\tthue E} \mbox{ or }
   t/_{\tthue {E\cup F}}=  t/_{\tthue F}\right)$. \qedhere
   \end{quote}
   \end{proof}

 \begin{exa} \label{5uborkaharmadikfolytatas} \rm 
 We continue  Example \ref{5uborka} and its sequels, Examples \ref{5uborkafolytatas},  \ref{5uborkamasodikfolytatas}.
  There does not exist a path of positive length in  $ AUX[ E;  F] $. There exists  one  cycle   in  $ AUX[ E;  F] $. The cycle  is of length $1$, it starts from the  vertex    
   $\#/_{\Theta\langle E\cup F  ,\emptyset \rangle}$, and leads to   $\#/_{\Theta\langle E\cup F  ,\emptyset \rangle}$.
   By Lemma \ref{rozmaring}, 
    $ R \langle E\cup F, \emptyset \rangle$ reaches      $\#/_{\Theta\langle E\cup F  ,\emptyset \rangle}$
   from a proper extension of 
      $\#/_{\Theta\langle E\cup F  ,\emptyset \rangle}$. 
        Recall that 
    $\#/_{\Theta\langle E \cup F, \emptyset \rangle}.equal_E=\mathrm{false}$ and 
    $\#/_{\Theta\langle E \cup F, \emptyset \rangle}.equal_F=\mathrm{false}$ and 
    $\#/_{\Theta\langle F, E \rangle}.keeps_F =\mathrm{false}$. 
   Then  the formula 
   \begin{quote}
    $\#/_{\Theta\langle E \cup F, \emptyset \rangle}.equal_E=\mathrm{true}$  or    $\#/_{\Theta\langle E \cup F, \emptyset \rangle}.equal_F=\mathrm{true}$
    \end{quote} 
  does not hold.  
             Thence the following statement does not  hold:            
            for each    ${b}\in C\langle  E \cup F, \emptyset \rangle $, 
   \begin{itemize}
  \item
        ${b}.equal_E=\mathrm{true} $ or      ${b}.equal_F=\mathrm{true} $ 
     and   
     \item if 
        ${b}\not \in C\langle  E, F \rangle $, then    
    for each    ${c} \in C\langle  E \cup F, \emptyset \rangle $ if there is a  path of positive length or a cycle     from  
     ${b}$ to    ${c}$
 in the auxiliary dpwpa   
$  AUX[ E;  F] =
(C\langle E \cup F, \emptyset\rangle, A[E; F])$, then  
   ${c}.keeps_F=\mathrm{true}$.  
 \end{itemize}       
  By Lemma \ref{telex}, GTRS $R\langle E, F \rangle$ and GTRS $R\langle F, E \rangle$
together do not   simulate the  GTRS $R\langle E \cup F, \emptyset \rangle$.
  \end{exa}

   \begin{exa} \label{6korteharmadikfolytatas} \rm 
   We continue 
   Example \ref{6korte}, and its sequels,  Examples \ref{6kortefolytatas}, and \ref{6kortemasodikfolytatas}.    
       There exist  one path of positive length and two cycles in $ AUX[ E;  F] $.

 The first cycle:  Its length is $1$, and it  starts from the  vertex    
   $\#/_{\Theta\langle E\cup F  ,\emptyset \rangle}$ and leads to $\#/_{\Theta\langle E\cup F  ,\emptyset \rangle}$. By Lemma \ref{rozmaring}, 
  $ R \langle E\cup F, \emptyset \rangle$ reaches   
   $\#/_{\Theta\langle E\cup F  ,\emptyset \rangle}$ from a proper extension of  $\#/_{\Theta\langle E\cup F  ,\emptyset \rangle}$.
 Recall that 
 \begin{quote}   
 $\#/_{\Theta\langle E \cup F, \emptyset \rangle}.equal_E=\mathrm{false}$,
    $\#/_{\Theta\langle E \cup F, \emptyset \rangle}.equal_F =\mathrm{true}$, and
      $\#/_{\Theta\langle E \cup F, \emptyset \rangle}.keeps_F =\mathrm{true}$.   
 \end{quote}

 The path of positive length:  Its length is $1$, and it  starts from the  vertex    
   $\$/_{\Theta\langle E\cup F  ,\emptyset \rangle}$ and leads to   $\#/_{\Theta\langle E\cup F  ,\emptyset\rangle}$.  By Lemma \ref{rozmaring}, 
     $ R \langle E\cup F, \emptyset \rangle$ reaches  $\$/_{\Theta\langle E\cup F  ,\emptyset \rangle}$ from a proper extension of 
  $\#/_{\Theta\langle E\cup F  ,\emptyset \rangle}$.   Recall that 
 \begin{quote} 
   $\$/_{\Theta\langle E \cup F, \emptyset \rangle}.equal_E=\mathrm{true}$,
    $\$/_{\Theta\langle E \cup F, \emptyset \rangle}.equal_F =\mathrm{false}$, and 
      $\#/_{\Theta\langle E \cup F, \emptyset \rangle}.keeps_F =\mathrm{true}$.   
 \end{quote}

 The second cycle:   Its length is $1$, and it  starts from the  vertex    
   $\$/_{\Theta\langle E\cup F  ,\emptyset \rangle}$ and leads to  $\$/_{\Theta\langle E\cup F  ,\emptyset \rangle}$.  By Lemma \ref{rozmaring}, 
  $ R \langle E\cup F, \emptyset \rangle$ reaches      $\$/_{\Theta\langle E\cup F  ,\emptyset \rangle}$
   from a proper extension of 
      $\$/_{\Theta\langle E\cup F  ,\emptyset \rangle}$.  Recall that 
   \begin{quote} 
    $\$/_{\Theta\langle E \cup F, \emptyset \rangle}.equal_E=\mathrm{true}$, 
    $\$/_{\Theta\langle E \cup F, \emptyset \rangle}.equal_F =\mathrm{false}$, and 
      $\$/_{\Theta\langle E \cup F, \emptyset \rangle}.keeps_F =\mathrm{false}$.   
 \end{quote}


       Observe that the following statement holds: 
       for each    ${b}\in C\langle  E \cup F, \emptyset \rangle $, 
   \begin{itemize}
  \item
        ${b}.equal_E=\mathrm{true} $ or      ${b}.equal_F=\mathrm{true} $ 
     and    
           \item if 
        ${b}\not \in C\langle  E, F \rangle $, then    
    for each    ${c} \in C\langle  E \cup F, \emptyset \rangle $ if there is a  path of positive length or a cycle     from  
     ${b}$ to    ${c}$
 in the auxiliary dpwpa   
$ AUX[ E;  F] =
(C\langle E \cup F, \emptyset\rangle, A[E; F])$, then  
   ${c}.keeps_F=\mathrm{true}$.  
 \end{itemize}      
     By Lemma \ref{telex}, GTRS $R\langle E, F \rangle$ and GTRS $R\langle F, E \rangle$  
together  simulate the  GTRS $R\langle E \cup F, \emptyset \rangle$.
  \end{exa}

\section{Main Cases of the Decision Algorithm}\label{esetek}

We now  recall the four main cases of our decision algorithm.  We show that we can  decide  in   $O(n \, \mathrm{log} \, n )$   time which main cases hold.
\begin{itemize} 
\item [Main Case 1:] $\S$ is a unary  signature. 
\item  [Main Case 2:] Both GTRS $R\langle E, F\rangle$ and GTRS $R\langle F, E\rangle$ are  total. 
\item  [Main Case 3:] 
   $\S_k\neq \emptyset$ for some $k\geq 2$, and  at least  one of $R\langle E, F\rangle$ 
and   $R\langle F, E\rangle$ is  total. 
\item   [Main Case 4:]  $\S$ has a symbol of arity at least $2$, 
  and GTRS  $ R \langle E, F \rangle$ and GTRS $ R \langle F, E \rangle$  are not total.   
  \end{itemize}
 Recall that Main  Case 2 overlaps with Main Case 1 and Main Case 3, and is a subcase of the union of 
 Main Case 1 and Main Case 3.

 \begin{lem}\label{juh} Let  $E$ and $F$ be GTESs over a  signature   $\Sigma$.
 We can  decide which Main Cases hold in   $O(n \, \mathrm{log} \, n )$   time.
 \end{lem}
\begin{proof}
We run Algorithm CAD  on $\Sigma$, $E$, and $F$. 
By Proposition \ref{kecskemet}, we get the following output in   $O(n \, \mathrm{log} \, n )$ time: 
\begin{itemize}
\item the relations  $\rho\langle E, F\rangle$, $\rho\langle F, E\rangle$,  and $\rho\langle E\cup F, \emptyset\rangle$;
\item the  set of nullary symbols  $C\langle E, F\rangle$, $C\langle F, E\rangle$, and
$C\langle E \cup F, \emptyset\rangle$;
\item  $INC(  E\cup F, E) $ stored in a red-black tree saved in a  variable  $RBT1$;
\item  $INC(E\cup F, F) $ stored in a red-black tree  saved in a  variable    $RBT2$;
\item $NUM( E\cup F, E)$ stored in a red-black tree  saved in a  variable    $RBT3$;
\item $NUM(  E\cup F, F)$ stored in a red-black tree  saved in a  variable    $RBT4$;
\item the GTRS $R\langle E\cup F, \emptyset\rangle$  stored in a red-black tree saved in a  variable  $RBT5$;
 each node  in  $RBT5$ contains a search key  which is  a rule in $R\langle E\cup F, \emptyset\rangle$, and satellite data $counter_E$ and $counter_F$;
\item $BSTEP\langle E\cup F, \emptyset \rangle$ stored in a red-black tree saved in a  variable   $RBT6$; 
for every ordered pair $({b}, {c})\in BSTEP\langle E\cup F, \emptyset \rangle$, the search key is the first component    ${b}$, the second component,    ${c}$ is  satellite data;
 \item the GTRS $R\langle E, F\rangle$ stored in a red-black tree saved in a  variable  $RBT7$;
  \item the GTRS $R\langle  F, E\rangle$ stored in a red-black tree  saved in a  variable   $RBT8$; 
 \item 
the auxiliary dpwpa $ AUX[ E;  F] =
(C\langle E\cup F, \emptyset \rangle,A[E; F])$ for the GTESs $E$ and $F$, and 
   for every vertex    ${b}\in C\langle E \cup  F, \emptyset  \rangle$,
the values of  the  attributes    ${b}.equal_E$,    ${b}.equal_F$,    ${b}.keeps_E$, and    ${b}.keeps_F$.
\end{itemize} 
By Lemma \ref{decidetotal}, we decide in  $O(n)$   time whether    
 $R\langle E, F\rangle$ is total, and whether    
 $R\langle F, E\rangle$ is total. 
  Then we decide  in a total of   $O(n \, \mathrm{log} \, n )$   time  which Main Cases hold.   
 \end{proof} 
 We study Main Cases 1, 2, 3, and 4 one by one in Sections \ref{tavasz}, \ref{nyar}, \ref{osz}, and \ref{tel}, respectively.
 First we consider Main Case 1.

\section{Main Case 1}\label{tavasz} 
In this main case,  $\Sigma$ is a unary  signature, and $E$ and $F$ are GTESs over   $\Sigma$. 
 We define  the  Algorithm  Nested   Partial Depth First Search (NPDFS for short) 
on the Auxiliary  dpwpa $AUX[ E;  F] $  for the GTESs $E$ and $F$.   NPDFS runs   in $O(n^2)$ time. Furthermore,   $\tthue {E\cup F}\subseteq \tthue E \cup \tthue F$ if and only if 
GTRS $R\langle E, F \rangle$ and GTRS $R\langle F, E \rangle$
together  simulate the  GTRS $R\langle E \cup F, \emptyset \rangle)$ if and only if 
  NPDFS outputs true. Thus,  
we can decide in  $O(n^2)$ time whether
   $\tthue {E\cup F}\subseteq \tthue E \cup \tthue F$.     First we note that 
  we can also apply the decision algorithm of   Champav{\`{e}}re  et al \cite{franciak}. 
  
Lemma \ref{alapvetes}, Proposition \ref{franciaeldontsimple}, and Proposition \ref{egyenlo} imply the following theorem.
\begin{thm}\label{1maincase1eldont}
Let $E$ and $F$ be GTESs over a unary  signature     $\Sigma$.
We can decide in $O(n^3)$ time whether
  $\tthue {E\cup F}= \tthue E \cup \tthue F$.
 \end{thm}
 We present a faster and simpler algorithm for solving our decidability problem. 
    We assume that we have already run  Algorithm CAD  on $\Sigma$, $E$, and $F$,  and have obtained all the output it produces, see Proposition \ref{kecskemet}. 
Example \ref{1oneegy} and its sequel, Example
 \ref{1oneegyfolytatas}, furthermore  Example \ref{2szepen} and its sequel,  \ref{2szepenfolytatas} 
 illustrate this main case.
 
   \begin{lem}\label{repulo} Let $E$ and $F$ be GTESs over  a unary  signature   $\Sigma$. 
   Let    ${b} \in C\langle E\cup F, \emptyset \rangle$ be
    such that GTRS $R \langle  E, F \rangle $ simulates the  GTRS $R \langle  E \cup F, \emptyset  \rangle  $ on the constant     ${b}\in C\langle  E \cup F, \emptyset \rangle $.  Assume that Conditions (a) and (b) hold:
\begin{itemize}
 \item [{\rm (a)}]
 $\delta_n[{b}_{n1}, \ldots, {b}_{nk_n} ]
 \red {R \langle  E \cup F, \emptyset  \rangle } \delta_{n-1}[{b}_{(n-1)1}, \ldots, {b}_{(n-1)k_{n-1} } ]\red {R \langle  E \cup F, \emptyset  \rangle }\delta_2[{b}_{21}, \ldots, {b}_{2k_2}]
 \red {R \langle E \cup  F, \emptyset  \rangle } 
 \cdots \newline \red {R \langle E \cup  F, \emptyset  \rangle } \delta_2[{b}_{21}, \ldots, {b}_{2k_2}]    \red {R \langle  E \cup F, \emptyset  \rangle }\delta_1[{b}_{11}, \ldots, {b}_{1k_1}]
  \red {R \langle  E \cup F, \emptyset  \rangle }{b}$, \newline 
where 
 $n\geq 1$ and
 for every $i\in \{\, 1, \ldots, n\, \}$, 
 $\left (k_i \in \Nat, \; \delta_i \in CON_{\Sigma, k_i},  \mbox{ and  } {b}_{i1}, \ldots, {b}_{ik_i} \in  C\langle {E\cup F}, \emptyset \rangle\right )$.
  \end{itemize} 
 \begin{itemize}
 \item [{\rm (b)}]    ${c}_{n1}, \ldots, {c}_{nk_n} \in C \langle E, F\rangle$ 
   such that     ${c}_{ni}  \subseteq {b}_{ni} $ for every $i\in \{\, 1, \ldots, k_n\, \}$.
 \end{itemize} 
  Then there exists    ${c}\in C \langle E, F \rangle$ with    ${c} \subseteq{b}$ and   for every  $j\in \{\, 1, \ldots, n\, \}$,
   there exists    ${c}_{j1}, \ldots, {c}_{jk_j} \in C \langle E, F \rangle$
    such that    ${c}_{ji} \subseteq {b}_{ji}$ for every $i \in \{\, 1, \ldots, k_j\, \}$, and the following sequence of reductions holds:
\begin{quote}
  $\delta_n[{c}_{n1}, \ldots, {c}_{nk_n} ]
 \red {R \langle  E \cup F, \emptyset  \rangle } \delta_{n-1}[{c}_{(n-1)1}, \ldots, {c}_{(n-1)k_{n-1} } ]\red {R \langle  E \cup F, \emptyset  \rangle }\delta_2[{c}_{21}, \ldots, {c}_{2k_2}]
 \red {R \langle E \cup  F, \emptyset  \rangle } 
 \cdots \newline \red {R \langle E \cup  F, \emptyset  \rangle } \delta_2[{c}_{21}, \ldots, {c}_{2k_2}]    \red {R \langle  E \cup F, \emptyset  \rangle }\delta_1[{c}_{11}, \ldots, {c}_{1k_1}]
  \red {R \langle  E \cup F, \emptyset  \rangle }{c}$,
 \end{quote} 
 \end{lem}
 \begin{proof}
 We proceed by induction on $n$.
 
 {\em Base Case:} $n=1$.  Then $\delta_1=\s(\diamondsuit, \ldots, \diamondsuit)$ for some $\s\in \Sigma_{k_1}$, 
 $\delta_1[{b}_{11}, \ldots, {k_1}]= \s({b}_{11}, \ldots, {k_1})$, and  
 $\s({b}_{11}, \ldots, {b}_{1k_1}) \dot{\rightarrow}  {b}\in  R\langle {E \cup F}, \emptyset \rangle$.
 We assumed that GTRS $R\langle E, F \rangle$ simulates the  GTRS $R\langle E \cup F, \emptyset \rangle$
 on the constant     ${b}\in C\langle  E \cup F, \emptyset \rangle $. 
 Hence  by Condition (b), and  the rewrite rule above  there exists a rewrite rule
  \begin{quote}
  $ \s({c}_{11}, \ldots, {c}_{1k_1})  \dot{\rightarrow} {c} \in R \langle E, F \rangle$ 
  \end{quote} 
for some     ${c}\in C \langle E, F \rangle$ with    ${c} \subseteq{b}$.
 Then we have 
  $\delta_1[{c}_{11}, \ldots, {c}_{1k_1}]\red {R \langle E, F \rangle }{c}$. 
 
 {\em Induction Step:} Let $n\geq 1$, and  assume that the statement of the lemma is true for $n$. We now show 
 that it is also true for $n+1$ as well.
   Assume that Conditions $(\mathrm{a}^\prime)$ and $(\mathrm{b}^\prime)$  hold:
 \begin{itemize}
 \item [($\mathrm{a}^\prime$)] $\delta_{n+1}[{b}_{(n+1)1}, \ldots, {b}_{(n+1)k_{(n+1)} }]
 \red {R \langle  E \cup F, \emptyset  \rangle } 
  \delta_n[{b}_{n1}, \ldots, {b}_{nk_n} ]
 \red {R \langle  E \cup F, \emptyset  \rangle } \delta_{n-1}[{b}_{(n-1)1}, \ldots, {b}_{(n-1)k_{n-1} } ]\red {R \langle  E \cup F, \emptyset  \rangle }
 \newline 
 \cdots \red {R \langle E \cup  F, \emptyset  \rangle } \delta_2[{b}_{21}, \ldots, {b}_{2k_2}]    \red {R \langle  E \cup F, \emptyset  \rangle }\delta_1[{b}_{11}, \ldots, {b}_{1k_1}]
  \red {R \langle  E \cup F, \emptyset  \rangle }{b}$, \newline 
where 
 $n\geq 1$ and
 for every $i\in \{\, 1, \ldots, n+1\, \}$, 
 $\left (k_i \in \Nat, \; \delta_i \in CON_{\Sigma, k_i},  \mbox{ and  } {b}_{i1}, \ldots, {b}_{ik_i} \in  C\langle {E\cup F}, \emptyset \rangle\right )$.
  \item [ ($\mathrm{b}^\prime$)  ]  
   ${c}_{(n+1)1}, \ldots, {c}_{(n+1)k_{n+1}} \in C \langle E, F\rangle$ 
   such that     ${c}_{(n+1)i}  \subseteq {b}_{(n+1)i} $ for every $i\in \{\, 1, \ldots, k_{n+1}\, \}$.
  
 \end{itemize}
   Then we obtain  
 $\delta_{n+1}$ from $\delta_n$ replacing the $\ell$th occurrence from left-to right of $\diamondsuit$ with the term $ \s(\diamondsuit, \ldots, \diamondsuit)$ for some $\s\in \Sigma_m$, 
 $m\in \Nat$. Consequently,
  \begin{quote}
  $k_{n+1}-m+1=k_n$, \newline
 $\delta_{n+1}[{b}_{(n+1)1}, \ldots, {b}_{(n+1)k_{(n+1)} }]= \delta_n[ {b}_{(n+1)1}, \ldots, {b}_{(n+1)(\ell-1) },\s( {b}_{(n+1)\ell}, \ldots, {b}_{(n+1)(\ell+m)}),  {b}_{(n+1)(\ell+1)}, \ldots  {b}_{(n+1)k_{(n+1)} }]$,
 \newline 
$\delta_n[{b}_{n1}, \ldots, {b}_{nk_n} ]=  \delta_n[ {b}_{(n+1)1}, \ldots, {b}_{(n+1)(\ell-1) }, {b}_{n, \ell}, {b}_{(n+1)(\ell+1)}, \ldots  {b}_{(n+1)k_{(n+1)} }]$,\newline 
     ${b}_{n1}={b}_{(n+1)1}, \ldots,
    {b}_{n(\ell-1)}={b}_{(n+1)(\ell-1)}, 
    {b}_{n(\ell+1)}={b}_{(n+1)(\ell-m+1)}, \ldots,
    {b}_{nk_n}= {b}_{(n+1)k_{(n+1)} }$
  \end{quote} 
   and 
  \begin{equation}\label{szoreg}
  \s( {b}_{(n+1)\ell}, \ldots, {b}_{(n+1)(\ell+m)}) \dot{\rightarrow}   {b}_{n\ell}\in  R\langle {E \cup F}, \emptyset\rangle.
    \end{equation}
 We assumed that 
 GTRS $R\langle E, F \rangle$ simulates the  GTRS $R\langle E \cup F, \emptyset \rangle$ 
  on  the constant     ${b}\in C\langle  E \cup F, \emptyset \rangle $. 
  By Condition ($\mathrm{a}^\prime$), $R \langle E\cup F,  \emptyset   \rangle $ reaches     ${b}$
 from a proper extension of    ${b}_{n\ell} $.   Therefore  by the rewrite rule (\ref{szoreg})  and Condition ($\mathrm{b}^\prime$),
 there exists a rewrite rule   
 \begin{equation}\label{szatymaz}
 \s( {c}_{(n+1)\ell}, \ldots, {c}_{(n+1)(\ell+m)}) \dot{\rightarrow}   {c}_{n\ell}\in  R\langle {E \cup F},\emptyset \rangle.
 \end{equation} 
 for some  $  {c}_{n\ell}
 \in  C \langle E\cup F, E \rangle $ with     ${c}_{n\ell}\subseteq {b}_{n \ell} $.
 Let 
 \begin{quote}
     ${c}_{n1}={c}_{(n+1)1}, \ldots,
    {c}_{n(l-1)}={c}_{(n+1)(l-1)}, \ldots,
    {c}_{n(l+1)}={c}_{(n+1)(l-m+1)}, \ldots,
    {c}_{nk_n}= {c}_{(n+1)k_{(n+1)} }$
  \end{quote}
 By the induction hypothesis, 
 there exists    ${c}\in C \langle E, F \rangle$ with    ${c} \subseteq{b}$ such that  the following sequence of reductions holds:
\begin{quote}
  $\delta_n[{c}_{n1}, \ldots, {c}_{nk_n} ]
 \red {R \langle  E \cup F, \emptyset  \rangle } \delta_{n-1}[{c}_{(n-1)1}, \ldots, {c}_{(n-1)k_{n-1} } ]\red {R \langle  E \cup F, \emptyset  \rangle }
 \cdots \newline \red {R \langle E \cup  F, \emptyset  \rangle } \delta_2[{c}_{21}, \ldots, {c}_{2k_2}]    \red {R \langle  E \cup F, \emptyset  \rangle }\delta_1[{c}_{11}, \ldots, {c}_{1k_1}]
  \red {R \langle  E \cup F, \emptyset  \rangle }{c}$,
 \end{quote}

   In this way  
   \begin{quote} 
   $\delta_{n+1}[{c}_{(n+1)1}, \ldots, {c}_{(n+1)k_(n+1)} ]
   =
   \newline
   \delta_n[ {c}_{(n+1)1}, \ldots, {c}_{(n+1)(\ell-1) },\s( {c}_{(n+1)\ell}, \ldots, {c}_{(n+1)(\ell+m)}),  {c}_{(n+1)(\ell+1)}, \ldots  {c}_{(n+1)k_{(n+1)} }]
  \red {R \langle  E \cup F, \emptyset  \rangle } \newline
   \delta_n[ {c}_{(n+1)1}, \ldots, {c}_{(n+1)(\ell-1) }, {c}_{n, \ell}, {c}_{(n+1)(\ell+1)}, \ldots  {c}_{(n+1)k_{(n+1)} }]=\delta_n[{c}_{n1}, \ldots, {c}_{nk_n} ]   
 \red {R \langle  E \cup F, \emptyset  \rangle } \newline 
 \delta_{n-1}[{c}_{(n-1)1}, \ldots, {c}_{(n-1)k_{n-1} } ]\red {R \langle  E \cup F, \emptyset  \rangle }
 \cdots \red {R \langle E \cup  F, \emptyset  \rangle } \delta_2[{c}_{21}, \ldots, {c}_{2k_2}]    \red {R \langle  E \cup F, \emptyset  \rangle }\delta_1[{c}_{11}, \ldots, {c}_{1k_1}]
  \red {R \langle  E \cup F, \emptyset  \rangle }{c}$.
    \qedhere
  \end{quote}
    \end{proof}

   \begin{lem}\label{tanacs} Let $E$ and $F$ be GTESs over  a unary  signature   $\Sigma$. 
  \begin{quote}
    For each   
  $t\in W $, 
  $\left( t/_{\tthue {E\cup F}}=  t/_{\tthue E} \mbox{ or }
   t/_{\tthue {E\cup F}}=  t/_{\tthue F}\right)$ 
   \end{quote}
  if and only if 
  GTRS $R\langle E, F \rangle$ and GTRS $R\langle F, E \rangle$
together  simulate the  GTRS $R\langle E \cup F, \emptyset \rangle$.
 \end{lem}
 \begin{proof} 
 $(\Rightarrow)$ 
 Let $t\in W $. Then by our assumption, $t/_{\tthue {E\cup F}}=  t/_{\tthue E}$  or 
  $ t/_{\tthue {E\cup F}}=  t/_{\tthue F}$. Assume that $ t/_{\tthue {E\cup F}}=  t/_{\tthue E}$. The case $ t/_{\tthue {E\cup F}}=  t/_{\tthue F}$ is similar.
 Let     ${c} \in C\langle E\cup F, \emptyset \rangle$ be such that 
 $R\langle {E \cup F}, \emptyset \rangle$ reaches   $t/_{\Theta\langle {E\cup F}, \emptyset \rangle}$ from   a proper extension    of the constant    ${c}$. Then there exists $\delta\in CON_{\Sigma, 1}$
 such that $\delta[{c}] {\downarrow}_{R\langle E \cup F, \emptyset \rangle}= t/_{\Theta\langle {E\cup F}, \emptyset \rangle}$. Consider any rewrite rule 
 $\s({c}_1, \ldots, {c}_m)  \dot{\rightarrow}   {c}$ in $R\langle E \cup F, \emptyset \rangle$.
 with right-hand side    ${c}$.
 We now distinguish two cases depending on the value of $m$.  
 
{\em Case 1:} $m=1$. Then we  consider  the rewrite rule 
 $\s({c}_1) \dot{\rightarrow} {c} \in  R\langle {E \cup F}, \emptyset\rangle$. Then  let
     ${d}_1\in C \langle E,  F \rangle$   with
    ${d}_1 \subseteq {c}_1$, and let  $t_1\in \ts$ such that $t_1\tred {R\langle E, F \rangle}{d}_1 $. 
  Then we have 
 \begin{quote}
   $\delta[\s(t_1)] \tred {R\langle E\cup F, \emptyset \rangle}\delta[\s({c}_1)]   \red {R\langle E\cup F, \emptyset \rangle} \delta[{c}]  \tred {R\langle E\cup F, \emptyset \rangle}
   t/_{\Theta\langle E\cup F, \emptyset \rangle}$.
    \end{quote}
    By Proposition \ref{harmadik}, $\delta[\s(t_1)] \tthue {E \cup F} t$. 
  As a consequence, by our assumption that $ t/_{\tthue {E\cup F}}=  t/_{\tthue E}$, we have 
  $\delta[\s(t_1)] \tthue E  t$. Consequently, by Proposition \ref{otodik},
  \begin{quote} 
         $\delta[\s(t_1)] \tred {R\langle E, F \rangle}t/_{\Theta\langle E, F \rangle}$.
   \end{quote}
 Hence 
 \begin{quote}
        $t_1\tred {R\langle E, F \rangle}{d}_1 $, \newline
         $\s(t_1)\tred {R\langle E, F \rangle}\s({d}_1) \red {R\langle E, F \rangle}{d}$ for some    ${d}\in C \langle E,  F \rangle$, 
  where the rewrite rule $\s({d}_1) \dot{\rightarrow} {d}$ is in $R\langle E, F \rangle$,        
         and  \newline
          $\delta[\s(t_1)] \tred {R\langle E, F \rangle} \delta[\s({d}_1)] \tred {R\langle E, F \rangle} 
          \delta[{d}] \tred {R\langle E, F \rangle}t/_{\Theta\langle E, F \rangle}$.
     \end{quote}

{\em Case 2:}   $m=0$. Then we consider the rewrite rule  
 $\s \dot{\rightarrow} {c} \in  R\langle {E \cup F}, \emptyset \rangle$. 
 Then we have 
 \begin{quote}
  $\delta[\s] \red {R\langle E\cup F, \emptyset \rangle} \delta[{c}] \tred {R\langle E\cup F, \emptyset \rangle}t/_{\Theta\langle E\cup F, \emptyset \rangle}$.
     \end{quote} 
By Proposition \ref{harmadik}, $\delta[\s] \tthue {E \cup F} t$. 
  As a consequence, by our assumption that $ t/_{\tthue {E\cup F}}=  t/_{\tthue E}$, we have 
  $\delta[\s] \tthue E  t$. Consequently, by Proposition \ref{otodik},
 \begin{quote}
 $\delta[\s] \tred {R\langle E, F \rangle} t/_{\Theta\langle E, F \rangle}$.
     \end{quote} 
  Therefore
 \begin{quote}
                 $\s\red {R\langle E, F \rangle} {d}$ for some    ${d}\in C \langle E, F\rangle$, 
  where the rewrite rule $\s \dot{\rightarrow} {d}$ is in $R\langle E, F \rangle$,        
         and  \newline
          $\delta[\s] \red {R\langle E, F \rangle} \delta[{d}] \tred {R\langle E, F \rangle}t/_{\Theta\langle E, F \rangle}$.
     \end{quote}

  $(\Leftarrow)$    Let $t \in W$ and   $s\in  t/_{\tthue {E\cup F}}$.  
 Then
  $s\tred {R \langle E\cup F, \emptyset   \rangle } t/_{\Theta\langle {E\cup F}, \emptyset \rangle}$. We assumed that GTRS $R\langle E, F \rangle$ and GTRS $R\langle F, E \rangle$
together  simulate the  GTRS $R\langle E \cup F, \emptyset \rangle$. Consequently,  
 $R\langle E, F \rangle$  or  $R\langle F, E \rangle$
    simulate   $R\langle E \cup F, \emptyset \rangle$ on  the constant  $t/_{\Theta\langle {E\cup F}, \emptyset\rangle}$. Assume that $R\langle E, F \rangle$ 
    simulates  $R\langle E \cup F, \emptyset \rangle$ on  the constant  $t/_{\Theta\langle {E\cup F},\emptyset \rangle}$. The case  $R\langle F, E \rangle$ 
    simulates  $R\langle E \cup F, \emptyset \rangle$ on  the constant  $t/_{\Theta\langle {E\cup F}, \emptyset \rangle}$ is similar. 
  By Lemma \ref{repulo}, there exists     ${e}\in C \langle E, F \rangle$ with    ${e} \subseteq  t/_{\Theta\langle {E \cup F}, \emptyset \rangle}$ such that $s\tred {R \langle E, F \rangle } {e}$. 
  By Proposition \ref{szazadik},     ${e}=  t/_{\Theta\langle E, F \rangle}$. 
   For this reason 
     $s\in  t/_{\tthue E}$.  
  Consequently, 
  for any  $s\in  t/_{\tthue {E\cup F}}$, we have
    $s\in  t/_{\tthue E}$. 
As a consequence
  $ t/_{\tthue {E\cup F}}\subseteq  t/_{\tthue E}$, and hence  $ t/_{\tthue {E\cup F}}=  t/_{\tthue E} $.
    \end{proof}

  \subsection{Nested Partial Depth First Search  on the Auxiliary  dpwpa $  AUX[ E;  F] $  for the GTESs $E$ and $F$}
  Algorithm  Nested   Partial Depth First Search (NPDFS for short) on the Auxiliary  dpwpa $AUX[ E;  F] $  for the GTESs $E$ and $F$ implements a nested loop. 
 Algorithm NPDFS calls the  function  NPDFS-VISIT, and in  this way it  implements the outer loop, which is 
  a prefix of a depth first search    on  $AUX[ E;  F] $,
 and  can terminate before visiting all vertices of $AUX[ E;  F] $ in case that  GTRS $R\langle E, F \rangle$ and GTRS $R\langle F, E \rangle$
together do not  simulate the  GTRS $R\langle E \cup F, \emptyset \rangle$. 
\begin{itemize} 
\item 
Algorithm NPDFS calls the functions INNER-LOOP-E, INNER-LOOP-F,   and 
 NPDFS-VISIT.  
 \item  Function  NPDFS-VISIT calls the function INNER-LOOP-E,  INNER-LOOP-F,   and 
 NPDFS-VISIT.       
 \end{itemize}
  In this way, algorithm  NPDFS decides the following formula:  
    For each    ${b}\in C\langle  E \cup F, \emptyset \rangle $, 
   \begin{itemize}
  \item
        ${b}.equal_E=\mathrm{true} $ and     
    for each    ${c} \in C\langle  E \cup F, \emptyset \rangle $ if there is a  path of positive length or a cycle     from  
     ${b}$ to    ${c}$
 in the auxiliary    dpwpa   
$AUX[ E;  F] =
(C\langle E \cup F, \emptyset\rangle, A[E; F])$, then      ${c}.keeps_E=\mathrm{true}$, or
 \item
        ${b}.equal_F=\mathrm{true} $ and     
    for each    ${c} \in C\langle  E \cup F, \emptyset \rangle $ if there is a  path of positive length or a cycle     from  
     ${b}$ to    ${c}$
 in the  auxiliary     dpwpa   
$AUX[ E;  F] =
(C\langle E \cup F, \emptyset\rangle, A[E; F])$, then  
   ${c}.keeps_F=\mathrm{true}$.
 \end{itemize}

Function 
 INNER-LOOP-E    is an inner  loop inside the body of the outer loop, it 
 also implements a prefix of a depth first search    on those vertices    ${c}$ which are positive-step reachable 
  from the current node    ${b}$ of the outer loop,   it can terminate before visiting all vertices
 of $ AUX[ E;  F] $ which are positive-step 
 reachable from    ${b}$ in case that it visits a vertex    ${c}$ such that    ${c}.keeps_E=\mathrm{false}$. 
 It checks whether  for each vertex    ${c}$ which are positive-step reachable 
  from the current node    ${b}$ of the outer loop, 
    ${c}.keeps_E=\mathrm{true}$. 
  INNER-LOOP-E decides the following formula:
  
  For each vertex     ${c}\in C\langle  E \cup F, \emptyset \rangle $ which is positive-step reachable 
  from the current node    ${b}$ of the outer loop,     ${c}.keeps_E=\mathrm{true}$.

 Function 
 INNER-LOOP-F    is an inner  loop inside the body of the outer loop, it 
 also implements a prefix of a depth first search    on those vertices    ${c}$ which are positive-step reachable 
  from the current node    ${b}$ of the outer loop,   it can terminate before visiting all vertices
 of $ AUX[ E;  F] $ which are positive-step 
 reachable from    ${b}$ in case that in case that it visits a vertex    ${c}$ such that    ${c}.keeps_F=\mathrm{false}$.  It checks whether  for each vertex    ${c}$ which are positive-step reachable 
  from the current node    ${b}$ of the outer loop, 
    ${c}.keeps_F=\mathrm{true}$. 
  INNER-LOOP-F decides the following formula:
  
  For each vertex     ${c}\in C\langle  E \cup F, \emptyset \rangle $ which is positive-step reachable 
  from the current node    ${b}$ of the outer loop,     ${c}.keeps_F=\mathrm{true}$.

 \begin{df}\label{szamar} \rm
  Let $\Sigma$ be a signature,  $E$ and $F$ be  GTESs over $\S$.  \newline
Algorithm Nested   Partial Depth First Search on Auxiliary dpwpa $ AUX[ E;  F] $ (NPDFS for short)  \newline
 Input:  the auxiliary dpwpa  $ AUX[ E;  F] $.
  \newline
Output: Boolean.\newline
\{\ $\mathrm{NPDFS}(AUX[ E;  F] )= \mathrm{true }$ if GTRS $R\langle E, F \rangle$ and GTRS $R\langle F, E \rangle$
together  simulate the  GTRS $R\langle E \cup F, \emptyset \rangle$, otherwise $\mathrm{NPDFS}(AUX[ E;  F] )= \mathrm{false}$.\ \}
\newline
{\bf var}      ${b}: C \langle E \cup F, \emptyset \rangle$; 
\begin{lstlisting}
  (*{\bf for each} vertex    ${b}\in C\langle E \cup F, \emptyset\rangle$ {\bf do} *)
    (*{\bf begin}*) 
      (*   ${b}.visited\_outer\_loop:=\mathrm{false}$*) 
      (*   ${b}.visited\_inner\_loop_E:=\mathrm{false}$*)
      (*   ${b}.visited\_inner\_loop_F:=\mathrm{false}$*)
    (*{\bf end}*)
  (*{\bf for each} vertex    ${b}\in C\langle E \cup F, \emptyset\rangle$ {\bf do} *)
    (*{\bf if} {\bf not}    ${b}.visited\_outer\_loop$*) 
      (*{\bf then} {\bf begin}*)
        (*   ${b}.visited\_outer\_loop:=\mathrm{true}$*) 
        (*   ${\bf if}$ {\bf not} (   ${b}.equal_E$ {\bf or }    ${b}.equal_F$)*)
          (*{\bf then} {\bf return} $\mathrm{false}$*)
        (*   ${\bf for}$    ${\bf  each}$ vertex    ${c}\in Adj({b})$ \bf{do}*)
           (*   ${\bf if}$ {\bf not} ((${b}.equal_E$  {\bf and} $\mathrm{INNER-LOOP-E}(AUX[ E;  F] , {c})$ {\bf and }  $(\mathrm{NPDFS-VISIT}( AUX[ E;  F] , {c})$*) 
           (*{\bf or } (${b}.equal_F$  {\bf and} $\mathrm{INNER-LOOP-F}(AUX[ E;  F] , {c})$ {\bf and }  $(\mathrm{NPDFS-VISIT}( AUX[ E;  F] , {c}))$*)
             (*{\bf then} {\bf return} $\mathrm{false}$*)
        (*{\bf end}*) 
  (*{\bf output}  $(\mathrm{true})$*)   
     
   
                       
(*{\bf function} $\mathrm{NPDFS-VISIT}(AUX[ E;  F] , {b}):\; \mathrm{Boolean}$*)
  (*{\bf if} {\bf not}    ${b}.visited\_outer\_loop$*) 
      (*{\bf then} {\bf begin}*)
        (*   ${b}.visited\_outer\_loop:=\mathrm{true}$*)
        (*   ${\bf if}$ {\bf not} (${b}.equal_E$ {\bf or }    ${b}.equal_F$)*)
          (*{\bf then} {\bf return} $\mathrm{false}$*)        
        (*   ${\bf for}$    ${\bf  each}$ vertex    ${c}\in Adj({b})$ \bf{do}*)
           (*   ${\bf if}$ {\bf not} ((${b}.equal_E$  {\bf and} $\mathrm{INNER-LOOP-E}(AUX[ E;  F] , {c})$ {\bf and }  $(\mathrm{NPDFS-VISIT}( AUX[ E;  F] , {c})$*) 
           (*{\bf or } (${b}.equal_F$  {\bf and} $\mathrm{INNER-LOOP-F}(AUX[ E;  F] , {c})$ {\bf and }  $(\mathrm{NPDFS-VISIT}( AUX[ E;  F] , {c}))$*)
             (*{\bf then} {\bf return} $\mathrm{false}$*)
        (*{\bf end}*) 
  (*{\bf output}  $(\mathrm{true})$*)   

 
   
(*{\bf function} $\mathrm{INNER-LOOP-E}( AUX[ E;  F] , {b}):\; \mathrm{Boolean}$*)    
  (*   ${b}.visited\_inner\_loop_E:=\mathrm{true}$;*)
  (*{\bf if}    ${b}.keeps_E$*)  
    (*   ${\bf then }$ {\bf for} {\bf  each} vertex    ${c}\in Adj({b})$ \bf{do}*)
      (*   ${\bf if }$ {\bf not}\; $\mathrm{INNER-LOOP-E}(AUX[ E;  F] , {c})$*) 
         (*   ${\bf then}\; {\bf begin}$*) 
           (*   ${b}.visited\_inner\_loop_E:=\mathrm{false}$;*)
           (*   ${\bf return}\; \mathrm{false}$*) 
         (*{\bf end}*)
  (*   ${\bf else}\; {\bf begin}$*)
    (*   ${b}.visited\_inner\_loop_E:=\mathrm{false}$;*)
    (*   ${\bf return}\; \mathrm{false}$*)  
  (*{\bf end}*)
  (*   ${b}.visited\_inner\_loop_E:=\mathrm{false}$;*)  
  (*$ {\bf output}(\mathrm{true})$*)  
                          
(*{\bf function} $\mathrm{INNER-LOOP-F}( AUX[ E;  F] , {b}):\; \mathrm{Boolean}$*)    
  (*   ${b}.visited\_inner\_loop_F:=\mathrm{true}$;*)
  (*{\bf if}    ${b}.keeps_F$*)  
    (*   ${\bf then }$ {\bf for} {\bf  each} vertex    ${c}\in  Adj({b})$ \bf{do}*)
      (*   ${\bf if }$ {\bf not}\; $\mathrm{INNER-LOOP-F}(AUX[ E;  F] , {c})$*) 
         (*   ${\bf then}\; {\bf begin}$*) 
           (*   ${b}.visited\_inner\_loop_F:=\mathrm{false}$;*)
           (*   ${\bf return}\; \mathrm{false}$*) 
         (*{\bf end}*)
  (*   ${\bf else}\; {\bf begin}$*)
    (*   ${b}.visited\_inner\_loop_F:=\mathrm{false}$;*)
    (*   ${\bf return}\; \mathrm{false}$*)  
  (*{\bf end}*)
  (*   ${b}.visited\_inner\_loop_F:=\mathrm{false}$;*)  
  (*$ {\bf output}(\mathrm{true})$*)  
\end{lstlisting}
\end{df}
  By direct inspection of Definition \ref{szamar}, we get the following result.
       \begin{lem}\label{attekint}  Let $E$ and $F$ be GTESs over  a unary  signature   $\Sigma$. 
       NPDFS outputs true on  $AUX[ E;  F] $ if and only if 
    for each    ${b}\in C\langle  E \cup F, \emptyset \rangle $, Conditions 1 or 2 hold:
   \begin{itemize}
  \item[1.]
        ${b}.equal_E=\mathrm{true} $ and     
    for each    ${c} \in C\langle  E \cup F, \emptyset \rangle $ if there is a  path of positive length or a cycle     from  
     ${b}$ to    ${c}$
 in the auxiliary    dpwpa   
$ AUX[ E;  F] =
(C\langle E \cup F, \emptyset\rangle,A[E; F])$, then      ${c}.keeps_E=\mathrm{true}$.
 \item[2.]
        ${b}.equal_F=\mathrm{true} $ and     
    for each    ${c} \in C\langle  E \cup F, \emptyset \rangle $ if there is a  path of positive length or a cycle     from  
     ${b}$ to    ${c}$
 in the  auxiliary     dpwpa   
$AUX[ E;  F] =
(C\langle E \cup F, \emptyset\rangle, A[E; F])$, then  
   ${c}.keeps_F=\mathrm{true}$.
\end{itemize}
    \end{lem} 
 
We get the following lemma by direct inspection of Definition \ref{szamar}.
   \begin{lem}\label{megkeres}  Let $E$ and $F$ be GTESs over  a unary  signature   $\Sigma$. \newline
If GTRS $R\langle E, F \rangle$ and GTRS $R\langle F, E \rangle$
together  simulate the  GTRS $R\langle E \cup F, \emptyset \rangle$,  then Conditions 1--3  hold: 
\begin{itemize}
\item[1.]
 Algorithm NPDFS   on   $AUX[ E;  F] $   in the outer loop visits once   each vertex    ${b} \in C\langle E \cup F, \emptyset\rangle $. 
 \item [2.]    For each  vertex    ${b}\in C\langle  E \cup F, \emptyset \rangle $ visited in the outer loop, 
   \begin{itemize}
  \item 
 [({b})]       ${b}.equal_E=\mathrm{true} $ and     
    for each    ${c} \in C\langle  E \cup F, \emptyset \rangle $ if there is a  path of positive length or a cycle     from  
     ${c}$ to    ${b}$
 in the auxiliary dpwpa   
$ AUX[ E;  F] =
(C\langle E \cup F, \emptyset\rangle, A[E; F])$, then      ${c}.keeps_E=\mathrm{true}$, or
 \item [({c})]
        ${b}.equal_F=\mathrm{true} $ and     
    for each    ${c} \in C\langle  E \cup F, \emptyset \rangle $ if there is a  path of positive length or a cycle     from  
     ${c}$ to    ${b}$
 in the auxiliary dpwpa   
$AUX[ E;  F] =
(C\langle E \cup F, \emptyset\rangle, A[E; F])$, then  
   ${c}.keeps_F=\mathrm{true}$.
\end{itemize}
  
   \item [3.] Algorithm NPDFS     on   $AUX[ E;  F] $ outputs true.
 \end{itemize} 
 If GTRS $R\langle E, F \rangle$ and GTRS $R\langle F, E \rangle$ 
together do not  simulate the  GTRS $R\langle E \cup F, \emptyset \rangle$, then Conditions 4--6  hold: 
 \begin{itemize}
  \item [4.]
  Algorithm NPDFS   on  $AUX[ E;  F] $  visits at most once each vertex    ${b} \in C\langle E \cup F, \emptyset\rangle $  in the outer loop. 
  
  \item [5.]    There exists a  vertex    ${b}\in C\langle  E \cup F, \emptyset \rangle $ such that Algorithm NPDFS   on $AUX[ E;  F] $  
   visits    ${b} \in C\langle E \cup F, \emptyset\rangle $ in the outer loop, and   Conditions 2.(a) and 2.(b) are false for    ${b}$. 
  \item[6.] Algorithm NPDFS   on  $AUX[ E;  F] $  outputs false.
  \end{itemize}
  \end{lem}

\begin{exa}\label{1oneegynegyedikfolytatas}\rm We now continue
Example \ref{1oneegy}, Example \ref{1oneegyfolytatas}, Example \ref{1oneegymasodikfolytatas}, and Example \ref{1oneegyharmadikfolytatas}. 
We run Algorithm NPDFS on  $AUX[ E;  F] $. Then  Algorithm NPDFS outputs true.
 By the last line of  Example \ref{1oneegyharmadikfolytatas},  GTRS $R\langle E, F \rangle$ and GTRS $R\langle F, E \rangle$
together  simulate the  GTRS $R\langle E \cup F, \emptyset \rangle$.
   Consequently  Algorithm NPDFS outputs the  answer  to the question whether GTRS $R\langle E, F \rangle$ and GTRS $R\langle F, E \rangle$ 
together  simulate the  GTRS $R\langle E \cup F, \emptyset \rangle$. 
   \end{exa}

\begin{exa}\label{2szepennegyedikfolytatas}\rm 
We now continue Examples  \ref{2szepen},  \ref{2szepenfolytatas},   \ref{2szepenmasodikfolytatas}, and   \ref{2szepenharmadikfolytatas}. 
We run Algorithm NPDFS  on $AUX[ E;  F] $.   
 Then  Algorithm NPDFS outputs false. By the last lines of Example  \ref{2szepenharmadikfolytatas},  GTRS $R\langle E, F \rangle$ and GTRS $R\langle F, E \rangle$ 
together do not  simulate the  GTRS $R\langle E \cup F, \emptyset \rangle$.
    Thus  Algorithm NPDFS outputs the  answer  to the question whether GTRS $R\langle E, F \rangle$ and GTRS $R\langle F, E \rangle$ 
together  simulate the  GTRS $R\langle E \cup F, \emptyset \rangle$. 
   \end{exa}

\begin{lem}\label{supportfut}{\em \cite{cormen}} Let $E$ and $F$ be GTESs over a unary  signature   $\Sigma$.
 Algorithm NPDFS runs on $AUX[ E;  F] $  in $O(n^2)$ time.
 \end{lem}
\begin{proof} Recall that  $n=\mathrm{size}(E)+\mathrm{size}(F)$.
 By Statement \ref{darabsz}, 
$|C\langle E \cup F, \emptyset\rangle | \leq n$. By   Statement \ref{iszap},  
  $| A[E; F]|\leq 2  n$.
    The running time of depth first search on   $AUX[ E;  F] $ is $O(|C\langle E \cup F, \emptyset\rangle | + | A[E; F]|)$   
\cite[Section~22.3]{cormen}. Consequently, the running time of depth first search on   $ AUX[ E;  F] $ is $O(n)$.

The outer loop of 
 Algorithm NPDFS on  $AUX[ E;  F] $ implements a prefix of 
the depth first search  on $AUX[ E;  F] $  \cite{cormen}, it can terminate before visiting all vertices of $ AUX[ E;  F] $.
Consequently, the outer loop  on   $AUX[ E;  F] $ takes  $O(n)$ time.

Function  INNER-LOOP-E  on  $AUX[ E;  F] $ and any  vertex     ${b}$
 implements a prefix of a depth first search    on $ AUX[ E;  F] $ starting from    ${b}$. 
Accordingly the running time of function INNER-LOOP-E is $O(n)$.  
 Similarly, the running time of  function INNER-LOOP-F is
 $O(n)$.  
 By  Statement \ref{darabsz}, $|C\langle E \cup F, \emptyset\rangle | \leq n$, therefore
INNER-LOOP-E and INNER-LOOP-F run at most $n$ times. Accordingly,   Algorithm NPDFS runs on   $ AUX[ E;  F] $  in $O(n^2)$ time.
    \end{proof}

\subsection{Decision Algorithm and its verification}
      We present our decision algorithm and show its correctnes.
 
\begin{lem}\label{pazsit}
  Let $E$ and $F$ be GTESs over  a unary  signature   $\Sigma$. 
  Then the following conditions are pairwise equivalent: 
  \begin{itemize}
  \item[(i)]
  $\tthue {E\cup F}\subseteq \tthue E \cup \tthue F$, 
  \item[(ii)]
 For each $t\in W  $, 
   $( t/_{\tthue {E\cup F}}=  t/_{\tthue E} \mbox{  or }
 t/_{\tthue {E\cup F}}=  t/_{\tthue F})$.

 \item[(iii)] GTRS $R\langle E, F \rangle$ and GTRS $R\langle F, E \rangle$
together  simulate the  GTRS $R\langle E \cup F, \emptyset \rangle$.

 \item[(iv)]  
  For each    ${b}\in C\langle  E \cup F, \emptyset \rangle $, GTRS $R\langle E, F \rangle$ or  GTRS $R\langle F, E \rangle$  simulate the  GTRS $R\langle E \cup F, \emptyset \rangle$
   on  the constant       ${b}$.

 \item[(v)]     For each    ${b}\in C\langle  E \cup F, \emptyset \rangle $, 
   \begin{itemize}
  \item
        ${b}.equal_E=\mathrm{true} $ and     
    for each    ${c} \in C\langle  E \cup F, \emptyset \rangle $ if there is a  path of positive length or a cycle     from  
     ${b}$ to    ${c}$
 in the auxiliary    dpwpa   
$ AUX[ E;  F] =
(C\langle E \cup F, \emptyset\rangle,A[E; F])$, then      ${c}.keeps_E=\mathrm{true}$, or
 \item
        ${b}.equal_F=\mathrm{true} $ and     
    for each    ${c} \in C\langle  E \cup F, \emptyset \rangle $ if there is a  path of positive length or a cycle     from  
     ${b}$ to    ${c}$
 in the  auxiliary     dpwpa   
$AUX[ E;  F] =
(C\langle E \cup F, \emptyset\rangle,A[E; F])$, then  
   ${c}.keeps_F=\mathrm{true}$.
 \end{itemize}

\item[(vi)] $NPDFS(AUX[ E;  F] )=\mathrm{true}$.

\end{itemize}
\end{lem}
\begin{proof}
$\mathrm{(i)} \Leftrightarrow \mathrm{(ii)}$: By Lemma \ref{alapvetes}.

$\mathrm{(ii)} \Leftrightarrow \mathrm{(iii)}$:  By Lemma \ref{tanacs}.

$\mathrm{(iii)} \Leftrightarrow \mathrm{(iv)}$: By Definition \ref{defsimulate}. 

$\mathrm{(iv)} \Leftrightarrow \mathrm{(v)}$: By Definition \ref{singlesimulate} and  Definition 
\ref{mimika} and Lemma \ref{eso} and   Lemma \ref{rozmaring}.
 
$\mathrm{(v)} \Leftrightarrow \mathrm{(vi)}$: By Lemma \ref{attekint}.
 \end{proof}
 \begin{lem}\label{1tetharmadik1}  Let $E$ and $F$ be GTESs over  a unary  signature   $\Sigma$. 
    Then we can decide in  $O(n^2)$ time whether
   $\tthue {E\cup F}\subseteq \tthue E \cup \tthue F$.
 \end{lem}   
 \begin{proof}
 Recall that we have run Algorithm CAD  on $\Sigma$, $E$, and $F$. 
By Proposition \ref{kecskemet}, we have obtained in   $O(n \, \mathrm{log} \, n )$ time
the auxiliary dpwpa $ AUX[ E;  F] =
(C\langle E\cup F, \emptyset \rangle,A[E; F])$ for the GTESs $E$ and $F$, and 
   for every vertex    ${b}\in C\langle E \cup  F, \emptyset  \rangle$,
the values of  the  attributes    ${b}.equal_E$,    ${b}.equal_F$,    ${b}.keeps_E$, and    ${b}.keeps_F$.

3. We run Algorithm NPDFS on $AUX[ E;  F] $ in $O(n^2)$ time by Lemma \ref{supportfut}. Then  by Lemma \ref{pazsit},
  NPDFS outputs true if and only if 
 $\tthue {E\cup F}\subseteq \tthue E \cup \tthue F$.

 Summing up the time complexities of the three steps above, we get that  the time  complexity of our entire computation is   $O(n^2)$.
 \end{proof}

\begin{exa}\label{1oneegyotodikfolytatas}\rm We now continue
Examples \ref{1oneegy},  \ref{1oneegyfolytatas},  \ref{1oneegymasodikfolytatas},   \ref{1oneegyharmadikfolytatas}. and  \ref{1oneegynegyedikfolytatas}. 
Recall that we run Algorithm NPDFS on the auxiliary dpwpa $ AUX[ E;  F] $, and  Algorithm NPDFS outputs true to the question whether 
 GTRS $R\langle E, F \rangle$ and GTRS $R\langle F, E \rangle$ 
together  simulate the  GTRS $R\langle E \cup F, \emptyset \rangle$.
  Hence Condition (iii) in Lemma \ref{pazsit} holds.
   Consequently, by Lemma \ref{pazsit},
    $\tthue {E\cup F} \subseteq \tthue E\cup \tthue F$.
     \end{exa}

\begin{exa}\label{2szepenotodikfolytatas}\rm 
We now continue Examples \ref{2szepen},  \ref{2szepenfolytatas},   \ref{2szepenmasodikfolytatas},  \ref{2szepenharmadikfolytatas},  and   \ref{2szepennegyedikfolytatas}. 
Recall that we  run Algorithm NPDFS  on the auxiliary dpwpa  $AUX[ E;  F] $, and   Algorithm NPDFS outputs false to the question whether 
 GTRS $R\langle E, F \rangle$ and GTRS $R\langle F, E \rangle$ 
together  simulate the  GTRS $R\langle E \cup F, \emptyset \rangle$.
  For this reason Condition (iii) in Lemma \ref{pazsit} does not  hold.
   Consequently, by Lemma \ref{pazsit},
    $\tthue {E\cup F} \not \subseteq \tthue E\cup \tthue F$. 
    
\end{exa}

\section{Main Case 2}\label{nyar} 
In this main case, $E$ and $F$ are  GTESs over  a  signature   $\Sigma$ such that  
 GTRS $R\langle E, F\rangle$ and GTRS $R\langle F, E\rangle$ are  total. 
  We assume that we have already run  Algorithm CAD  on $\Sigma$, $E$, and $F$,  and have obtained all the output it produces, see Proposition \ref{kecskemet}.
 We showed that   $\tthue {E\cup F}\subseteq \tthue E\cup \tthue F$
if and only if each $\Theta\langle E\cup F, \emptyset \rangle$  equivalence    class is equal to a 
 $\Theta\langle E, F \rangle$ equivalence    class or a $\Theta\langle F, E \rangle$ equivalence    class.
By  direct inspection of  
  $INC( E\cup F, E)$ stored in a red-black tree and 
$INC( E\cup F, F)$ stored in a red-black tree 
  we decide in    $O(n)$ 
 time whether each $\Theta\langle E\cup F, \emptyset \rangle$  equivalence    class is equal to a 
 $\Theta\langle E, F \rangle$ equivalence    class or a $\Theta\langle F, E \rangle$ equivalence    class. Thence we decide in   
 $O(n\,\mathrm{log}\, n)$ time whether 
 $\tthue {E\cup F}\subseteq \tthue E\cup \tthue F$. 
 Example  \ref{3szilva} and its sequels, Examples \ref{3szilvafolytatas} and  \ref{3szilvamasodikfolytatas}, in addition 
 Example  \ref{4u2staring} and its sequels, Examples \ref{4u2staringfolytatas} and  \ref{4u2staringmasodikfolytatas}
 exemplify this case. 
First we note that 
  we can also apply the decision algorithm of   Champav{\`{e}}re  et al \cite{franciak}.
 Lemma \ref{asas}  and Proposition \ref{franciaeldontsimple},  and Proposition \ref{egyenlo} imply the following theorem.
    \begin{thm}\label{2maincase22} 
 Let $E$ and $F$ be GTESs over  a  signature   $\Sigma$ such that  
 GTRS $R\langle E, F\rangle$ and GTRS $R\langle F, E\rangle$ are  total. 
 We can decide in $O(n^3)$ time whether
  \begin{quote}
  $\tthue {E\cup F}= \tthue E \cup \tthue F$.
  \end{quote}
 \end{thm}
 We present a faster  and simpler algorithm for solving our decidability problem.
 \begin{exa}\label{3szilvamasodikfolytatas} \rm  We now continue Examples
 \ref{3szilva} and  \ref{3szilvafolytatas}. Recall that 
 \begin{itemize}
 \item  both GTRS $R\langle E, F\rangle$ and GTRS $R\langle F, E\rangle$ are  total,
  \item  $\#/_{\Theta\langle E, F \rangle} \subset \#/_{\Theta\langle {E\cup F}, \emptyset \rangle}$,
 \item   $\$/_{\Theta\langle E, F \rangle} = \$/_{\Theta\langle {E\cup F}, \emptyset\rangle}$, 
  \item   $\mathsterling/_{\Theta\langle E, F \rangle} = \mathsterling/_{\Theta\langle {E\cup F}, \emptyset \rangle}$, 
 \item  $\#/_{\Theta\langle F, E \rangle} = \#/_{\Theta\langle {E\cup F}, \emptyset \rangle}$,
 \item   $\$/_{\Theta\langle F, E \rangle} \subset\$/_{\Theta\langle {E\cup F}, \emptyset \rangle}$,  
  \item   $\mathsterling/_{\Theta\langle F, E \rangle} = \mathsterling/_{\Theta\langle {E\cup F}, \emptyset \rangle}$, and   
   \item 
  $ \tthue {E\cup F}\subseteq \tthue E\cup \tthue F$.
\end{itemize}
 \end{exa}

 \begin{exa}\label{4u2staringmasodikfolytatas}  \rm  We now continue Examples
 \ref{4u2staring} and  \ref{4u2staringfolytatas}. Recall that 
 \begin{itemize}
 \item  both GTRS $R\langle E, F\rangle$ and GTRS $R\langle F, E\rangle$ are  total,
 \item 
 the only equivalence class of  $\Theta\langle E\cup F, \emptyset \rangle$ is $W$, 
  \item  $\#/_{\Theta\langle E, F \rangle} \subset \#/_{\Theta\langle {E\cup F}, \emptyset \rangle}$  
 and $ \#/_{\Theta\langle {F}, E\rangle}
  \subset  \#/_{\Theta\langle {E\cup F}, \emptyset\rangle}$, and 
  \item 
  $\tthue E\cup \tthue F\subset \tthue {E\cup F}$.
\end{itemize}
 \end{exa}

\begin{lem} \label{abcd}
Let $E$ and $F$ be GTESs over  a  signature   $\Sigma$ such that  
 GTRS $R\langle E, F\rangle$ and GTRS $R\langle F, E\rangle$ are  total. 
  Then \begin{quote} 
  for each $t\in W $, $\left( t/_{\tthue {E\cup F}}=  t/_{\tthue E} \mbox{ or }
   t/_{\tthue {E\cup F}}=  t/_{\tthue F}\right)$ 
  if and only if
  \newline 
for each $t\in W $, $ \left( t/_{\Theta\langle {E\cup F}, \emptyset \rangle}=
  t/_{\Theta\langle {E, F \rangle} } \mbox{  or  }
   t/_{\Theta\langle {E\cup F}, \emptyset \rangle}=  t/_{\Theta\langle {F}, E \rangle}\right)$.
  \end{quote}
\end{lem}
\begin{proof}
($\Rightarrow$)
If
$ t/_{\tthue {E\cup F}}=  t/_{\tthue E} $, then by Definition \ref{alfa}, 
$ t/_{\Theta\langle {E\cup F}, \emptyset \rangle}= t/_{\tthue {E\cup F}}\cap  W\times W
=  t/_{\tthue E} \cap  W\times W= t/_{\Theta\langle {E, F \rangle} } $.

 If
$ t/_{\tthue {E\cup F}}=  t/_{\tthue F} $, then 
$ t/_{\Theta\langle {E\cup F}, \emptyset \rangle}= t/_{\tthue {E\cup F}}\cap  W\times W
=  t/_{\tthue F} \cap  W\times W= t/_{\Theta\langle {E, F \rangle} } $.

($\Leftarrow$)
 Let $t\in W $ and $s \in  t/_{\tthue {E\cup F}}$.
 By Proposition \ref{otodik},
\begin{quote}
$s \tred {R\langle E\cup F, \emptyset \rangle}   t/_{\Theta\langle E\cup F, \emptyset \rangle}$.
\end{quote}
As $R\langle E, F\rangle$ and $R\langle F, E\rangle$ 
are total, by Lemma \ref{ezredik}, 
\begin{quote} 

$s{\downarrow}_ {R\langle E, F \rangle} = p/_{\Theta\langle E, F \rangle}$ for some $p \in W$,
\end{quote}
and 
\begin{quote}
  $s {\downarrow}_ {R\langle F, E \rangle}=  q/_{\Theta\langle F, E \rangle}$ for some
  $q \in W$.
\end{quote}
Consequently,
by Proposition \ref{harmadik},
\begin{quote}
  $s\tthue {E\cup F} t$, $s\tthue E p$ and $s\tthue F q$.
  \end{quote}
Therefore  $t\tthue {E\cup F}s\tthue E p$ and $t\tthue {E\cup F} s
\tthue E q$.
Therefore,
\begin{quote} 
$p \in  t/_{\Theta\langle {E\cup F}, \emptyset \rangle}$ and
$q \in  t/_{\Theta\langle {E\cup F}, \emptyset \rangle}$.
\end{quote}
By our assumption, 
\begin{quote}
$t/_{\Theta\langle {E\cup F}, \emptyset\rangle}=
  t/_{\Theta\langle {E, F \rangle} } $ or  $
   t/_{\Theta\langle {E\cup F}, \emptyset\rangle}=  t/_{\Theta\langle {F}, E \rangle})$.
   \end{quote}
First, suppose that 
 $ t/_{\Theta\langle {E\cup F}, \emptyset\rangle}=
 t/_{\Theta\langle {E, F \rangle} } $. Therefore,  we have 
$p \in  t/_{\Theta\langle {E, F \rangle}}$. 
Thus, by  $s\tthue E p$, we have
$s\tthue E p\tthue E t$.
Accordingly,  $s\in  t/_{\tthue E}$.  By the definition of $s$,  
\begin{quote}
$ t/_{\tthue {E\cup F}} \subseteq  t/_{\tthue E} $.
\end{quote}
 Since $E \subseteq E \cup F
$,  $ t/_{\tthue E} \subseteq  t/_{\tthue {E\cup F}}$.
Thus,  $ t/_{\tthue {E\cup F}}=  t/_{\tthue E} $.

Second, suppose  that  $ t/_{\Theta\langle {E\cup F}, \emptyset \rangle}=
  t/_{\Theta\langle {F, E \rangle} } $. Then by a similar argument as above we can show that
 $ t/_{\tthue {E\cup F}}=  t/_{\tthue F} $.
\end{proof}

 \begin{lem} \label{lennon}  Let $E$ and $F$ be GTESs over  a  signature   $\Sigma$ such that   GTRSs $R\langle E, F\rangle$ and  $R\langle F, E\rangle$ are total.
   Then we can decide in $O(n\,\mathrm{log}\, n)$  time whether
   $\tthue {E\cup F}\subseteq \tthue E \cup \tthue F$.
   \end{lem}
   \begin{proof}
   By Lemma \ref{teljes},
  $R\langle E\cup F, \emptyset\rangle$ 
is total as well. Consequently, by Lemma \ref{het},
\begin{quote}
  $\tthue {E\cup F}\subseteq \tthue E \cup \tthue F$
  if and only if  for each $p\in W  $, 
  $\left( p/_{\tthue {E\cup F}}=  p/_{\tthue E} \mbox{ or }
 p/_{\tthue {E\cup F}}=  p/_{\tthue F}\right)$.
\end{quote} Consequently, by Lemma \ref{abcd},
   \begin{quote}
     $\tthue {E\cup F}\subseteq \tthue E \cup \tthue F$
   if and only if
 for each $p\in W $, $ \left( p/_{\Theta\langle {E\cup F}, \emptyset \rangle}=
  p_{\Theta\langle {E, F \rangle} } \mbox{  or  }
   p/_{\Theta\langle {E\cup F}, \emptyset \rangle}=  p/_{\Theta\langle {F}, E \rangle}\right)$.
  \end{quote}
  Justified by the above observations and by Statement \ref{dan}, we compute as follows.
  Recall that  $NUM( E\cup F, E)$ stored in a red-black tree   $RBT3$, 
 and $NUM(  E\cup F, F)$ stored in a red-black tree    $RBT4$.
  Alternately in a coordinated way, we traverse the red-black trees $RBT3$ and $RBT4$ 
in in-order  traversal,  
For each    ${b}\in C\langle E\cup F, \emptyset\rangle$, we read the pair  $({b}, k)$, where 
$ k=|\{ {c} \in C\langle E, F\rangle \mid    {c}\subseteq {b}\, \}|$, and the pair 
 $({b}, \ell)$, where
$ \ell=|\{ {c} \in C\langle F, E\rangle \mid    {c}\subseteq {b}\, \}|$.
 In this way, 
  we decide in    $O(n)$ 
 time whether each $\Theta\langle E\cup F, \emptyset \rangle$  equivalence    class is equal to a 
 $\Theta\langle E, F \rangle$ equivalence    class or a $\Theta\langle F, E \rangle$ equivalence    class, see Definition \ref{szeder}. 
    If the answer is \enquote{yes}, then  $\tthue {E\cup F}\subseteq \tthue E \cup \tthue F$, otherwise, $\tthue {E\cup F}\not\subseteq \tthue E \cup \tthue F$.
\end{proof}

  \section{Main Case 3}\label{osz} 
 In this main case, $E$ and $F$ are  GTESs over  a  signature   $\Sigma$ such that  
 $\S$ has a symbol of arity at least $2$, and $ R \langle E, F \rangle$ or   $ R \langle F, E\rangle$ are total. 
Without  loss of generality, we may assume that  $ R \langle E, F \rangle$  is total from now on 
   throughout this section. 
  Example \ref{5uborka} and its sequels, Examples \ref{5uborkafolytatas},  \ref{5uborkamasodikfolytatas}, \ref{5uborkaharmadikfolytatas}, 
\ref{5uborkanegyedikfolytatas}, and \ref{5uborkaotodikfolytatas} plus
 Example \ref{6korte}, and its sequels,  Examples \ref{6kortefolytatas}, \ref{6kortemasodikfolytatas}, \ref{6korteharmadikfolytatas},  
 \ref{6kortenegyedikfolytatas}, and 
\ref{6korteotodikfolytatas} present this case.  We assume that we have already run  Algorithm CAD  on $\Sigma$, $E$, and $F$,  and have obtained all the output it produces, see Proposition \ref{kecskemet}. 
First we note that 
  we can also apply the decision algorithm of   Champav{\`{e}}re  et al \cite{franciak}.
    \begin{thm}\label{3maincase33} 
 Let $\S$ be a  signature having a symbol of arity at least $2$,  and let $E$ and $F$ be GTESs over   $\Sigma$ such that  
 GTRS $R\langle E, F\rangle$ or GTRS $R\langle F, E\rangle$ are  total. 
 We can decide in $O(n^3)$ time whether
  \begin{quote}
  $\tthue {E\cup F}=\tthue E \cup \tthue F$.
  \end{quote}
 \end{thm}
 \begin{proof}
 By Lemma \ref{teljes}, $ R\langle E\cup F, \emptyset \rangle$ is total as well. Then by Lemma \ref{het} and Proposition \ref{franciaeldontsimple}, and Proposition \ref{egyenlo}  we have the result.
 \end{proof}
  We present a faster and simpler decision algorithm.

 \subsection{Partial Depth First Search  on the Auxiliary  dpwpa  }

     We now define a partial depth first search on the auxiliary dpwpa   
$AUX[ E;  F] =
(C\langle E \cup F, \emptyset\rangle, A[E; F])$, which is
 a restricted and slightly modified verison of the depth first search  on graphs \cite{cormen}.
We  start the partial depth first search 
from the elements  of $C(E\cup F, \emptyset)\setminus C(E, F)$.
 \begin{df}\label{lo} \rm Let   $\Sigma$ be a signature such that 
   $\S_k\neq \emptyset$ for some $k\geq 2$, and let  $E$ and $F$ be GTESs over   $\Sigma$. \newline
   Algorithm Partial Depth First Search on Auxiliary dpwpa $AUX[ E;  F]  $  (PDFS for short){\rm \newline
 Input: the auxiliary dpwpa  $ AUX[ E;  F] $.
 Output: Boolean. \newline
\{ Output is true if GTRS $R\langle E, F \rangle$ and GTRS $R\langle F, E \rangle$ together  simulate the  GTRS $R\langle E \cup F, \emptyset \rangle$, false otherwise. \}\newline
{\bf var}      ${b}: C \langle E \cup F, \emptyset \rangle$;   
\begin{lstlisting}
(*$\mathrm{PDFS}( G \langle  E, F \rangle): \; \mathrm{Boolean}$*) 
  (*{\bf for each} vertex    ${b}\in C\langle E \cup F, \emptyset\rangle$ {\bf do} *) 
    (*{\bf begin}*)
      (*{\bf if} {\bf not}     ${b}.equal_E$ {\bf and}  {\bf not}     ${b}.equal_F$*)
        (*{\bf then} {\bf return}\; $\mathrm{false}$*)
      (*\{\, if    ${b}\not\in C\langle  E, F \rangle $ \mbox{ and }   ${b}\not \in C\langle  F, E \rangle $ \mbox{ then GTRS } $R\langle E, F \rangle$ \mbox{ and GTRS } $R\langle F, E \rangle$ \mbox{ together do not  simulate the  GTRS }  $R\langle E \cup F, \emptyset \rangle$. \}*)
      (*   ${b}.visited:=\mathrm{false}$*)
    (*{\bf end}*) 
  (*{\bf for each} vertex    ${b}\in C\langle E \cup F, \emptyset\rangle$ {\bf do}*) 
    (*{\bf if} {\bf not}    ${b}.equal_E$ and {\bf not}    ${b}.visited$*)   
    (*\{\ if       ${b}\in C\langle E \cup F, \emptyset\rangle\setminus C(E, F)$ and    ${b}$ has not been visited yet \}*)
      (*   ${\bf then }$    ${\bf for}$    ${\bf  each}$ vertex    ${c}\in Adj({b})$ \bf{do}*)
         (*   ${\bf if }$ {\bf not}    ${c}.visited$*)  
           (*   ${\bf then}$ if {\bf not} $\mathrm{PDFS-VISIT}( AUX[ E;  F] , {c})$ {\bf then}*) 
             (*   ${\bf return}\; \mathrm{false}$*)  
  (*$ {\bf output}(\mathrm{true})$*)        
                        

                    
(*   ${\bf function} \;\mathrm{PDFS-VISIT}(AUX[ E;  F] , {b}): \; \mathrm{Boolean}$*)    
  (*   ${b}.visited:=\mathrm{true}$;*)
  (*{\bf if}    ${b}.keeps_F=\mathrm{true}$*)   
    (*   ${\bf then }$ {\bf begin}*)
      (*{\bf for} {\bf  each} vertex    ${c}\in Adj({b})$ \bf{do}*)
        (*   ${\bf if }$ {\bf not}    ${c}.visited$*) 
          (*   ${\bf then}$    ${\bf if}$    ${\bf not}$ $\mathrm{PDFS-VISIT}(AUX[ E;  F] , {c})$*) 
            (*   ${\bf then}$    ${\bf return}$ $ \mathrm{false}$*)
    (*{\bf end}*)   
  (*   ${\bf else }\; {\bf return}\; \mathrm{false}$*)
  (*$ {\bf output}(\mathrm{true})$*)   
 \end{lstlisting}}
\end{df}
When  Algorithm PDFS starts the search from a vertex    ${b}\in C\langle E \cup F, \emptyset\rangle\setminus C(E, F)$, it does not set the attribute    ${b}.visited $ to true. 
 Consequently if for this  vertex    ${b}\in C\langle E \cup F, \emptyset\rangle\setminus C(E, F)$, there exists    a cycle leading from the vertex    ${b}$ to itself, 
   then  Algorithm PDFS visits vertex    ${b}$ twice. 
 
We run Algorithm PDFS on the auxiliary dpwpa  $ AUX[ E;  F] $. We shall show that PDFS   on $AUX[ E;  F] $ outputs true  if and only if 
GTRS $R\langle E, F \rangle$ and GTRS $R\langle F, E \rangle$ 
together  simulate the  GTRS $R\langle E \cup F, \emptyset \rangle$.

\begin{lem}\label{gezenguz}{\em \cite{cormen}}   Let $\S$ be a  signature   such that
   $\S_k\neq \emptyset$ for some $k\geq 2$,   and let 
 $E$ and $F$ be  GTESs over   $\Sigma$ such that $R\langle E, F\rangle$ is total.
  Let    ${b}\in  C\langle E\cup F, \emptyset \rangle  \setminus  C\langle E, F \rangle$.

(i)  During the run of Algorithm PDFS on $AUX[ E;  F] $, Function $\mathrm{PDFS-VISIT}$ on $( AUX[ E;  F] , {b})$ 
 visits at most once those vertices    ${c} \in C\langle E \cup F, \emptyset\rangle $ such that
 there is a path of positive length or a  cycle  from 
     ${b}$  to    ${c}$. 
  
  (ii) During the run of Algorithm PDFS on $AUX[ E;  F] $, Function $\mathrm{PDFS-VISIT}$ visits at most once  those vertices    ${c} \in C\langle E \cup F, \emptyset\rangle $ such that
 there is a path of positive length or a  cycle  from some vertex
     ${b}  \in C\langle E\cup F, \emptyset \rangle\setminus C\langle E, F \rangle $ to    ${c}$. 
\end{lem} 
\begin{proof} We obtain Statement (i) by direct inspection of Definition \ref{lo}. Statement (i) implies Statement (ii).
\end{proof}

\begin{prop}\label{fut}{\em \cite{cormen}}  Let $\S$ be a  signature   such that
   $\S_k\neq \emptyset$ for some $k\geq 2$,   and let 
 $E$ and $F$ be  GTESs over   $\Sigma$ such that $R\langle E, F\rangle$ is total.
 Algorithm PDFS  runs on $AUX[ E;  F] $  in $O(n)$ time.
 \end{prop}
\begin{proof} 
The running time of
PDFS is $O(|C\langle E \cup F, \emptyset\rangle | + |A[E; F]|)$ \cite{cormen}. Recall that  $n=\mathrm{size}(E)+\mathrm{size}(F)$. 
By Definition \ref{szeder} and Statement \ref{drrszam}, 
 $|C\langle E \cup F, \emptyset\rangle | \leq |W|\leq n$.   
  By Statement \ref{iszap}  
  $  | A[E; F]|\leq 2n$. Hence $|C\langle E \cup F, \emptyset\rangle | + |A[E; F]|\leq 3 n$. Then   
   we have the proposition.
 \end{proof}

  \begin{exa}\label{5uborkanegyedikfolytatas}\rm
  We now continue Examples  \ref{5uborka}, \ref{5uborkafolytatas},  \ref{5uborkamasodikfolytatas}, and \ref{5uborkaharmadikfolytatas}. First 
  Algorithm PDFS decides whether for each vertex    ${b}\in C\langle  E \cup F, \emptyset \rangle $, it holds that 
         $ {b}\in C\langle  E, F \rangle $ or  
       $ {b}\in C\langle  F, E \rangle $, that is, $\left({b}.equal_E=\mathrm{true} \mbox{  or } a.equal_F=\mathrm{true}\right)$. 
            Recall that $ \#/_{\Theta\langle E \cup F, \emptyset \rangle}.equal_E=\mathrm{false}$ and 
       $ \#/_{\Theta\langle E \cup F, \emptyset \rangle}.equal_F=\mathrm{false}$.  As a result  
       \begin{quote}
       $\left(\#/_{\Theta\langle E \cup F, \emptyset \rangle}.equal_E=\mathrm{false}\;\, \mbox{and} \;\, \#/_{\Theta\langle E \cup F, \emptyset \rangle}.equal_F=\mathrm{false}
       \right)=\mathrm{true}$.
 \end{quote}    
    Consequently, Algorithm PDFS outputs false. 
 \end{exa}

  \begin{exa}\label{6kortenegyedikfolytatas}\rm
   We now continue Examples \ref{6korte},  \ref{6kortefolytatas},  \ref{6kortemasodikfolytatas}, and \ref{6korteharmadikfolytatas}.
    Recall that there exist  one path of positive length and two cycles in $AUX[ E;  F] $.
        
 The first cycle:  Its length is $1$, and it  starts from the  vertex    
   $\#/_{\Theta\langle E\cup F  ,\emptyset \rangle}$ and leads to $\#/_{\Theta\langle E\cup F  ,\emptyset \rangle}$. By Lemma \ref{rozmaring}, 
  $ R \langle E\cup F, \emptyset \rangle$ reaches   
   $\#/_{\Theta\langle E\cup F  ,\emptyset \rangle}$ from   a proper extension    of the constant  $\#/_{\Theta\langle E\cup F  ,\emptyset \rangle}$.
 Recall that 
 \begin{quote}   
 $\#/_{\Theta\langle E \cup F, \emptyset \rangle}.equal_E=\mathrm{false}$,
    $\#/_{\Theta\langle E \cup F, \emptyset \rangle}.equal_F =\mathrm{true}$, and
      $\#/_{\Theta\langle E \cup F, \emptyset \rangle}.keeps_F =\mathrm{true}$.   
 \end{quote}

 The path of positive length:  Its length is $1$, and it  starts from the  vertex    
   $\$/_{\Theta\langle E\cup F  ,\emptyset \rangle}$ and leads to   $\#/_{\Theta\langle E\cup F  ,\emptyset \rangle}$.  By Lemma \ref{rozmaring}, 
     $ R \langle E\cup F, \emptyset \rangle$ reaches  $\$/_{\Theta\langle E\cup F  ,\emptyset \rangle}$ from  a proper extension  of the constant  
  $\#/_{\Theta\langle E\cup F  ,\emptyset \rangle}$.   Recall that 
 \begin{quote}    
$\$/_{\Theta\langle E \cup F, \emptyset \rangle}.equal_E=\mathrm{true}$,
    $\$/_{\Theta\langle E \cup F, \emptyset \rangle}.equal_F =\mathrm{false}$, and 
      $\#/_{\Theta\langle E \cup F, \emptyset \rangle}.keeps_F =\mathrm{true}$.   
 \end{quote}

 The second cycle:   Its length is $1$, and it  starts from the  vertex    
   $\$/_{\Theta\langle E\cup F  ,\emptyset \rangle}$ and leads to  $\$/_{\Theta\langle E\cup F  ,\emptyset \rangle}$.  By Lemma \ref{rozmaring}, 
  $ R \langle E\cup F, \emptyset \rangle$ reaches      $\$/_{\Theta\langle E\cup F  ,\emptyset \rangle}$
   from    a proper extension    of the constant  
      $\$/_{\Theta\langle E\cup F  ,\emptyset\rangle}$.  Recall that 
   \begin{quote} 
    $\$/_{\Theta\langle E \cup F, \emptyset \rangle}.equal_E=\mathrm{true}$, 
    $\$/_{\Theta\langle E \cup F, \emptyset \rangle}.equal_F =\mathrm{false}$, and 
      $\$/_{\Theta\langle E \cup F, \emptyset \rangle}.keeps_F =\mathrm{false}$.   
 \end{quote}
        First Algorithm PDFS decides whether for each    ${b}\in C\langle  E \cup F, \emptyset \rangle $, it holds that 
         $ {b}\in C\langle  E, F \rangle $ or  
       $ {b}\in C\langle  F, E \rangle $, that is, $\left({b}.equal_E=\mathrm{true} \mbox{  or } a.equal_F=\mathrm{true}\right)$.
            Observe that   
     $\#/_{\Theta\langle E \cup F, \emptyset \rangle}.equal_E=\mathrm{false}$ and 
    $\#/_{\Theta\langle E \cup F, \emptyset \rangle}.equal_F =\mathrm{true}$.  
     Thus
    \begin{quote}
    $\left(\#/_{\Theta\langle E \cup F, \emptyset \rangle}.equal_E \;  \mbox{or} \; \#/_{\Theta\langle E \cup F, \emptyset \rangle}.equal_F  \right)=\mathrm{true}$.
 \end{quote}

  Observe that   
     $\$/_{\Theta\langle E \cup F, \emptyset \rangle}.equal_E=\mathrm{true}$ and 
    $\$/_{\Theta\langle E \cup F, \emptyset \rangle}.equal_F =\mathrm{false}$.  
     As a result
    \begin{quote}
    $\left(\$/_{\Theta\langle E \cup F, \emptyset \rangle}.equal_E \;   \mbox{or} \; \$/_{\Theta\langle E \cup F, \emptyset \rangle}.equal_F \right)=\mathrm{true}$.
 \end{quote} 
            Algorithm PDFS examines the vertices of $ AUX[ E;  F] $, it looks for  vertices from which it can start  the search. 
           To this end,  Algorithm PDFS decides  whether for each    ${b}\in C\langle  E \cup F, \emptyset \rangle $, it holds that 
         $ {b}\in C\langle  E, F \rangle $ or   $ {b}\in C\langle  F, E \rangle $, that is,  $\left({b}.equal_E=\mathrm{true} \mbox{  or } a.equal_F=\mathrm{true}\right)$.
        First,  Algorithm PDFS examines the  vertex $\#/_{\Theta\langle E \cup F, \emptyset \rangle}$. It observes that
        $\#/_{\Theta\langle E \cup F, \emptyset \rangle}.equal_E=\mathrm{false}$
        and     $\#/_{\Theta\langle E \cup F, \emptyset \rangle}.equal_F=\mathrm{true}$. Then         
          Algorithm PDFS starts the search  from  $\#/_{\Theta\langle E \cup F, \emptyset \rangle}$.
It goes along the   cycle   $(\#/_{\Theta\langle E \cup F, \emptyset \rangle}, \#/_{\Theta\langle E \cup F, \emptyset \rangle})$. 
  Algorithm PDFS  observes that  $\#/_{\Theta\langle E \cup F, \emptyset \rangle}.keeps_F =\mathrm{true}$.
Second,  Algorithm PDFS examines the  vertex $\$/_{\Theta\langle E \cup F, \emptyset \rangle}$.
Since $\$/_{\Theta\langle E \cup F, \emptyset \rangle}.equal_E=\mathrm{true}$,  Algorithm PDFS does not start the search starting  from $\$/_{\Theta\langle E \cup F, \emptyset \rangle}$.   
       Consequently, Algorithm PDFS outputs true.
   \end{exa}

\subsection{Decision Algorithm and its verification}
      We present our decision algorithm and show its correctnes.

\begin{lem}\label{kert}
  Let $\S$ be a  signature   such that
   $\S_k\neq \emptyset$ for some $k\geq 2$.
    Let $E$ and $F$ be GTESs over  $\Sigma$ such that  
    GTRS     $R\langle E, F\rangle$   is total. 
 Then the following conditions are pairwise equivalent: 
  \begin{itemize}
  \item[(i)]
  $\tthue {E\cup F}\subseteq \tthue E \cup \tthue F$, 
    \item[(ii)]
 for each $t\in W  $, 
 $\left( 
     t/_{\tthue {E\cup F}}=  t/_{\tthue E}  \mbox{  or } 
t/_{\tthue {E\cup F}}=  t/_{\tthue F}\right) $.

\item[(iii)]
   For each $t\in W  $, 
    $\left( t/_{\Theta\langle E\cup F, \emptyset \rangle}= t/_{\Theta\langle E, F \rangle}\mbox{  or } t/_{\Theta\langle E\cup F, \emptyset \rangle}= t/_{\Theta\langle F, E \rangle} \right) $.

 \item[(iv)]  GTRS $R\langle E, F \rangle$ and GTRS $R\langle F, E \rangle$ 
together  simulate the  GTRS $R\langle E \cup F, \emptyset \rangle$.
 
  \item[(v)]  
    for each    ${b}\in C\langle  E \cup F, \emptyset \rangle $, 
   \begin{itemize}
   \item
     $ {b}\in C\langle  E, F \rangle $ or  
       $ {b}\in C\langle  F, E \rangle $  and 
      \item if 
        ${b}\not \in C\langle  E, F \rangle $, then  
  for each    ${c} \in C\langle  E \cup F, \emptyset \rangle $ if there is a  path of positive length or a cycle     from  
     ${b}$ to    ${c}$
 in the auxiliary dpwpa   
$ AUX[ E;  F] =
(C\langle E \cup F, \emptyset\rangle,A[E; F])$, then  $ R \langle F, E \rangle$  keeps up with  $ R \langle E \cup  F, \emptyset  \rangle$  writing    ${c} $.
\end{itemize}

 \item[(vi)]     for each    ${b}\in C\langle  E \cup F, \emptyset \rangle $, 
   \begin{itemize}
  \item
        ${b}.equal_E=\mathrm{true} $ or      ${b}.equal_F=\mathrm{true} $ 
     and 
      \item if 
        ${b}\not \in C\langle  E, F \rangle $, then    
    for each    ${c} \in C\langle  E \cup F, \emptyset \rangle $ if there is a path of positive length or a  cycle       from  
     ${b}$ to    ${c}$
 in the auxiliary dpwpa   
$AUX[ E;  F] =
(C\langle E \cup F, \emptyset\rangle, A[E; F])$, then  
   ${c}.keeps_F=\mathrm{true}$.
\end{itemize}

\end{itemize}
\end{lem}
\begin{proof}

$\mathrm{(i)} \Leftrightarrow \mathrm{(ii)}$: By Lemma \ref{teljes} and Lemma \ref{het}.

$\mathrm{(ii)} \Leftrightarrow \mathrm{(iii)}$: By Lemma \ref{duna}.

$\mathrm{(iii)} \Leftrightarrow \mathrm{(iv)}$: 
 By Lemma \ref{sportos}. 
 
$\mathrm{(iv)} \Leftrightarrow \mathrm{(v)}$:  By Lemma \ref{telex}. 


$\mathrm{(v)} \Leftrightarrow \mathrm{(vi)}$: By Statement 4 in Lemma \ref{eso}.
 \end{proof}
We get the following lemma by direct inspection of Definition \ref{lo} and by the equivalence  $\mathrm{(iv)} \Leftrightarrow \mathrm{(vi)}$ in Lemma \ref{kert}.
\begin{lem}\label{latogat}  Let $\S$ be a  signature   such that
   $\S_k\neq \emptyset$ for some $k\geq 2$, and let 
   $E$ and $F$ be GTESs over  $\Sigma$ such that   GTRS     $R\langle E, F\rangle$   is total.\newline
1. If GTRS $R\langle E, F \rangle$ and GTRS $R\langle F, E \rangle$ 
together  simulate the  GTRS $R\langle E \cup F, \emptyset \rangle$, 
then Conditions (a) and (b) hold:
\begin{itemize}
\item[\rm (a)]  For each vertex     ${b} \in C\langle E \cup F, \emptyset\rangle $, 
    ${b}.equal_E=\mathrm{true}$ or    ${b}.equal_F=\mathrm{true}$, and 
\item[\rm (b)]
 Algorithm PDFS   on $AUX[ E;  F] $ visits exactly once each vertex    ${b} \in C\langle E \cup F, \emptyset\rangle $ such that there is a path of positive length or a cycle  from some vertex
     ${c}  \in C\langle E\cup F, \emptyset \rangle\setminus C\langle E, F \rangle $ to    ${b}$. For each visited vertex    ${b}$, 
    ${b}.keeps_F=\mathrm{true}$.
 \item[{\rm (c)}] Algorithm PDFS   on $AUX[ E;  F] $ outputs true.
 \end{itemize}
 2. 
 If
 GTRS $R\langle E, F \rangle$ and GTRS $R\langle F, E \rangle$ 
together do not  simulate the  GTRS $R\langle E \cup F, \emptyset \rangle$, 
then  Algorithm PDFS   on $AUX[ E;  F] $ outputs false and Conditions (d) or (e) hold:
  \begin{itemize}
\item[{\rm (d)}] 
There exists a  vertex     ${b} \in C\langle E \cup F, \emptyset\rangle $ such that 
    ${b}.equal_E=\mathrm{true}$ or    ${b}.equal_F=\mathrm{true}$.
  \item[{\rm (e)}]
Algorithm PDFS   on $AUX[ E;  F] $ visits at most once each vertex    ${b} \in C\langle E \cup F, \emptyset\rangle $ such that there is a path of positive length or a cycle  from some vertex
     ${c}  \in C\langle E\cup F, \emptyset \rangle\setminus C\langle E, F \rangle $ to    ${b}$. Moreover, 
 Algorithm PDFS  visits a  vertex
     ${b} \in C\langle E \cup F, \emptyset\rangle $ such that 
   ${b}.keeps_F=\mathrm{false}$.
  \end{itemize}
  \end{lem}
We get the following result by   the equivalence  $\mathrm{(i)} \Leftrightarrow \mathrm{(iv)}$  in Lemma \ref{kert}   and by Lemma \ref{latogat}. 
\begin{lem}\label{kovetkeztet} Let $\S$ be a  signature   such that
   $\S_k\neq \emptyset$ for some $k\geq 2$, and let 
   $E$ and $F$ be GTESs over  $\Sigma$ such that   GTRS     $R\langle E, F\rangle$   is total. 
  $\tthue {E\cup F}\subseteq \tthue E \cup \tthue F$ if and only if Algorithm PDFS   on $ AUX[ E;  F] $ outputs true. 
 \end{lem}

\begin{lem}\label{tetel2}Let $\S$ be a  signature   such that
   $\S_k\neq \emptyset$ for some $k\geq 2$, and let 
   $E$ and $F$ be GTESs over  $\Sigma$. 
            Assume that 
  $R\langle E, F\rangle$   is  total.
 Then we can decide in  $O(n)$ time whether
   $\tthue {E\cup F}\subseteq \tthue E \cup \tthue F$.
   \end{lem}
\begin{proof}
We run  Algorithm PDFS  on $ AUX[ E;  F] $ in $O(n)$ time by Lemma \ref{fut}. By  Lemma  \ref{kovetkeztet}, 
Algorithm PDFS on the  auxiliary dpwpa  $ AUX[ E;  F] $  outputs the answer  to the question whether 
 $\tthue {E\cup F}\subseteq \tthue E \cup \tthue F$. 
 \end{proof}

 \begin{lem}\label{csimpanz} 
 Let $\S$ be a  signature   such that
   $\S_k\neq \emptyset$ for some $k\geq 2$, and let 
   $E$ and $F$ be GTESs over  $\Sigma$ such that   GTRS     $R\langle E, F\rangle$   is total.
 Then we can decide in  $O(n\, \mathrm{log}\, n)$ time whether
   $\tthue {E\cup F}\subseteq \tthue E \cup \tthue F$.
 \end{lem}   
 \begin{proof} Recall that we have run Algorithm CAD  on $\Sigma$, $E$, and $F$. 
 By Proposition \ref{kecskemet}, we got 
the auxiliary dpwpa $ AUX[ E;  F] =
(C\langle E\cup F, \emptyset \rangle, A[E; F])$ for the GTESs $E$ and $F$, and 
   for every vertex    ${b}\in C\langle E \cup  F, \emptyset  \rangle$,
the values of  the  attributes    ${b}.equal_E$,    ${b}.equal_F$,    ${b}.keeps_E$, and    ${b}.keeps_F$.
Then we run  Algorithm PDFS  on $AUX[ E;  F] $ in $O(n)$ time. 
Algorithm PDFS on input  $AUX[ E;  F] $  outputs the  answer  to the question whether 
 $\tthue {E\cup F}\subseteq \tthue E \cup \tthue F$. 
   Thence the  complexity of our entire computation is   $O(n\, \mathrm{log}\, n)$.
 \end{proof}
 
\begin{exa}\label{5uborkaotodikfolytatas} 
  \rm   We now continue Examples  \ref{5uborka}, \ref{5uborkamasodikfolytatas},  \ref{5uborkaharmadikfolytatas}, and \ref{5uborkanegyedikfolytatas}.
    By the last lines of Example  \ref{5uborkaharmadikfolytatas}, GTRS $R\langle E, F \rangle$ and GTRS $R\langle F, E \rangle$ 
together do not  simulate the  GTRS $R\langle E \cup F, \emptyset \rangle$.
 Accordingly Condition (iii) in Lemma \ref{kert} does not  hold.
   For this reason, by Lemma \ref{kert},
    $\tthue {E\cup F} \not \subseteq \tthue E\cup \tthue F$.  By the last line of Example, \ref{5uborkanegyedikfolytatas}
 Algorithm PDFS outputs false  to the question whether  $\tthue {E\cup F}\subseteq \tthue E \cup \tthue F$.  
   As a consequence, our decision algorithm also outputs false  to the question whether  $\tthue {E\cup F}\subseteq \tthue E \cup \tthue F$.  
          \end{exa}

\begin{exa}\label{6korteotodikfolytatas}
  \rm   We now continue Examples \ref{6korte},  \ref{6kortefolytatas},  \ref{6kortemasodikfolytatas},  \ref{6korteharmadikfolytatas}, and \ref{6kortenegyedikfolytatas}.
  By the last line of Example  \ref{6korteharmadikfolytatas},
  GTRS $R\langle E, F \rangle$ and GTRS $R\langle F, E \rangle$
together  simulate the  GTRS $R\langle E \cup F, \emptyset \rangle$.
  Therefore Condition (iii) in Lemma \ref{kert}  holds. Consequently, by Lemma \ref{kert},
    $\tthue {E\cup F} \subseteq \tthue E\cup \tthue F$. By the last line of Example, \ref{6kortenegyedikfolytatas}
 Algorithm PDFS outputs true to the question whether  $\tthue {E\cup F}\subseteq \tthue E \cup \tthue F$.  
    As a consequence, our decision algorithm also outputs true  to the question whether  $\tthue {E\cup F}\subseteq \tthue E \cup \tthue F$.  
    \end{exa}

\section{Main Case 4}\label{tel}
 In this main case, $E$ and $F$ are  GTESs over  a  signature   $\Sigma$ such that  
  $\S$ has a symbol of arity at least $2$,
  and GTRS  $ R \langle E, F \rangle$ and GTRS $ R \langle F, E \rangle$  are not total. 
 We assume that we have already run  Algorithm CAD  on $\Sigma$, $E$, and $F$,  and have obtained all the output it produces, see Proposition \ref{kecskemet}.
 We show that $\tthue {E\cup F}\subseteq \tthue E \cup \tthue F$
  if and only if
  $E \subseteq \Theta \langle F, E \rangle   $  or
  $ F \subseteq \Theta\langle E, F \rangle   $. 
By direct inspection of the relations  $ \tau\langle E, F \rangle$ and  $\rho\langle F, E\rangle$, we decide whether $E \subseteq \Theta \langle F, E \rangle   $.
Similarly, by direct inspection of the relations  $ \tau\langle F, E \rangle$ and  $\rho\langle E, F\rangle$, we decide whether $E \subseteq \Theta \langle F, E \rangle   $.
 Example
\ref{7twoketto} and its  sequel Example   \ref{7twokettofolytatas},  furthermore, Example 
 \ref{8threeharom} and its sequels Examples \ref{8threeharomfolytatas} and  \ref{8threeharommasodikfolytatas}
demonstrate  this case.
\begin{exa}\label{8threeharommasodikfolytatas}\rm 
We continue Examples \ref{8threeharom} and  \ref{8threeharomfolytatas}. We recall that
\begin{quote}
$f(\#, \mathsterling) \tthue {E\cup F} f(\$, \flat)$,\;
$(f(\#, \mathsterling), f(\$, \flat))\not\in  \tthue E  $, \, and \, 
$(f(\#, \mathsterling), f(\$, \flat))\not\in  \tthue F  $.
\end{quote}
The idea of the  $(\Rightarrow)$ direction of the  proof of the next theorem comes from the proof of the above observation.
\end{exa}
\begin{thm} \label{morandi}
  Let $\S$ be a  signature   such that 
   $\S_k\neq \emptyset$ for some $k\geq 2$. Let 
  $E$ and $F$ be GTESs over $\S$ such that
     GTRS  $ R \langle E, F \rangle$ and GTRS $ R \langle F, E\rangle$  are not total. 
  Then   \begin{quote}
    $\tthue {E\cup F}\subseteq \tthue E \cup \tthue F$
  if and only if
  $E \subseteq \tthue F$ or
    $ F \subseteq \tthue E $.
    \end{quote}
\end{thm}
\begin{proof}
$(\Rightarrow)$  We proceed by contrapositive. Assume that
$E \not\subseteq \tthue F$ and 
    $ F \not\subseteq \tthue E $.
That is, there exists
$l_1\doteq r_1\in E$ such that  $(l_1, r_1)\not \in \tthue F$ and 
there exists
$l_2\doteq r_2\in F$ such that  $(l_2, r_2)\not \in \tthue E$.
Then $l_1\neq r_1$ and $l_2\neq r_2$.  Let $\s\in \S_k$ with $k\geq 2$, and  by Lemma  \ref{halle}, let $v\in \ts \setminus  \bigcup  W /_{\tthue E}$,
and   $ z\in  \ts \setminus \bigcup W /_{\tthue F}$.
Then
\begin{quote}
$\s(\s(v, l_1, \ldots, l_1),  \s(z, l_2, \ldots, l_2), v, \ldots,  v)\tthue {E\cup F}
\s(\s(v, r_1, \ldots, r_1), \s(z, r_2, \ldots, r_2),  v, \ldots,  v)$.
\end{quote}
We now show that
\begin{quote}
$(\s(\s(v, l_1, \ldots, l_1),  \s(z, l_2, \ldots, l_2),  v, \ldots,  v), 
\s(\s(v, r_1, \ldots, r_1), \s(z, r_2, \ldots, r_2),  v, \ldots,  v))\not\in 
\tthue E\cup \tthue F$,
  \end{quote}
  and hence  we have 
  \begin{quote}
  $\tthue {E\cup F}\not \subseteq \tthue E \cup \tthue F$.
  \end{quote}
 To this end, first we show that 
  \begin{quote}
 $(\s(\s(v, l_1, \ldots, l_1),  \s(z, l_2, \ldots, l_2),  v, \ldots,  v), 
\s(\s(v, r_1, \ldots, r_1), \s(z, r_2, \ldots, r_2),  v, \ldots,  v))
 \not \in \tthue E $.
\end{quote}  
As $l_1\neq r_1$ and $l_2\neq r_2$,         
\begin{quote}
$\s(\s(v, l_1, \ldots, l_1),  \s(z, l_2, \ldots, l_2), v, \ldots,  v    )\neq 
\s(\s(v, r_1, \ldots, r_1), \s(z, r_2, \ldots, r_2), v, \ldots,  v)$.
  \end{quote}
We now show that 
  \begin{quote}
$(\s(\s(v, l_1, \ldots, l_1), \s(z, l_2, \ldots, l_2), v, \ldots,  v),
 \s(\s(v, r_1, \ldots, r_1), \s(z, r_2, \ldots, r_2), v, \ldots,  v))\not \in \plusthue E $.
\end{quote} 
To this purpose we now show the following statement.
\begin{cla}\label{muszaki}
Let \begin{equation}\label{huszka}
\s(\s(v, l_1, \ldots, l_1),  \s(z, l_2, \ldots, l_2), v, \ldots,  v)
  =s_0\thue {[ \alpha_1, p_1\doteq q_1]}
  s_1 \thue { [\alpha_2, p_2\doteq q_2 ]} \cdots \thue {[ \alpha_n, p_n \doteq q_n ]}
  s_n=t, \;n\geq 1.
\end{equation}
be a sequence of  $E$-rewrite steps.
Then for each $i\in \{\, 1, \ldots, n\,\}$, $\alpha_i$ is not a prefix of $11$.
\end{cla}
\begin{proof} 
By induction on $n$, we show that  for each $i\in \{\, 1, \ldots, n\,\}$, $\alpha_i$ is not a prefix of $11$.

{\em Base Case:} $n=1$.
By the definition of $v$, $v\not\in \bigcup  W /_{\tthue E}$, hence 
\begin{quote} 
$v$ is not a subterm of
$\s(\s(v, l_1, \ldots, l_1),  \s(z, l_2, \ldots, l_2), v, \ldots,  v) |_{\alpha_1}$,
\end{quote}
 thus 
$\alpha_1$ is not a prefix of $11$.

{\em Induction Step:} Let $n\geq 2$ and let us assume that we have shown the
statement for $n-1$. Then by  the induction hypothesis, for each $i\in \{\, 1, \ldots, n-1\,\}$, $\alpha_i$ is not a prefix of $11$. Accordingly
we have $v=s_0|_{11 } \tthue E s_n|_{11 }$. Consequently, by 
$v\not \in \bigcup  W /_{\tthue E}$, we have   $ s_n|_{11}\not \in  \bigcup W /_{\tthue E}$.
Thus 
$ s_n|_{11}$ is not a subterm of $p_n\in W$. As a consequence,
$\alpha_n$ is not a prefix of $11$. Then the  induction step is complete. 
\end{proof}

Consider the  sequence (\ref{huszka}) of  $E$ rewrite steps.
By Claim \ref{muszaki},
for each $i\in \{\, 1, \ldots, n\,\}$, $\alpha_i$ is not a prefix of 
$11$. Consequently $t=\s(\s( t_1, t_2, \ldots t_k), p_1, \ldots, p_{k-1} )$ for some $p_1, \ldots, p_{k-1}\in \ts$, and 
for each $j\in \{\, 2, \ldots, k\,\}$,
$l_1 \tthue E  t_j$.  Recall  that  $(l_1, r_1)\not \in \tthue E$. Therefore 
 for each $j\in \{\, 2, \ldots, k\, \}$, $t_j\neq r_1$. Hence 
  \begin{quote}
 $t\neq \s(\s(v, r_1, \ldots, r_1), \s(z, r_2, \ldots, r_2),  v, \ldots,  v)$.
 \end{quote}
 Then
 \begin{quote}
$(\s(\s(v, l_1, \ldots, l_1), \s(z, l_2, \ldots, l_2),  v, \ldots,  v),
 \s(\s(v, r_1, \ldots, r_1), \s(z, r_2, \ldots, r_2),  v, \ldots,  v))\not \in \tthue E$.
\end{quote}
 
 We can show similarly that
  \begin{quote}
$(\s(\s(v, l_1, \ldots, l_1), \s(z, l_2, \ldots, l_2),  v, \ldots,  v)),
 \s(\s(v, r_1, \ldots, r_1), \s(z, r_2, \ldots, r_2),  v, \ldots,  v))\not \in \tthue F $.
\end{quote}
 Consequently, 
  $\tthue {E\cup F}\not \subseteq \tthue E \cup \tthue F$.

$(\Leftarrow)$
 If   $E \subseteq \tthue F$,  then $\tthue {E\cup F}\subseteq \tthue F$, and hence $\tthue {E\cup F}\subseteq \tthue E \cup \tthue F$.
 If  $ F \subseteq \tthue E $, 
 then $\tthue {E\cup F}\subseteq \tthue E$, and hence  $\tthue {E\cup F}\subseteq \tthue E \cup \tthue F$.
 \end{proof}

\begin{cons} \label{cabello}
  Let $\S$ be a  signature   such that 
   $\S_k\neq \emptyset$ for some $k\geq 2$. Let 
  $E$ and $F$ be GTESs over $\S$ such that
  $R\langle E, F\rangle$ and 
 $R\langle F, E\rangle$ are not  total.
  Then   \begin{quote}
    $\tthue {E\cup F}\subseteq \tthue E \cup \tthue F$
  if and only if
  $E \subseteq \Theta \langle F, E \rangle   $  or
  $ F \subseteq \Theta\langle E, F \rangle   $. 
    \end{quote}
\end{cons}
\begin{proof} By  Definition \ref{alfa},   we have
$\Theta\langle F,E \rangle=\tthue F \cap W\times W$ and
$\Theta\langle E,F\rangle=\tthue E \cap W\times W$.
Hence Theorem  \ref{morandi} implies our consequence.
\end{proof}

\begin{lem} \label{orangutan} Let $\S$ be a  signature   such that
   $\S_k\neq \emptyset$ for some $k\geq 2$, and let 
   $E$ and $F$ be GTESs over  $\Sigma$ such that 
 $R\langle E, F\rangle$ and 
 $R\langle F, E\rangle$  are  not  total.
  We can decide in  $O(n \, \mathrm{log} \, n )$  time whether   $\tthue {E\cup F}\subseteq \tthue E \cup \tthue F$.
\end{lem}
\begin{proof} By Consequence \ref{cabello}, we  proceed as follows.
Recall that we have run Algorithm CAD  on $\Sigma$, $E$, and $F$. 
By Proposition \ref{kecskemet}, we got 
the relations  $\rho\langle E, F\rangle$, $\rho\langle F, E\rangle$,  and $\rho\langle E\cup F, \emptyset\rangle$ in   $O(n \, \mathrm{log} \, n )$ time.

First we decide whether $E \subseteq \Theta \langle F, E \rangle   $.
By  Proposition \ref{kozen}, 
\begin{quote}
 $E \subseteq \Theta \langle F, E \rangle   $ if and only if \newline for  each  $l \doteq  r $ in $ E$, $ (l, r)\in \Theta\langle F, E\rangle$
 if and only if \newline
for each  $(u, v) \in \tau\langle E, F \rangle$, 
$u/_{\rho\langle F, E\rangle}=v/_{\rho\langle F, E\rangle}$ if and only if \newline 
$\mathrm{FIND}_{\rho\langle F, E\rangle}(u)= \mathrm{FIND}_{\rho\langle F, E\rangle}(v)$. 
\end{quote} 
We decide whether
for each  $(u, v) \in \tau\langle E, F \rangle$, $\mathrm{FIND}_{\rho\langle F, E\rangle}(u)= \mathrm{FIND}_{\rho\langle F, E\rangle}(v)$.
If so,   $E \subseteq \Theta \langle F, E \rangle   $,  otherwise   $E \not \subseteq \Theta \langle F, E \rangle   $. 

Second, similarly to above, we decide whether $F \subseteq \Theta \langle E, F \rangle   $.

By Consequence \ref{cabello} \begin{quote}
    $\tthue {E\cup F}\subseteq \tthue E \cup \tthue F$
  if and only if
  $E \subseteq \Theta \langle F, E \rangle   $  or
  $ F \subseteq \Theta\langle E, F \rangle   $. 
    \end{quote}

We call above the functions   $\mathrm{FIND}_{\rho\langle F, E\rangle}$ and $\mathrm{FIND}_{\rho\langle E, F\rangle}$ a total of 
at most $2  n$ times,
 where  $n=\mathrm{size}(E)+\mathrm{size}(F)$.
  These callings,  by the main
result of \cite{journals/jacm/Tarjan75},
 require $O(n)$ time. 
 As a result, the time complexity of our entire computation is     $O(n\, \mathrm{log}\, n)$.
\end{proof}



\section{Main Result}\label{algoritmus}
 We  now sum up our decidability and complexity results.

 \begin{thm}\label{fofo} Let $E$ and $F$ be GTESs over  a   signature   $\Sigma$. 
 It is decidable  in  $O(n^2)$ 
 time  whether $\tthue {E\cup F}\subseteq\tthue E\cup \tthue F$, where  $n=\mathrm{size}(E)+\mathrm{size}(F)$.
   \end{thm}
  \begin{proof}
We run Algorithm CAD  on $\Sigma$, $E$, and $F$. 
By Proposition \ref{kecskemet}, we get the following output in   $O(n \, \mathrm{log} \, n )$ time: 
\begin{itemize}
\item the relations  $\rho\langle E, F\rangle$, $\rho\langle F, E\rangle$,  and $\rho\langle E\cup F, \emptyset\rangle$;
\item the  set of nullary symbols  $C\langle E, F\rangle$, $C\langle F, E\rangle$, and
$C\langle E \cup F, \emptyset\rangle$;
\item  $INC(  E\cup F, E) $ stored in a red-black tree saved in a  variable  $RBT1$;
\item  $INC(E\cup F, F) $ stored in a red-black tree  saved in a  variable    $RBT2$;
\item $NUM( E\cup F, E)$ stored in a red-black tree  saved in a  variable    $RBT3$;
\item $NUM(  E\cup F, F)$ stored in a red-black tree  saved in a  variable    $RBT4$;
\item the GTRS $R\langle E\cup F, \emptyset\rangle$  stored in a red-black tree saved in a  variable  $RBT5$;
 each node  in  $RBT5$ contains a search key  which is  a rule in $R\langle E\cup F, \emptyset\rangle$, and satellite data $counter_E$ and $counter_F$;
\item $BSTEP\langle E\cup F, \emptyset \rangle$ stored in a red-black tree saved in a  variable   $RBT6$; 
for every ordered pair $({b}, {c})\in BSTEP\langle E\cup F, \emptyset \rangle$, the search key is the first component    ${b}$, the second component,    ${c}$ is  satellite data;
 \item the GTRS $R\langle E, F\rangle$ stored in a red-black tree saved in a  variable  $RBT7$;
  \item the GTRS $R\langle  F, E\rangle$ stored in a red-black tree  saved in a  variable   $RBT8$; 
 \item 
the auxiliary dpwpa $ AUX[ E;  F] =
(C\langle E\cup F, \emptyset \rangle,A[E; F])$ for the GTESs $E$ and $F$, and 
   for every vertex    ${b}\in C\langle E \cup  F, \emptyset  \rangle$,
the values of  the  attributes    ${b}.equal_E$,    ${b}.equal_F$,    ${b}.keeps_E$, and    ${b}.keeps_F$.
\end{itemize} 
By Lemma \ref{decidetotal}, we decide in  $O(n)$   time whether    
 $R\langle E, F\rangle$ is total, and whether    
 $R\langle F, E\rangle$ is total. 
  \begin{itemize}
\item If both of them are total, then we  proceed as in main Case 2, see Lemma \ref{lennon}. 
\item Otherwise,
\begin{itemize}
\item  if $\S$ is a unary alphabet, then we proceed as in Main Case 1,  see  Lemma \ref{1tetharmadik1}.
\item  If  $\S_k\neq \emptyset$ for some $k\geq 2$, and  at least  one of $R\langle E, F\rangle$ 
and   $R\langle F, E\rangle$ is  total, then we proceed as in  Main Case 3,  see   Lemma \ref{csimpanz}. 
\item If $\S$ has a symbol of arity at least $2$,
  and GTRS  $ R \langle E, F \rangle$ and GTRS $ R \langle F, E \rangle$  are not total,  then we proceed as in  Main Case 4,
 see Lemma \ref{orangutan}. \qedhere
\end{itemize}
\end{itemize}
      \end{proof}   
   By Theorem \ref{fofo} and Proposition \ref{egyenlo}, we have the following consequence.
\begin{cons} \label{zab} Let 
  $E$ and $F$ be GTESs over a  signature   $\S$.
 Then we can decide in   $O(n^2)$     time whether
 there exists a GTES $H$  over $\S$ such that $\tthue H=\tthue E\cup \tthue F$, where  $n=\mathrm{size}(E)+\mathrm{size}(F)$.
  If the answer is \enquote{yes}, then $\tthue E\cup \tthue F$ is a congruence on the ground term algebra      ${\cal T}(\Sigma)$ and 
$\tthue {E\cup F}= \tthue E\cup \tthue F$. Otherwise
 $\tthue E\cup \tthue F$ is not a congruence on the ground term algebra      ${\cal T}(\Sigma)$.
  \end{cons}
\begin{proof}
By Theorem \ref{fofo}, we decide in  $O(n^2)$   
time whether
$\tthue {E\cup F}\subseteq \tthue E \cup \tthue F$.
By Proposition \ref{egyenlo},
if the answer is \enquote{yes}, then there exists a GTES $H$  over $\S$ such that $\tthue H=\tthue E\cup \tthue F$, and
$\tthue {E\cup F}=\tthue E \cup \tthue F$;
otherwise
 $\tthue E\cup \tthue F$ is not a congruence on the ground term algebra      ${\cal T}(\Sigma)$.
 \end{proof}

 \section{Conclusion}\label{tanulsag}
 We showed that for any 
  GTESs  $E$ and $F$ over a  signature   $\S$,
  the following four statements are equivalent:
    \begin{itemize}
    \item[1.] $\tthue {E\cup F}\subseteq \tthue E\cup \tthue F$.
    \item[2.] $\tthue {E\cup F}= \tthue E\cup \tthue F$.
    \item[3.]
    There exists a GTES $H$ over $\S$ such that $\tthue H=\tthue E\cup \tthue F$.
    \item[4.] $\tthue E\cup \tthue F$ is a congruence on the ground term algebra      ${\cal T}(\Sigma)$.
      \end{itemize}
  We showed that we can decide in 
 $O(n^2)$  
 time for given GTESs
  $E$ and $F$ over a  signature   $\S$ whether Statement 1 above holds. Here $n=\mathrm{size}(E)+\mathrm{size}(F)$.
 We can construct the  GTES $E \cup F$ in  $O(n^2)$ time.

 We used the fast ground completion algorithm introducing additional constants
of F\"ul\"op and V\'agv\"olgyi \cite{eatcs/FuloopV91},
which, given a GTESs $E$ and $F$ over a  signature   $\S$
 produces in $O(n \,\mathrm{log}\, n)$ time a finite set $C\langle E, F\rangle$ of  new  constants 
 and 
a reduced GTRS  $ R \langle E, F \rangle  $
over $\S\cup  C\langle E, F\rangle$ 
such  that $\tthue  E = \tthue {R\langle E, F \rangle }  \cap \ts \times \ts $. 
Along the proof, we distinguished four  main cases.  Main  Case 2 overlaps with Main Case 1 and Main Case 3, and is a subcase of the union of 
 Main Case 1 and Main Case 3. 
 
 Given GTES $E$ and GTES $F$, for every  $p \in W$,  we can construct the deterministic bottom-up tree automata
       ${\cal A}$,    ${\cal B}$, and    ${\cal C}$         such that 
        ${\cal A}$ recognizes  $ p/_{\tthue E}$,     ${\cal B}$ recognizes  $ p/_{\tthue F}$, and 
         ${\cal C}$ recognizes  $ p/_{\tthue E\cup F}$, see Proposition \ref{mezo}.
     Thence, in Main Cases 1, 2, and 3, we can reduce our decidability 
     problem to deciding the equivalence of two deterministic bottom-up tree automata.
  Champav{\`{e}}re  et al \cite{franciak}
 presented an  algorithm for deciding inclusion between two tree languages recognized by given bottom-up 
  tree automata.  Applying their algorithm  to our decision problem in Main Cases 1, 2, and 3  yields an alternative  decision  algorithm, which
    is more sophisticated than ours, and in general takes more time than ours, see Lemmas \ref{1maincase1eldont}, \ref{2maincase22}, and 
    \ref{3maincase33}.

 We raise the following more general problem:
 \begin{prob}
 Let $D$,   $E$,  and $F$ be GTESs over a  signature   $\S$. 
Can we decide whether
  $\tthue {D \cup E\cup F}=\tthue D \cup  \tthue E \cup \tthue F$?
If we can, then  what is the time complexity of our decision algorithm?
 \end{prob}
One might think that we can generalize our constructions and decidability results in a straightforward way. However, the following example shows that this belief is not justified. 
\begin{exa}\rm Consider the unary signature 
    $\S=\S_0=\{\,  \#,  \; \$,\; \mathsterling \, \}$. Let
    \begin{quote}
  $D=\{\,  \# \doteq \$\, \}$,\,
$E=\{\, \# \doteq \mathsterling\, \}$,\,
   $F=\{\, \$ \doteq \mathsterling\, \}$,\,
     $H= D \cup E \cup F $.
      \end{quote}
    Then \begin{itemize}
      \item $ST\langle H \rangle= \S_0$, 
  \item  $\tthue H= \tthue {D}\cup \tthue {E} \cup \tthue {F}$, 
  \item  $ \#/_{\tthue {D} } \subset   \#/_{\tthue H}$, \,  $ \#/_{\tthue {E} } \subset   \#/_{\tthue H}$, \,  $ \#/_{\tthue {F} } \subset   \#/_{\tthue H}$,
     \item  $ \#/_{\tthue {D} }, \; \#/_{\tthue {E} }, \;\#/_{\tthue {F} }$  are pairwise incomparable for the inclusion relation, and 
   \item for
 each $t\in ST\langle H \rangle $, \,
   $  t/_{\tthue D} \subset t/_{\tthue H}$, \,  $  t/_{\tthue E} \subset t/_{\tthue H}$, \, and  \,  $  t/_{\tthue F} \subset t/_{\tthue H}$.
     \end{itemize}
  Then the generalization of Condition (ii) in  Lemma \ref{pazsit} does not hold.
    
\end{exa}

\bibliographystyle{elsarticle-num-names} 
\bibliography{a0}

\begin{thebibliography}{21}
\expandafter\ifx\csname natexlab\endcsname\relax\def\natexlab#1{#1}\fi
\providecommand{\url}[1]{\texttt{#1}}
\providecommand{\href}[2]{#2}
\providecommand{\path}[1]{#1}
\providecommand{\DOIprefix}{doi:}
\providecommand{\ArXivprefix}{arXiv:}
\providecommand{\URLprefix}{URL: }
\providecommand{\Pubmedprefix}{pmid:}
\providecommand{\doi}[1]{\href{http://dx.doi.org/#1}{\path{#1}}}
\providecommand{\Pubmed}[1]{\href{pmid:#1}{\path{#1}}}
\providecommand{\bibinfo}[2]{#2}
\ifx\xfnm\relax \def\xfnm[#1]{\unskip,\space#1}\fi
\bibitem[{Kozen(1977)}]{Kozen77}
\bibinfo{author}{D.~Kozen},
\newblock \bibinfo{title}{Complexity of finitely presented algebras},
\newblock in: \bibinfo{editor}{J.~E. Hopcroft}, \bibinfo{editor}{E.~P.
  Friedman}, \bibinfo{editor}{M.~A. Harrison} (Eds.),
  \bibinfo{booktitle}{Proceedings of the 9th Annual {ACM} Symposium on Theory
  of Computing, May 4-6, 1977, Boulder, Colorado, {USA}},
  \bibinfo{publisher}{{ACM}}, \bibinfo{year}{1977}, pp.
  \bibinfo{pages}{164--177}. \URLprefix
  \url{https://doi.org/10.1145/800105.803406}.
  \DOIprefix\doi{10.1145/800105.803406}.
\bibitem[{Nelson and Oppen(1980)}]{jacm/NelsonO80}
\bibinfo{author}{G.~Nelson}, \bibinfo{author}{D.~C. Oppen},
\newblock \bibinfo{title}{Fast decision procedures based on congruence
  closure},
\newblock \bibinfo{journal}{J. {ACM}} \bibinfo{volume}{27}
  (\bibinfo{year}{1980}) \bibinfo{pages}{356--364}. \URLprefix
  \url{https://doi.org/10.1145/322186.322198}.
  \DOIprefix\doi{10.1145/322186.322198}.
\bibitem[{Tarjan(1975)}]{journals/jacm/Tarjan75}
\bibinfo{author}{R.~E. Tarjan},
\newblock \bibinfo{title}{Efficiency of a good but not linear set union
  algorithm.},
\newblock \bibinfo{journal}{J. ACM} \bibinfo{volume}{22} (\bibinfo{year}{1975})
  \bibinfo{pages}{215--225}. \URLprefix
  \url{http://dblp.uni-trier.de/db/journals/jacm/jacm22.html#Tarjan75}.
\bibitem[{Snyder(1993)}]{jsc/Snyder93}
\bibinfo{author}{W.~Snyder},
\newblock \bibinfo{title}{A fast algorithm for generating reduced ground
  rewriting systems from a set of ground equations},
\newblock \bibinfo{journal}{J. Symb. Comput.} \bibinfo{volume}{15}
  (\bibinfo{year}{1993}) \bibinfo{pages}{415--450}. \URLprefix
  \url{https://doi.org/10.1006/jsco.1993.1029}.
  \DOIprefix\doi{10.1006/jsco.1993.1029}.
\bibitem[{F{\"{u}}l{\"{o}}p and V{\'{a}}gv{\"{o}}lgyi(1991)}]{eatcs/FuloopV91}
\bibinfo{author}{Z.~F{\"{u}}l{\"{o}}p},
  \bibinfo{author}{S.~V{\'{a}}gv{\"{o}}lgyi},
\newblock \bibinfo{title}{Ground term rewriting rules for the word problem of
  ground term equations},
\newblock \bibinfo{journal}{Bull. {EATCS}} \bibinfo{volume}{45}
  (\bibinfo{year}{1991}) \bibinfo{pages}{186--201}.
\bibitem[{Brainerd(1969)}]{Brainerd69}
\bibinfo{author}{W.~S. Brainerd},
\newblock \bibinfo{title}{Tree generating regular systems},
\newblock \bibinfo{journal}{Inf. Control.} \bibinfo{volume}{14}
  (\bibinfo{year}{1969}) \bibinfo{pages}{217--231}. \URLprefix
  \url{https://doi.org/10.1016/S0019-9958(69)90065-5}.
  \DOIprefix\doi{10.1016/S0019-9958(69)90065-5}.
\bibitem[{V{\'{a}}gv{\"{o}}lgyi(1993)}]{tcs/Vagvolgyi93}
\bibinfo{author}{S.~V{\'{a}}gv{\"{o}}lgyi},
\newblock \bibinfo{title}{A fast algorithm for constructing a tree automaton
  recognizing a congruential tree language},
\newblock \bibinfo{journal}{Theor. Comput. Sci.} \bibinfo{volume}{115}
  (\bibinfo{year}{1993}) \bibinfo{pages}{391--399}. \URLprefix
  \url{https://doi.org/10.1016/0304-3975(93)90127-F}.
  \DOIprefix\doi{10.1016/0304-3975(93)90127-F}.
\bibitem[{V{\'{a}}gv{\"{o}}lgyi(2003)}]{tcs/Vagvolgyi03b}
\bibinfo{author}{S.~V{\'{a}}gv{\"{o}}lgyi},
\newblock \bibinfo{title}{Intersection of finitely generated congruences over
  term algebra},
\newblock \bibinfo{journal}{Theor. Comput. Sci.} \bibinfo{volume}{300}
  (\bibinfo{year}{2003}) \bibinfo{pages}{209--234}. \URLprefix
  \url{https://doi.org/10.1016/S0304-3975(02)00080-4}.
  \DOIprefix\doi{10.1016/S0304-3975(02)00080-4}.
\bibitem[{Downey et~al.(1980)Downey, Sethi, and Tarjan}]{jacm/DowneyST80}
\bibinfo{author}{P.~J. Downey}, \bibinfo{author}{R.~Sethi},
  \bibinfo{author}{R.~E. Tarjan},
\newblock \bibinfo{title}{Variations on the common subexpression problem},
\newblock \bibinfo{journal}{J. {ACM}} \bibinfo{volume}{27}
  (\bibinfo{year}{1980}) \bibinfo{pages}{758--771}. \URLprefix
  \url{https://doi.org/10.1145/322217.322228}.
  \DOIprefix\doi{10.1145/322217.322228}.
\bibitem[{Champav{\`{e}}re et~al.(2009)Champav{\`{e}}re, Gilleron, Lemay, and
  Niehren}]{franciak}
\bibinfo{author}{J.~Champav{\`{e}}re}, \bibinfo{author}{R.~Gilleron},
  \bibinfo{author}{A.~Lemay}, \bibinfo{author}{J.~Niehren},
\newblock \bibinfo{title}{Efficient inclusion checking for deterministic tree
  automata and {XML} schemas},
\newblock \bibinfo{journal}{Inf. Comput.} \bibinfo{volume}{207}
  (\bibinfo{year}{2009}) \bibinfo{pages}{1181--1208}. \URLprefix
  \url{https://doi.org/10.1016/j.ic.2009.03.003}.
  \DOIprefix\doi{10.1016/j.ic.2009.03.003}.
\bibitem[{Sanders et~al.(2019)Sanders, Mehlhorn, Dietzfelbinger, and
  Dementiev}]{Sanders}
\bibinfo{author}{P.~Sanders}, \bibinfo{author}{K.~Mehlhorn},
  \bibinfo{author}{M.~Dietzfelbinger}, \bibinfo{author}{R.~Dementiev},
  \bibinfo{title}{Sequential and Parallel Algorithms and Data Structures - The
  Basic Toolbox}, \bibinfo{publisher}{Springer}, \bibinfo{year}{2019}.
  \URLprefix \url{https://doi.org/10.1007/978-3-030-25209-0}.
  \DOIprefix\doi{10.1007/978-3-030-25209-0}.
\bibitem[{Skiena(2020)}]{skiena20}
\bibinfo{author}{S.~Skiena}, \bibinfo{title}{The Algorithm Design Manual, Third
  Edition}, Texts in Computer Science, \bibinfo{publisher}{Springer},
  \bibinfo{year}{2020}. \URLprefix
  \url{https://doi.org/10.1007/978-3-030-54256-6}.
  \DOIprefix\doi{10.1007/978-3-030-54256-6}.
\bibitem[{Cormen et~al.(2009)Cormen, Leiserson, Rivest, and Stein}]{cormen}
\bibinfo{author}{T.~H. Cormen}, \bibinfo{author}{C.~E. Leiserson},
  \bibinfo{author}{R.~L. Rivest}, \bibinfo{author}{C.~Stein},
  \bibinfo{title}{Introduction to Algorithms, 3rd Edition},
  \bibinfo{publisher}{{MIT} Press}, \bibinfo{year}{2009}. \URLprefix
  \url{http://mitpress.mit.edu/books/introduction-algorithms}.
\bibitem[{Johnsonbaugh(2016)}]{johnsonbaugh}
\bibinfo{author}{R.~Johnsonbaugh}, \bibinfo{title}{Discrete Mathematics},
  \bibinfo{edition}{8} ed., \bibinfo{publisher}{Pearson},
  \bibinfo{address}{Upper Saddle River, NJ}, \bibinfo{year}{2016}.
\bibitem[{Baader and Nipkow(1999)}]{baader1999term}
\bibinfo{author}{F.~Baader}, \bibinfo{author}{T.~Nipkow}, \bibinfo{title}{Term
  Rewriting and All That}, Term Rewriting and All that,
  \bibinfo{publisher}{Cambridge University Press}, \bibinfo{year}{1999}.
  \URLprefix \url{https://books.google.hu/books?id=N7BvXVUCQk8C}.
\bibitem[{Sedgewick and Wayne(2011)}]{books/daglib/0029345}
\bibinfo{author}{R.~Sedgewick}, \bibinfo{author}{K.~Wayne},
  \bibinfo{title}{Algorithms, 4th Edition.},
  \bibinfo{publisher}{Addison-Wesley}, \bibinfo{year}{2011}.
\bibitem[{Engelfriet(1975)}]{Engelfriet75}
\bibinfo{author}{J.~Engelfriet},
\newblock \bibinfo{title}{Bottom-up and top-down tree transformations - {A}
  comparison},
\newblock \bibinfo{journal}{Math. Syst. Theory} \bibinfo{volume}{9}
  (\bibinfo{year}{1975}) \bibinfo{pages}{198--231}. \URLprefix
  \url{https://doi.org/10.1007/BF01704020}. \DOIprefix\doi{10.1007/BF01704020}.
\bibitem[{Gécseg and Steinby(2015)}]{gs}
\bibinfo{author}{F.~Gécseg}, \bibinfo{author}{M.~Steinby},
  \bibinfo{title}{Tree automata}, \bibinfo{year}{2015}. \URLprefix
  \url{https://arxiv.org/abs/1509.06233}.
  \DOIprefix\doi{10.48550/ARXIV.1509.06233}.
\bibitem[{Bang-Jensen and Gutin(2010)}]{bang}
\bibinfo{author}{J.~Bang-Jensen}, \bibinfo{author}{G.~Gutin},
  \bibinfo{title}{Digraphs: Theory, algorithms and applications},
  \bibinfo{publisher}{Springer}, \bibinfo{year}{2010}.
\bibitem[{Sipser(2006)}]{sipser2006}
\bibinfo{author}{M.~Sipser}, \bibinfo{title}{Introduction to the Theory of
  Computation}, \bibinfo{edition}{second} ed., \bibinfo{publisher}{Course
  Technology}, \bibinfo{year}{2006}.
\bibitem[{Engelfriet(2015)}]{eng2}
\bibinfo{author}{J.~Engelfriet}, \bibinfo{title}{Tree automata and tree
  grammars}, \bibinfo{year}{2015}. \URLprefix
  \url{https://arxiv.org/abs/1510.02036}.
  \DOIprefix\doi{10.48550/ARXIV.1510.02036}.

\end{thebibliography}


\end{document}